\documentclass[12pt, reqno]{amsart}

\usepackage{amsmath, amsfonts, amssymb, amsthm}
\usepackage{amstext} 
\usepackage{mathrsfs}
\usepackage{mathtools}
\usepackage{xcolor}
\usepackage{ulem}
\usepackage{natbib}
\usepackage{graphicx}
\usepackage{subfig}
\usepackage{url}
\usepackage{comment}
\usepackage{algorithm}
\usepackage{algpseudocode}

\usepackage{fullpage}
\def\T{{ \mathrm{\scriptscriptstyle T} }}

\newcommand{\given}{\,|\,}

  \newcommand{\beq}{\begin{equation}}
    \newcommand{\eeq}{\end{equation}}
    
    \newcommand{\bal}{\begin{align}}
    \newcommand{\eal}{\end{align}}
    \newcommand{\bals}{\begin{align*}}
    \newcommand{\eals}{\end{align*}}
    
    \newcommand{\calD}{\mathcal{D}}

    \newcommand{\calS}{{\mathcal S}}
    \newcommand{\calM}{{\mathcal M}}

\DeclareMathOperator\cov{Cov}
\DeclareMathOperator\tr{Tr}

\DeclareMathOperator\var{Var}
\DeclareMathOperator\argmin{\arg \min}

\newtheorem{theorem}{Theorem}[section]
\newtheorem{lemma}[theorem]{Lemma}
\newtheorem{proposition}[theorem]{Proposition}

\newtheorem{assumption}[theorem]{Assumption}

\numberwithin{equation}{section}
\usepackage{xr}

\makeatletter
\DeclareRobustCommand*\cal{\@fontswitch\relax\mathcal}
\makeatother

\newtheorem{innercustomgeneric}{\customgenericname}
\providecommand{\customgenericname}{}
\newcommand{\newcustomtheorem}[2]{%
  \newenvironment{#1}[1]
  {%
   \renewcommand\customgenericname{#2}%
   \renewcommand\theinnercustomgeneric{##1}%
   \innercustomgeneric
  }
  {\endinnercustomgeneric}
}

\newcustomtheorem{customthm}{Theorem}
\newcustomtheorem{customlemma}{Lemma}
\newcustomtheorem{customcorollary}{Corollary}
\newcustomtheorem{customprop}{Proposition}

\graphicspath{{./pics/}}

\usepackage{siunitx}

\title{Bayesian Geostatistics Using Predictive Stacking}

\author{Lu Zhang}
\author{Wenpin Tang}
\author{Sudipto Banerjee}

\begin{document}

\begin{abstract}
We develop Bayesian predictive stacking for geostatistical models, where the primary inferential objective is to provide inference on the latent spatial random field and conduct spatial predictions at arbitrary locations. We exploit analytically tractable posterior distributions for regression coefficients of predictors and the realizations of the spatial process conditional upon process parameters. We subsequently combine such inference by stacking these models across the range of values of the hyper-parameters. We devise stacking of means and posterior densities in a manner that is computationally efficient without resorting to iterative algorithms such as Markov chain Monte Carlo (MCMC) and can exploit the benefits of parallel computations. We offer novel theoretical insights into the resulting inference within an infill asymptotic paradigm and through empirical results showing that stacked inference is comparable to full sampling-based Bayesian inference at a significantly lower computational cost.
\end{abstract}

\maketitle


\section{Introduction}
Geostatistics \citep{cres93, Chiles1999, zimmerman2010handbook, banerjee2019handbook} refers to the study of a spatially distributed variable of interest, which in theory is defined at every point over a bounded study region of interest. Customary geostatistical modeling proceeds from a latent stochastic process over space that specifies the probability law for the measurements on the variable as a partial realization of the process over a finite set of locations. Inference is sought for the underlying spatial process, which is subsequently used for spatial predictions \citep[][]{Steinbook} to grasp the scientific phenomenon under study. The spatial process is often assumed to be stationary and specified by parameters representing the sill, the nugget, the range and, possibly, the smoothness of the process. We collectively refer to these as process parameters that are often empirically estimated from measurements at sampled locations using the ``variogram''.

Likelihood-based inference for this process is, however, thwarted by the absence of classical consistent estimators of the process parameters in a customarily preferred infill asymptotic paradigm \citep[see, e.g.,][]{Steinbook, Zhang04, zhang2005towards, KS13, TZB2021}. Bayesian inference for geostatistical data \citep{handcock1993technometrics, berger2001, banerjee2014hierarchical, li2023jasa}, while not relying upon asymptotic inference, is also not entirely straightforward. Specifically, irrespective of how many spatial locations yield measurements, the likelihood fails to mitigate the prior distributions' impact on the inference. This is undesirable since prior elicitation for the process parameters is challenging. Objective priors for spatial process models have also been pursued, but interpreting such information in practice and their implications in scientific contexts are not uncontroversial. The related question of how effectively (or poorly) the realized data identify these process parameters (in an exact sense from finite samples) has also received commentary \citep[][]{hodges2013book, bose2018biocs, oliveira2022jabes}.

It is, therefore, not unreasonable to pursue methods that will yield robust inference for the spatial process and for spatial predictions of the outcome at arbitrary points (``kriging'') while circumventing inference on the weakly identified parameters. Instead of seeking families of prior distributions for such parameters, recent efforts at computationally efficient algorithms for geostatistical models have proposed multi-fold cross-validation methods \citep{finley2019efficient} to fix the values of weakly identified parameters. However, the metrics for ascertaining optimal values of such parameters are somewhat arbitrary and may not offer robust inference. 

Our primary contribution here is to develop and explore Bayesian predictive stacking of geostatistical models. Stacking is a model averaging procedure for generating predictions \citep{wolpert1992stacked, breiman1996stacked, clyde2013bayesian}. Stacking methods and algorithms in diverse data analytic applications are rapidly evolving and a comprehensive review is beyond the scope of this manuscript. Significant developments of stacking methodology in Bayesian analysis have been achieved in recent years \citep{le2017bayes, yao2018using, yao2020stacking, yao21}, but, to the best of our knowledge, developments in the context of spatial data analysis are lacking. Stacking can be regarded as an alternative to Bayesian model averaging \citep{madigan1996bayesian, hoeting1999bma}. Assume that there are $G$ 
candidate models ${\mathbb{M}} = \{\calM_1, \ldots, \calM_G\}$ and each model $\calM_g$ is indexed according to a set of fixed values of certain spatial covariance parameters so that the exact posterior distribution is analytically tractable (Section~\ref{sec: spatial_models}). 
We follow a specific formulation of a spatial hierarchical linear model that seamlessly evinces the familiar closed-form posterior distributions of regression coefficients, random effects and the spatial variance component. The nugget, too, is available in closed form while stacking over all other intractable parameters. 

The inferential properties of conventional stacking are largely available for exchangeable models that do not apply to geostatistics. Therefore, we offer theoretical insights within an infill asymptotic paradigm by exploiting tractability offered by the Mat\'ern covariance kernel and conjugate Bayesian linear regression (Section~\ref{sc3}). Section~\ref{sec: stacking_alg} offers implementation details, including two algorithms: (i) stacking of means, which combines posterior predictive means; and (ii) stacking of posterior predictive densities \citep{yao2018using}, which combines posterior predictive densities. 
These methods are evaluated theoretically and empirically through simulation experiments (Section~\ref{sec: simulation}) demonstrating that stacked inference is comparable to full Bayesian inference using MCMC at significantly less computational expense. An illustrative data analysis is presented in Section~\ref{sec: aod} including comparisons with machine learning interpolation. Section~\ref{sec: conclusion} concludes with pointers to future research.

\section{Bayesian spatial models and stacking algorithms}
\label{sec: spatial_models}
\subsection{Overview}\label{sec: spatial_models_stacking_overview}
Let $y(s)$ be a spatially indexed outcome at location $s \in \mathcal{D} \subset \mathbb{R}^d$ and $x(s)$ is a $p \times 1$ vector of predictors observed at $s$. A customary geostatistical model is
\begin{equation}\label{eq:spatialG}
y(s) = x(s)^\T \beta + z(s) + \varepsilon(s),
\end{equation} 
where $\beta$ is the $p \times 1$ vector of slopes, $z(s) \sim \mbox{GP}(0, \sigma^2 R_{\Phi}(\cdot, \cdot))$ is a zero-centered spatial Gaussian process on $\mathbb{R}^d$ with spatial correlation function $R_{\Phi} (\cdot, \cdot)$ indexed by parameters $\Phi$, $\sigma^2$ is the spatial variance parameter (``partial sill'') and $\varepsilon(s) \sim \mathcal{N}(0,\tau^2)$ is a white noise process with variance $\tau^2$ (``nugget'') capturing measurement error. 
Processes $z(\cdot)$ and $\epsilon(\cdot)$ are assumed to be independent.
Let $\chi = \{s_1, \ldots, s_n\} \in \mathcal{D}$ be a set of $n$ spatial locations yielding measurements $y = (y(s_1), \ldots, y(s_n))^\T$ with known values of predictors at these locations collected in the $n\times p$ full rank matrix $X = (x(s_1), \cdots, x(s_n))^\T$. 
We let $z = (z(s_1), \ldots, z(s_n))^\T$ denote the realization of $z(s)$ over $\chi$ and let $R_{\Phi}(\chi) = (R_{\Phi}(s_i, s_j))_{1 \le i,j \le n}$ be the $n\times n$ spatial correlation matrix constructed from the correlation function. 


Bayesian modeling extends \eqref{eq:spatialG} by assigning proper prior distributions to the parameters $\{\beta, z, \sigma^2, \Phi, \tau^2\}$ and drawing samples from the posterior distribution $p(\beta, \sigma^2, \Phi, \tau^2 \given y)$. However, evaluating the posterior distribution, particularly when employing iterative Markov chain Monte Carlo algorithms \citep[MCMC,][]{robertcasella}, is cumbersome due to slow convergence of weakly identified process parameters. Instead, we exploit exact distributions for conjugate Bayesian spatial models by fixing some process parameters and subsequently implement stacked posterior inference over candidate values of the process parameters. 

Our approach can be broadly described as follows. We separate our model parameters $\Theta = \{\theta_1, \theta_2\}$ into two sets $\theta_1 = \{\beta, z, \sigma^2\}$ and $\theta_2 = \{\Phi, \delta^2\}$, where $\delta^2 = \tau^2/\sigma^2$. We use conjugate distribution theory so that $p(\theta_1 \given y, \theta_2)$ is available in closed form. We are primarily concerned with Bayesian inference on $\theta_1$ averaging out the effects of $\theta_2$. Therefore, our desired posterior distribution is $p(\theta_1 \given y) = \int p(\theta_1\given y, \theta_2)p(\theta_2\given y) d\theta_2$. The key bottleneck here is that $p(\theta_2\given y)$ is intractable and will require either MCMC or iterative quadrature such as Integrated Nested Laplace Approximations \citep[INLA,][]{inla2009} or variational inference \citep[][]{renBanerjeeEtAl2011csda, blei2017vbreview}. However, such algorithms are thwarted by convergence issues arising especially from weakly identified parameters $\Phi$ in the spatial correlation kernel. Therefore, we reformulate the inference problem by writing $p(\theta_1 \given y) = \int p(\theta_1\given y, \theta_2)p(\theta_2\given y) d\theta_2 \approx \sum_{g=1}^G w_g p(\theta_1\given y, \theta_{2g})$,
where the collection of weights $w_g$ replace $p(\theta_2\given y)$. While this may seem to resemble quadrature, a key distinction is that we find the weights using convex optimization with scoring rules and do not attempt to approximate $p(\theta_2\given y)$. Once the optimal weights, say $\hat{w}_g$ are computed, posterior inference for quantities of interest subsequently proceed from the ``stacked posterior'' $\tilde{p}( \cdot \given y) = \sum_{g = 1}^G \hat{w}_g p(\cdot \given y, \theta_{2g})$. To circumvent iterative algorithms, we employ conjugate Bayesian spatial models with closed form expressions for $p(\theta_1 \given y, \theta_2)$ and then proceed to obtain the stacked posterior by averaging over candidate values of $\theta_2$. Section~\ref{sc21} derives the analytically accessible models that are used in Section~\ref{sc22} where we devise the stacking methodology and algorithms.

\subsection{Conjugate Bayesian spatial model}\label{sc21}
We extend \eqref{eq:spatialG} to a conjugate Bayesian hierarchical spatial model,
\begin{equation}\label{eq:BayesianG}
\begin{aligned}
 y \given &z, \beta, \sigma^2  \sim \mathcal{N}(X \beta + z, \delta^2 \sigma^2 I_n), \quad z \given \sigma^2 \sim \mathcal{N}(0, \sigma^2 R_{\Phi}(\chi)), \\
 \beta \given &\sigma^2 \sim \mathcal{N}(\mu_\beta, \sigma^2 V_{\beta}), \quad
\sigma^2 \sim \mbox{IG}(a_\sigma, b_{\sigma}),
\end{aligned}
\end{equation}
where we fix $\Phi$, the noise-to-spatial variance ratio $\delta^2: = \frac{\tau^2}{\sigma^2}$, and $\mu_\beta$, $V_\beta$, $a_\sigma$, and $b_{\sigma}$ are fixed hyper-parameters specifying the prior distributions for $\beta$ and $\sigma^2$. This design enables closed-form posterior distributions, as 
summarized in the following lemma \citep[also see][for exact Bayesian inference]{kitanidis1986parameter, handcock1993technometrics, gaudard1999bayesian, Banerjee20}.

\begin{lemma}\label{lem:structure}
Let $\gamma = (\beta^{\T}, z^{\T})^{\T}$. The posterior distribution of $(\gamma, \sigma^2)$ from (\ref{eq:BayesianG}) is 
\begin{equation}
\label{eq:posterior}
p(\gamma, \sigma^2 \,|\, y,\Phi, \delta^2) = \underbrace{\mbox{IG}(\sigma^2;\, a^*_\sigma, b^*_\sigma)}_{p(\sigma^2 \,|\, y)} \times 
\underbrace{\mathcal{N}(\gamma;\, \hat{\gamma}, \sigma^2 M_*)}_{p(\gamma \given \sigma^2, y)},
\end{equation}
where $\hat{\gamma} = M_{\ast}X_*^{\T}V_{*}^{-1}y_*$, $y_* = [y, \mu_\beta, 0]^\T$, $a^*_\sigma = a_\sigma + n/2$, 
$b^*_\sigma = b_\sigma + \frac{1}{2}(y_*-X_{*}\hat{\gamma})^{\T}V_*^{-1}(y_* - X_*\hat{\gamma})$, 
$M_*^{-1} = X_*^\T V_{*}^{-1} X_*$, 
$X_{*}^\T = \begin{bmatrix}
X^\T & I_p & 0 \\ 
I_n & 0 & I_n \end{bmatrix}$ and $V_{*} = \begin{bmatrix}
\delta^2 I_n & 0 & 0 \\ 
0 & V_\beta & 0 \\ 
0 & 0 & R_{\Phi}(\chi)
\end{bmatrix}$.  The posterior distribution $p(\gamma \given y, \Phi, \delta^2)$ obtained after integrating out $\sigma^2$ is multivariate Student's t 
{(i.e. $t_{2a_\sigma^\ast}(\gamma;\, \hat{\gamma}, (b_\sigma^\ast / a_\sigma^\ast)  M_\ast)$)} with degrees of freedom $2a_\sigma^\ast$, location  $\hat{\gamma}$ and scale matrix $(b_\sigma^\ast / a_\sigma^\ast)  M_\ast$.  
\end{lemma}
\begin{proof}
The proof is a straightforward adaptation of familiar results from the Normal-Gamma family of distributions \citep{kitanidis1986parameter,handcock1993technometrics}. 
The posterior distributions remain well-defined in the limit as $\delta^2 \rightarrow 0$, as elaborated in Appendix~\ref{sec: deltasq0limit}.
\end{proof}

Furthermore, let $\Tilde{\chi} = \{\Tilde{s}_{1},\ldots,\Tilde{s}_{m}\}$ be a set of $m$ unknown points in ${\cal D}$, $\Tilde{z}$ and $\Tilde{y}$ be the $m\times 1$ vectors with elements $z(\Tilde{s}_{i})$ and $y(\Tilde{s}_{i})$ for $i=1,2,\ldots,m$. Let $\tilde{X} =(x(\Tilde{s}_{1}), \ldots, x(\Tilde{s}_{m}))^\T$ be the $m \times p$ matrix that carries the values of predictors at $\Tilde{\chi}$ and let 
{$J_\Phi(\chi, \Tilde{\chi}) = (R_{\Phi}(s, s'))_{\{s \in \chi, s' \in \Tilde{\chi}\} }$} Then, spatial predictive inference follows from the posterior distribution
\begin{equation}\label{eq: posterior_predictive_joint}
p(\Tilde{z}, \Tilde{y}\given y, \Phi, \delta^2) = \int p(\Tilde{y} \given \Tilde{z}, \beta, \sigma^2, \Phi, \delta^2)p(\Tilde{z}\given z, \sigma^2, \Phi, \delta^2)p(\gamma, \sigma^2 \given y, \Phi, \delta^2)\mathop{d\gamma}\mathop{d\sigma^2}\;,
\end{equation}
which is again a multivariate $t$ distribution with degrees of freedom $2a_\sigma^\ast$, location  $\tilde{\mu}$ and scale matrix $(b_\sigma^\ast / a_\sigma^\ast)  \tilde{M}$ where $\tilde{\mu} = W\hat{\gamma}$, $\tilde{M}= W M_\ast W^\T + M_2 $, $ M_1 =  R_\Phi(\Tilde{\chi}) - J_\Phi^\T(\chi,\Tilde{\chi})R_\Phi^{-1}(\chi)J_\Phi(\chi,\Tilde{\chi})$, $\displaystyle W = \begin{bmatrix} 0& J_\Phi^\T(\chi,\Tilde{\chi})R_\Phi^{-1}(\chi) \\
\tilde{X} & J_\Phi^\T(\chi,\Tilde{\chi})R_\Phi^{-1}(\chi)
\end{bmatrix}$ and $\displaystyle M_2^{-1} = \begin{bmatrix}
\frac{1}{\delta^2}I_m + M_1^{-1} & -\frac{1}{\delta^2}I_m \\
-\frac{1}{\delta^2}I_m & \frac{1}{\delta^2}I_m
\end{bmatrix}$.
The predictive distributions $p(z(s_0)\given y, \Phi, \delta^2)$ and $p(y(s_0)\given y, \Phi, \delta^2)$ are also available in analytic form as a univariate t distributions for any single point $s_0 \in {\cal D}$. Bayesian inference proceeds from exact posterior samples obtained from (\ref{eq:posterior}). We first draw values of $\sigma^2 \sim IG(a_{\sigma}^{\ast}, b_{\sigma}^{\ast})$ followed by a single draw of $\gamma \sim N(\hat{\gamma},\sigma^2 M_{\ast})$ for each drawn value of $\sigma^2$. This yields samples $\{\gamma,\sigma^2\}$ from (\ref{eq:posterior}). Predictive inference for the latent process $z(s_0)$ and the outcome $y(s_0)$ is obtained by sampling from (\ref{eq: posterior_predictive_joint}) by drawing a value of $\Tilde{z} \sim N(\mu_z(\gamma), \sigma^2 M_1)$ with $\mu_z(\gamma):= J_\Phi^\T(\chi,\Tilde{\chi})R^{-1}_\Phi(\chi) z$
for each value of $\{\gamma, \sigma^2\}$ drawn above \citep[see Section 3.4 in][]{Banerjee20}, then drawing a value of $\Tilde{y} \sim N(\tilde{X}\beta + \Tilde{z}, \sigma^2\delta^2 I_m)$ for each drawn value of $\beta$ (extracted from $\gamma$), $\sigma^2$ and $\Tilde{z}$.  

This direct sampling is possible if $\Phi$ and $\delta^2$ are fixed. However, these parameters are not consistently estimable \citep{Zhang04, TZB2021}, and trying to estimate them from the data impedes the convergence of the MCMC algorithms. 
{\citet{diggle2007springer} proposed inference with discrete priors on these parameters \citep{ribeiro2007geor}, which still entails evaluating potentially numerically unstable conditional posterior densities. Alternate approaches that use $K$-fold cross-validation have been explored with limited success \citep{finley2019efficient}. Instead, we avoid numerically computing marginal posterior distributions and pursue optimization based on stacking over a set of fixed values of $\{\Phi, \delta^2\}$. 

\subsection{Stacking algorithms for Bayesian spatial models}\label{sc22}

{Let $\{\calM_g, g = 1, \ldots, G\}$ be the set of candidate models. The Bayes predictor for $y(s_0)$ under model $\calM_g$, for each $g = 1, \ldots, G$, is  $\mathbb{E}_g(y(s_0) \given y, \calM_g)$, where $\mathbb{E}_g(\cdot \given y, \calM_g)$ is the expectation with respect to $p(y(s_0)\given y, \calM_g) = t_{2a_\sigma^\ast}(y(s_0);\, h_g^{\T}\hat{\gamma}_g, \;(b_\sigma^\ast / a_\sigma^\ast) h_g^{\T}M_\ast h_g)$
with $h_g^{\T} = [x^{\T}(s_0), J_{\Phi_g}(s_0, \chi)R^{-1}_{\Phi_g}(\chi)]$ and $a_\sigma^\ast$, $b_\sigma^\ast$, $\hat{\gamma}_g$ and $M_\ast$ given by Lemma~\ref{lem:structure} with $\Phi = \Phi_g$ and $\delta^2 = \delta^2_g$.} Stacking will combine the $G$ Bayes predictors as a weighted average, 
\begin{equation}\label{eq: stack_predictor}
    \sum_{g = 1}^G w_g \mathbb{E}_g(y(s_0) \given y, \calM_g) = 
    {\sum_{g = 1}^G w_g h_g^{\T}\hat{\gamma}_g},
\end{equation}
where $\{w_1, \ldots, w_G\}$ are the stacking weights. 
{We refer to \eqref{eq: stack_predictor} as the stacked predictor. Subject to the constraint that stacking weights are non-negative and their sum equals one, we define the corresponding stacked predictive density as $\sum_{g = 1}^G w_g p(y(s_0) \given y, \calM_g)$.} 
We consider two stacking algorithms:
{\em stacking of means} and {\em stacking of predictive densities}. 

\medskip
\noindent
{\bf Stacking of means}:
This is the most natural stacking algorithm adapted from \citet{breiman1996stacked}. 
{Define the leave-one-out (LOO) Bayes predictor for $y(s_i)$ under model $\calM_g$ as $\hat{y}_g(s_i) = \mathbb{E}_g(y(s_i) \given y_{-i}, \calM_g) = h_{g, -i}^\T \hat{\gamma}_{g, -i}$, where $y_{-i}$ is the data without the $i$-th observation, $h^{\T}_{g,-i}$ is defined as in \eqref{eq: stack_predictor} but with $s_i$ and $\chi_{-i} = \chi\setminus\{s_i\}$ replacing $s_0$ and $\chi$, respectively, and $\hat{\gamma}_{g,-i}$ is obtained from Lemma~\ref{lem:structure} applied to the data without $s_i$. The expectation $\mathbb{E}_g(\cdot \given y_{-i}, \calM_g)$ is calculated with respect to $p(y(s_i)\given y_{-i}, \calM_g) = t_{2a_{\sigma, -i}^\ast}\left(y(s_i); h_{g, -i}^\T \hat{\gamma}_{g, -i}, \frac{b_{\sigma, -i}^\ast}{a_{\sigma, -i}^\ast} h_{g, -i}^{\T} M_{\ast, -i} h_{g, -i} + \delta^2\right)$, where $a_{\sigma, -i}^\ast$, $b_{\sigma, -i}^\ast$ and $M_{\ast, -i}$ are, again, provided in Lemma~\ref{lem:structure} for the data excluding $s_i$.} Stacking of means determines the optimal weights as 
\begin{equation}\label{eq: weights_cal1}
    \underset{w}{\argmin} \sum_{i = 1}^n \left(y(s_i) - \sum_{g = 1}^G w_g \hat{y}_g(s_i)\right)^2\;.
\end{equation}

\medskip
\noindent
{\bf Stacking of predictive densities}: Following the generalized Bayesian stacking framework in \citet{yao2018using}, we devise a second stacking algorithm for spatial analysis, which we refer to as stacking of predictive densities. This algorithm finds the distribution in the convex hull ${\cal C} = \{\sum_{g = 1}^G w_g \times p(\cdot \given \calM_g): \sum_{g} w_g = 1, w_g \geq 0\}$ that is optimal according to some proper scoring functions. Here $p(\cdot \given \calM_g)$ refers to the distribution of interest under model $\calM_g$. Let ${\calS}_1^G = \{w \in [0, 1]^G: \sum_{g = 1}^G w_g = 1\}$ and $p_t(\cdot \given y)$ be the true posterior predictive distribution. Using the logarithmic score (corresponding to the KL divergence), we seek $w$ so that 
\begin{equation}\label{eq:defpd}
    {\underset{w \in \calS_1^G }{\mbox{max}}
    \frac{1}{n}\sum_{i = 1}^n\log\left(\sum_{g = 1}^G w_g \underbrace{t_{2a_{\sigma, -i}^\ast}\left(y(s_i); h_{g, -i}^\T \hat{\gamma}_{g, -i}, \frac{b_{\sigma, -i}^\ast}{a_{\sigma, -i}^\ast} h_{g, -i}^{\T} M_{\ast, -i} h_{g, -i} + \delta^2 \right)}_{p(y(s_i) \given y_{-i}, \calM_g)}\right) \;.}
\end{equation}
The optimal distribution $\sum_{g = 1}^G w_g p(y(s_i) \given y_{-i}, \calM_g)$ provides a ``likelihood'' of observing $y(s_i)$ on location $s_i$ given other data. Therefore, $\prod_{i = 1}^n\sum_{g = 1}^G w_g p(y(s_i) \given y_{-i}, \calM_g)$ serves as a pseudo-likelihood that measures the performance of prediction based on the weighted average of the LOO predictors for all observed locations. 


\medskip
\noindent
{\bf Stacking using K-fold cross-validation predictors}:
Solving stacking weights relies upon computing the Bayes predictor and predictive density. Computing the exact LOO Bayes predictor and predictive densities for all observed locations $\{s_1, \ldots, s_n\}$ requires refitting a model $n$ times. For a Gaussian latent variable model with the number of parameters larger than the sample size $n$, there are limited choices for approximating LOO predictors accurately without the onerous computation \citep[see, e.g.,][]{vehtari2016bayesian}. Instead of using the LOO predictors, computing predictors through $K$-fold cross-validation is more practical. Using $K$-fold cross-validation instead of LOO in stacking is explored in \citet{breiman1996stacked}, who demonstrated that $10$-fold cross-validation provides more efficient predictors than LOO. 
{If the data is partitioned into $K$ folds and $y[-k]$ denotes the observed outcomes that are not included in the $k$-th fold, then, following \eqref{eq:defpd}, we have $p(y(s_i) \given y[-k], \calM_g) = t_{2a_{\sigma}^\ast[-k]}\left(y(s_i); h_{g}[-k]^\T \hat{\gamma}_{g}[-k], \frac{b_{\sigma}^\ast[-k]}{a_{\sigma}^\ast[-k]} h_{g}[-k]^{\T} M_{\ast}[-k] h_{g}[-k] + \delta^2 \right)$, where $s_i$ is in $k$-th folder. Here, $a_{\sigma}^\ast[-k]$, $b_{\sigma}^\ast[-k]$, $h_{g}[-k]$, and $M_{\ast}[-k]$ correspond to $a_{\sigma, -i}^\ast$, $b_{\sigma, -i}^\ast$, $h_{g, -i}$, and $M_{\ast, -i}$ in $p(y(s_i) \given y_{-i}, \calM_g)$, but are derived using data excluding the $k$-th folder instead of just the $i$-th observation. The K-fold cross-validation Bayes predictor for $y(s_i)$ under model $\calM_g$ is $\hat{y}_g(s_i) = \mathbb{E}_g(y(s_i) \given y[-k], \calM_g) = h_{g}[-k]^\T \hat{\gamma}_{g}[-k]$.}
{
For stacking of predictive densities, the optimal distribution changes into $\sum_{g = 1}^G w_g p(y(s_i) \given y[-k], \calM_g)$.}


\noindent {\bf Reconstructing stacked posterior distributions:}
Once the stacking weights are calculated using either stacking of means or stacking of predictive densities, we use them to reconstruct the posterior distributions of interest as
\begin{equation}\label{eq: posterior density_stacked}
    p(\cdot \given y) = \sum_{g=1}^G \hat{w}_g p(\cdot \given y, {\cal M}_g)\;,
\end{equation}
where $\cdot$ represents the inferential quantity of interest. We refer to \eqref{eq: posterior density_stacked} as the \emph{stacked posterior density}. For parameter inference we take $\cdot$ as $\{\beta, z, \sigma^2\}$, while for predictive inference we use $y(s_0)$ at an arbitrary location.

\section{Theoretical results}\label{sc3}
We focus on posterior inference for a Mat\'ern model without trend to justify the stacking algorithms for these models. Subsequently, we extend these investigations to \eqref{eq:BayesianG}.
{
It should be noted that formal theory on spatial asymptotics is extremely challenging and usually adheres to either an expanding domain or an infill paradigm \citep[with attempts at reconciliation][]{zhang2005towards}. Several of the theoretical results, in their current form, depend on assumptions or conjectures that may not strictly adhere to a single paradigm and are difficult or impossible to verify in practice. We intend these results to provide some insight into the asymptotic behavior of stacking weights and to generate more formal theoretical research in this domain.}

\subsection{Posterior inference for the Mat\'ern model} \label{sc31}
The Mat\'ern model without trend is a special case of \eqref{eq:spatialG} with $\beta = 0$. Hence,
\begin{equation}\label{eq:Matern}
y(s) = z(s) + \varepsilon(s),
\end{equation}
where $z(s)$ is modeled with the isotropic Mat\'ern correlation function,
\begin{equation}\label{eq: Matern}
R_\Phi(s, s'):= \frac{(\phi |s - s'|)^{\nu}}{\Gamma(\nu) 2^{\nu - 1}} K_{\nu}(\phi |s - s'|), \quad \color{black}{\Phi = \{\phi, \nu\} \;.}
\end{equation}
Here $\phi$ is the decay parameter, $\nu > 0$ is a fixed smoothness parameter, $\Gamma(\cdot)$ is the Gamma function, and $K_{\nu}(\cdot)$ is the modified Bessel function of the second kind of order $\nu$ \cite[Section 10]{AS65}. 
We refer to \eqref{eq:Matern} as the Mat\'ern model with parameters $\{\sigma^2, \phi, \tau^2\}$. The conjugate Bayesian model \eqref{eq:BayesianG} simplifies to
\begin{equation}
\label{eq:BayesianG2}
y \,|\, \sigma^2, \phi, \delta^2 \sim \mathcal{N}(0, \sigma^2 (R_\phi(\chi) + \delta^2 I_n)), \quad \sigma^2 \sim \mbox{IG}(a_\sigma, b_{\sigma}).
\end{equation}

We consider posterior inference for the conjugate Bayesian model \eqref{eq:BayesianG2} using  \eqref{eq: Matern}. Let $\mathbb{P}_0$ be the probability distribution of the Mat\'ern model \eqref{eq:Matern} with $(\sigma_0^2, \phi_0, \tau_0^2)$ that generates the data $y$.  The following theorem shows the posterior inconsistency of $\sigma^2$  under this model. 

\begin{theorem}[Posterior inference for the Mat\'ern model]
\label{thm:main2}
Assume that the location set $\chi = \{s_1, \ldots, s_n\}$ satisfies {
the infill condition:}
\begin{equation}
\tag{3.4}
\max_{s \in \mathcal{D}} \min_{1 \le i \le n} |s - s_i| \asymp n^{-\frac{1}{d}}.
\end{equation}
Under the true data generating distribution $\mathbb{P}_0$ as in \eqref{eq:Matern} with $(\sigma_0^2, \phi_0, \tau_0^2)$, 
\begin{equation}
\label{eq:postin}
\lim_{n\to\infty} p(\sigma^2 \given y, \phi, \delta^2) = \texttt{Dirac} (\tau_0^2/\delta^2)\quad  \mbox{and}\quad \lim_{n\to\infty} p(\tau^2 \given y, \phi, \delta^2) = \texttt{Dirac}(\tau_0^2)
\end{equation}
where 
$\texttt{Dirac}(\cdot)$ denotes the Dirac mass point.
\end{theorem}

\begin{proof}[Proof of Theorem \ref{thm:main2}]
See Appendix~\ref{sec: thm:main2}.
\end{proof}

Notably, the {asymptotic} posterior inference of $\sigma^2$ is independent of the range decay $\phi$ chosen in the 
Mat\'ern model. 
The scale $\sigma^2$ is posterior inconsistent 
unless the noise-to-spatial variance ratio $\delta^2 = {\tau_0^2/\sigma_0^2}$,
whereas the nugget $\tau^2$ is posterior consistent. 

\subsection{Posterior prediction for the Mat\'ern model}
\label{sc32}

We consider Bayesian posterior predictive inference at a new location $s_0 \in \mathcal{D}$,
under the Mat\'ern model \eqref{eq: Matern}--\eqref{eq:BayesianG2}. 
We study the posterior predictive consistency of the conjugate model with the misspecified prefixed parameters. 
Let $Z_n(s_0)$ be a random variable distributed as $p(z(s_0)\,|\, y, \phi, \delta^2)$ and $Y_n(s_0)$ be distributed as $p(y(s_0)\,|\, y, \phi, \delta^2)$, and let $\mathbb{E}_0(Z_n(s_0) - z(s_0))^2$ and $\mathbb{E}_0(Y_n(s_0) - y(s_0))^2$ denote expected prediction errors for the latent process and outcome variable, respectively, where $\mathbb{E}_0(\cdot)$ denotes expectation with respect to the Mat\'ern model $\mathbb{P}_0$ that integrates over the generating process for $y$. 

\begin{theorem}[Posterior predictive consistency for the Mat\'ern model]
\label{prop:main2bis}
Let $s_0 \in \mathcal{D}$.
For any given $\phi>0$, denote $\cov(z, z(s_0) \given \sigma^2)$ and $R_{\phi}(\chi)$ by $\sigma^2 J_{\phi, n}$ and $R_{\phi, n}$, respectively. Then, 
\begin{equation}
\label{eq:decomperr22}
\mathbb{E}_0 (Z_n(s_0) - z(s_0))^2 = E_{1,n} + E_{2,n} + o(1),
\end{equation}
where $E_{1, n}$ is the prediction error of the best linear predictor for a Mat\'ern model with any parameters $\{\sigma'^2, \phi, \tau'^2\}$ satisfying $\delta^2 = \frac{\tau'^2}{\sigma'^2}$,
and 
\begin{equation}
\label{eq:E2}
E_{2, n}: 
=\frac{\tau_0^2}{\delta^2}\left[1 - J_{\phi, n}^\T (\delta^{2}I_n + R_{\phi, n})^{-1} J_{\phi, n} \right]
\end{equation}
\end{theorem}
\begin{proof}
See Appendix~\ref{sec: prop:main2bis}.
\end{proof}

Theorem~\ref{prop:main2bis} is similar in spirit to \citet{ZC92, Abt99},
exploiting the rich structure of the Mat\'ern model.
In the decomposition \eqref{eq:decomperr22},
the term $E_{1,n}$ arises in the deviation from the posterior mean, 
whereas the term $E_{2,n}$ is from the posterior uncertainty. 
Moreover,
the posterior mean of the Mat\'ern model is identified as 
the best linear predictor of {\em any} Mat\'ern model 
with parameters $\{\sigma'^2, \phi, \tau'^2\}$ provided $\tau'^2/\sigma'^2 = \delta^2$.
This observation connects the Bayesian modeling to a frequentist approach
in that the deviation error $E_{1, n}$ is viewed as the prediction error of the best linear predictor of a Mat\'ern model in the presence of a nugget. 

If $E_{1, n}, E_{2, n} \to 0$ as $n \to \infty$, then the latent process $z(s)$ is posterior predictive consistent in the sense that $\mathbb{E}_0(Z_n(s_0) - z(s_0))^2 \to 0$ as $n \to \infty$ and hence, $\mathbb{E}_0(Y_n(s_0) - y(s_0))^2 \to 2 \tau_0^2$ as $n \to \infty$.
However, the conditions $E_{1,n}, E_{2,n} \to 0$ are analytically intractable. In Appendix~\ref{sec: prop:main2bis}, we provide theoretical and numerical evidence to support these conditions which justifies posterior predictive consistency for the latent process. Specifically, we provide theoretical support for $E_{1, n} \to 0$ as $n \to \infty$, and illustrate empirically that $E_{2, n}$ decreases rapidly as the sample size grows within finite domains for spatial dimensions $d = 1$ and $2$.

\subsection{Stacking algorithms}
\label{sc34}

We attend to predictive inference using stacking. For ease of presentation, we focus on the LOO cross validation. Extending to $K$-fold cross validation is straightforward, albeit tedious. We offer the following result concerning the stacked mean square posterior prediction error.

\begin{proposition}[Posterior prediction error for stacking]
\label{prop:stackingpost1}
Let $s_0 \in \mathcal{D}$, and $w_g^*(y): = (w^*_1(y), \ldots, w^*_G(y))$ be the stacking weights 
(e.g. defined by \eqref{eq: weights_cal1})
such that $\mathbb{P}_0$ almost surely,
$$\sum_{g=1}^G w^*_g(y) = 1 \quad \mbox{and} \quad w^*_g(y) \ge 0  \mbox{ for each } 1 \le g \le G.$$
Recall that $E^g_{1,n}$ is the prediction error of the best linear predictor for model $\mathcal{M}_g$, and assume that $E^g_{1,n} \to 0$ as $n \to \infty$, for each model $\mathcal{M}_g$.
We have
\begin{equation}
\label{eq:stackingerr}
\mathbb{E}_0\left(y(s_0) - \sum_{g = 1}^G w^*_g(y) \mathbb{E}_g(y(s_0) \,|\, y) \right)^2 \to \tau_0^2 \quad \mbox{as } n \to \infty \;. 
\end{equation}
\end{proposition}

\begin{proof}
See Appendix~\ref{sec: prop:stackingpost1}.
\end{proof}
The proof of Proposition \ref{prop:stackingpost1} uses  the posterior predictive consistency of $z(\cdot)$ (Theorem \ref{prop:main2bis}). It implies that if each candidate model yields reasonable prediction,  then the stacking of means will also produce good predictive inference.
Next, we show that the stacking predictor asymptotically minimizes the mean square posterior prediction error.
\begin{theorem}
\label{prop:stackingpost2}
Let $s_0 \in \mathcal{D}$, and $w_g^*(y): = (w^*_1(y), \ldots, w^*_G(y))$ be the stacking weights 
(e.g. defined by \eqref{eq: weights_cal1})
such that $\mathbb{P}_0$ almost surely,
$$\sum_{g=1}^G w^*_g(y) = 1 \quad \mbox{and} \quad w^*_g(y) \ge 0  \mbox{ for each } 1 \le g \le G.$$
{
Assume that $E^g_{1,n} \to 0$ as $n \to \infty$, for each model $\mathcal{M}_g$.}
For $1 \le g \le G$ and $1 \le i \le n$, 
let 
$E^g_{1,n,i}: = \mathbb{E}_0(z(s_i) - \widehat{y}_g(s_i))^2$
be the deviation error for the latent process $z(s)$ by leaving the $i^{th}$ observation out under the model $\mathcal{M}_g$.
Assume that for each $1 \le g \le G$,
\begin{equation}
\label{eq:avezero}
\frac{1}{n}\sum_{i = 1}^n E^g_{1,n,i} \to 0 \quad \mbox{as } n \to \infty.
\end{equation}
Also let the assumptions in 
Theorem \ref{prop:main2bis} hold for each model $\calM_g$.
Then, as $n \to \infty$,
\begin{equation}
\label{eq:asymplimit}
\mathbb{E}_0\left(y(s_0) - \sum_{g = 1}^G w^*_g(y) \mathbb{E}_g(y(s_0) \,|\, y) \right)^2
- \mathbb{E}_0\left( \frac{1}{n}\sum_{i = 1}^n \left(y(s_i) - \sum_{g = 1}^G w^*_g(y) \hat{y}_g(s_i)\right)^2\right) \to 0.
\end{equation}
\end{theorem}
\begin{proof}
See Appendix~\ref{sec: prop:stackingpost2}.
\end{proof}

Equation~\eqref{eq:avezero} implies that for each $\calM_g$ the average deviation error for the latent process goes to $0$ as the sampling resolution becomes finer. This condition is consistent with the fact that in the one-dimensional grid we typically have $E_{1, g, i}  \asymp \Delta^{\min( 2 \nu_0, \frac{2 \nu}{2 \nu + 1})}$ (see Proposition~\ref{prop:errbd}); hence $\frac{1}{n}\sum_{i = 1}^n E_{1,g,i} \to 0$ as $n \to \infty$. Theorem~\ref{prop:stackingpost2} holds for the candidate models with a misspecified smoothness parameter $\nu$ in the Mat\'ern kernel for $d = 1$. \citet{clyde2013bayesian} proves a similar result, while the established theoretical results about stacking assume exchangeability which is generally not available in geostatistical models. 
{Lack of exchangeablility limits formal theoretical developments. Section~\ref{sec: thm:KLpd} offers limited discussion on stacking of predictive densities without exchangeable assumptions.} Therefore, our theoretical results emerge from studying the behavior of the posterior and predictive distributions in the conjugate Bayesian linear model framework within the infill-paradigm. 
{In this regard, our investigations differ from \cite{de2013semiparametric} who assume a true underlying function of arbitrary smoothness.} The proof of Theorem~\ref{prop:stackingpost2} can be extended, fairly straightforwardly to the case where the LOO Bayes prediction $\hat{y}_g(s_i)$ is replaced by a much cheaper prediction based on $K$-fold cross-validation.

\paragraph{Stacked predictive densities in the general setting.}
We extend our theoretical analysis to the general conjugate Bayesian spatial model \eqref{eq:BayesianG} and establish asymptotic results of the posterior distribution obtained through stacking, detailed in Appendix~\ref{sc33}. In line with Theorem~\ref{thm:main2}, we show that the posterior distribution of the scale parameter $\sigma^2$ in the general conjugate model does not necessarily concentrate on the true generating value. Our analysis reveals that the posterior distribution for $\sigma^2$ approaches a value influenced by $\delta^2$ as sample size increases. This is especially pronounced when the stacking weights for the candidate models exhibit variability in their respective $\delta^2$ values, leading to an anticipation of multi-modality in the stacked posterior density of $\sigma^2$. Such multi-modality suggests that employing a stacking algorithm may not yield dependable posterior inferences for $\sigma^2$, underscoring limitations in stacking over fixed values of $\sigma^2$. 
Regarding posterior prediction, we offer further theoretical discussion on the asymptotic behavior of posterior predictions for the general conjugate Bayesian spatial models in Appendix~\ref{sec: thm:main1bis}. Section~\ref{sec: simulation} presents simulations showing that stacking algorithms serve as efficient alternatives to more expensive MCMC algorithms for posterior prediction.

\section{Implementation of stacking algorithms}
\label{sec: stacking_alg}

We outline algorithms that compute the weights for spatial stacking. We first partition the data into $K$-folds based on locations. 
Let $X = [x(s_1): \cdots: x(s_n)]^\T$ be the design matrix with $X[k]$, $y[k]$ and $\chi[k]$ denoting the predictors, outcome and observed locations from $k$-th fold, respectively, and $X[-k]$, $y[-k]$ and $\chi[-k]$ denoting respective data not in $k$-th fold. Let $n_k$ be the number of observed locations for the $k$-th fold. The values of the prefixed hyper-parameters $\{\phi, \nu, \delta^2\}$ of the conjugate Bayesian spatial regression model are picked from the grid $G_{all}$, which is expanded over the grids of candidate values as the Cartesian product $G_\phi \times G_\nu \times G_\delta^2$ for $\{\phi, \nu, \delta^2\}$. We compute the posterior expectation $\mathbb{E}(y[k] \given y[-k], \phi, \nu, \delta^2)$ for $k = 1, \ldots, K$ and all candidate $\{\phi, \nu, \delta^2\}$ when using stacking of means to obtain the stacking weights. Algorithm~\ref{code: w_stacking_mean} describes the procedure with additional details in Section~\ref{app: pseudo_code}. Algorithm~\ref{code: w_stacking_mean} structures the computation for different choices of $\delta^2$ to be nested in each folder $k$. This structure allows the re-use of the correlation matrix for the same $\{\phi, \nu\}$ in the same folder for different $\delta^2$, but it only works for the shared candidate values $G_\delta^2$. 

{
Let $\hat{Y}_{kg}$ be the $n_k\times 1$ vector with elements $\mathbb{E}(y[k] \given y[-k], \phi, \nu, \delta^2)$ for each $\{\phi, \nu, \delta^2\} \in G_{all}$, where $g = 1, \ldots, G$ indexes the distinct combination of $\{\phi_g, \nu_g, \delta^2_g\} \in G_{all}$ and $G = |G_{all}|$. We suppress the index $g$ in Algorithm~\ref{code: w_stacking_mean} for ease of notation. We construct $\hat{Y}_g = (\hat{Y}_{1g}^{\T},\ldots,\hat{Y}_{Kg}^{\T})^{\T}$ as the $n\times 1$ vector with $n=\sum_{k=1}^K n_k$, and $\hat{Y} = [\hat{Y}_{1}:\ldots : \hat{Y}_G]$ as an $n\times G$ matrix. Algorithm~\ref{code: w_stacking_mean} describes the explicit steps for computing $\mathbb{E}(y[k] \given y[-k], \phi, \nu, \delta^2)$. Let $w = (w_1, w_2, \ldots, w_G)^\T$ be the stacking weights obtained as $\underset{w }{\mbox{argmin}}\{(y - \hat{Y}w)^\T(y - \hat{Y}w)\}$ under constraints $\sum_{g = 1}^{G} w_g = 1$ and $w_g \geq 0$ for $g = 1, \ldots, G$. We formulate this as a quadratic programming (QP) problem  and use the \texttt{quadprog} package in \texttt{R} and the solver \texttt{Mosek} \citep{andersen2000mosek} in \texttt{Julia} to solve for the weights (see Section~\ref{app: QP_w}).}

\begin{algorithm}[ht]
    \caption{Computing stacking weights using stacking of means}\label{code: w_stacking_mean}
    \begin{algorithmic}[1] 
    \State \textbf{Input:}, $X$, $y$, $\chi$, prior parameters $\mu_\beta$, $V_\beta$, $a_\sigma$, $b_\sigma$, $G_{\phi}, G_{\nu}, G_{\delta^2}$ and $K$
    \State \textbf{Output:} $w = \{w_{\phi, \nu, \delta^2}\}_{(\phi, \nu, \delta^2) \in G_{all}}$: Stacking weights
    \vspace{2mm}
    \State  Compute $X_{\mbox{prod}}^{(k)} = X^\T[-k]X[-k]$, $X_{y}^{(k)} = X^\T[-k]y[-k]$ and $n_k$ for $k = 1, \ldots, K$ 
    \For{$\{\phi,\nu\} \in G_\phi \times G_\nu$}
        \For{$k=1$ to $K$}
        \State Calculate $R^{-1}_{\phi, \nu}(\chi[-k])$ and store $J_{\phi, \nu}(\chi[k], \chi[-k])$
            \For{$\delta^2$ in $G_{\delta^2}$}
            \State Set $W = \begin{bmatrix}
                X[k]^{\T} \\
                R^{-1}_{\phi, \nu}(\chi[-k])J_{\phi, \nu}(\chi[k], \chi[-k])^{\T}
            \end{bmatrix}$ and $m_{\ast} = \begin{bmatrix}
              V_\beta^{-1}\mu_\beta + \delta^{-2}X_y^{(k)} \\ \delta^{-2}y[-k] \end{bmatrix}$
            \State Set $M_{\ast}^{-1} = \begin{bmatrix}
              \delta^{-2} X_{\mbox{prod}}^{(k)} + V_\beta^{-1} & \delta^{-2} X^\T[-k] \\
              \delta^{-2} X[-k] & R^{-1}_{\phi, \nu}(\chi[-k]) + \delta^{-2}I_{n - n_k} \end{bmatrix} $
              \State Set $\mathbb{E}(y[k] \given y[-k], \phi, \nu, \delta^2) = W^{\T}M_{\ast}^{-1}m_{\ast}$
            \EndFor
        \EndFor
    \EndFor
    \State Construct $\hat{Y}$ using $\mathbb{E}(y[k] \given y[-k], \phi, \nu, \delta^2)$ for $k=1,\ldots K$ and $\{\phi, \nu, \delta^2\} \in G_{all}$. 
    \State Solve convex optimization problem: $\arg\min_{w}(y - \hat{Y}{w})^\T
    (y - \hat{Y}{w})$ under constraints $\sum_{g = 1}^{G} w_g = 1$ and $w_g \geq 0$ for $g = 1, \ldots, G$, where $G = |G_{all}|$. 
    \end{algorithmic}
\end{algorithm}

We obtain the stacking weights for predictive densities by evaluating the log point-wise predictive density, ($lp_{(\phi, \nu, \delta^2)}(s)$), of $y(s)$ for all locations in each fold for all candidate models. The log point-wise predictive density is derived explicitly in Section~\ref{app: lppd}, while Section~\ref{app: MC_stacking_alg} devises a Monte Carlo algorithm for stacking of predictive densities. The stacking weights for predictive densities are calculated using \textit{Mosek} in \texttt{R} and \textit{Ipopt} in \texttt{Julia}, both employing interior-point methods for optimization problems with logarithmic objectives and linear constraints.} Algorithm~\ref{code: w_stacking_pds} summarizes the procedure with further details provided in Section~\ref{app: pseudo_code}.

{\footnotesize
\begin{algorithm}[h]
    \caption{
    {Stacking weights calculation using stacking of predictive densities}}\label{code: w_stacking_pds}
    \begin{algorithmic}[1] 
    \State \textbf{Input:}, $X$, $y$, $\chi$, prior parameters $\mu_\beta$, $V_\beta$, $a_\sigma$, $b_\sigma$, $G_{\phi}, G_{\nu}, G_{\delta^2}$ and $K$
    \State \textbf{Output:} $w = \{w_{\phi, \nu, \delta^2}\}_{(\phi, \nu, \delta^2) \in G_{all}}$: Stacking weights
    \For{$\{\phi, \nu\}$ in grid expanded by $G_\phi$ and $G_\mu$}
        \For{$k=1$ to $K$}
            \For{$\delta^2$ in $G_\delta^2$}
            \State Compute $\mathbb{E}(y[k] \given y[-k], \phi, \nu, \delta^2)$ (Follow line 1-7 in Algorithm~\ref{code: w_stacking_mean})
            \For{$s$ in $\chi[k]$}
            \State Compute $lp_{(\phi, \nu, \delta^2)}(s) = \log(p(y(s) \given y[-k], \phi, \nu, \delta^2))$
            \EndFor
            \EndFor
        \EndFor
    \EndFor
    \State Calculate weights by maximizing $\sum_{s \in \chi}\log\left(\sum_{(\phi, \nu, \delta^2) \in G_{all}} \exp{\{lp_{(\phi, \nu, \delta^2)}(s)\}} * w_{(\phi, \nu, \delta^2)}\right)$ under constrains $\sum_{(\phi, \nu, \delta^2) \in G_{all}} w_{(\phi, \nu, \delta^2)}  = 1$ and $w_{(\phi, \nu, \delta^2)} > 0$
    \end{algorithmic}
\end{algorithm}
}

{While inferential performance of these algorithms is promising (theoretically and asymptotically) when the candidate values of the correlation parameters fall in a reasonable domain, the choice of $G_\phi$, $G_\nu$ and $G_\delta^2$ still impact the performance of stacking in practical analysis. Here, we offer guidance on specifying $G_\phi, G_\nu, G_\delta^2$. More comprehensive evaluation and discussion are presented in Section~\ref{sec: simulation}. Based on the property of Mat\'ern kernel, popular choices for $\nu$ include $0.5$, $1.0$, $1.5$ and $1.75$ or $2$. Mat\'ern kernels with $\nu > 2$ generate overly smoothed processes and cause numerical instabilities. The range of candidate values for $\phi$ are determined by a lower and upper bound of range along with the choices for $\nu$. In simulations, we choose equally spaced candidate values for $G_\phi$. We note that an even grid is not equivalent to a uniform prior 
{unlike, for example, \cite{kazianka2012objective}}. $G_\phi$ serves as a discretized domain for $\phi$ and, as we described and showed in Section~\ref{app: Inf_hyper}, one can hardly obtain inference about hyper-parameters through stacking. The choice of candidate values for $\delta^2$ is more subtle. In our implementation, we use quantiles of beta distribution $\mbox{Beta}(a_1, a_2)$ to select candidate values of $\delta^2/(1+\delta^2)$, which falls between 0 and 1. This is based on the fact that when $\sigma^2 \sim \mbox{IG}(a_1, b)$, $\tau^2 \sim \mbox{IG}(a_2, b)$ and the two parameters are independent, $\delta^2/(1+\delta^2) \sim \mbox{Beta}(a_1, a_2)$. The two shape parameters $a_1, a_2$ are determined from values of nugget and partial sill estimated from an empirical semivariogram. We choose $b$ to be the larger value of the estimated nugget and partial sill in our implementation. Since the posteriors of $\sigma^2$ and $\tau^2$ are not independent, other choices for $G_\delta^2$ are preferred when additional information about $\delta^2$ is available.}

\section{Simulation}\label{sec: simulation}
\subsection{Simulation settings}\label{subsec: sim_setting}
We present four simulation experiments to evaluate predictive performance using our stacking algorithms. The data for these experiments are generated using \eqref{eq:spatialG} on locations sampled uniformly over a unit square $[0, 1]^2$ with $R_{\Phi}$ being the Mat\'ern covariogram in \eqref{eq: Matern}. The sample size $n$ of the simulated data sets ranges from $200$ to $900$, and we randomly pick $n_h = 100$ observations for checking predictive performance. The vector $x(s)$ consists of an intercept and a predictor generated from a standard normal distribution. We use the parameter values $\beta = (1, 2)^\T$, $\sigma^2 = 1$, $\tau^2 = 1$, $\nu = 1$, and $\phi = 7$ and $2$ to generate data for the first and second simulation, respectively. For the third and fourth simulation, we set $\phi = 20$ and $2$, respectively, with $\sigma^2 = 1$, $\tau^2 = 0.3$ and $\nu = 0.5$. 

We analyze our data using the $K=10$-fold stacking Algorithms~\ref{code: w_stacking_mean}~and~\ref{code: w_stacking_pds} with candidate values $\nu \in G_{\nu} = \{0.5, 1, 1.5, 1.75\}$. The candidate values for $\phi$ are selected so that the ``effective spatial range'', which refers to the distance where spatial correlation drops below 0.05, covers $0.1$ and $0.6$ times $\sqrt{2}$ (the maximum inter-site distance within a unit square) for all candidate values of $\nu$. Here we set $G_{\phi} = \{3, 14, 25, 36\}$. It is worth noting that the actual value of $\phi$ in the second and fourth simulation are $2$, which are smaller than the lowest candidate value in $G_{\phi}$. The values of $\phi$ were deliberately chosen to be large and small to investigate the behavior of the proposed algorithms.
{Finally, we specify $G_{\delta^2}$ to comprise the $0.05$, $0.35$, $0.65$ and $0.95$th quantiles of a beta distribution with expectations of $\sigma^2$ and $\tau^2$ equal to their data generating values.} We assign an $\mbox{IG}(a_\sigma, b_\sigma)$ prior with $a_\sigma = b_\sigma = 2$ for $\sigma^2$. The prior of $\beta$ is $\mbox{N}(\mu_\beta, V_\beta)$ where $\mu_\beta = 0$ and $V_\beta = 4\cdot {I}$. For each simulated data set, we implement stacking of means and of predictive densities to obtain the expected outcome $\hat{y}(s)$ based on the held out observed locations. The predictive accuracy is evaluated by the root mean squared prediction error over a set of $n_h$ hold-out locations in set ${\cal S}_h$  ($\text{RMSPE} = \sqrt{\sum_{s \in {\cal S}_h}((\hat{y}(s) - y(s))^2)/n_h}$). We also compute the posterior expected values of the latent process $\hat{z}(s)$ for $z(s)$ on all of the $n$ sampled locations in ${\cal S}$ and evaluate the root mean squared error for $z(s)$ ($\text{RMSEZ} = \sqrt{\sum_{s \in {\cal S}}(\hat{z}(s) - z(s))^2/n}$). To further evaluate the distribution of predicted values, we compute the mean log point-wise predictive density for the $n_h$ held out locations ($\text{MLPD} = \sum_{s \in {\cal S}_h}\{\log (\sum_{g = 1}^G w_g p(y(s) \given y, \calM_g))\}/n_h$). 

\begin{figure}[t]
\centering
{\includegraphics[width=0.49\textwidth, keepaspectratio]{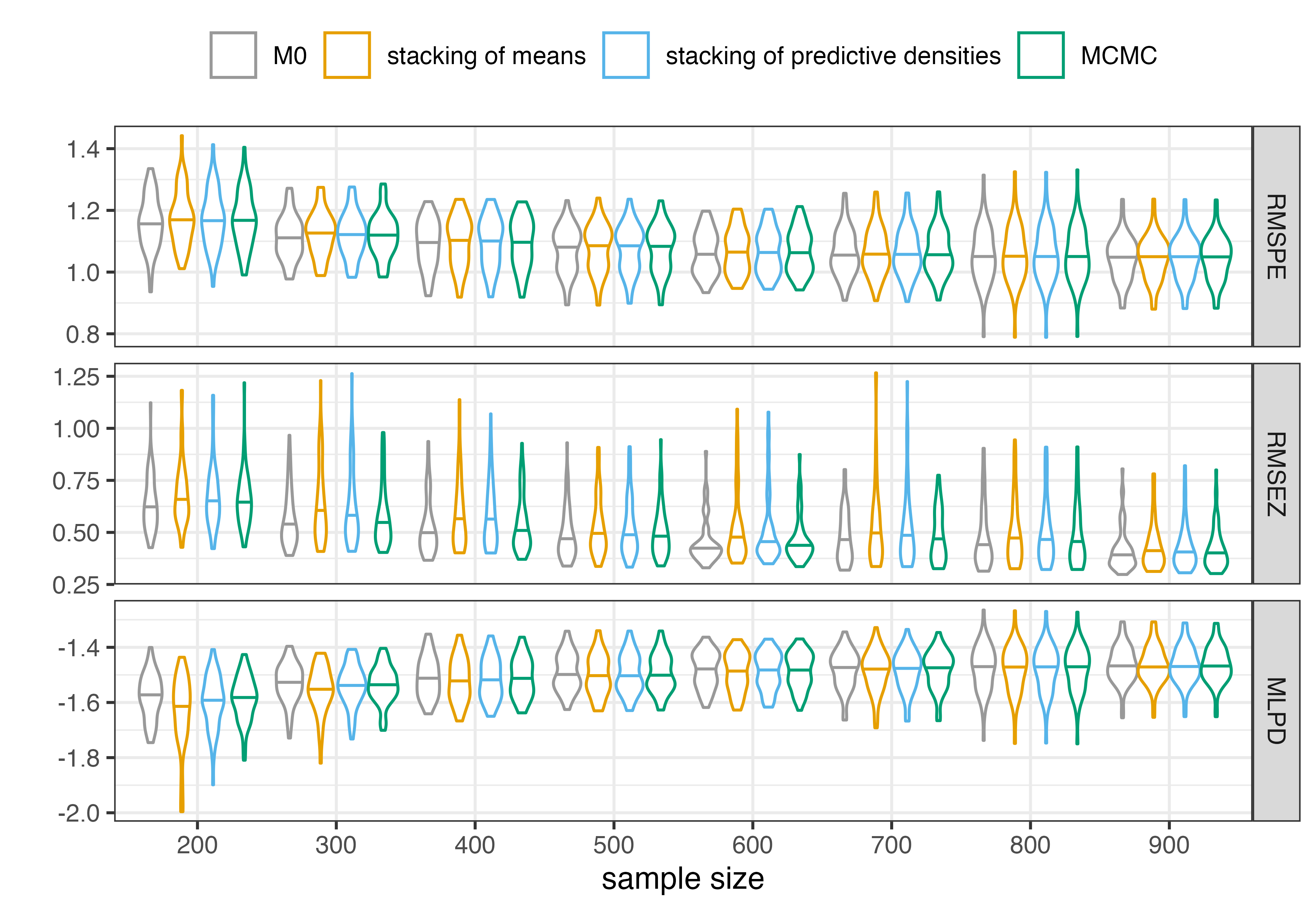}}
{\includegraphics[width=0.49\textwidth, keepaspectratio]{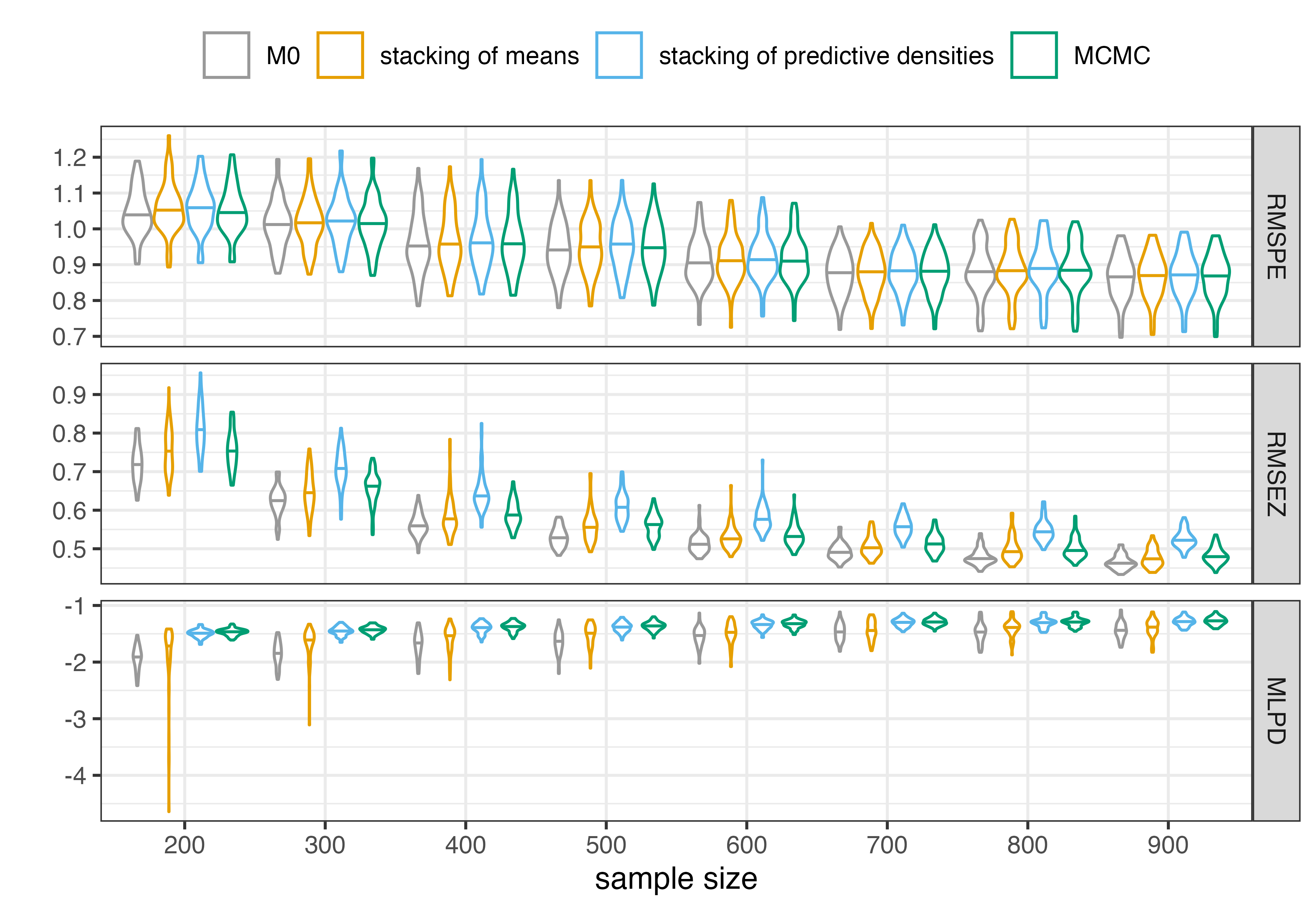}}
\caption{Distributions of the diagnostic metrics for prediction performance for the first simulation (left) and the third simulation (right). Each distribution is depicted through a violin plot. The horizontal line in each violin plot indicates the median. }
\label{fig: sim_compar}
\end{figure}

Apart from stacking, we also implement a fully Bayesian model with priors on the hyperparameters using Markov chain Monte Carlo (MCMC) sampling for comparison. In addition, we carry out exact Bayesian inference using the conjugate model in Section~\ref{sc21} with hyper-parameters fixed at the exact value (denoted as $\mathcal{M}_0$). We use the same priors for $\sigma^2$ and $\beta$ as those in stacking implementations. For the rest of the priors needed in full MCMC sampling, we assign uniform priors $\mbox{U}(3, 36)$ for $\phi$ and $\mbox{U}(0.25, 2)$ for $\nu$, and an $\mbox{IG}(2, 2)$ prior for $\tau^2$. Sampling is fitted through the \textit{spLM} function in the \textit{spBayes} package in \texttt{R}. The diagnostic metrics are computed based on 1,000 posterior samples retained after convergence was diagnosed over a burn-in period of 10,000 initial iterations. The algorithm for recovering the expected $z$ and the log point-wise predictive density based on the output of \texttt{spLM} is presented in Appendix~\ref{app: recover_z_MCMC}.  We monitor all diagnostic metrics for prediction for all competing algorithms. To measure uncertainty of the diagnostic metrics, we generate 60 data sets for each sample size in each simulation, fit each data set with the four competing methods and record the diagnostic metrics of each model fitting. 

\subsection{Predictive performances}\label{subsec: Pred_perform}
The aforementioned methods exhibit different behaviors in the four simulation studies. Figure~\ref{fig: sim_compar} compares predictive performance for the first and the third study. The results for Simulation 2 \& 4 closely mirror those of the first simulation and are included in Appendix~\ref{appsub: figs_sum} for brevity. There are no pronounced distinctions in predictive performance among the competing models in the first simulation. In the third simulation, however, stacking of means seems to deliver better estimates for the latent process at observed and unobserved locations than stacking of predictive densities (based on RMSEZ), while stacking of predictive densities outperforms stacking of means in terms of the log point-wise predictive density (based on MLPD). This is unsurprising as we optimize prediction error in the stacking of means and we maximize the log predictive densities in the stacking of predictive densities. Treating the fully Bayesian model with priors on all hyperparameters (fitted using MCMC) as a benchmark, we find that stacking of predictive densities is very competitive in terms of MLPD. The performance of latent process estimation for the full Bayesian model falls between stacking of means and stacking of predictive densities. All competing algorithms deliver very similar prediction accuracy for the outcome at unobserved locations, while stacking of means slightly outperforms the full Bayesian model with regard to the medians of the RMSPEs for all fittings; both are slightly better than stacking of predictive densities, The conjugate Bayesian model $M_0$ provides the best point estimates for the outcome and latent processes based on RMSPE and RMSEZ, but is less impressive in terms of MLPD. These results indicate that stacking of means is excels in point estimation, while stacking of predictive densities is preferable for interval estimation. 

We check the counts of the non-zero weights in stacking and find that stacking of means tends to produce a slightly smaller number of non-zero weights than stacking of predictive densities. The number of non-zero weights is small for both stacking algorithms. {
This sparsity is not an artifact of our methodology but rather a known characteristic of stacking. As first reported by \cite[Section~9][]{breiman1996stacked}, the author observed that stacking combines a ``surprisingly few'' number of models. Our findings are consistent with this observation: }
on average, there are around 3.6 and 4.3 out of 64 weights that are greater than 0.001 in the simulation studies for stacking of means and stacking of predictive densities, respectively. This number is relatively consistent when the sample sizes increase. Plots for the distributions of nonzero weights counts are provided in Figure~\ref{fig: nonzero_weights_compar}. To further explore this, we visualize the distributions of the non-zero weight values in Supplement~\ref{appsub: w_distr}. These plots reveal a strongly right-skewed distribution, confirming that stacking not only selects a small set of models but further concentrates the predictive influence on just a few top performers. This inherent parsimony highlights stacking's ability to perform implicit model selection. 

\begin{figure}[h]
\centering
\subfloat[\label{subfig: sim1_nonzero}]{\includegraphics[width=0.49\textwidth, keepaspectratio]{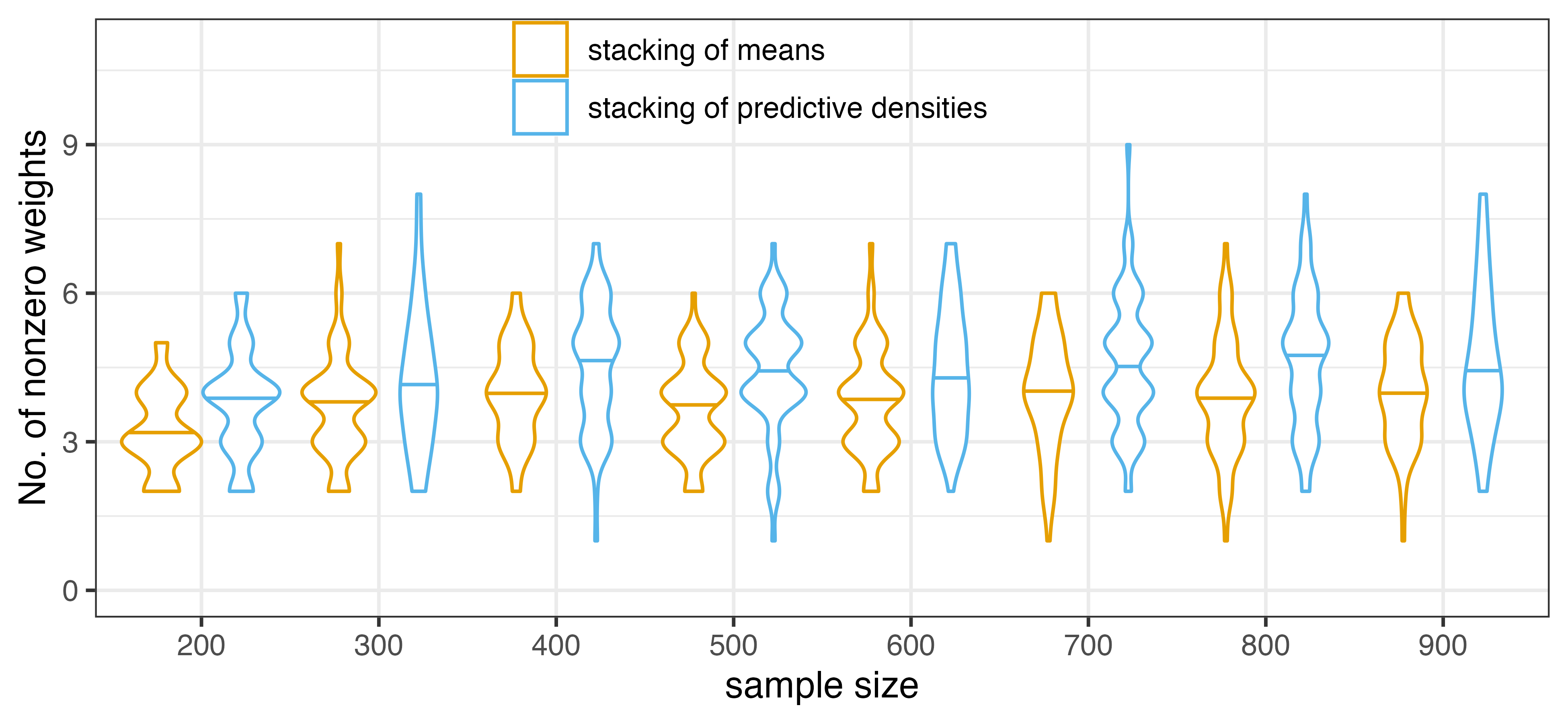}}
\subfloat[ \label{subfig: sim4_nonzero}]{\includegraphics[width=0.49\textwidth, keepaspectratio]{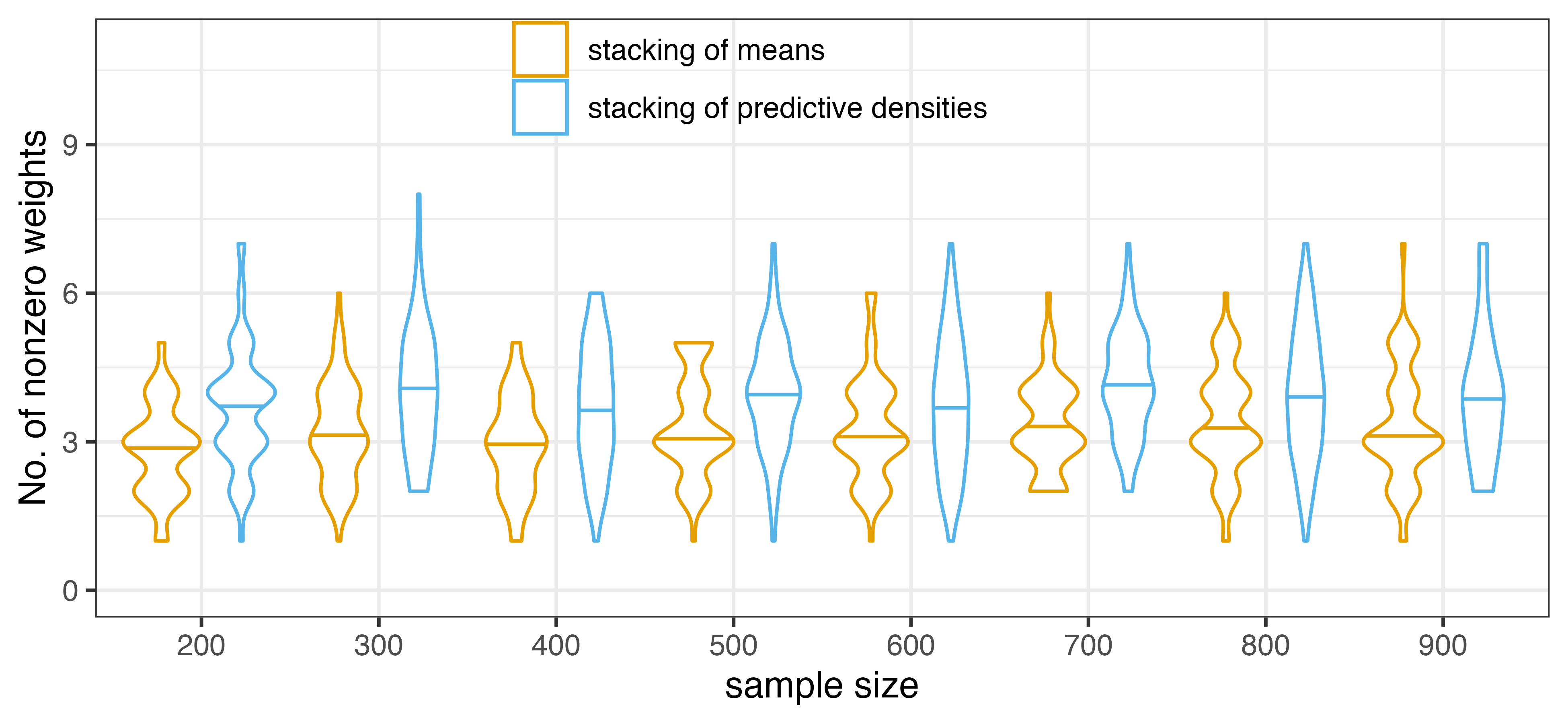}} \\
\subfloat[ \label{subfig: sim2_nonzero}]
{\includegraphics[width=0.49\textwidth, keepaspectratio]{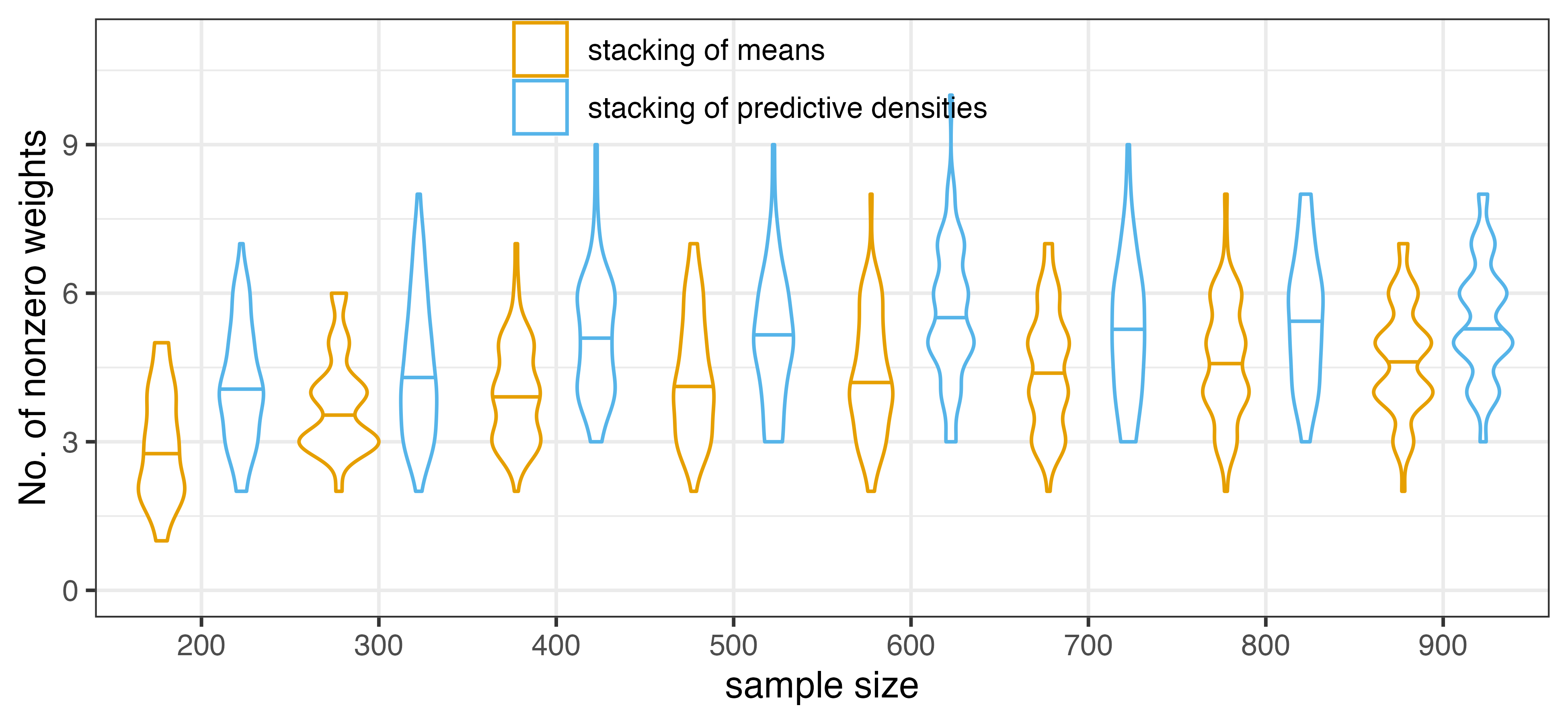}}
\subfloat[ \label{subfig: sim3_nonzero}]
{\includegraphics[width=0.49\textwidth, keepaspectratio]{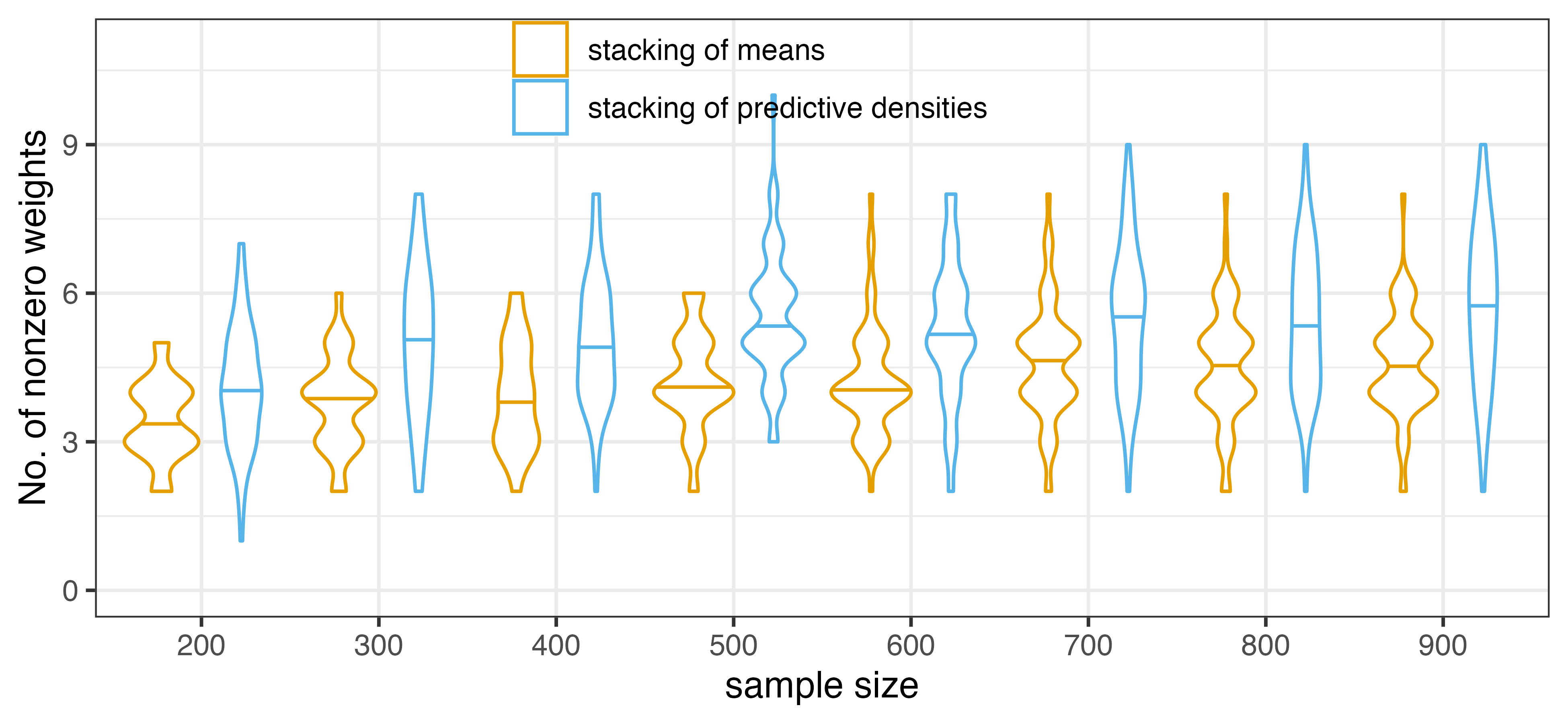}}
\caption{Distributions of the counts of nonzero weights in the first (a), second (b), third (c), and fourth (d) simulation. The distribution of the counts are described through violin plots whose  horizontal lines indicate the medians.}\label{fig: nonzero_weights_compar}
\end{figure}
\captionsetup{labelfont={color=black},textfont={color=black}}

Building on the preceding broader analysis, we now attend to more specific scrutiny of inferential performance. We examine a case from Simulation 1 featuring 800 observations, another from Simulation 2 with 600 observations, a third from Simulation 3 containing 400 observations, and a fourth from simulation 4 containing 200 observations to enhance our evaluation of stacking's predictive performance. These four examples are intentionally selected to represent typical inferential behavior across varying parameter settings and sample sizes. The results from these selected examples are consistent with our empirical findings across 60 replicates and align well with the diagnostic metrics for predictive performance reported for the full set of simulations. Figure~\ref{fig: y_U_CI_compar} directly compares 95\% credible intervals (CIs) and point estimates obtained through stacking and MCMC methods for the first and the third cases. Results for the second and fourth cases, which closely resemble those of the first case, are provided in Appendix~\ref{appsub: figs_sum} for brevity. Here, the 95\% CIs for stacking are based on 900 draws from the stacked posterior distribution. Specifically, we generate 900 posterior samples from each candidate model with non-zero stacking weights and then randomly sample 900 draws from this pool according to the stacking weights. Stacking of predictive density appears to closely align with inference from MCMC, while stacking of means tends to marginally underestimate the CI widths, particularly in the setting with a smaller range. 

\begin{figure}[t]
\centering
\includegraphics[width=0.8\textwidth, height=0.275\textwidth]{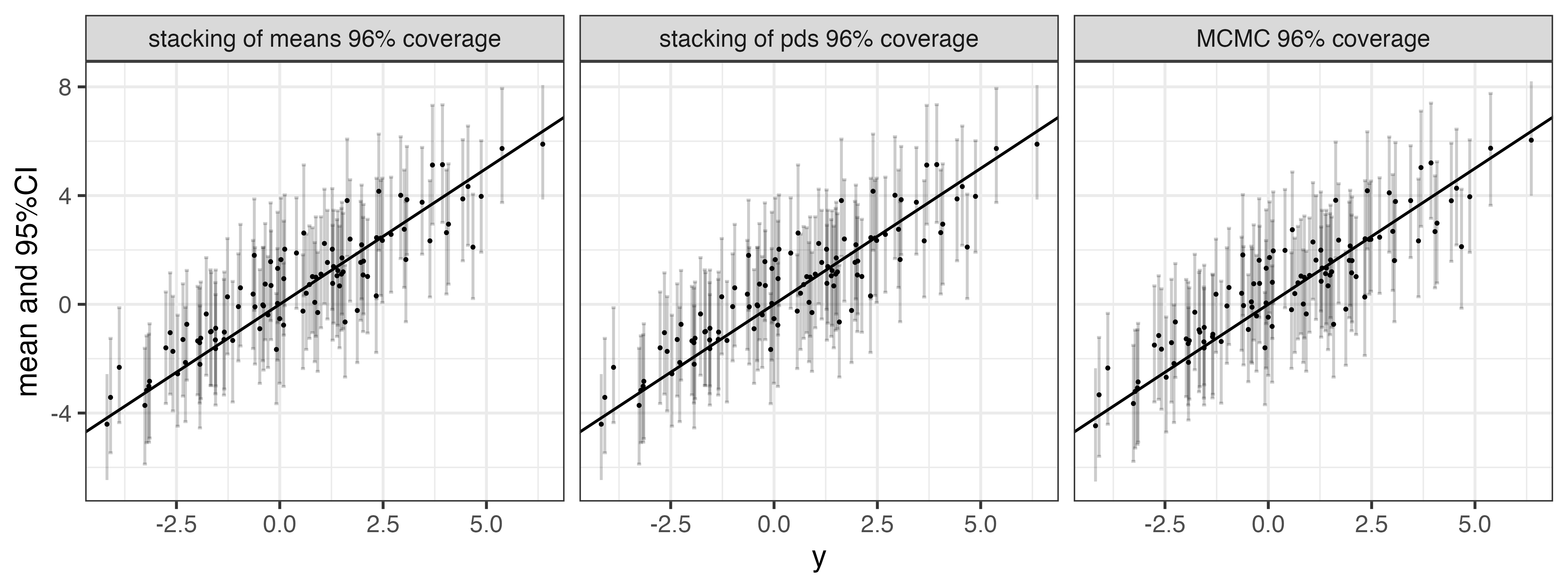}\\
\includegraphics[width=0.8\textwidth, height=0.275\textwidth]{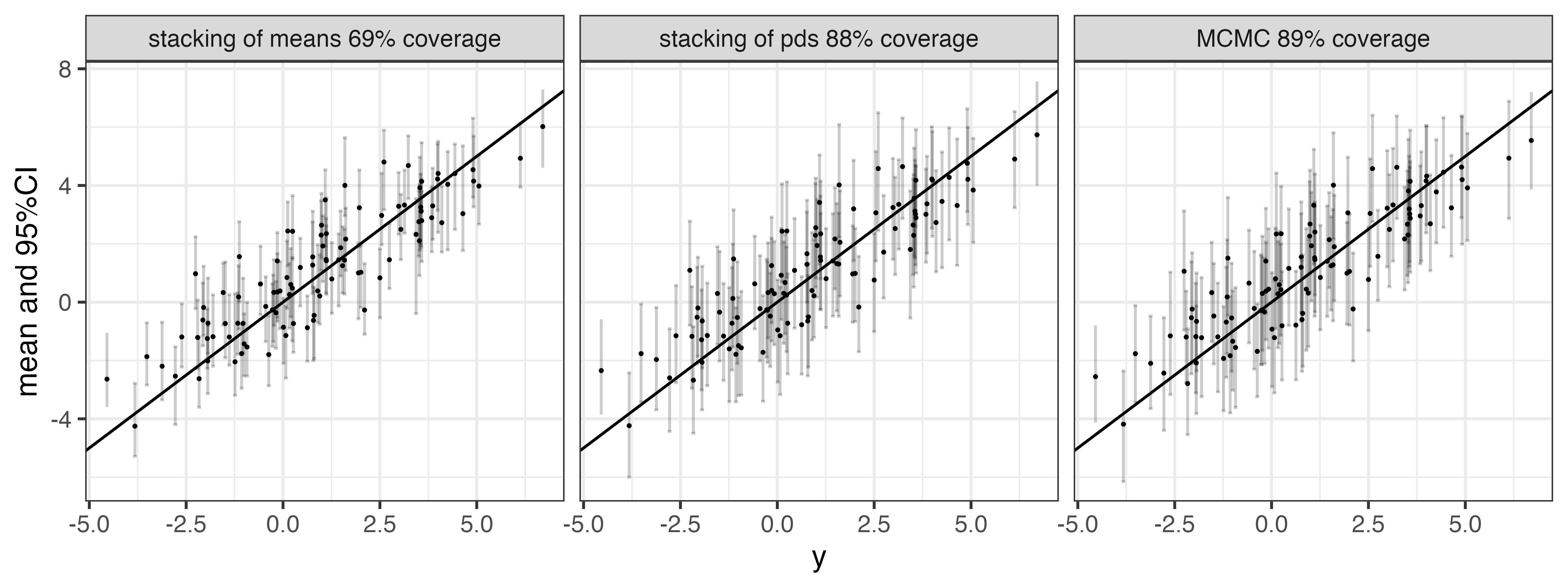}
\caption{ 95\% credible intervals for predicted versus actual outcomes at 100 unobserved locations: Simulation 1 with 800 observations (upper row) and another example from Simulation 3 with 400 observations (lower row). Each plot includes a solid black line representing the 45-degree reference line, with captions indicating 95\% CI coverage. `pds' denotes predictive densities.}\label{fig: y_U_CI_compar}
\end{figure}

Section~\ref{app: maps_w} illustrates the interpolated maps of the predicted outcome at held out locations and the expected latent processes over all locations generated by different fitting algorithms. The posterior predictive means, $\mathbb{E}(y(s)\given y)$, and $\mathbb{E}(z(s)\given y)$, share similar patterns with the raw data. The mean of $z(s)$ estimated by stacking of predictive densities appears smoother than those estimated from $\mathcal{M}_0$, full Bayes and stacking of means, while the predicted mean of $y(s)$ at unobserved locations are indistinguishable across all methods.

\subsection{Running time comparisons}
Computer programs for reproducing the simulation studies are hosted in the GitHub repository. Comparisons in predictive performances presented above are conducted in \texttt{R}. For the running time comparisons reported here, the stacking algorithms are implemented using \texttt{Julia-1.11.1}. The MCMC sampling algorithms are executed through the package \texttt{spBayes} in \texttt{R-4.3.2}, which relies on underlying functions written in C++ for computational efficiency. 
{W}e report the time for obtaining weights for stacking, and we consider the sampling time for $\{\phi, \nu, \sigma^2, \tau^2\}$ using MCMC (no sampling of $\{\beta, z\}$ and no predictions). The timing comparisons are based upon experiments on a Linux system equipped with 64 AMD EPYC 7513 32-Core Processors. Parallel computing is performed using 8 threads. 
Figure~\ref{fig: time_compar} summarizes the running time for the three competing algorithms. On average, the stacking of means is 496 times faster than MCMC, while stacking of predictive densities is only slightly slower being around 483 times faster than MCMC sampling. These experiments clearly establish that predictive stacking algorithms are efficient alternatives to MCMC for estimating latent spatial processes and predicting spatial outcomes.

\begin{figure}[t]
\centering
\includegraphics[width=\textwidth, keepaspectratio]{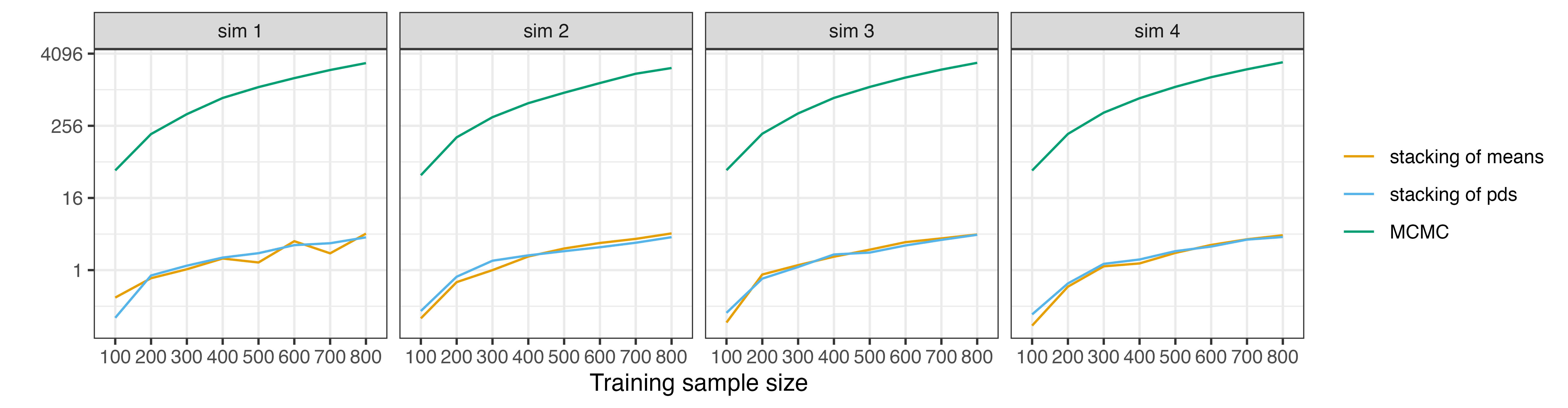}
\caption{Running time comparison for stacking and MCMC sampling}\label{fig: time_compar}
\end{figure}

\subsection{Impact of hyperparameter selection}\label{subsec:Impact_hyper_select}
{We further explore improving inferential performance of stacking by judicious choices of  $\{\phi, \nu, \delta^2\}$. We explore the inferential impact of selecting candidate values for the prefixed hyperparameters from their posterior distributions. These posterior distributions are evaluated from MCMC samples of the full Bayesian model. The method for choosing candidate values as outlined in previous subsections is now referred to as the `default' method. Given that the default algorithm in the simulation studies involves 64 candidate models, we randomly select 64 samples of $\{\phi, \nu, \delta^2\}$ as candidate values. Although obtaining marginal posterior samples is impractical when implementing the stacking algorithm, this approach serves as our benchmark, assuming full knowledge of the marginal posterior. We also present the posterior distributions recovered by MCMC as the gold standard for comparisons.}

\begin{figure}[t]
\centering
{\includegraphics[width=0.49\textwidth, keepaspectratio]{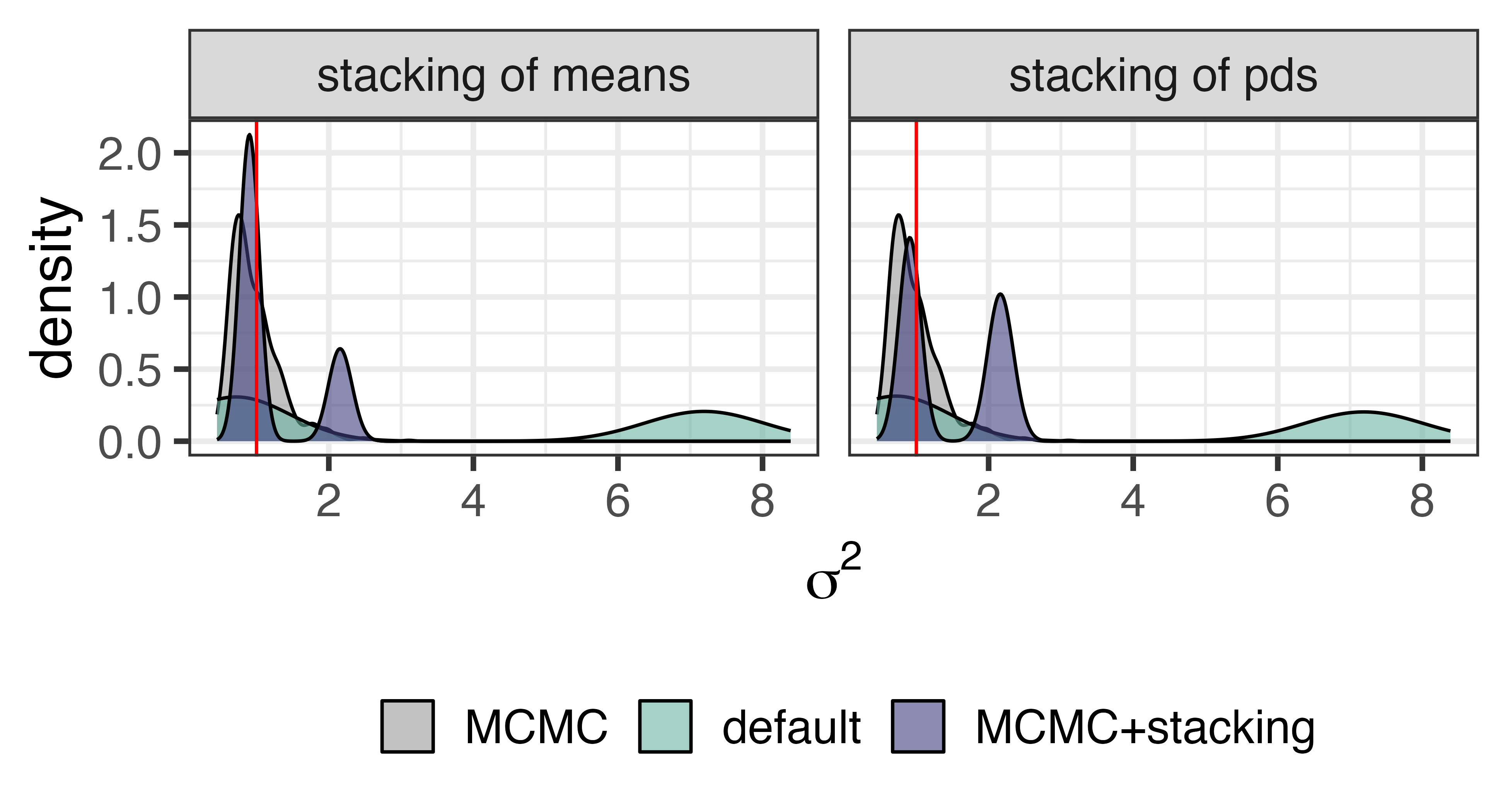}}
{\includegraphics[width=0.49\textwidth, keepaspectratio]{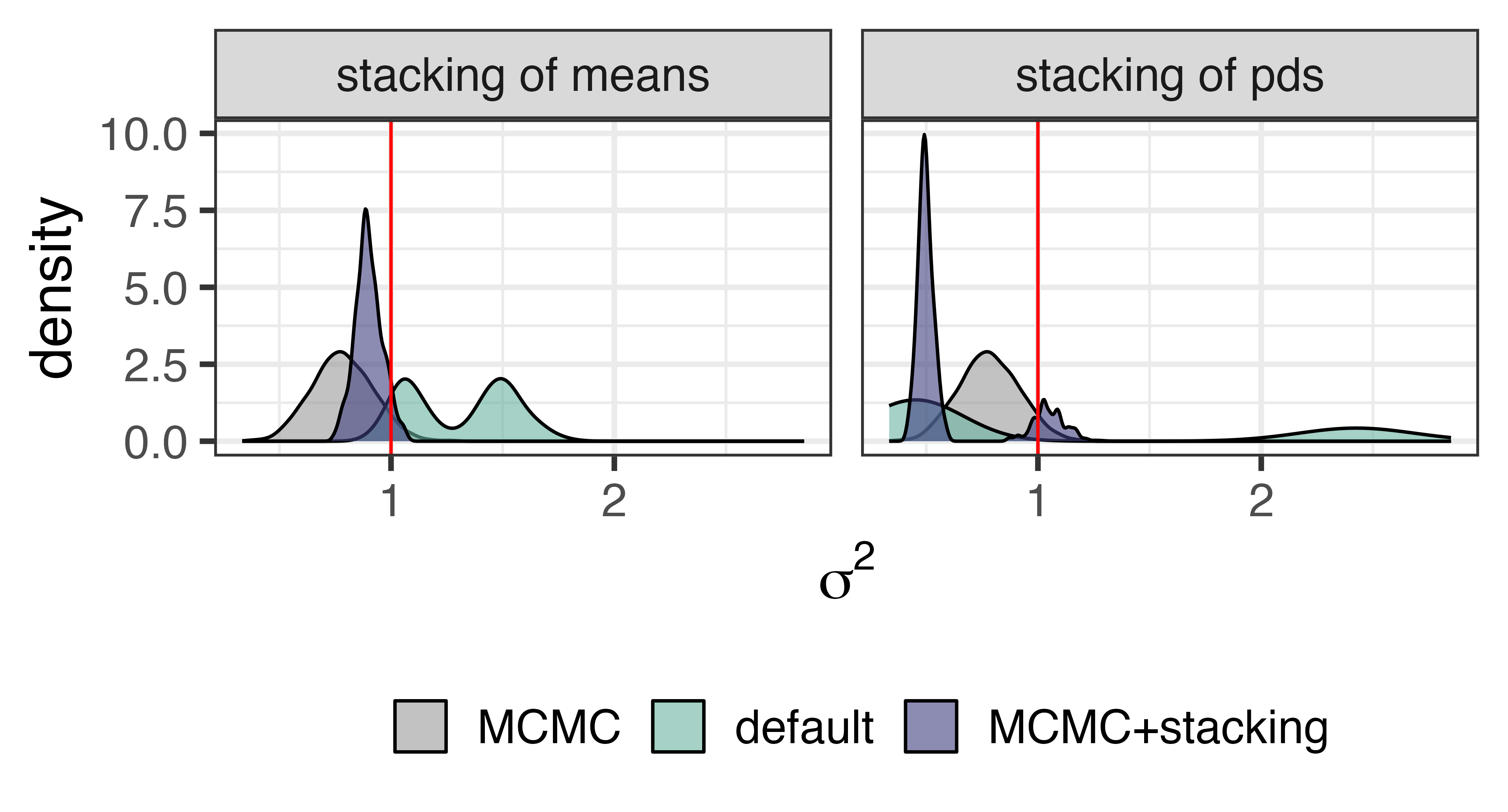}}
\caption{Densities of $\sigma^2$ for the example with 800 observations from simulation 1 (left) and the example with 400 observations from simulation 2 (right). Vertical red lines indicate the actual $\sigma^2$ values. Grey densities represent MCMC-recovered posterior distributions of $\sigma^2$. 'Default' and 'MCMC+Stacking' show stacking results using two methods for selecting ${\phi, \nu, \delta^2}$ candidates. Left panel: stacking of means. Right panel: stacking of predictive densities.}
\label{fig: sim_prefix_sigma_compar}
\end{figure}

Figure~\ref{fig: sim_prefix_compar} depicts a comparison of diagnostic metrics and reveals that there is no significant improvement in predictive performance by selecting the prefixed correlation parameters through posterior distributions. To further check the impact, we revisit the four selected examples in Section~\ref{subsec: Pred_perform}, featuring 800, 600, 400, and 200 observations from simulations~1,~2,~3~and~4, respectively. Representative results from the first and third examples are included in the main paper, with all simulation results detailed in the Appendix for completeness. Figure~\ref{fig: sim_prefix_sigma_compar} compares the posterior distributions for $\sigma^2$. Regardless of the method used for selecting candidate values for the prefixed hyper-parameters, the marginal distributions recovered through stacking invariably exhibits multimodal behavior. This reinforces the claim that stacking does not provide effective inference for covariance parameters, including $\sigma^2$. Furthermore, we observe that the distribution of $\sigma^2$ recovered by stacking is highly dependent on the choice of candidate values for the correlation parameter. Similarly, for $\tau^2$, we observe multimodality in the third example
, as shown in Figure~\ref{fig: sim_prefix_tau_compar}. Intriguingly, the variations in the distribution of $\sigma^2$ resulting from stacking do not impact the predictive distribution of outcomes. We examine the predictive distribution of the outcome at several unobserved points and present some typical examples in Figure~\ref{fig: sim_prefix_50_compar}~and~\ref{fig: sim_prefix_50_90_compar}. The predictive distributions obtained from stacking are largely consistent regardless of hyperparameter selection methods. Additionally, the predictive distributions recovered by stacking of means tend to concentrate around the mode compared to those recovered by stacking of predictive densities. The distributions for the intercept recovered by stacking may have larger variance according to Figure~\ref{fig: sim_prefix_beta_compar}, while those for regression coefficients $\beta_2$ closely align with inference from MCMC. We conclude that the improvement of selecting candidate values for $\{\phi, \delta^2, \nu\}$ from the posterior is limited based on these results.

\begin{figure}[t]
\centering
{\includegraphics[width=0.49\textwidth, keepaspectratio]{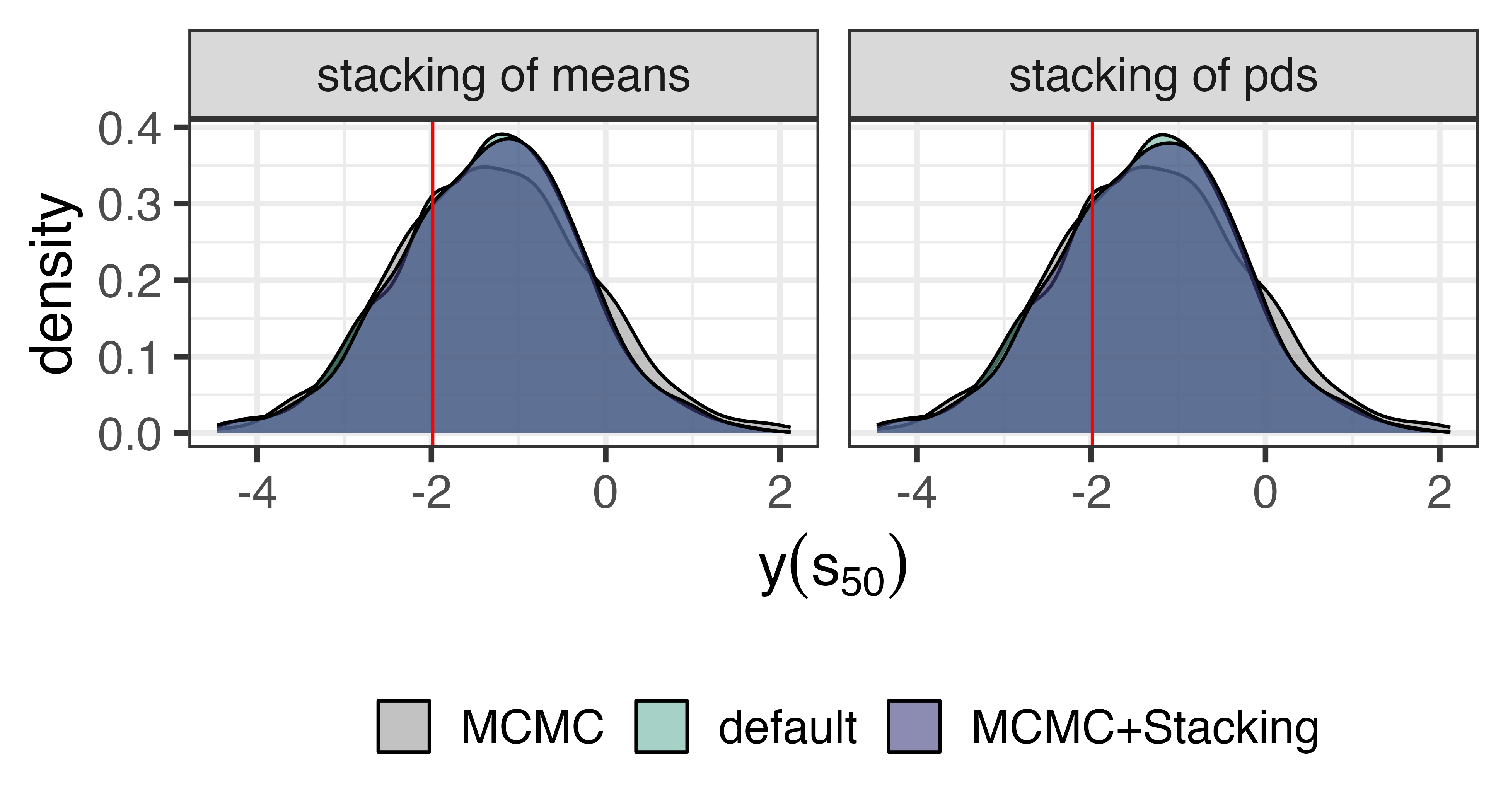}}
{\includegraphics[width=0.49\textwidth, keepaspectratio]{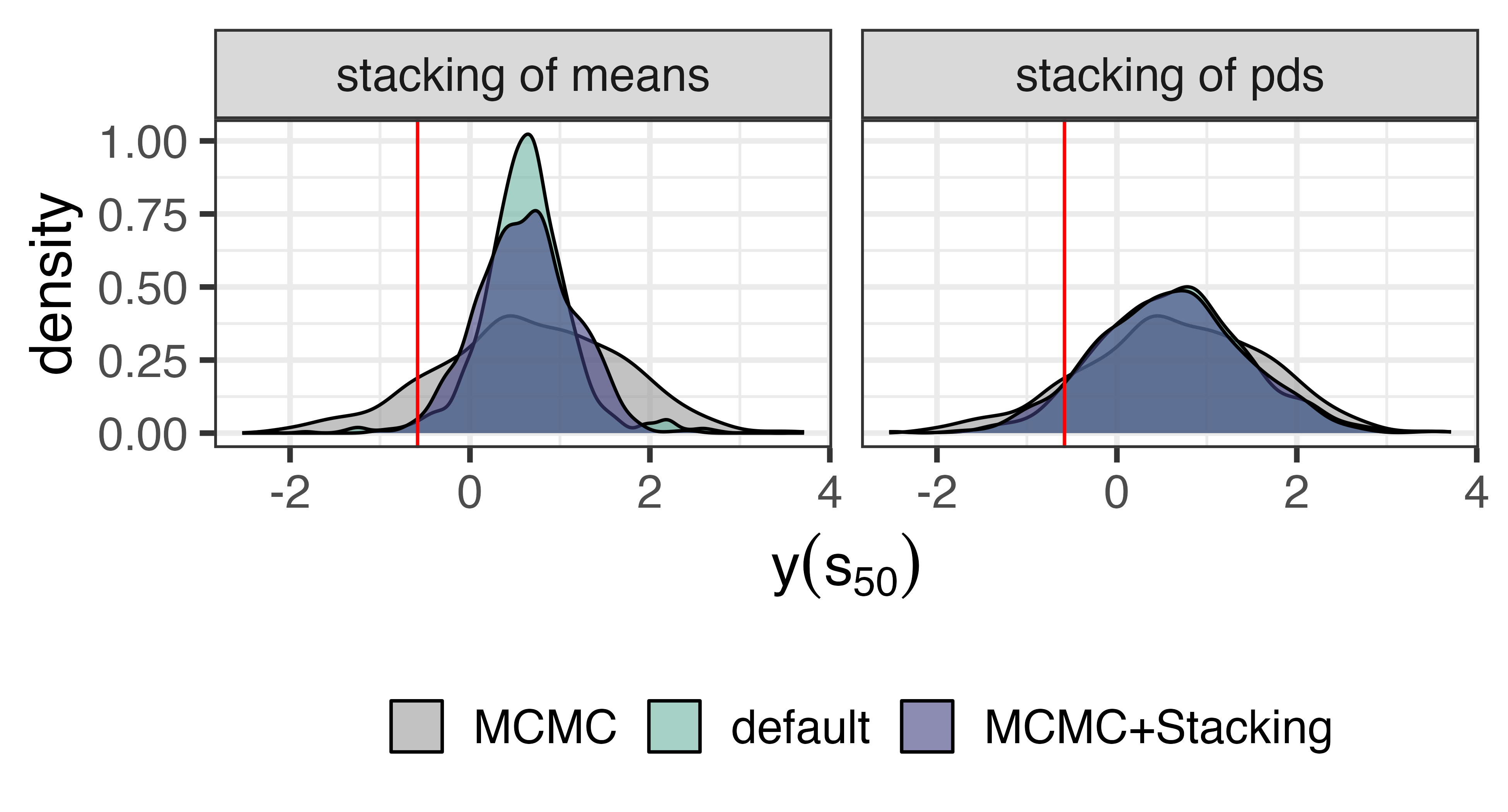}}
\caption{Predictive densities of the outcome at $50$-th point in the example with 800 observations from simulation 1 (left) and the example with 400 observations from simulation 3 (right). Vertical red lines indicate the actual values. Grey densities represent MCMC-recovered posterior distributions. 'Default' and 'MCMC+Stacking' show stacking results using two methods for selecting ${\phi, \nu, \delta^2}$ candidates. Left panel: stacking of means. Right panel: stacking of predictive densities}
\label{fig: sim_prefix_50_compar}
\end{figure}

\section{AOD prediction}\label{sec: aod}
{We use $K=10$-fold Bayesian stacking to analyze Aerosol Optical Depth (AOD) observations from satellite technologies in global aerosol research. This is an increasingly important field across multiple disciplines such as environmental health, climatology, atmospheric science, and remote sensing \citep[][]{voiland2010aerosols}. Unlike ground-based monitoring that is limited by regional coverage and budget, satellite-derived AOD data provides a more expansive picture of aerosol distributions on a global scale. However, cloud screening and conditions of high surface reflectance can result in a significant proportion of missing data in AOD satellite observations \citep{li2009uncertainties}.  
Prevailing AOD interpolation algorithms often rely upon random forests and neural networks \citep[][]{
fan2023satellite, aguilera2023novel}.  Nonetheless, Gaussian process models present a competitive alternative by accommodating data-driven processes and providing essential uncertainty quantification. We apply our Bayesian spatial regression model using stacking to analyze 1-km MODIS AOD products (MCD19A2) over the Greater Los Angeles area \citep{lyapustin2018modis} and assess their predictive performances in the context of large-scale AOD retrieval and uncertainty quantification.}

Figure~\ref{fig: RDA_AOD} is a base image with near-complete AOD coverage on September, 13th, 2018, encompassing 16,003 pixels that do not encroach over water bodies. We use the cloud pattern from August, 24th, 2018, to partition the data. Pixels not obscured by clouds comprise the training set (totaling 11,857 pixels), while the remaining 4,146 pixels form the testing set. We use log-transformed AOD as the outcome and five predictors (resampled to a 1-km resolution): the x-y coordinates, the Enhanced Vegetation Index (EVI) from the 16-day MODIS MOD13A2 products, the impervious surface percentage from the USGS National Land Cover Database (NLCD) 2018, and the weighted road network density from OpenStreetMap (Figure~\ref{fig: RDA_AOD}).

\begin{figure}[t]
\centering
\includegraphics[width=\textwidth, keepaspectratio]{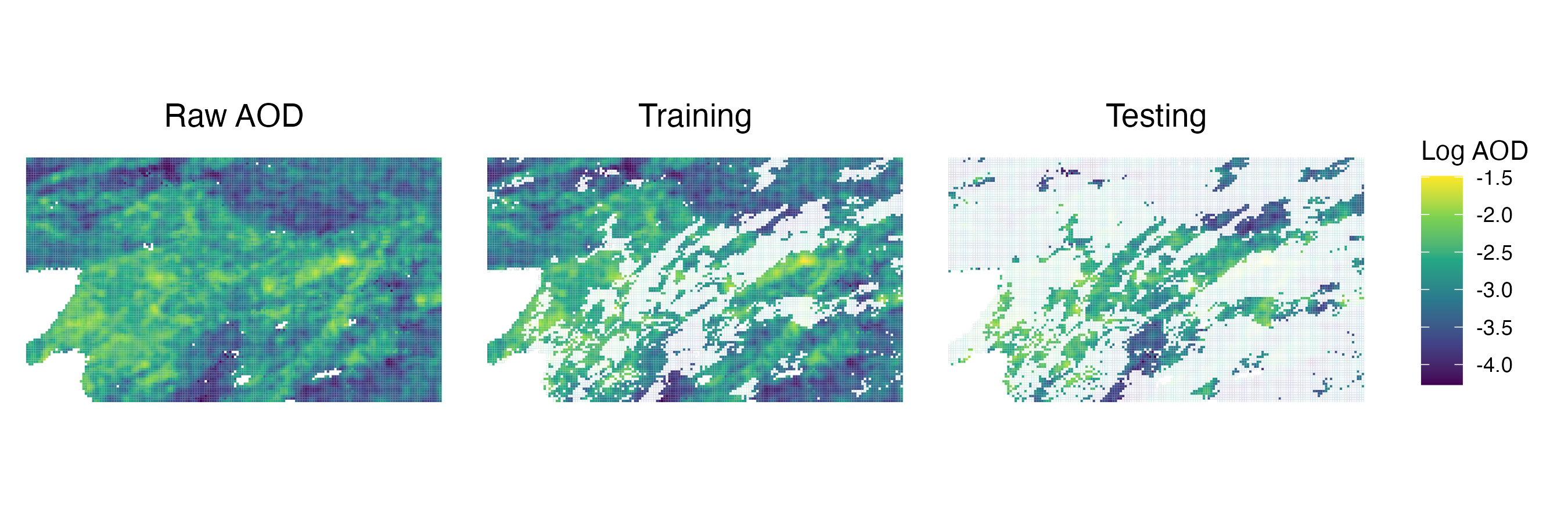}\\
\vspace{-1.1cm}
\includegraphics[width=\textwidth, keepaspectratio]{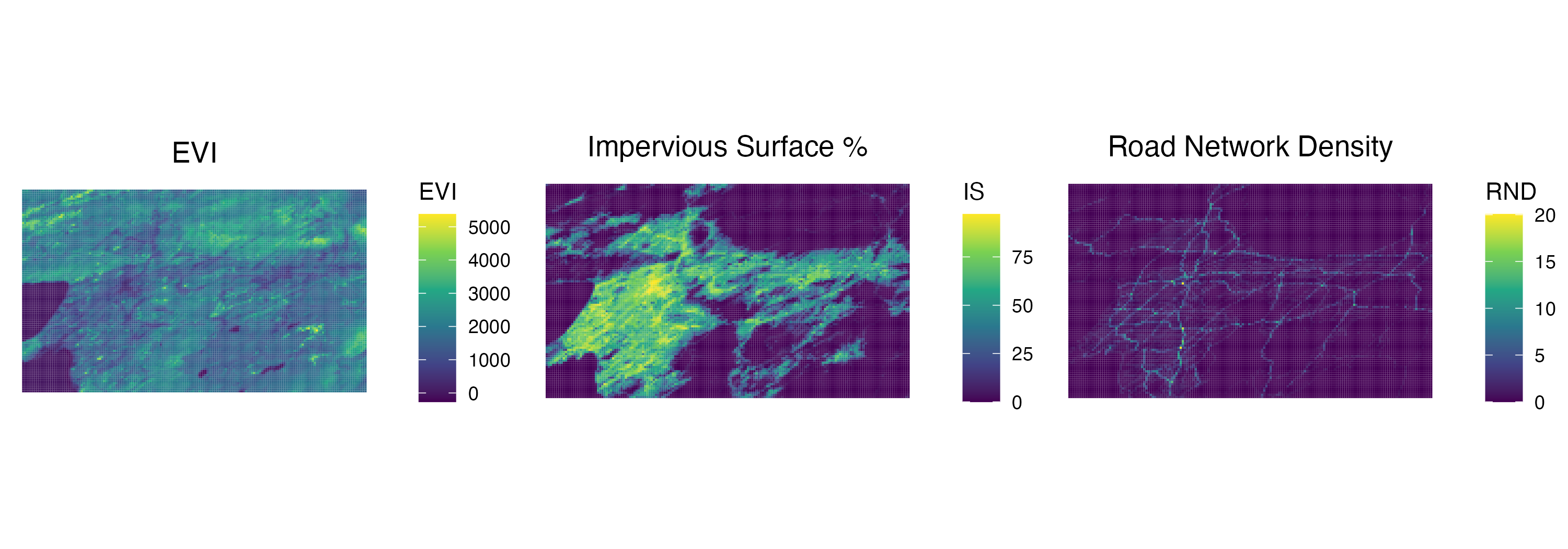}
\caption{Upper: MODIS Aerosol Optical Depth (AOD) Visualization. From left to right: Log-transformed AOD values, training data (pixels not covered by clouds), and testing data (cloud-covered pixels), as of September, 13th, 2018. Lower: Visualization of Regression Model Predictors. From left to right: Enhanced Vegetation Index (EVI) from MODIS MOD13A2 products, impervious surface percentage from USGS NLCD 2018, and weighted road network density from OpenStreetMap, all resampled to 1-km resolution.}\label{fig: RDA_AOD}
\end{figure}

\begin{figure}[t]
\centering
\includegraphics[width=0.7\textwidth, keepaspectratio]{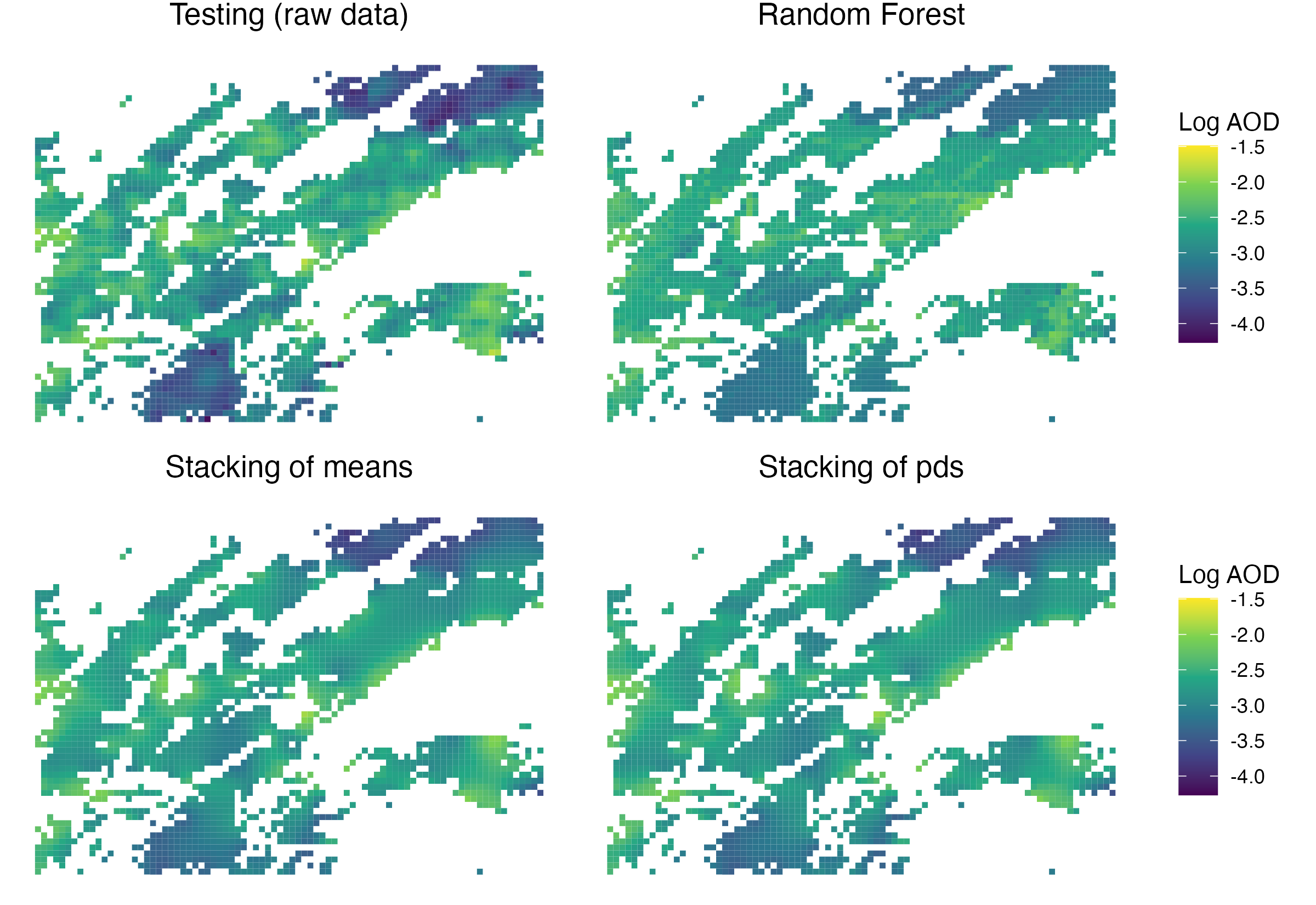} \caption{Interpolated and testing data AOD in the selected region}
\label{fig: RDA_pred}
\end{figure}

Following the prototype in simulation studies, we set $G_\nu = (0.5, 1.0, 1.5, 1.75)$.  The candidate values for $\phi$ were set at $(0.1, 0.4, 0.7, 1.0)$ so that the practical range of the process spans from 5km to 25km. This range was conservatively estimated based on the empirical semivariogram of the residual of a linear regression model (Figure~\ref{fig: RDA_variogram}). We used the estimated values for $\sigma^2$ and $\tau^2$ from the empirical semivariogram to determine $G_{\delta^2}$. We assigned $\sigma^2\sim \mbox{IG}(2, 0.05)$ and the remaining settings following those used in our simulation studies. In addition to stacking, we expanded the analysis by including five competitive algorithms: Deep Learning, Random Forest (RF) \citep{breiman2001random}, Gradient Boosting Machine (GBM) \citep{friedman2001greedy}, an Ensemble Model integrating the aforementioned three algorithms, and Bayesian Linear Regression (BLR) without spatial effects. We implemented the first four algorithms on H2O \citep{cook2016practical}, a popular open-source platform designed for big data analytics. Specifically, we utilized H2O's default settings for the Deep Learning algorithm \citep{candel2016deep} and similarly for the GBM.  For the RF model, we determined the optimal configuration, setting the number of trees to 120 and the maximum tree depth to 60—via a Cartesian grid search. We validated these learner models with 10-fold cross-validation. H2O's Ensemble method is another stacking algorithm which finds the optimal combination of predictions from the fitted DL, RF and GBM models. The BLR model was fitted through \texttt{R} package \texttt{brms}. We assigned $N(0, 2^2)$ to our regression coefficients and 
a half-Cauchy prior $\mbox{half-Cauchy}(0, 0.5)$ to the standard deviation of the error.

For stacking and BLR, the AOD interpolation was retrieved using posterior mean. We calculated the root mean squared prediction error (RMSPE), mean absolute error (MAE) and Pearson correlation coefficient (R) by comparing the interpolated AOD with the testing data AOD, and we have summarized these results in Table~\ref{tab:table1}. The two stacking algorithms significantly outperformed the other competitive algorithms in terms of Pearson correlation coefficients, RMSPE and MAE. Figure~\ref{fig: RDA_pred} showcases predictions at a selection of central testing locations, where we note that the stacking algorithms produce smoother interpolations. We compare the 95\% CI for the testing AOD in Figure~\ref{fig: RDA_CI} and provide the CI coverage (CIC) for 95\% CI and 99\% CI in Table~\ref{tab:table1}. Although the BLR model achieved excellent CI coverage, the CIs it produced are remarkably wider than those obtained through stacking. This analysis demonstrates the effectiveness of our stacking algorithms in AOD interpolation, notably outperforming prevailing models in accuracy and uncertainty quantification with limited predictors, showing great potential for future remote sensing data interpolation.

\begin{table}[t]
\begin{center}
\caption{Comparative performance metrics for AOD interpolation. Best results are in \textbf{bold}.}
\scalebox{0.9}[0.9]{ 
\begin{tabular}{|c|c|c|c|c|c|c|} 
 \hline
 Method & RMSPE & MAE & R & 95\% CIC & 99\% CIC \\
 \hline
Bayesian linear regression & 0.019 & 0.0147 & 0.693 & \textbf{94.9\%} & \textbf{99.4\%} \\
Deep Learning (H2O)  & 0.0166 & 0.0127 & 0.792 & NA & NA \\
Gradient Boosting (H2O)  & 0.0161 & 0.0122 & 0.797 & NA & NA \\
Random Forest (H2O)  & 0.0153 & 0.0115 & 0.821 & NA & NA \\
Ensemble model (H2O)  & 0.0154 & 0.0116 & 0.816 & NA & NA \\
Stacking of means & \textbf{0.010} & \textbf{0.007} & \textbf{0.927} & 75.9\% & 87.0\%\\
Stacking of posterior density& \textbf{0.010} & \textbf{0.007} &  0.926 & 82.1\% & 93.9\% \\
\hline
\end{tabular}}
\label{tab:table1}
\end{center}
\end{table}

\section{Conclusion and future work}\label{sec: conclusion}

We devised Bayesian inference for geostatistical data using predictive stacking. We offer some theoretical insights into the inferential behavior of posterior distributions in fixed-domain or infill settings and explore inferential performance through simulations and analysis of AOD data. The empirical results reveal that Bayesian predictive stacking delivers predictions comparable to full Bayesian inference obtained using MCMC samples, but at significantly lower computational costs. Our proposed algorithms are implemented in parallel using efficient storage and, hence, comprise an efficient alternative to full Bayesian inference using MCMC samples. 
We can build upon our current framework to extend Bayesian stacking for multivariate geostatistics using conjugate matrix-variate normal-Wishart families \citep{zhang2021high} and conjugate exponential families for non-Gaussian data \citep{bradley2020jasa}. A recent article by \cite{pan2024bayesian} extends the current current framework to spatial-temporal generalized linear models and another by \cite{presicceBanerjee2025geoAI} develops a Bayesian transfer learning framework for massive spatial datasets using predictive stacking. 
Stacked Bayesian inference for high-dimensional geostatistics \citep[e.g., building on the conjugate frameworks in][]{Banerjee20, zhang2019practical} is also possible as is the development of stacking methods for pooling inference across subsets of data \citep[e.g., as an alternative to meta-kriging described in][]{guhaniyogi2018a}. Finally, the R package \texttt{spStack} available for download from \url{https://cran.r-project.org/package=spStack} implements Bayesian predictive stacking for Gaussian and non-Gaussian data based on the methodology developed in this article.  

\bibliographystyle{ba}
\bibliography{stacking}  

\begin{thebibliography}{60}
\newcommand{\enquote}[1]{``#1''}
\expandafter\ifx\csname natexlab\endcsname\relax\def\natexlab#1{#1}\fi
\expandafter\ifx\csname url\endcsname\relax
  \def\url#1{{\tt #1}}\fi
\expandafter\ifx\csname urlprefix\endcsname\relax\def\urlprefix{URL }\fi
\ifx\endbibitem\undefined \let\endbibitem\relax\fi

\bibitem[{Abramowitz and Stegun(1965)}]{AS65}
Abramowitz, M. and Stegun, A. (1965).
\newblock {\em Handbook of {M}athematical {F}unctions: with {F}ormulas, {G}raphs, and {M}athematical {T}ables\/}.
\newblock Dover.
\endbibitem

\bibitem[{Abt(1999)}]{Abt99}
Abt, M. (1999).
\newblock \enquote{Estimating the prediction mean squared error in Gaussian stochastic processes with exponential correlation structure.}
\newblock {\em Scandinavian Journal of Statistics\/}, 26(4): 563--578.
\endbibitem

\bibitem[{Aguilera et~al.(2023)Aguilera, Luo, Basu, Wu, Clemesha, Gershunov, and Benmarhnia}]{aguilera2023novel}
Aguilera, R., Luo, N., Basu, R., Wu, J., Clemesha, R., Gershunov, A., and Benmarhnia, T. (2023).
\newblock \enquote{A novel ensemble-based statistical approach to estimate daily wildfire-specific {PM}2. 5 in {C}alifornia (2006--2020).}
\newblock {\em Environment International\/}, 171: 107719.
\endbibitem

\bibitem[{Andersen and Andersen(2000)}]{andersen2000mosek}
Andersen, E.~D. and Andersen, K.~D. (2000).
\newblock \enquote{{The MOSEK interior point optimizer for linear programming: an implementation of the homogeneous algorithm}.}
\newblock In {\em High performance optimization\/}, 197--232. Springer.
\endbibitem

\bibitem[{Banerjee(2019)}]{banerjee2019handbook}
Banerjee, S. (2019).
\newblock \enquote{Geostatistics for Environmental Processes.}
\newblock In Gelfand, A.~E., Fuentes, M., Hoeting, J.~A., and Smith, R.~L. (eds.), {\em Handbook of Environmental and Ecological Statistics\/}, 81--96. CRC press, Boca Raton, FL.
\endbibitem

\bibitem[{Banerjee(2020)}]{Banerjee20}
--- (2020).
\newblock \enquote{Modeling massive spatial datasets using a conjugate Bayesian linear modeling framework.}
\newblock {\em Spatial Statistics\/}, 37: 100417.
\endbibitem

\bibitem[{Banerjee et~al.(2014)Banerjee, Carlin, and Gelfand}]{banerjee2014hierarchical}
Banerjee, S., Carlin, B.~P., and Gelfand, A.~E. (2014).
\newblock {\em Hierarchical Modeling and Analysis for Spatial Data\/}.
\newblock CRC Press, Boca Raton, FL.
\endbibitem

\bibitem[{Berger et~al.(2001)Berger, Oliveira, and Sansó}]{berger2001}
Berger, J.~O., Oliveira, V.~D., and Sansó, B. (2001).
\newblock \enquote{Objective {B}ayesian Analysis of Spatially Correlated Data.}
\newblock {\em Journal of the American Statistical Association\/}, 96(456): 1361--1374.
\newline\urlprefix\url{https://doi.org/10.1198/016214501753382282}
\endbibitem

\bibitem[{Bose et~al.(2018)Bose, Hodges, and Banerjee}]{bose2018biocs}
Bose, M., Hodges, J.~S., and Banerjee, S. (2018).
\newblock \enquote{Toward a diagnostic toolkit for linear models with {G}aussian-process distributed random effects.}
\newblock {\em Biometrics\/}, 74(3): 863--873.
\newline\urlprefix\url{https://onlinelibrary.wiley.com/doi/abs/10.1111/biom.12848}
\endbibitem

\bibitem[{Bradley et~al.(2020)Bradley, Holan, and Wikle}]{bradley2020jasa}
Bradley, J.~R., Holan, S.~H., and Wikle, C.~K. (2020).
\newblock \enquote{Bayesian Hierarchical Models With Conjugate Full-Conditional Distributions for Dependent Data From the Natural Exponential Family.}
\newblock {\em Journal of the American Statistical Association\/}, 115(532): 2037--2052.
\newline\urlprefix\url{https://doi.org/10.1080/01621459.2019.1677471}
\endbibitem

\bibitem[{Breiman(1996)}]{breiman1996stacked}
Breiman, L. (1996).
\newblock \enquote{Stacked regressions.}
\newblock {\em Machine learning\/}, 24(1): 49--64.
\endbibitem

\bibitem[{Breiman(2001)}]{breiman2001random}
--- (2001).
\newblock \enquote{Random forests.}
\newblock {\em Machine learning\/}, 45: 5--32.
\endbibitem

\bibitem[{Candel et~al.(2016)Candel, Parmar, LeDell, and Arora}]{candel2016deep}
Candel, A., Parmar, V., LeDell, E., and Arora, A. (2016).
\newblock \enquote{Deep learning with H2O.}
\newblock {\em H2O. ai Inc\/}, 1--21.
\endbibitem

\bibitem[{Chil\'es and Delfiner(1999)}]{Chiles1999}
Chil\'es, J. and Delfiner, P. (1999).
\newblock {\em Geostatistics: Modeling Spatial Uncertainty\/}.
\newblock John Wiley: New York.
\endbibitem

\bibitem[{Clyde and Iversen(2013)}]{clyde2013bayesian}
Clyde, M. and Iversen, E.~S. (2013).
\newblock \enquote{Bayesian model averaging in the M-open framework.}
\newblock {\em Bayesian theory and applications\/}, 14(4): 483--498.
\endbibitem

\bibitem[{Cook(2016)}]{cook2016practical}
Cook, D. (2016).
\newblock {\em Practical machine learning with H2O: powerful, scalable techniques for deep learning and AI\/}.
\newblock " O'Reilly Media, Inc.".
\endbibitem

\bibitem[{Cressie(1993)}]{cres93}
Cressie, N. (1993).
\newblock {\em Statistics for Spatial Data\/}.
\newblock Wiley-Interscience, New York, revised edition.
\endbibitem

\bibitem[{David M.~Blei and McAuliffe(2017)}]{blei2017vbreview}
David M.~Blei, A.~K. and McAuliffe, J.~D. (2017).
\newblock \enquote{Variational Inference: A Review for Statisticians.}
\newblock {\em Journal of the American Statistical Association\/}, 112(518): 859--877.
\newline\urlprefix\url{https://doi.org/10.1080/01621459.2017.1285773}
\endbibitem

\bibitem[{de~Jonge and van Zanten(2013)}]{de2013semiparametric}
de~Jonge, R. and van Zanten, H. (2013).
\newblock \enquote{{Semiparametric Bernstein–von Mises for the error standard deviation}.}
\newblock {\em Electronic Journal of Statistics\/}, 7(none): 217 -- 243.
\newline\urlprefix\url{https://doi.org/10.1214/13-EJS768}
\endbibitem

\bibitem[{De~Oliveira and Han(2022)}]{oliveira2022jabes}
De~Oliveira, V. and Han, Z. (2022).
\newblock \enquote{On Information About Covariance Parameters in {G}aussian Matern Random Fields.}
\newblock {\em Journal of Agricultural, Biological and Environmental Statistics\/}, 27: 690--712.
\endbibitem

\bibitem[{Diggle and Ribeiro(2007)}]{diggle2007springer}
Diggle, P. and Ribeiro, P. (2007).
\newblock {\em Model-based Geostatistics\/}.
\newblock Springer.
\endbibitem

\bibitem[{Fan and Sun(2023)}]{fan2023satellite}
Fan, Y. and Sun, L. (2023).
\newblock \enquote{Satellite Aerosol Optical Depth Retrieval Based on Fully Connected Neural Network (FCNN) and a Combine Algorithm of Simplified Aerosol Retrieval Algorithm and Simplified and Robust Surface Reflectance Estimation (SREMARA).}
\newblock {\em IEEE J. Sel. Top. Appl. Earth Obs.\/}.
\endbibitem

\bibitem[{Finley et~al.(2019)Finley, Datta, Cook, Morton, Andersen, and Banerjee}]{finley2019efficient}
Finley, A.~O., Datta, A., Cook, B.~C., Morton, D.~C., Andersen, H.~E., and Banerjee, S. (2019).
\newblock \enquote{Efficient algorithms for {B}ayesian Nearest Neighbor Gaussian Processes.}
\newblock {\em Journal of Computational and Graphical Statistics\/}, 28(2): 401--414.
\endbibitem

\bibitem[{Friedman(2001)}]{friedman2001greedy}
Friedman, J.~H. (2001).
\newblock \enquote{Greedy function approximation: a gradient boosting machine.}
\newblock {\em Annals of statistics\/}, 1189--1232.
\endbibitem

\bibitem[{Gaudard et~al.(1999)Gaudard, Karson, Linder, and Sinha}]{gaudard1999bayesian}
Gaudard, M., Karson, M., Linder, E., and Sinha, D. (1999).
\newblock \enquote{Bayesian spatial prediction.}
\newblock {\em Environmental and Ecological Statistics\/}, 6(2): 147--171.
\endbibitem

\bibitem[{Golub and Van~Loan(2013)}]{GolubLoanMatrix4}
Golub, G.~H. and Van~Loan, C.~F. (2013).
\newblock {\em Matrix Computations - 4th Edition\/}.
\newblock Philadelphia, PA: Johns Hopkins University Press, 4th edition.
\newline\urlprefix\url{https://epubs.siam.org/doi/abs/10.1137/1.9781421407944}
\endbibitem

\bibitem[{Guhaniyogi and Banerjee(2018)}]{guhaniyogi2018a}
Guhaniyogi, R. and Banerjee, S. (2018).
\newblock \enquote{Meta-Kriging: Scalable {B}ayesian Modeling and Inference for Massive Spatial Datasets.}
\newblock {\em Technometrics\/}, 60(4): 430--444.
\endbibitem

\bibitem[{Handcock and Stein(1993)}]{handcock1993technometrics}
Handcock, M.~S. and Stein, M.~L. (1993).
\newblock \enquote{A {B}ayesian Analysis of Kriging.}
\newblock {\em Technometrics\/}, 35(4): 403--410.
\newline\urlprefix\url{https://www.tandfonline.com/doi/abs/10.1080/00401706.1993.10485354}
\endbibitem

\bibitem[{Hodges(2013)}]{hodges2013book}
Hodges, J.~S. (2013).
\newblock {\em Richly Parameterized Linear Models: Additive, Time Series, and Spatial Models Using Random Effects\/}.
\newblock Chapman and Hall/CRC, Boca Raton, FL.
\endbibitem

\bibitem[{Hoeting et~al.(1999)Hoeting, Madigan, Raftery, and Volinsky}]{hoeting1999bma}
Hoeting, J.~A., Madigan, D., Raftery, A.~E., and Volinsky, C.~T. (1999).
\newblock \enquote{{Bayesian model averaging: a tutorial (with comments by M. Clyde, David Draper and E. I. George, and a rejoinder by the authors}.}
\newblock {\em Statistical Science\/}, 14(4): 382 -- 417.
\newline\urlprefix\url{https://doi.org/10.1214/ss/1009212519}
\endbibitem

\bibitem[{Kaufman and Shaby(2013)}]{KS13}
Kaufman, C.~G. and Shaby, B.~A. (2013).
\newblock \enquote{The role of the range parameter for estimation and prediction in geostatistics.}
\newblock {\em Biometrika\/}, 100(2): 473--484.
\endbibitem

\bibitem[{Kazianka and Pilz(2012)}]{kazianka2012objective}
Kazianka, H. and Pilz, J. (2012).
\newblock \enquote{Objective {B}ayesian analysis of spatial data with uncertain nugget and range parameters.}
\newblock {\em Canadian Journal of Statistics\/}, 40(2): 304--327.
\endbibitem

\bibitem[{Kitanidis(1986)}]{kitanidis1986parameter}
Kitanidis, P.~K. (1986).
\newblock \enquote{Parameter uncertainty in estimation of spatial functions: {B}ayesian analysis.}
\newblock {\em Water resources research\/}, 22(4): 499--507.
\endbibitem

\bibitem[{Le and Clarke(2017)}]{le2017bayes}
Le, T. and Clarke, B. (2017).
\newblock \enquote{A {B}ayes interpretation of stacking for M-complete and M-open settings.}
\newblock {\em Bayesian Analysis\/}, 12(3): 807--829.
\endbibitem

\bibitem[{Li et~al.(2023)Li, Sun, and Zhu}]{li2023jasa}
Li, C., Sun, S., and Zhu, Y. (2023).
\newblock \enquote{Fixed-domain Posterior Contraction Rates for Spatial {G}aussian Process Model with Nugget.}
\newblock {\em Journal of the American Statistical Association\/}, 0(ja): 1--21.
\newline\urlprefix\url{https://doi.org/10.1080/01621459.2023.2191380}
\endbibitem

\bibitem[{Li et~al.(2009)Li, Zhao, Kahn, Mishchenko, Remer, Lee, Wang, Laszlo, Nakajima, and Maring}]{li2009uncertainties}
Li, Z., Zhao, X., Kahn, R., Mishchenko, M., Remer, L., Lee, K.-H., Wang, M., Laszlo, I., Nakajima, T., and Maring, H. (2009).
\newblock \enquote{Uncertainties in satellite remote sensing of aerosols and impact on monitoring its long-term trend: a review and perspective.}
\newblock In {\em Annales Geophysicae\/}, volume~27, 2755--2770. Copernicus GmbH.
\endbibitem

\bibitem[{Lyapustin et~al.(2018)}]{lyapustin2018modis}
Lyapustin, A. et~al. (2018).
\newblock \enquote{MODIS collection 6 MAIAC algorithm.}
\newblock {\em Atmos. Meas. Tech.\/}, 11(10): 5741--5765.
\endbibitem

\bibitem[{Madigan et~al.(1996)Madigan, Raftery, Volinsky, and Hoeting}]{madigan1996bayesian}
Madigan, D., Raftery, A.~E., Volinsky, C., and Hoeting, J. (1996).
\newblock \enquote{Bayesian model averaging.}
\newblock In {\em Proceedings of the AAAI Workshop on Integrating Multiple Learned Models, Portland, OR\/}, 77--83.
\endbibitem

\bibitem[{Pan et~al.(2024)Pan, Zhang, Bradley, and Banerjee}]{pan2024bayesian}
Pan, S., Zhang, L., Bradley, J.~R., and Banerjee, S. (2024).
\newblock \enquote{{Bayesian Inference for Spatial-temporal Non-{G}aussian Data Using Predictive Stacking}.}
\newblock {\em arXiv preprint arXiv:2406.04655\/}.
\endbibitem

\bibitem[{Presicce and Banerjee(2025)}]{presicceBanerjee2025geoAI}
Presicce, L. and Banerjee, S. (2025).
\newblock \enquote{Bayesian Transfer Learning for Artificially Intelligent Geospatial Systems: A Predictive Stacking Approach.}
\newline\urlprefix\url{https://arxiv.org/abs/2410.09504}
\endbibitem

\bibitem[{Rasmussen and Williams(2006)}]{RW06}
Rasmussen, C.~E. and Williams, C. K.~I. (2006).
\newblock {\em Gaussian processes for machine learning\/}.
\newblock MIT Press, Cambridge, MA.
\endbibitem

\bibitem[{Ren et~al.(2011)Ren, Banerjee, Finley, and Hodges}]{renBanerjeeEtAl2011csda}
Ren, Q., Banerjee, S., Finley, A.~O., and Hodges, J.~S. (2011).
\newblock \enquote{Variational {B}ayesian methods for spatial data analysis.}
\newblock {\em Computational Statistics \& Data Analysis\/}, 55(12): 3197--3217.
\newline\urlprefix\url{https://www.sciencedirect.com/science/article/pii/S0167947311002003}
\endbibitem

\bibitem[{Ribeiro~Jr et~al.(2007)Ribeiro~Jr, Diggle, Ribeiro~Jr, and Suggests}]{ribeiro2007geor}
Ribeiro~Jr, P.~J., Diggle, P.~J., Ribeiro~Jr, M. P.~J., and Suggests, M. (2007).
\newblock \enquote{The geo{R} package.}
\newblock {\em R news\/}, 1(2): 14--18.
\endbibitem

\bibitem[{Robert and Casella(1999)}]{robertcasella}
Robert, C.~P. and Casella, G. (1999).
\newblock {\em Monte Carlo Statistical Methods\/}.
\newblock Springer New York, NY.
\endbibitem

\bibitem[{Rue et~al.(2009)Rue, Martino, and Chopin}]{inla2009}
Rue, H., Martino, S., and Chopin, N. (2009).
\newblock \enquote{{Approximate {B}ayesian Inference for Latent {G}aussian models by using Integrated Nested Laplace Approximations}.}
\newblock {\em Journal of the Royal Statistical Society Series B: Statistical Methodology\/}, 71(2): 319--392.
\newline\urlprefix\url{https://doi.org/10.1111/j.1467-9868.2008.00700.x}
\endbibitem

\bibitem[{Stein(1999)}]{Steinbook}
Stein, M.~L. (1999).
\newblock {\em Interpolation of {S}patial {D}ata: {S}ome {T}heory for {K}riging\/}.
\newblock Springer-Verlag, New York.
\endbibitem

\bibitem[{Tang et~al.(2021)Tang, Zhang, and Banerjee}]{TZB2021}
Tang, W., Zhang, L., and Banerjee, S. (2021).
\newblock \enquote{On Identifiability and Consistency of The Nugget in Gaussian Spatial Process Models.}
\newblock {\em Journal of the Royal Statistical Society Series B: Statistical Methodology\/}, 83(5): 1044--1070.
\newline\urlprefix\url{https://doi.org/10.1111/rssb.12472}
\endbibitem

\bibitem[{Vehtari et~al.(2016)Vehtari, Mononen, Tolvanen, Sivula, and Winther}]{vehtari2016bayesian}
Vehtari, A., Mononen, T., Tolvanen, V., Sivula, T., and Winther, O. (2016).
\newblock \enquote{Bayesian leave-one-out cross-validation approximations for Gaussian latent variable models.}
\newblock {\em The Journal of Machine Learning Research\/}, 17(1): 3581--3618.
\endbibitem

\bibitem[{Voiland(2010)}]{voiland2010aerosols}
Voiland, A. (2010).
\newblock \enquote{Aerosols: Tiny particles, big impact.}
\newblock {\em NASA Earth Observatory\/}, 2.
\endbibitem

\bibitem[{Wolpert(1992)}]{wolpert1992stacked}
Wolpert, D.~H. (1992).
\newblock \enquote{Stacked generalization.}
\newblock {\em Neural networks\/}, 5(2): 241--259.
\endbibitem

\bibitem[{Yao et~al.(2021)Yao, Pir$\v{s}$, Vehtari, and Gelman}]{yao21}
Yao, Y., Pir$\v{s}$, G., Vehtari, A., and Gelman, A. (2021).
\newblock \enquote{Bayesian hierarchical stacking: {S}ome models are (somewhere) useful.}
\newblock {\em Bayesian Analysis\/}, 1(1): 1--29.
\endbibitem

\bibitem[{Yao et~al.(2022)Yao, Vehtari, and Gelman}]{yao2020stacking}
Yao, Y., Vehtari, A., and Gelman, A. (2022).
\newblock \enquote{Stacking for Non-mixing Bayesian Computations: The Curse and Blessing of Multimodal Posteriors.}
\newblock {\em Journal of Machine Learning Research\/}, 23(79): 1--45.
\newline\urlprefix\url{http://jmlr.org/papers/v23/20-1426.html}
\endbibitem

\bibitem[{Yao et~al.(2018)Yao, Vehtari, Simpson, and Gelman}]{yao2018using}
Yao, Y., Vehtari, A., Simpson, D., and Gelman, A. (2018).
\newblock \enquote{Using stacking to average {B}ayesian predictive distributions (with discussion).}
\newblock {\em Bayesian Analysis\/}, 13(3): 917--1007.
\endbibitem

\bibitem[{Zhang(2004)}]{Zhang04}
Zhang, H. (2004).
\newblock \enquote{Inconsistent estimation and asymptotically equal interpolations in model-based geostatistics.}
\newblock {\em Journal of the American Statistical Association\/}, 99(465): 250--261.
\endbibitem

\bibitem[{Zhang and Zimmerman(2005)}]{zhang2005towards}
Zhang, H. and Zimmerman, D.~L. (2005).
\newblock \enquote{Towards reconciling two asymptotic frameworks in spatial statistics.}
\newblock {\em Biometrika\/}, 92(4): 921--936.
\endbibitem

\bibitem[{Zhang and Banerjee(2022)}]{ZB21}
Zhang, L. and Banerjee, S. (2022).
\newblock \enquote{{Spatial factor modeling: A Bayesian matrix-normal approach for misaligned data}.}
\newblock {\em Biometrics\/}, 78(2): 560--573.
\endbibitem

\bibitem[{Zhang et~al.(2021)Zhang, Banerjee, and Finley}]{zhang2021high}
Zhang, L., Banerjee, S., and Finley, A.~O. (2021).
\newblock \enquote{High-dimensional multivariate geostatistics: {A} {B}ayesian matrix-normal approach.}
\newblock {\em Environmetrics\/}, 32(4): e2675.
\endbibitem

\bibitem[{Zhang et~al.(2019)Zhang, Datta, and Banerjee}]{zhang2019practical}
Zhang, L., Datta, A., and Banerjee, S. (2019).
\newblock \enquote{Practical {B}ayesian modeling and inference for massive spatial data sets on modest computing environments.}
\newblock {\em Statistical Analysis and Data Mining: The ASA Data Science Journal\/}, 12(3): 197--209.
\endbibitem

\bibitem[{Zimmerman and Cressie(1992)}]{ZC92}
Zimmerman, D. and Cressie, N. (1992).
\newblock \enquote{Mean squared prediction error in the spatial linear model with estimated covariance parameters.}
\newblock {\em Annals of the Institute of Statistical Mathematics\/}, 44: 27--43.
\endbibitem

\bibitem[{Zimmerman and Stein(2010)}]{zimmerman2010handbook}
Zimmerman, D. and Stein, M. (2010).
\newblock \enquote{Classical Geostatistical Methods.}
\newblock In Gelfand, A.~E., Diggle, P., Guttorp, P., and Fuentes, M. (eds.), {\em Handbook of spatial statistics\/}, 29--44. CRC press, Boca Raton, FL.
\endbibitem

\end{thebibliography}

\clearpage

\appendix

\section{Limiting behavior of Lemma~\ref{lem:structure} as $\delta^2 \rightarrow 0$}\label{sec: deltasq0limit}
We here present an approach that avoids singular matrices in \eqref{eq:posterior} in Lemma~\ref{lem:structure} when $\delta^2 = 0$. For brevity, we simplify the notation in this section by letting $R$ and $I$ denote $R_{\Phi}(\chi)$ and $I_n$, respectively.
Using familiar matrix identities, the $M_\ast$ in \eqref{eq:posterior} can be formulated as
\begin{align*}
 M_\ast = \begin{bmatrix} E & -E X D^{-1}\\
 -D^{-1} X E & \delta^2 D^{-1} + D^{-1} X E X^\T D^{-1}
 \end{bmatrix}\;,
\end{align*}
where $E = \{X^\T(\delta^2 I + R)^{-1}X + V_\beta^{-1}\}^{-1}$, $D = (I+\delta^{2}R^{-1})$. Checking $\hat{\gamma} = M_{\ast}X_*^{\T}V_{*}^{-1}y_*$, we have
\begin{align*}
  \hat{\gamma} &= \left\{\begin{bmatrix}
      0 & 0\\
      0 & \delta^2 D^{-1}\\
  \end{bmatrix}  + \begin{bmatrix}
      I\\ -D^{-1} X
  \end{bmatrix} E \begin{bmatrix}
      I& -X^\T D^{-1}
  \end{bmatrix}\right\}
  \begin{bmatrix}
      \delta^{-2}X^\T y + V_\beta^{-1}\mu_\beta \\
      \delta^{-2} y 
  \end{bmatrix} =\begin{bmatrix}
      B \\ D^{-1}(y - X B)
  \end{bmatrix}\;,
\end{align*}
where $B = E\{\delta^{-2}X^\T (I - D^{-1}) y + V_\beta^{-1}\mu_\beta\}$. Since $\delta^{-2}(I - D^{-1}) = \delta^{-2}\{I - (I + \delta^2 R^{-1})^{-1}\} = (\delta^2 I + R)^{-1}$, we can avoid having $\delta^{-2}$ in the formulation of the posterior distribution \eqref{eq:posterior} by letting $B = E\{X^\T (\delta^2 I + R)^{-1} y + V_\beta^{-1}\mu_\beta\}$, $b_\sigma^\ast = b_\sigma + 0.5\{(y - XB)^\T (I - D^{-1}) (\delta^2 + R)^{-1}(y - XB) + (\mu_\beta - B)^\T V_\beta^{-1}(\mu_\beta - B) + (y - XB)^\T D^{-1} R^{-1} D^{-1} (y - XB)\}$. As $\delta^2$ goes to zero, 
\begin{align*}
    E &\rightarrow \{X^\T R^{-1}X + V_\beta^{-1}\}^{-1}\quad B \rightarrow  E\{X^\T R^{-1} y + V_\beta^{-1}\mu_\beta\} \quad 
    D^{-1} \rightarrow I \\
    b_\sigma^\ast &\rightarrow b_\sigma + 0.5\{(\mu_\beta - B)^\T V_\beta^{-1}(\mu_\beta - B) + (y - XB)^\T R^{-1} (y - XB)\}\;.
\end{align*}
For model without nugget ($\delta^2 = 0$), the conjugate Bayesian hierarchical spatial model is constructed as
\begin{equation}
\begin{aligned}
 y \;|\; &\beta, \sigma^2  \sim \mathcal{N}(X \beta, \sigma^2 R), \quad
 \beta \;|\; \sigma^2 \sim \mathcal{N}(\mu_\beta, \sigma^2 V_{\beta}), \quad \sigma^2 \sim \mbox{IG}(a_\sigma, b_{\sigma}).
\end{aligned}
\end{equation}
The corresponding posterior distribution 
is
$p(\beta, \sigma^2 \,|\, y) = \underbrace{\mbox{IG}(\sigma^2 \,|\, \Tilde{a}_\sigma, \Tilde{b}_\sigma)}_{p(\sigma^2 \,|\, y)} \times 
\underbrace{\mathcal{N}(\beta \,|\, \hat{\beta}, \sigma^2 \Tilde{M})}_{p(\beta \,|\, \sigma^2, y)}$ where $\Tilde{a}_\sigma, \Tilde{b}_\sigma, \hat{\beta}, \Tilde{M}$ coincide with $a^\ast_\sigma, b^\ast_\sigma, B, E$, respectively, when $\delta^2$ reaches zero. A few algebraic simplifications reveal that the correct posterior distribution $p(\beta,\sigma^2 \mid y)$ for the model without the nugget is obtained using the limiting argument.

\section{Proof of Theorem~\ref{thm:main2}}\label{sec: thm:main2}

\begin{customthm}{\ref{thm:main2}}
Assume that the location set $\chi = \{s_1, \ldots, s_n\}$ satisfies the infill condition:
\begin{equation}
\label{eq:spacecond}
\max_{s \in \mathcal{D}} \min_{1 \le i \le n} |s - s_i| \asymp n^{-\frac{1}{d}}.
\end{equation}
Let $\mathbb{P}_0$ be the probability distribution of the Mat\'ern model \eqref{eq:Matern} with $(\sigma_0^2, \phi_0, \tau_0^2)$. 
Under $\mathbb{P}_0$, 
\begin{equation}
\tag{3.5}
\lim_{n\to\infty} p(\sigma^2 \given y) = \texttt{Dirac}(\tau_0^2/\delta^2)\;, \mbox{ and } \lim_{n\to\infty} p(\tau^2 \given y) = \texttt{Dirac}(\tau_0^2)
\end{equation}
where $y = (y(s_1),y(s_2),\ldots,y(s_n))^{\T}$,
and $\texttt{Dirac}(\cdot)$ denotes the Dirac mass point.
\end{customthm}

The proof of this theorem breaks into the following lemmas. Recall the definition of $b_\sigma^{*} = b_{\sigma,n}^{*}$ from Lemma \ref{lem:structure}. We derive a simple expression for $b_{\sigma,n}^{*}$, which is specific to the conjugate model 
\eqref{eq:BayesianG2}.

\begin{lemma}\label{lem:bsigsimple}
We have 
$b_{\sigma,n}^* = b_\sigma + \frac{1}{2} y^\T (\delta^2 I_n + R_\phi(\chi))^{-1} y$.
\end{lemma}
\begin{proof}
Note that 
\begin{equation}
\label{eq:newMm}
M_*^{-1} = X_*^\T V_{y_*}^{-1} X_* = \delta^{-2} I_n + R^{-1}_\phi(\chi), \quad
X_*^\T V_{*}^{-1} y_* = \delta^{-2} y \quad
\mbox{and} \quad y_*^\T V_{*}^{-1} y_* = \delta^{-2} y^\T y.
\end{equation}
By the Woodbury matrix identity, we can simplify 
\begin{align*}
(y_*-X_{*}\hat{\gamma})^{\T}V_*^{-1}(y_* - X_*\hat{\gamma}) &=\delta^{-2} y^\T y - \delta^{-2} y^\T(I_n + \delta^{2}R_\phi^{-1}(\chi))^{-1} y \\
&= y_*^\T (\delta^2I_n+ R_\phi(\chi))^{-1} y,
\end{align*}
which yields the desired result.
\end{proof}

The next lemma investigates the asymptotic behavior of $b_\sigma^*$ when the range decay $\phi = \phi_0$, i.e., $\phi$ is fixed at the value generating the data.

\begin{lemma}\label{lem:bstarasymp}
Let $\phi = \phi_0$, and assume that 
$\max_{s \in \mathcal{D}} \min_{1 \le i \le n} |s - s_i| \asymp n^{-\frac{1}{d}}$.
Then
\begin{equation}
\label{eq:blimit}
\frac{b_{\sigma,n}^* - b_{\sigma}}{n} \longrightarrow \frac{\tau_0^2}{2 \delta^2}, \quad \mathbb{P}_0\mbox{-almost surely}.
\end{equation}
\end{lemma}

\begin{proof}
Let $Q_n$ be the orthogonal matrix such that 
$Q_{n} R_{\phi_0}(\chi) Q_n^\T = 
\begin{pmatrix}
\lambda_1^{(n)} &  & \\
 & \ddots & \\
 & & &\lambda_n^{(n)}
\end{pmatrix}$,
where $\lambda_i^{(n)}$ is the $i$-th largest eigenvalue of matrix $R_{\phi_0}(\chi)$.
Thus, under $\mathbb{P}_0$, 
\begin{equation*}
Q_n y \sim \mathcal{N}\left(0, \begin{pmatrix}
\sigma_0^2 \lambda_1^{(n)} + \tau_0^2 &  & \\
 & \ddots & \\
 & & &\sigma_0^2 \lambda_n^{(n)} + \tau_0^2
\end{pmatrix}\right).
\end{equation*}
By Lemma \ref{lem:bsigsimple}, we get
\begin{equation}
\label{eq: sigma_post_b}
2(b_{\sigma,n}^* - b_\sigma) = \sum_{i = 1}^n \frac{\sigma_0^2 \lambda_i^{(n)} + \tau_0^2}{\lambda_i^{(n)} + \delta^2} u_i^2,
\end{equation}
where $u_i  \overset{i.i.d.}{\sim} \mathcal{N}(0, 1)$ for $i = 1, \ldots, n$.
By \citet[Corollary 2]{TZB2021}, there exists $C > 0$ independent of $n$ such that 
$\lambda_i^{(n)} \le C n i^{-\frac{2 \nu}{d} - 1}$ for all $1 \le i \le n$.
This implies that 
\begin{equation}
\label{eq:sumasymp}
\sum_{i = 1}^n \frac{\sigma_0^2 \lambda_i^{(n)} + \tau_0^2}{\lambda_i^{(n)} + \delta^2} \sim \frac{n \tau_0^2}{\delta^2} \quad \mbox{as } n \to \infty.
\end{equation}
By the law of large numbers, \eqref{eq:blimit} follows from \eqref{eq: sigma_post_b} and \eqref{eq:sumasymp}.
\end{proof}

\begin{proof}[Proof of Theorem \ref{thm:main2}]
Let $\sigma'^2: =\sigma_0^2 \phi_0^{2 \nu}/\phi^{2 \nu}$, 
and let $\mathbb{P}'$ be the probability distribution of the Mat\'ern model with parameters $(\sigma'^2, \phi, \tau_0^2)$.
By \citet[Theorem 1]{TZB2021}, $\mathbb{P}'$ is equivalent to $\mathbb{P}$.
Further by Lemma \ref{lem:bstarasymp}, 
\begin{equation*}
\frac{b_{\sigma,n}^* - b_\sigma}{n} \longrightarrow \frac{\tau_0^2}{2 \delta^2}, \quad \mathbb{P}'\mbox{-almost surely}.
\end{equation*}
which also holds $\mathbb{P}_0$-almost surely.
By Lemma \ref{lem:structure},
\begin{equation}
\mathbb{E}_0(\sigma^2 \,|\, y) = \frac{b_{\sigma,n}^*}{a_{\sigma,n}^*} \sim \frac{\tau_0^2}{\delta^2} \quad
\mbox{and} \quad
\mathbb{V}_0(\sigma^2 \,|\, y) = \frac{b_{\sigma,n}^{*2}}{(a_{\sigma,n}^*-1)^2 (a_{\sigma,n}^*-2)} \asymp \frac{1}{n},
\end{equation}
which yields \eqref{eq:postin} by Chebyshev's inequality and the posterior consistency of $\tau^2 = \delta^2 \sigma^2$.
\end{proof}
\section{Posterior predictive consistency for the Mat\'ern model}\label{sec: prop:main2bis}

\begin{customthm}{\ref{prop:main2bis}}
Let $s_0 \in \mathcal{D}$. 
For any given $\phi>0$, denote $\cov(z, z(s_0) \given \sigma^2)$ and $R_{\phi}(\chi)$ by $\sigma^2 J_{\phi, n}$ and $R_{\phi, n}$, respectively.  
Then we have the decomposition 
\begin{equation}
\tag{3.6}
\mathbb{E}_0(Z_n(s_0) - z(s_0))^2 = E_{1,n} + E_{2,n} + o(1),
\end{equation}
where $E_{1, n}$ is the prediction error of the best linear predictor for a Mat\'ern model with parameters $\{\sigma'^2, \phi, \tau'^2\}$ satisfying $\delta^2 = \frac{\tau'^2}{\sigma'^2}$,
and 
\begin{equation}
\tag{3.7}
E_{2, n}: 
=\frac{\tau_0^2}{\delta^2}\left[1 - J_{\phi, n}^\T (\delta^{2}I_n + R_{\phi, n})^{-1} J_{\phi, n} \right]
\end{equation}
\end{customthm}
\begin{proof}
By \eqref{eq:newMm}, we have
$$p(z \,|\, y, \sigma^2) = \mathcal{N}((I_n + \delta^2 R_\phi(\chi)^{-1})^{-1} y, \sigma^2(\delta^{-2} I_n +  R_\phi(\chi)^{-1})^{-1}).$$
Combining with \eqref{eq:ppmeanpf}, we get the posterior predictive mean
\begin{align*}
\mathbb{E}(z(s_0) \,|\, y) &= J_{\phi, n}^\T R_\phi(\chi)^{-1} \mathbb{E}(z \,|\, y) \\
& = J_{\phi, n}^\T R_\phi(\chi)^{-1}(I_n + \delta^2 R_\phi(\chi)^{-1})^{-1} y
= J_{\phi, n}^\T (\delta^2 I_n + R_\phi(\chi))^{-1} y,
\end{align*}
which is the best linear predictor corresponding to a Mat\'ern model with parameter values $\{\sigma'^2, \phi, \tau'^2\}$ satisfying $\delta^2 = \tau'^2/\sigma'^2$. Further by Theorem \ref{thm:main2}, the formula \eqref{eq:varerr} reduces to
\begin{equation*}
\mathbb{V}(z(s_0) \,|\, y) \to \frac{\tau_0^2}{\delta^2}(1 - J_{\phi, n}^\T R_\phi(\chi)^{-1} J_{\phi, n}) + \tau_0^2 J_{\phi, n}^\T (\delta^2 I_n + R_\phi(\chi))^{-1} R_\phi(\chi)^{-1} J_{\phi, n},
\end{equation*}
which gives \eqref{eq:E2}. 
The rest of the theorem easily follows. 
\end{proof}

\medskip
\noindent
{\bf The deviation error $E_{1, n}$}: It is expected that the prediction error of the best linear predictor of a Mat\'ern model in the presence of a nugget tends to $0$ as long as the fill distance condition \eqref{eq:spacecond} holds. However, it is hard to prove this statement. 

To provide some ideas, we consider a one-dimensional example in $\mathcal{D} = [-1,1]$. 
In particular, traditional paradigms for spatial asymptotics have relied on either fixed domain or expanding domain with attempts at reconciling the two paradigms \citep{zhang2005towards}. For a given integer $n$, let $\chi_n = \{i/n, \, -n \le i \le n\}$ and $\chi_{n, \infty} =\{i/n, \, -\infty < i < \infty\}$. 
Define 
$$\widehat{z}_{n}(0): = \mathbb{E}(z(0) \,|\, y(s), s \in \chi_n \setminus \{0\}, \phi, \delta^2) \quad \mbox{and} \quad 
\widehat{z}_{n, \infty}(0): = \mathbb{E}(z(0) \,|\, y(s), s \in \chi_{n,\infty} \setminus \{0\}, \phi, \delta^2)).$$ 
The corresponding prediction errors are:
$$e_n (0): = \mathbb{E}_0(z(0) - \widehat{z}_n(0))^2 \quad \mbox{and} \quad e_{n, \infty} (0): = \mathbb{E}_0(z(0) - \widehat{z}_{n, \infty}(0))^2, \mbox{ respectively}.$$
We conjecture that the difference between these two prediction errors will be negligible for large data, i.e.,  $e_n(0) - e_{n, \infty}(0) \to 0$, as $n \to \infty$
. If this is in place, the prediction error $e_n(0) \to 0$ provided $e_{n, \infty}(0) \to 0$ as $n \to \infty$, which we prove in Proposition~\ref{prop:errbd}.

\begin{proposition}
\label{prop:errbd}
We have $e_{n, \infty}(0) \to 0$ as $n \to \infty$.
\end{proposition}

\begin{proof}
To simplify the notation, we denote by $\Delta: = \frac{1}{n}$ the inter-spacing of the grid.
Let $f_0$ (resp. $f$) be the spectral density of the Mat\'ern model with true parameter values $\{\sigma_0^2, \phi_0, \tau_0^2\}$ and the smoothness parameter $\nu_0$ (resp. possibly misspecified parameter values $\{\sigma^2, \phi, \tau^2\}$ and the smoothness parameter $\nu$).
Define
\begin{equation*}
\widetilde{f}^\Delta_0(u) = \Delta^{-1} \sum_{k = -\infty}^{\infty} f_0\left( \frac{u + 2 \pi k}{\Delta} \right) \quad \mbox{and} \quad 
\widetilde{f}^\Delta(u) = \Delta^{-1} \sum_{k = -\infty}^{\infty} f\left( \frac{u + 2 \pi k}{\Delta} \right).
\end{equation*}
By \cite[Chapter 3, (13)]{Steinbook}, the prediction error of $y(0)$ based on $y(s)$, $s \in \chi_{n, \infty} \setminus \{0\}$ is 
\begin{equation*}
4\pi^2 \int_{-\pi}^{\pi} \frac{\widetilde{f}^\Delta_0(u)}{\widetilde{f}^\Delta(u)^2} du \bigg/ \left(\int_{-\pi}^{\pi}\widetilde{f}^{\Delta}(u)^{-1} du\right)^2,
\end{equation*}
and hence the prediction error of $z(0)$ based on $y(s)$, $s \in \chi_{n, \infty} \setminus \{0\}$ is 
\begin{equation*}
\label{eq:einfsd}
e_{n, \infty}(0) = \frac{4\pi^2 \int_{-\pi}^{\pi} \frac{\widetilde{f}^\Delta_0(u)}{\widetilde{f}^\Delta(u)^2} du}{\left(\int_{-\pi}^{\pi}\widetilde{f}^{\Delta}(u)^{-1} du\right)^2} - \tau_0^2.
\end{equation*}
The spectral density of the Mat\'ern model without the nugget is
\begin{equation*}
\label{eq:sdiMatern}
f_{\mbox{\tiny Mat\'ern}}(u) = C \frac{\sigma^2 \phi^{2 \nu}}{(\phi^2 + u^2)^{\nu + d/2}} \quad \mbox{for some } C > 0.
\end{equation*}
We write 
\begin{equation*}
\widetilde{f}_0^\Delta(u) = \sigma_0^2 g^\Delta_0(u) + \frac{\tau_0^2}{2 \pi} \quad \mbox{and} \quad
\widetilde{f}^\Delta(u) = \sigma^2 g^\Delta(u) + \frac{\tau^2}{2 \pi},
\end{equation*}
so that
\begin{equation}
\label{eq:erep}
e_{n, \infty}(0)
 = \frac{4\pi^2 \int_{-\pi}^{\pi} \frac{ \sigma_0^2 g^{\Delta}_0(u) + \frac{1}{2\pi}\tau_0^2}{(\sigma^2 g^{\Delta}(u) + \frac{1}{2\pi} \tau^2)^2} du}{\left(\int_{-\pi}^{\pi}(\sigma^2 g^{\Delta}(u) + \frac{1}{2\pi}\tau^2)^{-1} du\right)^2} - \tau_0^2
\end{equation}
Note that $g^\Delta(u) \sim c u^{-2 \nu - 1}$ for some $c > 0$.
It is known that (see e.g. \cite[Section 2.3]{TZB2021})
\begin{equation}
\label{eq:denom}
\int_{-\pi}^{\pi} \left(\sigma^2 g^{\Delta}(u) + \frac{1}{2\pi}\tau^2 \right)^{-1} du \sim \frac{4 \pi^2}{\tau^2 + C \Delta^{\frac{2 \nu}{2 \nu + 1}}} \quad \mbox{for some } C > 0.
\end{equation}
Furthermore, $g^\Delta(u) \asymp \Delta^{2 \nu}$ for $u$ large and $g^\Delta(u) \asymp \Delta^{-1}$ for $u$ small.
We prove that $e_\infty \to 0$ as $n \to \infty$.
Therefore, 
\begin{equation}
\label{eq:num}
 \int_{-\pi}^{\pi} \frac{ \sigma_0^2 g^{\Delta}_0(u) + \frac{1}{2\pi}\tau_0^2}{(\sigma^2 g^{\Delta}(u) + \frac{1}{2\pi} \tau^2)^2} du = \frac{4 \pi^2 \tau_0^2}{\tau^4} + \mathcal{O}(\Delta^{\min(2 \nu_0, 2 \nu,1)}).
\end{equation}
Combining \eqref{eq:erep}, \eqref{eq:denom} and \eqref{eq:num} yields
\begin{equation*}
e_{n, \infty}(0) =\left(\tau_0^2 + \mathcal{O}(\Delta^{\min(2 \nu_0, 2 \nu,1)})\right) \left(1 +  \frac{C \Delta^{\frac{2 \nu}{2 \nu + 1}} }{\tau^2} \right) - \tau_0^2 \asymp \Delta^{\min( 2 \nu_0, \frac{2 \nu}{2 \nu + 1})}.
\end{equation*}
Thus, we have $e_{n, \infty}(0) \to 0$ as $n \to \infty$
.
\end{proof}

\medskip
\noindent
{\bf The posterior variance $E_{2, n}$}: 
The error term $E_{2, n}$ defined by \eqref{eq:E2} is also analytically intractable. 
Here, we provide a numerical study to investigate the behavior of $E_{2, n}$ when $n \to \infty$ in the general case. We first generate the location set $\chi$ by uniformly sampling $n$ locations in $[0, 1]$ or $[0, 1]^2$, then we compute $E_{2, n}$ for every location in $\chi$ with $\delta^2 = 1, \tau^2 = 1$ and different values of $\phi$ and $\nu$. We expand the location set $\chi$ sequentially to sets with larger sample sizes by adding locations that are uniformly sampled in the study domain. For each expanded set $\chi$, we recompute the $E_{2, n}$ for all locations in $\chi$. Figure~\ref{fig: E2_check} plots the median, the 2.5th and 97.5th percentiles of $E_{2, n}$ for different sample sizes. The values of $E_{2, n}$ for points in a fixed domain, shown in Figure~\ref{fig: E2_check}, decrease rapidly as the sample size increases, although the rate diminishes when $d$ increases from $1$ to $2$. This suggests that the decreasing rate is related to the dimension.

\begin{figure}[t]
\centering
\includegraphics[width=0.49\textwidth, keepaspectratio]{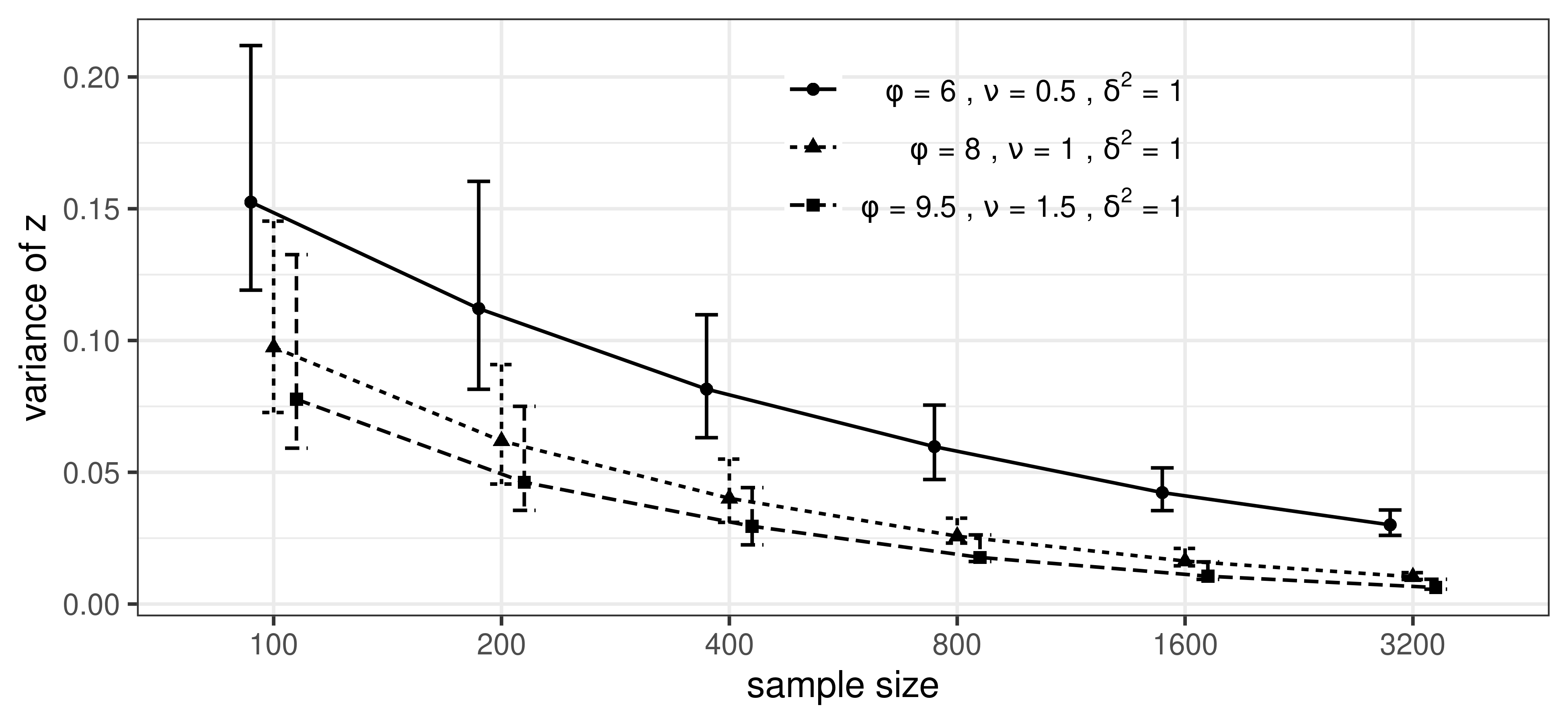}
\includegraphics[width=0.49\textwidth, keepaspectratio]{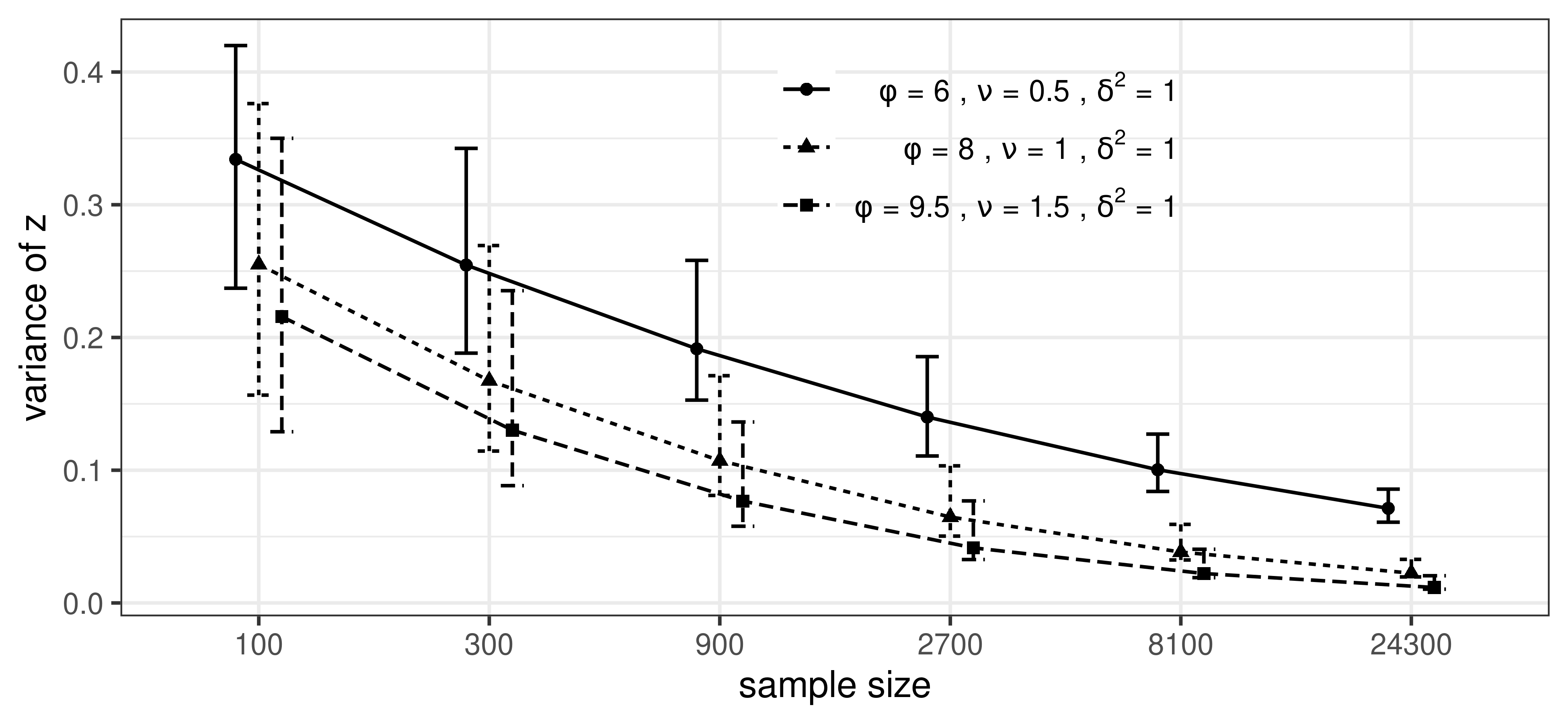}
\caption{The median of $E_{2,n}$ for locations uniformly sampled on $[0, 1]$ (left) and $[0, 1]^2$ (right). The error bars indicate the 97.5th and 2.5th percentiles. The sample size $n$ ranges from 100 to 3,200 and from 100 to 24,300 for the experiments on $[0, 1]$ and $[0, 1]^2$, respectively.}\label{fig: E2_check}
\end{figure}

\section{Proof of Proposition~\ref{prop:stackingpost1}}\label{sec: prop:stackingpost1}
\begin{customprop}{\ref{prop:stackingpost1}}
Let $s_0 \in \mathcal{D}$, and $w_g^*(y): = (w^*_1(y), \ldots, w^*_G(y))$ be the stacking weights 
(e.g. defined by \eqref{eq: weights_cal1})
such that $\mathbb{P}_0$ almost surely,
$$\sum_{g=1}^G w^*_g(y) = 1 \quad \mbox{and} \quad w^*_g(y) \ge 0  \mbox{ for each } 1 \le g \le G.$$
Recall that $E^g_{1,n}$ is the prediction error of the best linear predictor for model $\mathcal{M}_g$, and assume that $E^g_{1,n} \to 0$ as $n \to \infty$, for each model $\mathcal{M}_g$.
We have
\begin{equation}
\tag{3.8}
\mathbb{E}_0\left(y(s_0) - \sum_{g = 1}^G w^*_g(y) \mathbb{E}_g(y(s_0) \,|\, y) \right)^2 \to \tau_0^2 \quad \mbox{as } n \to \infty \;. 
\end{equation}
\end{customprop}
\begin{proof}
For ease of presentation, we give the proof in the setting of Theorem \ref{prop:main2bis}.
Recall that $y(s_0) = z(s_0) + \varepsilon(s_0)$.
We have:
\begin{equation}
\label{eq:CSbound}
\begin{aligned}
\mathbb{E}_0\left(y(s_0) - \sum_{g = 1}^G w^*_g \mathbb{E}_g(y(s_0) \,|\, y) \right)^2 
& = \mathbb{E}_0\left(\varepsilon(s_0) + z(s_0) - \sum_{g = 1}^G w^*_g(y) \mathbb{E}_g(z(s_0) \,|\, y)\right)^2  \\
& = \tau_0^2 + \mathbb{E}_0\left(z(s_0) - \sum_{g= 1}^G w^*_g(y) \mathbb{E}_g(z(s_0) \,|\, y) \right)^2 \\
&= \tau_0^2 + \mathbb{E}_0\left(\sum_{g= 1}^G w^*_g(y) \bigg(z(s_0) - \mathbb{E}_g(z(s_0) \,|\, y) \bigg) \right)^2 \\
& \le \tau_0^2 + \mathbb{E}_0 \left( \sum_{g=1}^G w^{*2}_g(y) \sum_{g = 1}^G \bigg(z(s_0) - \mathbb{E}_g(z(s_0) \,|\, y)\bigg)^2 \right) \\
& \le \tau_0^2 + G  \sum_{g = 1}^G \mathbb{E}_0 \bigg(z(s_0) - \mathbb{E}_g(z(s_0) \,|\, y)\bigg)^2, 
\end{aligned}
\end{equation}
where we use the fact that $\mathbb{E}_g(y(s_0)\,|\,y) = \mathbb{E}_g(z(s_0)\,|\,y)$ in the first equality (which is clear from the description below \eqref{eq: posterior_predictive_joint});
the second equality follows from the fact that 
$\varepsilon(s_0)$ and $(z(s_0), y)$ are independent;
the third equality is due to $\sum_{g = 1}^G w_g^*(y) = 1$, $\mathbb{P}_0$ a.s.;
we apply the Cauchy-Schwarz inequality in the fourth inequality;
and the final inequality is by the fact that 
$\sum_{g = 1}^G w_g^{*2}(y) \le G$, $\mathbb{P}_0$ a.s.
For each $1 \le g \le G$, $\mathbb{E}_0\bigg(z(s_0) - \mathbb{E}_g(z(s_0) \,|\, y)\bigg)^2$ corresponds to the deviation error $E_{1, n}$ for the model $\calM_g$, which goes to $0$ as $n \to \infty$.
The bound \eqref{eq:stackingerr} follows readily from \eqref{eq:CSbound}.
\end{proof}

\section{Proof of Theorem~\ref{prop:stackingpost2}}\label{sec: prop:stackingpost2}

\begin{customthm}{\ref{prop:stackingpost2}}
Let $s_0 \in \mathcal{D}$, and $w_g^*(y): = (w^*_1(y), \ldots, w^*_G(y))$ be the stacking weights 
(e.g. defined by \eqref{eq: weights_cal1})
such that $\mathbb{P}_0$ almost surely,
$$\sum_{g=1}^G w^*_g(y) = 1 \quad \mbox{and} \quad w^*_g(y) \ge 0  \mbox{ for each } 1 \le g \le G.$$
{
Assume that $E^g_{1,n} \to 0$ as $n \to \infty$, for each model $\mathcal{M}_g$.}
For $1 \le g \le G$ and $1 \le i \le n$, 
let 
$E^g_{1,n,i}: = \mathbb{E}_0(z(s_i) - \widehat{y}_g(s_i))^2$
be the deviation error for the latent process $z(s)$ by leaving the $i^{th}$ observation out under the model $\mathcal{M}_g$.
Assume that for each $1 \le g \le G$,
\begin{equation}
\label{eq:avezero_app}
\frac{1}{n}\sum_{i = 1}^n E^g_{1,n,i} \to 0 \quad \mbox{as } n \to \infty.
\end{equation}
Also let the assumptions in 
Theorem \ref{prop:main2bis} hold for each model $\calM_g$.
We have as $n \to \infty$,
\begin{equation}
\label{eq:asymplimit_app}
\mathbb{E}_0\left(y(s_0) - \sum_{g = 1}^G w^*_g(y) \mathbb{E}_g(y(s_0) \,|\, y) \right)^2
- \mathbb{E}_0\left( \frac{1}{n}\sum_{i = 1}^n \left(y(s_i) - \sum_{g = 1}^G w^*_g(y) \hat{y}_g(s_i)\right)^2\right) \to 0.
\end{equation}
\end{customthm}
\begin{proof}
By Proposition \ref{prop:stackingpost1}, it suffices to prove that
\begin{equation}
\label{eq:asymplimit2}
\mathbb{E}_0\left( \frac{1}{n}\sum_{i = 1}^n \left(y(s_i) - \sum_{g = 1}^G w^*_g(y) \hat{y}_g(s_i)\right)^2\right) \to \tau_0^2 \quad 
\mbox{as } n \to \infty.
\end{equation}
For ease of presentation, we prove \eqref{eq:asymplimit2} in the setting of Theorem \ref{prop:main2bis}.
Note that
\begin{align}
\label{eq:CSbound2}
\mathbb{E}_0\left( \frac{1}{n}\sum_{i = 1}^n \left(y(s_i) - \sum_{g = 1}^G w_g^*(y) \hat{y}_g(s_i)\right)^2\right)= \tau_0^2 + \underbrace{\frac{1}{n} \mathbb{E}_0 \sum_{i=1}^n \left(\sum_{g = 1}^G w^*_g(y) \left(z(s_i) - \widehat{y}_g(s_i)\right)\right)^2}_{(B)},
\end{align}
where $(B) \le \frac{G}{n}\sum_{i = 1}^n E^g_{1,n,i} \to 0$ 
by the condition \eqref{eq:avezero}. Hence, \eqref{eq:asymplimit2} follows from \eqref{eq:CSbound2}.
\end{proof}

Taking a closer look at \eqref{eq:CSbound2}, we can see that $B$ summarises the average squared prediction errors for $z(s_i), i = 1, \ldots, n$, in a LOO cross-validation. Hence, the stacking weights obtained from \eqref{eq: weights_cal1} minimises the average squared prediction errors for the latent process over the observed locations.

\section{Stacked predictive densities in the general setting}
\label{sc33}

We extend the discussion to the general conjugate Bayesian spatial model \eqref{eq:BayesianG} and establish asymptotic results of the posterior distribution obtained through stacking. For easier exposition, we refine the notation for \eqref{eq:posterior} by casting the spatial model in \eqref{eq:BayesianG} into an augmented linear system. Let $L_\beta$ be the Cholesky decomposition of $V_\beta$ such that $V_\beta = L_\beta L_\beta^\T$, and $L_\Phi$ a non-singular square matrix such that $R_{\Phi}^{-1}(\chi) = L_\Phi^\T L_\Phi$. We adapt \eqref{eq:BayesianG} as
\begin{equation}
\label{eq:spatialaug2}
\underbrace{\begin{bmatrix}
\frac{1}{\delta}y \\  L_\beta ^{-1} \mu_\beta \\ 0
\end{bmatrix}}_{y_{\dagger}}
= \underbrace{\begin{bmatrix}
\frac{1}{\delta} X & \frac{1}{\delta} I_n \\ 
L_\beta^{-1} & 0 \\ 
0 & L_\Phi 
\end{bmatrix}}_{X_{\dagger}} 
\underbrace{\begin{bmatrix}
\beta \\ z
\end{bmatrix}}_{\gamma} 
+ \underbrace{\begin{bmatrix}
\eta_{\dagger, 1} \\ \eta_{\dagger, 2} \\ \eta_{\dagger, 3}
\end{bmatrix}}_{\eta_{\dagger}},
\end{equation}
where $\eta_{\dagger} \sim \mathcal{N}(0, \sigma^2 I_{2n + p})$ 
{and $\sigma^2 \sim \mbox{IG}(a_\sigma, b_\sigma)$}. 
{This formulation simplifies the $V_\ast$ matrix to an identity matrix. The conclusions detailed in Section~\ref{sc21} are applicable with the replacement of $X_\ast$ and $y_\ast$ by $X_{\dagger}$ and $y_{\dagger}$.} We use the equivalence of probability measures to explore posterior concentrations within an in-fill paradigm.
\begin{assumption}[Equivalence]
\label{assump:equiv}
Let $\mathbb{P}_0$ be the probability distribution of the process $y(s)$ defined by the model \eqref{eq:spatialG} 
with true parameter values $\{\beta_0, \sigma_0^2, \Phi_0, \tau_0^2\}$. For each $\Phi$, there is $\sigma'^2 > 0$ such that the probability distribution of the process $y(s)$ defined by the model \eqref{eq:spatialG} with parameter values $\{\beta_0, \sigma'^2, \Phi, \tau_0^2\}$ is equivalent to $\mathbb{P}_0$.
\end{assumption}

The above Assumption holds when the latent process $z(s)$ follows a Mat\'ern model in dimension $d \in \{1,2,3\}$.
In this case, the probability distribution of the process $y(s)$ defined by the model \eqref{eq:spatialG} with parameters $\left\{\beta_0, \frac{\sigma_0^2 \phi_0^{2 \nu}}{\phi^{2 \nu}}, \phi, \tau_0^2\right\}$
is equivalent to $\mathbb{P}_0$ (see e.g. \citet[Section 2.1]{TZB2021}).
We denote $\mathbb{E}_0(\cdot)$ the expectation with respect to $\mathbb{P}_0$.

\medskip
\noindent
{\bf Posterior inference}:
The following theorem,
which extends Theorem \ref{thm:main2},
explores the posterior (in)consistency of 
the scale $\sigma^2$.
\begin{theorem}[Posterior inference (in)consistency]
\label{thm:main1}
Let $\mathbb{P}_0$ be the probability measure of the model \eqref{eq:spatialG} 
with parameter values $\{\beta_0, \sigma_0^2, \Phi_0, \tau_0^2\}$, and let Assumption \ref{assump:equiv} hold. Let $H= X_{\dagger} (X_{\dagger}^\T X_{\dagger})^{-1} X_{\dagger}^\T$ be the $(2n+p)\times (2n+p)$ orthogonal projector onto the column space of $X_{\dagger}$
and let 
$H = \begin{bmatrix} H_{11} & H_{12} \\ H_{12}^{\T} & H_{22} \end{bmatrix}$ be a $2\times 2$ partition of $H$ so that $H_{22}$ is the lower right $n \times n$ block formed by rows and columns indexed from $n+p+1$ to $2n+p$. Assume that $\tr(H_{22})/n \to \alpha$ as $n \to \infty$. Then under $\mathbb{P}_0$, 
\begin{equation}
\label{eq:sigmacv}
\lim_{n\to\infty} p(\sigma^2 \given y) = \texttt{Dirac}(\sigma^2_\alpha)\;,
\end{equation}
where $\sigma^2_\alpha:=\frac{\tau_0^2}{\delta^2} \alpha + \sigma'^2 (1 -\alpha)$.
\end{theorem}

\begin{proof}
By Assumption \ref{assump:equiv}, there is a probability distribution $\mathbb{P}' \equiv \mathbb{P}_0$, which corresponds to the model with parameters $(\beta_0, \sigma'^2, \Phi, \tau_0^2)$ for some $\sigma'^2$. Under $\mathbb{P}'$, there exists a $\gamma_0$ such that 
\begin{equation} \label{eq:tprim}
y_{\dagger} - X_{\dagger} \gamma_0 \given \sigma^2 \sim \mathcal{N}\left(0, V'\right),
\end{equation}
where $V' = \begin{pmatrix}
\frac{\tau_0^2}{\delta^2} I_n & 0 & 0 \\ 
0 & \sigma^2 I_p & 0  \\
0 &  0  &  \sigma'^2 I_n
\end{pmatrix}$. 
Consider the sequence ${\zeta}_n \sim p(\sigma^2 \given y)$. Under $\mathbb{P}'$,
${\zeta}_n \sim IG(a^*_{\sigma,n}, b^*_{\sigma,n})$, 
where $a^*_{\sigma,n} = a_{\sigma}+n/2$ and 
\begin{equation}
\begin{aligned}
\label{eq:brep}
b_{\sigma,n}^* & = b_\sigma + \frac{1}{2}(y_{\dagger} - X_{\dagger} \gamma_0)^\T(I_{2n + p} - H) (y_{\dagger} - X_{\dagger} \gamma_0) \\
& = b_\sigma + \frac{1}{2} [Q(y_{\dagger} - X_{\dagger} \gamma_0)]^\T \begin{pmatrix}
0 & 0 \\ 
0 & I_n 
\end{pmatrix}  [Q(y_{\dagger} - X_{\dagger} \gamma_0)], 
\end{aligned}
\end{equation}
The expectation under $\mathbb{P}'$ for $b_{\sigma,n}^*$ is
\begin{equation}
\begin{split}
    \mathbb{E}'(b_{\sigma,n}) &= b_{\sigma} + \frac{1}{2}  \mathbb{E}'(y_{\dagger}^{\T}(I_{2n+p} - H)y_{\dagger}) \\
    &=  b_{\sigma} + \frac{1}{2}\left\{\gamma_0^{\T}X_{\dagger}^{\T}(I_{2n+p}-H)X_{\dagger}\gamma_0 + \tr((I_{2n+p}-H)V')\right\} \\
    &= b_{\sigma} + \frac{1}{2}\tr((I_{2n+p}-H)V') \quad \mbox{(since $(I_{2n+p}-H)X_{\dagger} = O$)} \\
    &=  b_{\sigma} + \frac{1}{2}\tr\left(Q_{21}^{\T}Q_{21}\begin{pmatrix}\frac{\tau_0^2}{\delta^2}I_{n} & O \\ O & \sigma^2I_p\end{pmatrix}\right) + \frac{\sigma'^{2}}{2}\tr\left(Q_{22}^{\T}Q_{22}\right)\;,
\end{split}
\end{equation}
where $Q = \begin{pmatrix} Q_{11} & Q_{12} \\ Q_{21} & Q_{22}\end{pmatrix}$ is an orthogonal matrix such that $H = Q^\T \begin{pmatrix} I_{n+p} & 0 \\  0 & 0  \end{pmatrix} Q$ with $Q_{22}$ being the lower right $n \times n$ block of $Q$ in a $2\times 2$ partition.
Using $Q_{21}Q_{21}^{\T} = I_n - Q_{22}Q_{22}^{\T}$ and some further simplification we obtain
\begin{equation}\label{eq: expectation_b}
    \mathbb{E}'(b_{\sigma^*,n}) = b_{\sigma} + \frac{n}{2}\left(\frac{\tau_0^2}{\delta^2}\left(1 - \tr(Q_{22}^{\T}Q_{22}/n)\right) + \sigma'^2\tr(Q_{22}^{\T}Q_{22}/n)\right) + \underbrace{\frac{p}{2}\left(\sigma^2 - \frac{\tau_0^2}{\delta^2}\right)}_{o(n)}\;. 
\end{equation}
Since $\mathbb{E}'(\zeta_n) = \mathbb{E}'(b_{\sigma,n}^*)/(a_{\sigma,n}^* - 1)$ and $\tr(Q_{22}^{\T}Q_{22}) =  n - \tr(H_{22})$, where $\tr(H_{22})/n\to \alpha$ as $n\to\infty$, we obtain $\lim_{n\to\infty}\mathbb{E}'(\zeta_n) = \sigma_{\alpha}^2$. The variance of $\zeta_n$ under $\mathbb{P}'$ is given by
\begin{equation}\label{eq:varZn}
\begin{aligned}
\mathbb{V}'(\zeta_n) &= \mathbb{E}'[\mathbb{V}(\zeta_n \given y)] + \mathbb{V}'[\mathbb{E}(\zeta_n\given y)] \\
 &= \mathbb{E}'\left(\frac{b_{\sigma,n}^{*2}}{(a_{\sigma,n}^*-1)^2(a_{\sigma,n}^*-2)}\right) + \mathbb{V}'\left(\frac{b_{\sigma,n}^*}{a_{\sigma,n}^* - 1}\right)\;.
\end{aligned}
\end{equation}
Further note that
\begin{equation}
\begin{aligned}
\mathbb{V}'(b_{\sigma, n}^*) 
& = \frac{3}{4} \tr\left[\left(\frac{\tau_0^2}{\delta^2} (I_n - Q_{22}^{\T} Q_{22}) + \sigma'^2 Q_{22}^{\T} Q_{22}\right)^2\right] + o(n) \\
& \le C \left( n + \tr(Q_{22}^{\T} Q_{22}) + \tr\left[(Q_{22}^{\T} Q_{22})^2\right]\right),
\end{aligned}
\end{equation}
for some $C > 0$ independent of $n$.
Since $Q_{21}Q_{21}^{\T} +  Q_{22}Q_{22}^{\T} = I_n$, 
we have $\tr(Q_{22}^{\T} Q_{22})$ and $\tr\left[(Q_{22}^{\T} Q_{22})^2\right]$ are bounded from above by $n$.
Hence, 
\begin{equation}
\label{eq:varbest}
\frac{\mathbb{V}'(b_{\sigma, n}^*)}{n^2} \to 0 \quad \mbox{and} \quad \frac{\mathbb{E}'(b_{\sigma, n}^{*2})}{n^{2 + \kappa}} \to 0 \, \mbox{ for any } \kappa > 0 \;.
\end{equation}
Combining \eqref{eq:varZn} and \eqref{eq:varbest} yields 
$\lim_{n \to \infty} \mathbb{V}'(\zeta_n) = 0$.
By Chebyshev's inequality, $\zeta_n$ converges in probability under $\mathbb{P}'$, and hence under $\mathbb{P}_0$, to $\sigma_\alpha^2$.
\end{proof}

Similar to the Mat\'ern model, 
Theorem~\ref{thm:main1} suggests that the posterior distribution of the scale parameter $\sigma^2$ in the general conjugate model does not necessarily concentrate on the true generating value. The parameter $\alpha$ quantifies how the trend and the chosen parameter $\Phi$ affects the inference on $\sigma^2$. 
The assumption $\tr(H_{22})/n \to \alpha$ 
is generally hard to check analytically. 
Figure~\ref{fig: cond_check} summarizes some numerical experiments to empirically explore these assumptions. The study domain $\calD$ is $[0, 1]^2$, and locations in $\chi$ are chosen uniformly on $\calD$. We generate data using the Mat\'ern covariogram for $z(s)$ (see \eqref{eq: Matern} in Section~\ref{sc21}). We consider two types of predictors $x(s)$.
For the first type, titled ``with intercept'', $x(s)$ consists of a constant $1$ for intercept and a predictor generated by a standard normal.  For the second type, labeled ``without intercept'', $x(s)$ is composed of two predictors sampled from a standard normal.  We consider the trends of the target quantities with different hyper-parameter values in the covariogram of $z(s)$ and different types of $x(s)$ as sample size increases. Figure~\ref{fig: cond_check} shows that $\tr(H_{22})/n$ increases as sample size increases for all examples. Since $\tr(H_{22})/n$ is bounded above by $1$,  the assumption $\tr(H_{22})/n \to \alpha$ for some constant $\alpha$ is likely to hold in general. This is consistent with Theorem \ref{thm:main2},  where $\sigma^2$ converges to the Dirac measure at $\tau_0^2/\delta^2$.

\begin{figure}[t]
\centering
\includegraphics[width=0.8\textwidth, height = 0.3\textwidth]{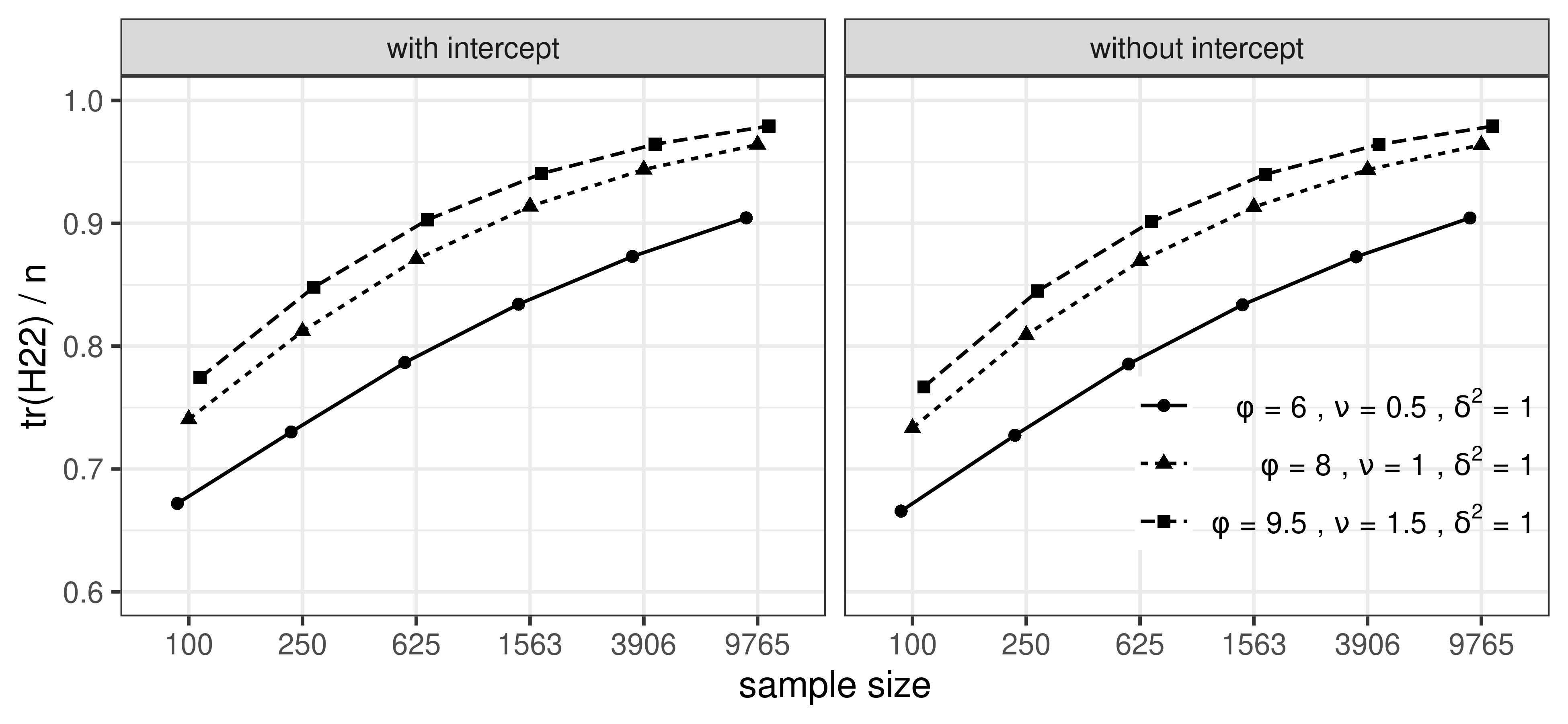}
\caption{Plots of $\tr(H_{22})/n$ when $x(s)$ consists of an intercept only (left), and with additional covariates (right) for different parameter values.}\label{fig: cond_check} 
\end{figure}

\section{Predictive consistency for general conjugate spatial models}\label{sec: thm:main1bis}
We provide a brief discussion on extending Theorem~\ref{prop:stackingpost2} to establish the results on posterior predictions under the general conjugate Bayesian spatial models. At the outset, it is worth remarking that the mean function of a Gaussian process under the Mat\'ern covariance kernel is, in general, not identifiable \citep[this follows from Theorem~6 in Chapter~4 of][]{Steinbook}, which generally suggests the lack of consistent estimators. Nevertheless, some investigations into the behavior of the posterior distributions may be possible under the following setup. We offer some brief guidelines below. 

\paragraph{Posterior inference:}
Let $\mathbb{P}_0$ be the probability distribution corresponding to the model in \eqref{eq:spatialG} with parameters $\{\beta_0, \sigma_0^2, \Phi_0, \tau_0^2\}$. 
Define $U = (X_{\dagger}^\T X_{\dagger})^{-1} = \begin{pmatrix} U_{11} & U_{12} \\ U_{12}^{\T} & U_{22} \end{pmatrix}$, where $U_{11}$ is a $p\times p$ matrix. Let Assumption~\ref{assump:equiv} hold and assume the following additional condition holds:
\begin{equation} \label{eq:betacond}
\tr\left(U_{11} \right) \to 0, \quad \mbox{as } n \to \infty\;.
\end{equation}
We can then claim that $\lim_{n\to \infty} p(\beta \given y) = \texttt{Dirac}(\beta_0)$ under $\mathbb{P}_0$.
To see this, let $\xi_n \sim p(\beta \given y)$ and define $B = U \begin{pmatrix} X^{\T}X & X^{\T} \\ X & I_n\end{pmatrix} U = \begin{pmatrix} B_{11} & B_{12} \\ B_{21} & B_{22} \end{pmatrix}$, where $B_{11}$ is $p\times p$, $C = U_{11} V_\beta^{-1} U_{11}$, and $D = U_{12} R_\Phi(\chi)^{-1} U_{12}^{\T}$. Straightforward algebra yields
\begin{equation}
\label{eq:betaL2}
\mathbb{E}'\|\xi_n- \beta_0\|^2  =  
\frac{\tau_0^2}{\delta^4}\tr \left(B_{11} \right) + \mathbb{E}'(\zeta_n) \tr(C)
+ \sigma'^2 \tr(D) +  \mathbb{E}'(\zeta_n) \tr\left(U_{11}\right).
\end{equation}
Since $\mathbb{E}'(\zeta_n) \to \sigma_\alpha^2< \infty$ from Theorem~\ref{thm:main1}, $\mathbb{E}'\|\xi_n- \beta_0\|^2$ converges to zero when 
$\tr\left(U_{11} \right)$, $\tr\left(B_{11}\right)$, $\tr(C)$,  $\tr(D) \longrightarrow 0$ as $n \to \infty$.
When $\delta^2 > 0$, we have $\tr\left(\delta^{-2}B_{11}+C+D\right) = \tr\left(U_{11}\right)$. Since $B_{11}$ $C$ and $D$ are positive semi-definite, we only require \eqref{eq:betacond} to establish consistency of $p(\beta \given y)$.

While verifying condition \eqref{eq:betacond} can be challenging in general, it is feasible in certain simplified cases where we are able to derive some closed-form results agnostic to the asymtotic behavior of $R_{\Phi}(\chi)$. For example, assume $X \sim \mbox{MN}(0, I_n, \Sigma_x)$ where each $x(s_i)$ follows a zero-centered Gaussian distribution with a covariance matrix $\Sigma_x$. Then $U_{11}^{-1} = X^{\T} (\delta^2 I_n + R_\Phi(\chi))^{-1}X + V_\beta^{-1}$. 
Given the conditions in Remark S.3 of \citet{ZB21}, the smallest eigenvalue of $U_{11}^{-1}$ goes to infinity as the sample size increases. Consequently, $\tr\left(U_{11} \right)$ converges to zero as sample size goes to infinity, satisfying condition~\eqref{eq:betacond}.

\paragraph{Posterior prediction:} Recall that $Z_n(s_0)$ is distributed as $p(z(s_0)\,|\, y)$, and $Y_n(s_0)$ is distributed as $p(y(s_0)\,|\, y)$ under $\mathbb{P}_0$. 
For any given $\Phi$, denote $\cov(z, z(s_0) \given \sigma^2)$ and $R_{\Phi}(\chi)$ by $\sigma^2 J_{\Phi, n}$ and $R_{\Phi, n}$, respectively. Let 
\begin{equation*}
\begin{aligned}
F_n &= J_{\Phi, n}^\T R_{\Phi, n}^{-1} (U X_{\dagger}^\T y_{\dagger})_{[p+1:p+n]} \quad \mbox{and} \quad \\
G_n &= 1 - J_{\Phi, n}^\T R_{\Phi, n}^{-1}(R_{\Phi, n} + U_{[p+1:p+n, \, p+1:p+n]})R_{\Phi, n}^{-1}J_{\Phi, n}.
\end{aligned}
\end{equation*}
Under the assumptions in Theorem \ref{thm:main1},
we have the decomposition \eqref{eq:decomperr22}, where 
$E_{1, n} = \mathbb{E}_0(z(s_0) -F_n)^2$ and $E_{2, n} = \sigma_\alpha^2 G_n$. To prove this, observe that 
\begin{equation}
\label{eq:decomppf}
\begin{aligned}
\mathbb{E}_0(Z_n(s_0) - z(s_0))^2 & = \mathbb{E}_0 \left\{Z_n(s_0) - \mathbb{E}(z(s_0) \given y) +  \mathbb{E}(z(s_0) \given y) - z(s_0) \right\}^2 \\
& = \mathbb{E}_0\left\{z(s_0) - \mathbb{E}(z(s_0) \given y) \right\}^2 + \mathbb{E}_0\{\mathbb{V}(z(s_0) \given y)\},
\end{aligned}
\end{equation}
where the second equality follows from the fact that $z(s_0) - \mathbb{E}(z(s_0) \given y)$ is independent of $y_n$.
Note that 
\begin{equation}
\label{eq:condprod}
p(z(s_0) \given y) = \int p(z(s_0) \given y, \sigma^2, \gamma) \times p(\gamma \given y, \sigma^2) \times p(\sigma^2 \given y) \, d\sigma^2 d\gamma.
\end{equation}
By standard Gaussian conditioning (see e.g. \cite[Section 2.2]{RW06}), 
\begin{equation}
\label{eq:G1}
p(z(s_0) \given y, \sigma^2, \gamma) = \mathcal{N}(J_{\Phi, n}^\T R_{\Phi, n}^{-1}z, \sigma^2(1 - J_{\Phi, n}^\T R_{\Phi, n}^{-1}J_{\Phi, n})).
\end{equation}
By \eqref{eq:condprod}, \eqref{eq:G1} and Lemma \ref{lem:structure}, the posterior predictive mean is 
\begin{equation}
\label{eq:ppmeanpf}
\mathbb{E}(z(s_0) \given y) = F_n.
\end{equation}
Further by the law of total variance and Theorem \ref{thm:main1}, we get
\begin{align}
\label{eq:varerr}
\mathbb{V}&(z(s_0) \given y)  = \mathbb{E}\{\mathbb{V}(z(s_0) \given y, \sigma^2, \gamma)\} + \mathbb{V}\{\mathbb{E}(z(s_0) \given y, \sigma^2, \gamma)\} \notag \\
& \to \sigma_\alpha^2 (1 - J_{\Phi, n}^\T R_{\Phi, n}^{-1}J_{\Phi, n}) + \sigma_\alpha^2 J_{\Phi, n}^\T R_{\Phi, n}^{-1} U_{[p+1:p+n, \, p+1:p+n]}   R_{\Phi, n}^{-1} J_{\Phi, n} = \sigma_\alpha^2 G_n.
\end{align}
Combining \eqref{eq:decomppf}, \eqref{eq:ppmeanpf} and \eqref{eq:varerr}
yields the decomposition \eqref{eq:decomperr22},
and hence the posterior predictive consistency for $z(s_0)$ holds if $E_{1, n}, E_{2, n} \to 0$ as $n \to \infty$.
Further, let $\xi_n$ have the density $p(\beta \given y)$. Under conditions~\eqref{eq:betacond}, we have $\mathbb{E}_0|| \xi_n - \beta_0||^2 \to 0$ as $n \to \infty$.
As a result, posterior predictive inference for $y(s_0)$ satisfies
\begin{equation*}
\mathbb{E}_0(Y_n(s_0) - y(s_0))^2 = \mathbb{E}_0(x(s_0)^\T(\xi_n - \beta_0))^2 + E_{1, n} + E_{2, n} + \tau_0^2 +
\delta^2 \mathbb{E}_0(\sigma^2 \given y),
\end{equation*}
which converges to $\tau_0^2 + \delta^2 \sigma_{\alpha}^2$ as $n \to \infty$.

\section{KL bound for stacking of predictive densities}\label{sec: thm:KLpd}
Let $y = (y(s_1), \cdots, y(s_n))^\T$ be sampled from a model $\mathcal{M}_0$.
Given $y$, define the probability measure $Q_{0,y}(y'):= \prod_{i = 1}^n p(Y(s_i) = y_i' \given y_{-i}, \mathcal{M}_0)$, where $y'=(y_1',\ldots,y_n')$ and $Y = (Y(s_1),\ldots,Y(s_n))^{\T}$ are random variables following $\mathcal{M}_0$.
Let $Q_{0}(\cdot):= \mathbb{E}_0 (Q_{0,y}(\cdot))$, where $\mathbb{E}_0$ denotes expectation over $y \sim \mathcal{M}_0$. The probability distribution $Q_0(\cdot)$ can be viewed as the in-sample predictor of the  distribution of $y$ under $\mathcal{M}_0$. We establish the KL bound for in-sample predictions using stacking of predictive densities.
\begin{proposition}
{\label{thm:KLpd}}
Define $P_{w,y}(y'):= \prod_{i = 1}^n \left(\sum_{g = 1}^G w_g p(Y(s_i) = y_i' \given y_{-i}, \calM_g) \right)$ for any set of stacking weights $w = (w_1, \ldots, w_G)$ and fixed $y$, where $Y(s_i) \overset{ind}{\sim} P_{y(s_i)\given y_{-i}, \mathcal{M}_0}(\cdot)$.
Let $P_{w}(\cdot):= \mathbb{E}_0(P_{w,y}(\cdot))$,  where $\mathbb{E}_0$ denotes expectation with respect to $y\sim \mathcal{M}_0$. For the stacking weights with $w_g = 1$, we abbreviate $P_w(y)$ with $P_g(y)$. Let $w^*:=(w_1^*,\ldots, w^*_G)$ be the stacking weights defined in \eqref{eq:defpd}. Then we have
\begin{equation}
\label{eq:KLest}
\mbox{KL}(Q_0, P_{w^*}) 
\le \sum_{g = 1}^G w_g^* \mbox{KL}(Q_0, P_g).
\end{equation}
where $\mbox{KL}(\cdot, \cdot)$ is the KL divergence between two distributions.
\end{proposition}

\begin{proof}
Note that for any stacking weights $w = (w_1, \ldots, w_G)$ and $y$,
we have
\begin{equation*}
\begin{aligned}
\log P_{w,y}(y') 
&= \sum_{i = 1}^n \log \left(\sum_{g = 1}^G w_g p(Y(s_i) = y_i'\given y_{-i}, M_g) \right) \\
& \ge \sum_{g = 1}^G w_g \log \left( \prod_{i = 1}^n p(Y(s_i)=y_i' \given y_{-i}, M_g) \right) =: \sum_{g = 1}^G w_g \log P_{g,y}(y'),
\end{aligned}
\end{equation*}
using concavity of the log function, where $P_{g,y}(y') = \prod_{i = 1}^n p(Y(s_i)=y_i' \given y_{-i}, M_g)$. Therefore, 
\begin{equation}
\label{eq:pdineq1}
\log Q_{0,y}(y') - \log P_{w,y}(y') 
\le \sum_{g = 1}^G w_g \left(\log Q_{0,y}(y') - \log P_{g,y}(y')\right).
\end{equation}
Taking expectations of both sides in \eqref{eq:pdineq1} with respect to $Q_0(\cdot)$ yields $\mbox{KL}(Q_{0,y}, P_{w,y}) \le \sum_{g = 1}^G w_g \mbox{KL}(Q_{0,y}, P_{g,y})$. The chain rule of the KL divergence with $w = w^*$ yields \eqref{eq:KLest}.
\end{proof}


\section{Pseudo-codes for Stacking algorithms}\label{app: pseudo_code}

Algorithm~\ref{code: stacking_expectaion} and Algorithm~\ref{code: stacking_LP} outline the implementation details for stacking of means and stacking of posterior densities, respectively. Additionally, we provide a Monte Carlo variant of Algorithm~\ref{code: stacking_LP}, detailed in Algorithm~\ref{code: stacking_LP_MC}, for readers interested in this approach. Following \citet{pan2024bayesian}, our Julia implementation computes $R^{-1}_{\phi, \nu}(\chi[-k])$ using a block Givens rotation algorithm \citep[][Section 5.1.8]{GolubLoanMatrix4}, enabling faster model evaluations during cross-validation.

{\footnotesize
\begin{algorithm}[h]
    \caption{Stacking weights calculation using stacking of means}\label{code: stacking_expectaion}
    \begin{algorithmic}[1] 
    \State \textbf{Input:}, $X$, $y$, $\chi$: Design matrix, outcome and location set
    $\mu_\beta$, $V_\beta$, $a_\sigma$, $b_\sigma$: Prior parameters; $G_{\phi}$, $G_{\nu}$, $G_{\delta^2}$: Grids of $\phi, \nu, \delta^2$; $K$: Number of folders
    \State \textbf{Output:} $w = \{w_{\phi, \nu, \delta^2}\}_{(\phi, \nu, \delta^2) \in G_{all}}$: Stacking weights; $G_{all}$: Grid spanned by $G_{\phi}, G_{\nu}, G_{\delta^2}$
    \vspace{2mm}
    \Function{\texttt{SPStacking}}{$X$, $y$, $\chi$,
    $\mu_\beta$, $V_\beta$, $a_\sigma$, $b_\sigma$, $G_{\phi}, G_{\nu}, G_{\delta^2}$, $\phi$, $\nu$, $\delta^2$, $K$}
    \State Compute $X_{\mbox{prod}}^{(k)} = X^\T[-k]X[-k]$, $X_{y}^{(k)} = X^\T[-k]y[-k]$ and record the number of observations $n_k$ in fold $k$ for $k = 1, \ldots, K$, where $X[-k]$ $y[-k]$ denotes the predictors and response for observations not in fold $k$\; 
    \For{$\{\phi, \nu\}$ in grid expanded by $G_\phi$ and $G_\mu$}
        \For{$k=1$ to $K$}
        \State Calculate $R^{-1}_{\phi, \nu}(\chi[-k]) = \{R(s, s';\; \phi, \nu)\}^{-1}_{s, s' \in \chi[-k]}$ \hfill{$\mathcal{O}(n^3)$}
        \State Store $J_{\phi, \nu}(\chi[k], \chi[-k]) = \{R(s, s'; \; \phi, \nu)\}_{s \in \chi[k], s' \in \chi[-k]}$ \hfill{$\mathcal{O}(n \cdot n_k)$}
            \For{$\delta^2$ in $G_\delta^2$}
            \State Compute the Cholesky decomposition $L_\ast$ of 
            \State \quad $M_\ast^{-1} = L_\ast L_\ast^\T = \begin{bmatrix}
\delta^{-2} X_{\mbox{prod}}^{(k)} + V_\beta^{-1} & \delta^{-2} X^\T[-k] \\
\delta^{-2} X[-k] & R^{-1}_{\phi, \nu}(\chi[-k]) + \delta^{-2}I_{n - n_k} \end{bmatrix}$ \hfill{$\mathcal{O}(n^3)$}
            \State Compute $m_\ast = \begin{bmatrix}
V_\beta^{-1}\mu_\beta + \delta^{-2}X_{y}^{(k)} \\ \delta^{-2}y[-k] \end{bmatrix}$; 
            Set $u = L_\ast^{-1}m_{\ast}$
            \State Update $u = L_\ast^{-\T} u$; Extract $\mu_{\beta}^{\ast\T}$ and $\mu_{z}^{\ast\T}$ from $u$ such that $u = (\mu_{\beta}^{\ast\T}, \mu_{z}^{\ast\T})^\T$ \hfill{$\mathcal{O}(n^2)$}
            \State Compute expected outcome on locations in fold $k$
            \State \quad $\mathbb{E}(y[k] \given y[-k], \phi, \nu, \delta^2) = X[k] \mu_{\beta}^\ast + J_{\phi, \nu}(\chi[k], \chi[-k]) \cdot R^{-1}_{\phi, \nu}(\chi[-k]) \cdot \mu^\ast_{z}$ 
            \EndFor
        \EndFor
    \EndFor
    \State {
    Solve convex optimization problem: $\arg\min_{w}(y - \hat{Y}{w})^\T
    (y - \hat{Y}{w})$ under constraints $\sum_{g = 1}^{G} w_g = 1$ and $w_g \geq 0$ for $g = 1, \ldots, G$, where $G = |G_{all}|$. } 
    \State \Return $\left\{w, G_{all}\right\}$, 
    \EndFunction
    \end{algorithmic}
\end{algorithm}
}

\section{Derive the closed form of point-wise predictive density}\label{app: lppd}
We derive the posterior predictive density of the outcome $y(s_0)$ on location $s_0$. We follow the notations in Section~\ref{sec: spatial_models}.
First, we know that $y(s_0) \given \gamma, \sigma^2, y, \Phi$ follows a Gaussian with mean $x(s_0)\beta + J_{\phi, \nu}(s_0, \chi)R^{-1}_{\phi, \nu}(\chi)z$ and variance $\delta^2\sigma^2$. Since the conditional posterior distribution $\gamma \given \sigma^2, y, \Phi$ follows $\mbox{N}(M_\ast m_\ast, \sigma^2 M_\ast)$, the conditional posterior distribution $y(s_0) \given \sigma^2, y, \Phi$ still follows a Gaussian $\mbox{N}(\mu_{s_0}, \sigma^2 V_{s_0})$ where
\begin{align*}
  \mu_{s_0} &= \underbrace{\begin{bmatrix}
  x(s_0) & J_{\phi, \nu}(s_0, \chi)R^{-1}_{\phi}(\chi)
 \end{bmatrix}}_{h_g^\T}M_\ast m_\ast  \;, \; V_{s_0}  = h_{g}^\T M_\ast h_{g} + \delta^2\;.
\end{align*}
Next, through equation \eqref{eq:posterior}
\begin{align*}
    p(y(s_0) \given y, \Phi) &= \int p(y(s_0) \given \sigma^2, y, \Phi) p(\sigma^2 \given \Phi, y) d\sigma^2 =\int \mbox{N}( \mu_{s_0}, \sigma^2 V_{s_0}) \mbox{IG}( a_\ast, b_\ast) d\sigma^2 \\
    & = \int \frac{1}{(2\pi V_{s_0}\sigma^2)^{1/2}} \exp \left\{-\frac{(y(s_0) - \mu_{s_0})^2}{2\sigma^2V_{s_0}}\right\} \frac{b_\ast^{a_\ast}}{\Gamma(a_\ast)} \sigma^{2(-a_\ast - 1)} \exp\left\{-\frac{b_\ast}{\sigma^2}\right\}d\sigma^2\\
    & = \frac{b_\ast^{a_\ast}}{(2\pi V_{s_0})^{1/2} \Gamma(a_\ast)} \int \sigma^{2(-a_\ast - 1/2 -1)} \exp\left\{-\frac{1}{\sigma^2}\left(b_\ast + \frac{(y(s_0) - \mu_{s_0})^2}{2V_{s_0}}\right) \right\}d\sigma^2\\
    &= \frac{\Gamma(a_\ast + 1/2) b_\ast^{ a_\ast}}{(2\pi V_{s_0})^{1/2} \Gamma(a_\ast)} \left(b_\ast + \frac{(y(s_0) - \mu_{s_0})^2}{2V_{s_0}}\right)^{-(a_\ast + 1/2)}
\end{align*}
The log point-wise predictive density is 
\begin{equation}\label{eq: lpd_closed}
\begin{aligned}
    lp(y(s_0) \given y, \Phi) = &-0.5 \log(2\pi V_{s_0}) + 
    a_\ast \log b_\ast - (a_\ast + 1/2) \log \left(b_\ast + \frac{(y(s_0) - \mu_{s_0})^2}{2V_{s_0}}\right) \\
    &+ \log \Gamma(a_\ast + 1/2) - \log \Gamma(a_\ast)
\end{aligned}
\end{equation}

{\footnotesize
\begin{algorithm}[t]
    \caption{Stacking weights calculation using stacking of predictive densities}\label{code: stacking_LP}
    \begin{algorithmic}[1] 
    \State \textbf{Input:}, $X$, $y$, $\chi$: Design matrix, outcome and location set; $\mu_\beta$, $V_\beta$, $a_\sigma$, $b_\sigma$: Prior parameters; $G_{\phi}$, $G_{\nu}$, $G_{\delta^2}$: Grids of $\phi$, $\nu$, $\delta^2$; $K$: Number of folders
    \State \textbf{Output:} $w = \{w_{\phi, \nu, \delta^2}\}_{(\phi, \nu, \delta^2) \in G_{all}}$: Stacking weights; $G_{all}$: Grid spanned by $G_{\phi}, G_{\nu}, G_{\delta^2}$
    \vspace{2mm}
    \Function{\texttt{SPStacking}}{$X$, $y$, $\chi$,
    $\mu_\beta$, $V_\beta$, $a_\sigma$, $b_\sigma$, $G_{\phi}, G_{\nu}, G_{\delta^2}$, $\phi$, $\nu$, $\delta^2$, $K$}
    \State Compute $X_{\mbox{prod}}^{(k)} = X^\T[-k]X[-k]$, $X_{y}^{(k)} = X^\T[-k]y[-k]$, $\|y[-k]\|^2 = y^\T[-k]y[-k]$; $X[-k]$ and $y[-k]$ are predictors and response, respectively, for observations not in fold $k$; $n_k$: number of observations in fold $k$ for $k = 1, \ldots, K$.
    \For{$\{\phi, \nu\}$ in grid expanded by $G_\phi$ and $G_\mu$}
        \For{$k=1$ to $K$}
        \State Calculate $R^{-1}_{\phi, \nu}(\chi[-k]) = \{R^{-1}(s, s';\; \phi, \nu)\}_{s, s' \in \chi[-k]}$ \hfill{$\mathcal{O}(n^3)$}
        \State Store $J_{\phi, \nu}(\chi[k], \chi[-k]) = \{R(s, s'; \; \phi, \nu)\}_{s \in \chi[k], s' \in \chi[-k]}$ \hfill{$\mathcal{O}(n \cdot n_k)$}
            \For{$\delta^2$ in $G_\delta^2$}
            \State Compute the Cholesky decomposition $L_\ast$ of 
            \State \quad $M_\ast^{-1} = L_\ast L_\ast^\T = \begin{bmatrix}
\delta^{-2} X_{\mbox{prod}}^{(k)} + V_\beta^{-1} & \delta^{-2} X^\T[-k] \\
\delta^{-2} X[-k] & R^{-1}_{\phi, \nu}(\chi[-k]) + \delta^{-2}I_{n - n_k} \end{bmatrix}$ \hfill{$\mathcal{O}(n^3)$}
            \State Compute $m_\ast = \begin{bmatrix}
V_\beta^{-1}\mu_\beta + \delta^{-2}X_{y}^{(k)} \\ \delta^{-2}y[-k] \end{bmatrix}$; 
Set $u = L_\ast^{-1} m_{\ast}$ \hfill{$\mathcal{O}(n^2)$}
            \State Compute $b_\ast = b_\sigma + 0.5(\delta^{-2}\|y[-k]\|^2 + \mu_\beta^\T V_\beta^{-1} \mu_\beta - u^\T u)$ and $a_\ast = a_\sigma + 0.5 (N - n_k)$
            \State Update $u = L_\ast^{-\T} u$; Extract $\mu_{\beta}^{\ast\T}$ and $\mu_{z}^{\ast\T}$ from $u$ such that $u = (\mu_{\beta}^{\ast\T}, \mu_{z}^{\ast\T})^\T$ \hfill{$\mathcal{O}(n^2)$}
            \State Generate the posterior expected outcome on locations in fold $k$
            \State \quad $\mathbb{E}(y[k] \given y[-k], \phi, \nu, \delta^2) = X[k] \cdot \mu_{\beta}^\ast + J_{\phi, \nu}(\chi[k], \chi[-k]) \cdot R^{-1}_{\phi, \nu}(\chi[-k]) \cdot \mu_{z}^\ast$
            \State Compute $lpc = -0.5\log(2\pi) + \log \Gamma(a_\ast + 1/2) - \log \Gamma(a_\ast) + a_\ast \log b_\ast$\;

            \For{$s \in S[k]$}
            \State Construct $h_s = \begin{bmatrix}
  x(s) & J_{\phi, \nu}(s, \chi[-k]) \cdot R^{-1}_{\phi, \nu}(\chi[-k])
  \end{bmatrix}$ where $J_{\phi, \nu}(s, \chi[-k])$ is the row 
            \State \quad with elements $\{R(s, s'; \; \phi, \nu)\}_{s' \in \chi[-k]}$  
            \State Compute $V_{s} = \|L_\ast^{-1} h_{s}^\T\|^2 + \delta^2$ and the log point-wise predictive density of $y(s)$, by \eqref{eq: lpd_closed}
            \State $lp_{(\phi, \nu, \delta^2)}(s) = lpc - \frac{1}{2}\log(V_{s}) - (a_\ast + 1/2) \log\{b_\ast + \frac{1}{2V_s} (y(s) - \mathbb{E}(y(s) \given y[-k], \phi, \nu, \delta^2))^2\}$
            \EndFor
            \EndFor
        \EndFor
    \EndFor
    \State Solve convex optimization problem: Maximize $\sum_{s \in \chi}\log(\sum_{(\phi, \nu, \delta^2) \in G_{all}} \exp{\{lp_{(\phi, \nu, \delta^2)}(s)\}} * w_{(\phi, \nu, \delta^2)})$ under constrains $\sum_{(\phi, \nu, \delta^2) \in G_{all}} w_{(\phi, \nu, \delta^2)}  = 1$ and $w_{(\phi, \nu, \delta^2)} > 0$
    \State \Return $\left\{w, G_{all}\right\}$, 
    \EndFunction
    \end{algorithmic}
\end{algorithm}
}

\section{Stacking of predictive densities (Monte Carlo version)}\label{app: MC_stacking_alg}
We present a Monte Carlo algorithm to estimate the log of point-wise predictive density for outcome in fold $k$ given observations not in fold $k$. For each $k$, we generate $J$ posterior samples of $\sigma^2$ and $(\beta^\T, z^\T)^\T = \gamma$, (i.e., $\{\sigma^{2(j)}, \gamma^{(j)}\}$ for $j = 1, \ldots, J$), using data not in fold $k$. Then we calculate the corresponding expected outcome for location $s$ in fold $k$, $\hat{y}_{(\phi, \nu, \delta^2)}^{k, j}(s)$ for $j = 1, \ldots, J$. Next, we compute the predictive density of $y(s)$ conditional on the prediction $\hat{y}_{(\phi, \nu, \delta^2)}^{k, j}(s)$ and the nugget (variance of the noise process, which equals the product of $\delta^2$ and the $j$-th posterior sample $\sigma^{2(j)}$) for each $j$. The conditional predictive distribution of $y(s)$ is
\begin{equation}\label{eq: cond_pred_distr}
    p(y(s) \given \sigma^{2(j)}, \gamma^{(j)}) = \mbox{N}(y(s) \given \hat{y}_{(\phi, \nu, \delta^2)}^{k, j}, \delta^2 \sigma^{2(j)})\;.
\end{equation}
Finally, the log point-wise predictive density (LPD) of $y(s)$ at location $s$ is estimated by
\begin{equation}\label{eq: lpd}
    \begin{aligned}
     lp_{(\phi, \nu, \delta^2)}(s) &= 
\log \int_{\sigma^2, \gamma} p(y(s) \given \sigma^2, \gamma) p(\sigma^2, \gamma \given y[-k]) d\sigma^2d\gamma \\ &\approx
\log \left\{ \frac{1}{J} \sum_{j = 1}^J p_{(\phi, \nu, \delta^2)}(y(s) \given \sigma^{2(j)}, \gamma^{(j)}) \right\}
    \end{aligned}
\end{equation}
and we can compute the stacking weights based on the estimated LPDs. Algorithm~\ref{code: stacking_LP_MC} presents the Monte Carlo version of the stacking of predictive densities.

{\footnotesize
\begin{algorithm}[h]
    \caption{Stacking of predictive densities (Monte Carlo Version)}\label{code: stacking_LP_MC}
    \begin{algorithmic}[1] 
    \State \textbf{Input:}
    $X$, $y$, $\chi$: Design matrix, outcome and location set; $\mu_\beta$, $V_\beta$, $a_\sigma$, $b_\sigma$: Prior parameters; $G_{\phi}$, $G_{\nu}$, $G_{\delta^2}$: Grids for $\{\phi,\nu,\delta^2\}$; $K$ Number of folders; $J$: number of samples for estimating log point-wise predictive density
    \State \textbf{Output:} $w = \{w_{\phi, \nu, \delta^2}\}_{(\phi, \nu, \delta^2) \in G_{all}}$: Stacking weights; $G_{all}$: Grid spanned by $G_{\phi}, G_{\nu}, G_{\delta^2}$
    \vspace{2mm}
    \Function{\texttt{SPStackingMC}}{$X$, $y$, $\chi$,
    $\mu_\beta$, $V_\beta$, $a_\sigma$, $b_\sigma$, $G_{\phi}, G_{\nu}, G_{\delta^2}$, $\phi$, $\nu$, $\delta^2$, $K$, $J$}
    \State Compute $X_{\mbox{prod}}^{(k)} = X^\T[-k]X[-k]$, $X_{y}^{(k)} = X^\T[-k]y[-k]$, $\|y[-k]\|^2 = y^\T[-k]y[-k]$ and record the number of observations $n_k$ in fold $k$ for $k = 1, \ldots, K$, where $X[-k]$ $y[-k]$ denotes the predictors and response for observations not in fold $k$
    \For{$\{\phi, \nu\}$ in grid expanded by $G_\phi$ and $G_\mu$}
        \For{$k=1$ to $K$}
        \State Calculate $R^{-1}_{\phi, \nu}(\chi[-k]) = \{R^{-1}(s, s';\; \phi, \nu)\}_{s, s' \in \chi[-k]}$ \hfill{$\mathcal{O}(n^3)$}
        \State Store $J_{\phi, \nu}(\chi[k], \chi[-k]) = \{R(s, s'; \; \phi, \nu)\}_{s \in \chi[k], s' \in \chi[-k]}$ \hfill{$\mathcal{O}(n \cdot n_k)$}
            \For{$\delta^2$ in $G_\delta^2$}
            \State Compute the Cholesky decomposition $L_\ast$ of 
            \State \quad $M_\ast^{-1} = L_\ast L_\ast^\T = \begin{bmatrix}
\delta^{-2} X_{\mbox{prod}}^{(k)} + V_\beta^{-1} & \delta^{-2} X^\T[-k] \\
\delta^{-2} X[-k] & R^{-1}_{\phi, \nu}(\chi[-k]) + \delta^{-2}I_{n - n_k} \end{bmatrix}$ \hfill{$\mathcal{O}(n^3)$}
            \State Compute $m_\ast = \begin{bmatrix}
V_\beta^{-1}\mu_\beta + \delta^{-2}X_{y}^{(k)} \\ \delta^{-2}y[-k] \end{bmatrix}$; 
            Set $u = L_\ast^{-1} m_{\ast}$ \hfill{$\mathcal{O}(n^2)$}
            \State Compute $b_\ast = b_\sigma + 0.5(\delta^{-2}\|y[-k]\|^2 + \mu_\beta^\T V_\beta^{-1} \mu_\beta - u^\T u)$ and $a_\ast = a_\sigma + 0.5 (N - n_k)$
            \State Generate  $\sigma^{2(1)}, \ldots, \sigma^{2(J)} \sim \mbox{Inverse-Gamma}(a_\ast, b_\ast)$
            \State Generate $v^{(j)} \sim \mbox{N}(0, \sigma^{2(j)} I_{n - n_k + p})$; Set $\gamma^{(j)} = L_\ast^{-\T}(v^{(j)} + u)$ for $j = 1, \ldots, J$\;
            \State Generate the posterior samples of the expected outcome on locations in fold $k$
            \State \quad $\hat{y}_{(\phi, \nu, \delta^2)}^{(k, j)} = X[k] \gamma^{(j)}_\beta + J_{\phi, \nu}(\chi[k], \chi[-k]) \cdot R^{-1}_{\phi, \nu}(\chi[-k]) \gamma^{(j)}_z$, $j = 1, \ldots, J$\;
            \For{$s \in S[k]$}
            \State Compute the posterior samples of the log-density of observation $y(s)$ at location $s$, 
            \State \quad $p_{(\phi, \nu, \delta^2)}(y(s) \given \sigma^{2(j)}, \gamma^{(j)}) := \mbox{N}(y(s)\given \hat{y}_{(\phi, \nu, \delta^2)}^{(k, j)}(s),\delta^{2}\sigma^{2(j)})$ 
            for $j = 1, \ldots, J$ \eqref{eq: cond_pred_distr}\;
            \State Compute expected log point-wise predictive density of $y(s)$ at location $s$, by 
            \State \quad $lp_{(\phi, \nu, \delta^2)}(s) = \log \left[\frac{1}{J}\sum_{j = 1}^J p_{(\phi, \nu, \delta^2)}(y(s) \given \sigma^{2(j)}, \gamma^{(j)})\right]$ \eqref{eq: lpd}
            \EndFor
            \EndFor
        \EndFor
    \EndFor
    \State Solve convex optimization problem: Maximize $\sum_{s \in \chi}\log(\sum_{(\phi, \nu, \delta^2) \in G_{all}} \exp{\{lp_{(\phi, \nu, \delta^2)}(s)\}} * w_{(\phi, \nu, \delta^2)})$ under constrains $\sum_{(\phi, \nu, \delta^2) \in G_{all}} w_{(\phi, \nu, \delta^2)}  = 1$ and $w_{(\phi, \nu, \delta^2)} > 0$
    \State \Return $\left\{w, G_{all}\right\}$, 
    \EndFunction
    \end{algorithmic}
\end{algorithm}
}

\section{Recover expected $z(s)$ and lppd for MCMC sampling}\label{app: recover_z_MCMC}

The package \textit{spBayes} does not record posterior samples of the latent process $z(s)$. Here, we recover the expected $z(s)$ for the observed and unobserved locations and compute the MLPD for the simulation studies based on the outputs returned by \textit{spLM}. To achieve our goal, we need to recover the posterior samples of $z(s)$ at all locations given the recorded MCMC samples of the parameters $\phi, \nu, \sigma^2, \tau^2$ and $\beta$. Let $z_{o}$ and $z_{u}$ be the values of $z(s)$ at observed and unobserved locations, respectively, and let $z^\ast = \begin{bmatrix}z_{o} \\ z_{u} \end{bmatrix}$. Based on \eqref{eq:BayesianG},
{\footnotesize
\begin{align*}
& p(z^{\ast} \given y, \beta, \sigma^2 , \phi, \nu) \propto  
\mathcal{N}\left(y \left | X\beta + \begin{bmatrix}I_n : 0\end{bmatrix} \begin{bmatrix}z_{o} \\ z_{u} \end{bmatrix}\right., \tau^{2}I\right) \times \mathcal{N}\left(\left.\begin{bmatrix}z_{o} \\ z_{u} \end{bmatrix} \right | 0, \sigma^2 R_{\phi, \nu}(\chi^\ast)\right)\\
&\quad \propto \exp\left[-\frac{1}{2} \left\{
\tau^{-2}\left(\begin{bmatrix}I_n : 0\end{bmatrix} z^\ast - (y - X\beta)\right)^\T\left(\begin{bmatrix}I_n : 0\end{bmatrix} z^\ast -  (y - X\beta)\right) + z^{\ast\T} \sigma^{-2}R^{-1}_{\phi, \nu}(\chi^\ast) z^\ast \right\} \right]\\
&\quad \propto \exp\left[ -\frac{1}{2}\left\{
z^{\ast\T}\left(\sigma^{-2}R^{-1}_{\phi, \nu}(\chi^\ast) + \begin{bmatrix}\tau^{-2}I_n & 0 \\
0 & 0 \end{bmatrix} \right) z^\ast \right\} \right] \\
&\quad \quad \quad \times \exp\left[ -\frac{1}{2} \left\{- z^{\ast\T} \begin{bmatrix}(y - X\beta)/\tau^2\\ 0 \end{bmatrix} - \begin{bmatrix}(y - X\beta)^\T/\tau^2 : 0 \end{bmatrix}z^\ast\right\}\right]  \propto \mathcal{N}(z^\ast \given M_{z}^\ast m_{z}^\ast, M_{z}^\ast),
\end{align*}}
where $\chi^\ast$ combines the observed and unobserved location sets and
\begin{align*}
    M_{z}^\ast = \left(\sigma^{-2}R^{-1}_{\phi, \nu}(\chi^\ast) + \begin{bmatrix}\tau^{-2}I_n & 0 \\
0 & 0 \end{bmatrix} \right)^{-1}, \quad
m_{z}^\ast = \begin{bmatrix}(y - X\beta)/\tau^2 \\ 0 \end{bmatrix}\;.
\end{align*}
Let $\{\beta^{(j)}, \sigma^{2(j)}, \tau^{2(j)}, \phi^{(j)}, \nu^{(j)}\}$ for $j = 1, \ldots, J$ denote the recorded MCMC samples. We generate posterior samples for $z^\ast$ using the above full conditional posterior distribution for each iteration $j$ and then compute the average as the expected $z^\ast$. We further compute the LPPD of $y(s)$ any held out location $s$ by 
\begin{align*}
   lp(s) &= \log \int_{\beta, z(s), \tau^2} p(y(s)\given \beta, z(s), \tau^2) p(\beta, z(s), \tau^2 \given y) d\beta dz(s) d\tau^2 \\
   &= \log \left\{ \frac{1}{J} \sum_{j = 1}^J p(y(s) \given \beta^{(j)}, z(s)^{(j)}, \tau^{2(j)}) \right\} \\
   &= \log \left\{ \frac{1}{J} \sum_{j = 1}^J \mathcal{N}(y(s) \given x(s)^\T\beta^{(j)} + z(s)^{(j)}, \tau^{2(j)}) \right\}
\end{align*}

\section{Stacking weights for Stacking of means (in R code)}\label{app: QP_w}
We format the expected outcome $\{\mathbb{E}(y[k] \given y[-k], \phi, \nu, \delta^2)\}^{k = 1, \ldots, K}_{(\phi, \nu, \delta^2) \in G_{all}}$ computed in Algorithm~\ref{code: stacking_expectaion} by an $n \times G$ matrix $\hat{Y}$. Each column of $\hat{Y}$ stores  $\{\mathbb{E}(y[k] \given y[-k], \phi, \nu, \delta^2)\}^{k = 1, \ldots, K}$ for each candidate model, and it shares the same order of observed locations as the outcome $y$. Let $w = (w_1, w_2, \ldots, w_G)^\T$ be the stacking weights, we need to find the weights that satisfy
\[
\underset{w }{\mbox{argmin}}\{(y - \hat{Y}w)^\T(y - \hat{Y}w)\}\;,
\]
under the constrain $\sum_{g}^G w_g = 1$. We cast this as a quadratic programming (QP) problem.
\begin{align*}
    (y - &\hat{Y}w)^\T(y - \hat{Y}w) \\
    &= \{(y - \sum_{g}^{G} w_g\hat{Y}_G) - \sum_{g}^{G-1} w_g(\hat{Y}_g - \hat{Y}_G)\}^\T \{(y - \sum_{g}^{G} w_g\hat{Y}_G) - \sum_{g}^{G-1} w_g(\hat{Y}_g - \hat{Y}_G)\} \\
    &= (\Tilde{y} - \Tilde{Y}\Tilde{w})^\T
    (\Tilde{y} - \Tilde{Y}\Tilde{w})\;
\end{align*}
where $\Tilde{y} = y - \hat{Y}_G$, $\Tilde{Y} = [(\hat{Y}_1 - \hat{Y}_G) : \cdots : (\hat{Y}_{G-1} - \hat{Y}_G)]$, and $\Tilde{w} = (w_1, \ldots, w_{G-1})^\T$. And the QP problem has constrains $-\sum_{g = 1}^{G-1} w_g \geq -1$ and $w_g \geq 0$ for $g = 1, \ldots, G - 1$.


\section{Figures for simulation studies}\label{app: figs}

\subsection{Distributions of the diagnostic metrics for prediction performance for simulation studies}\label{appsub: figs_sum}

See Figure~\ref{fig: sim_compar2}.

\begin{figure}[t]
\centering
\subfloat[\label{subfig: sim1_compar}]
{\includegraphics[width=0.49\textwidth, keepaspectratio]{pics/CVexperiment_sim1.png}}
\subfloat[\label{subfig: sim4_compar}]
{\includegraphics[width=0.49\textwidth, keepaspectratio]{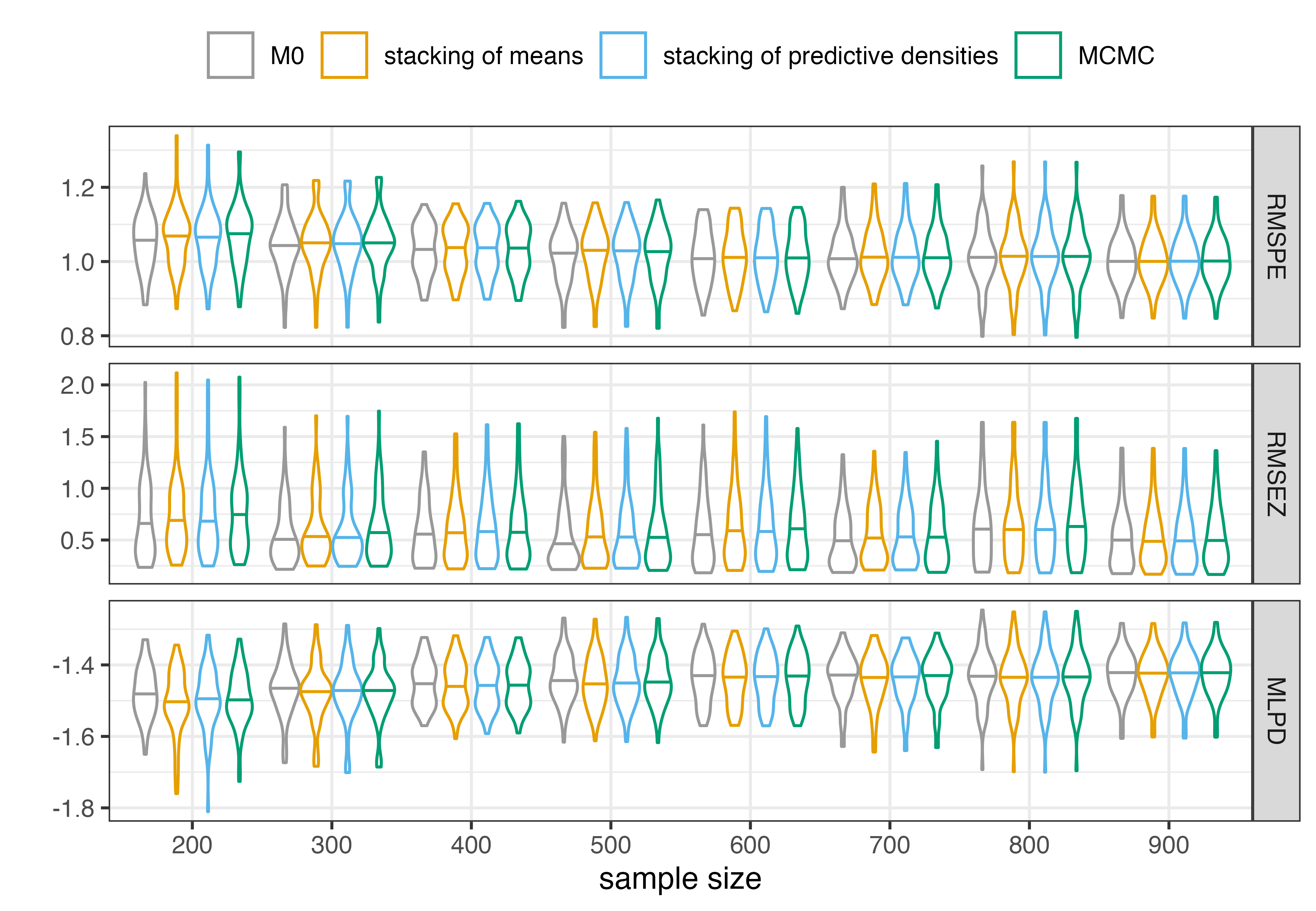}}\\
\subfloat[ \label{subfig: sim2_compar}]
{\includegraphics[width=0.49\textwidth, keepaspectratio]{pics/CVexperiment_sim2.png}}
\subfloat[ \label{subfig: sim3_compar}]
{\includegraphics[width=0.49\textwidth, keepaspectratio]{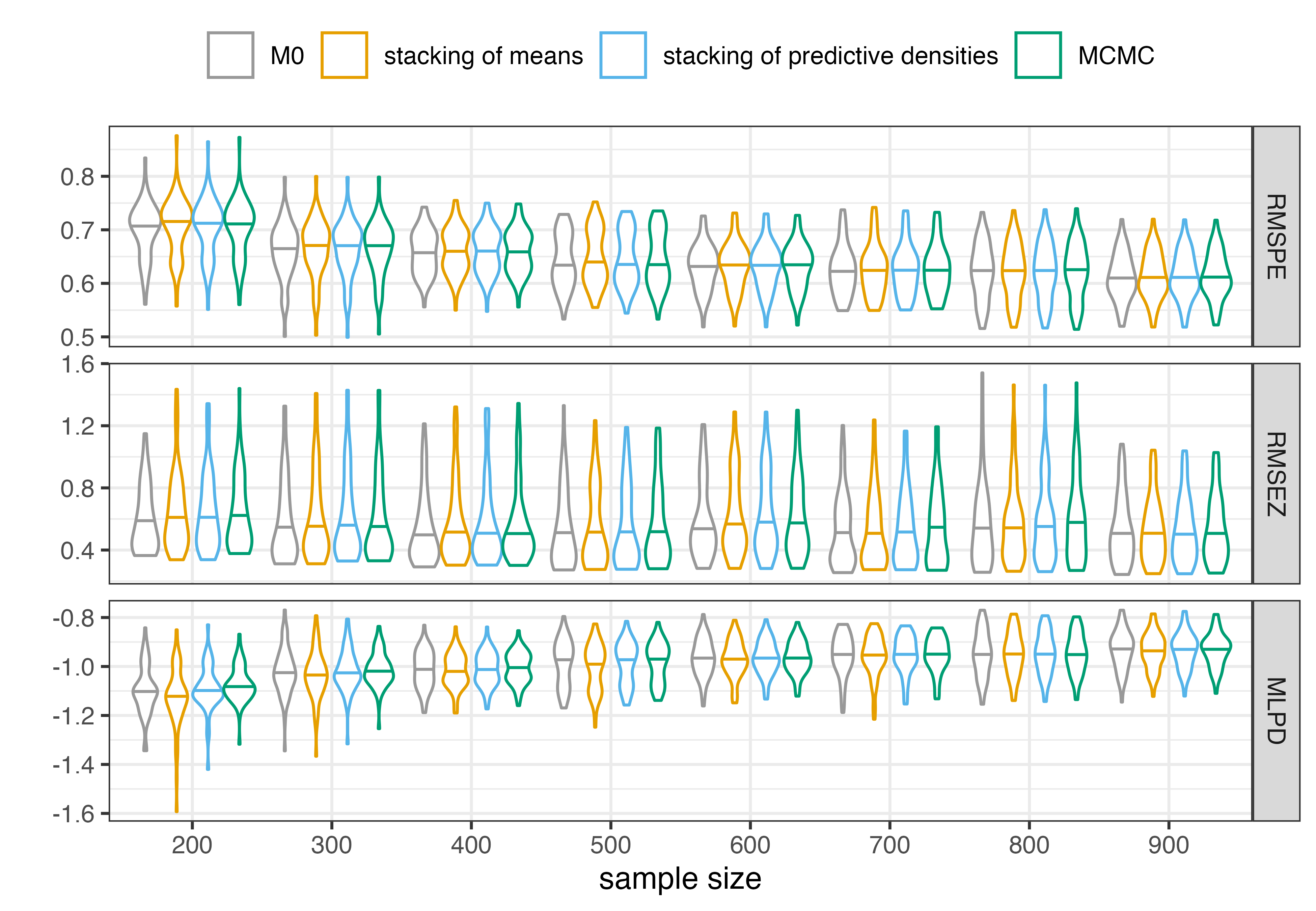}}
\caption{Distributions of the diagnostic metrics for prediction performance for the first (a), second (b), third (c) and fourth (d) simulation. Each distribution is depicted through a violin plot. The horizontal line in each violin plot indicates the median. }
\label{fig: sim_compar2}
\end{figure}

\subsection{Distributions of stacking weights ($>0.001$ only) for simulation studies}\label{appsub: w_distr}

See Figure~\ref{fig: w_distr}.

\begin{figure}[t]
\centering
\subfloat[\label{subfig: sim1_compar_w}]
{\includegraphics[width=0.49\textwidth, keepaspectratio]{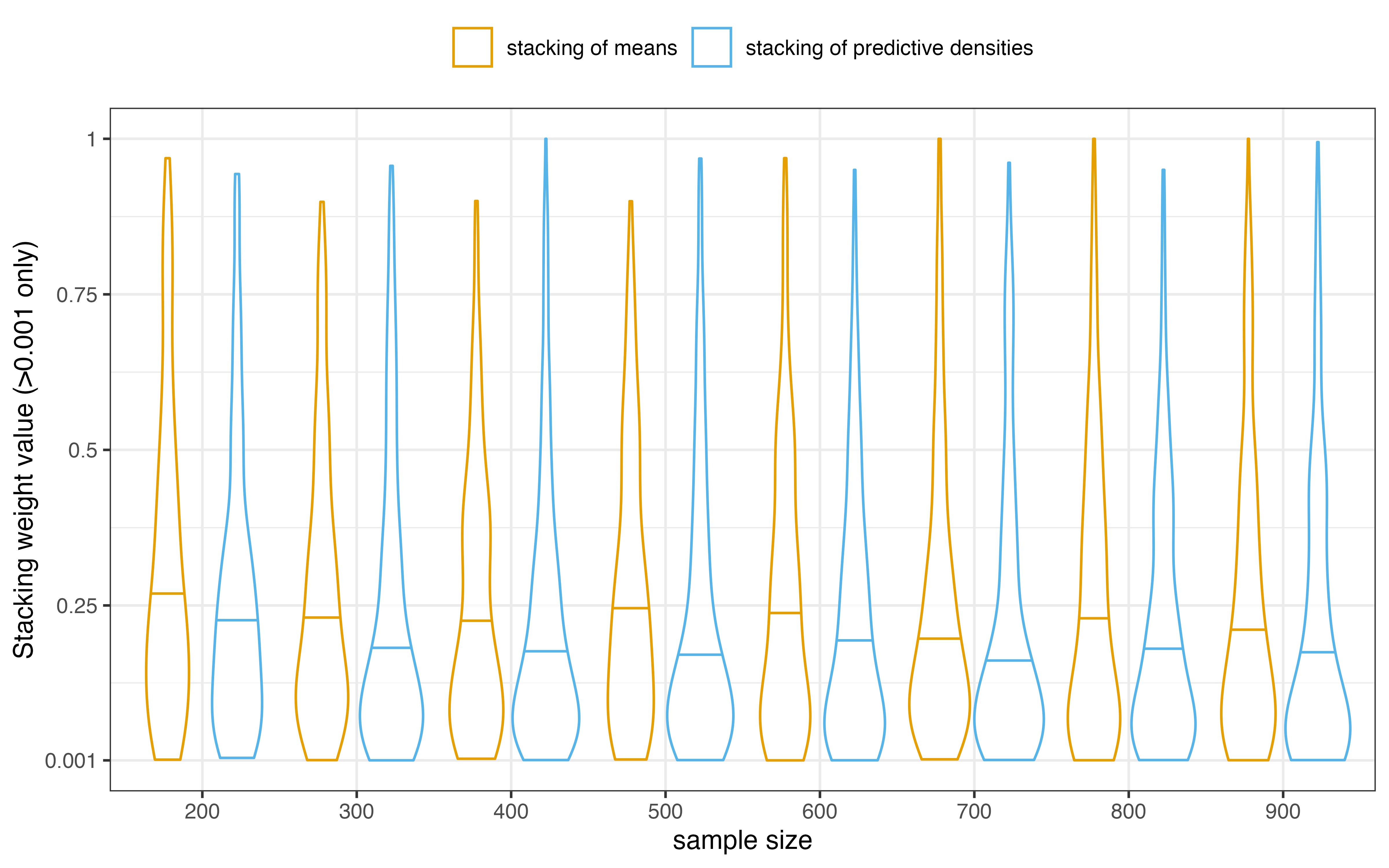}}
\subfloat[\label{subfig: sim4_compar_w}]
{\includegraphics[width=0.49\textwidth, keepaspectratio]{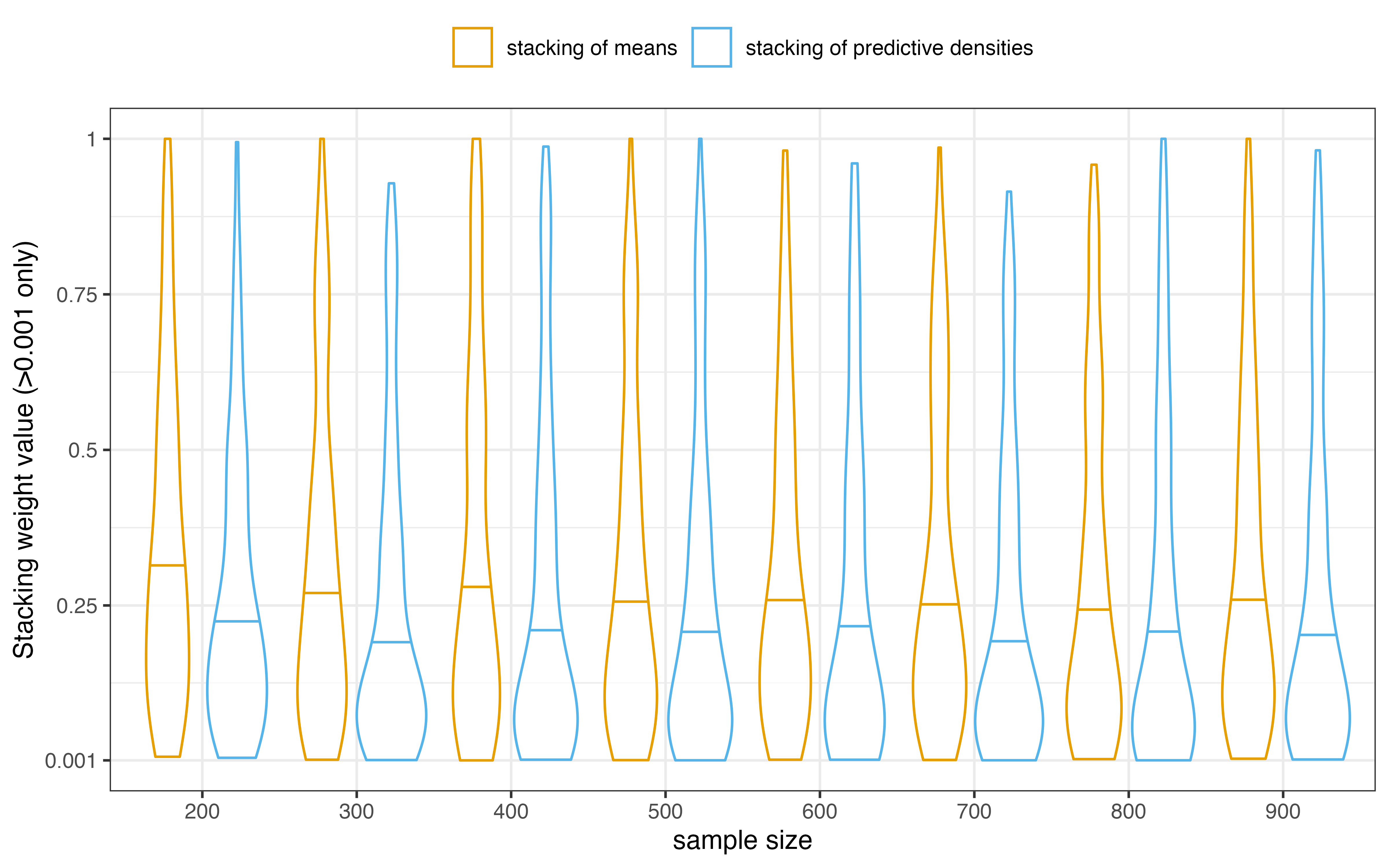}}\\
\subfloat[ \label{subfig: sim2_compar_w}]
{\includegraphics[width=0.49\textwidth, keepaspectratio]{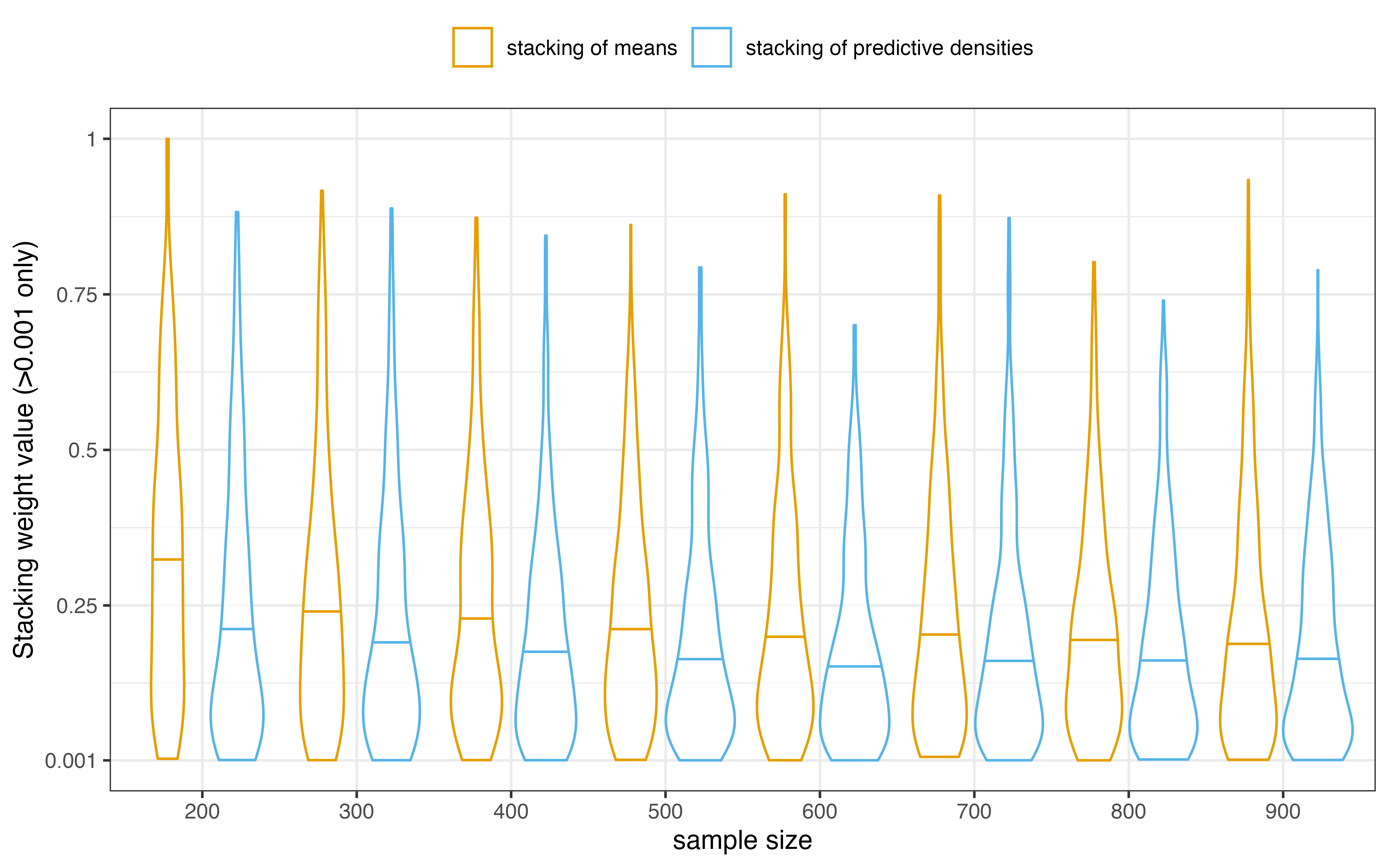}}
\subfloat[ \label{subfig: sim3_compar_w}]
{\includegraphics[width=0.49\textwidth, keepaspectratio]{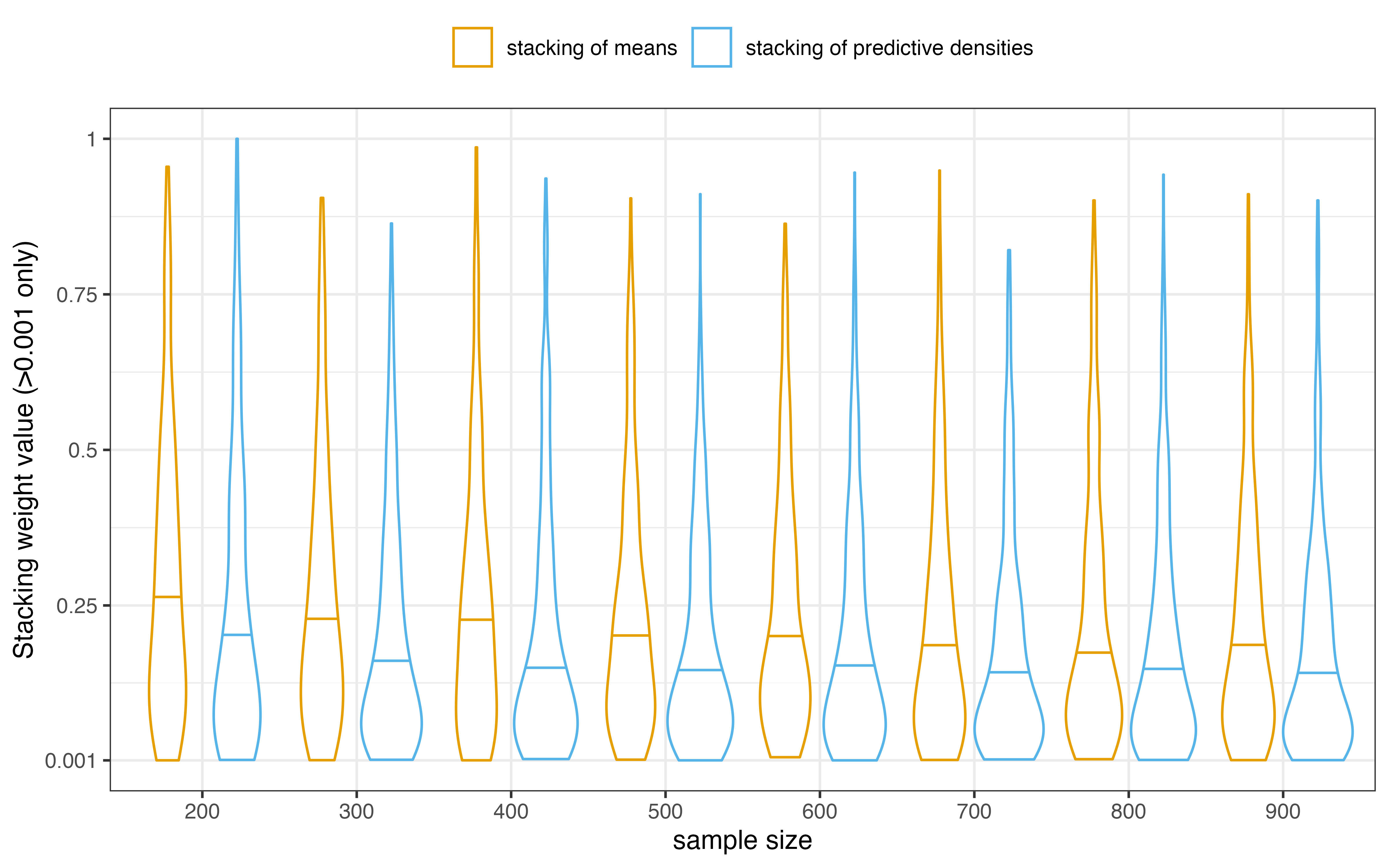}}
\caption{Distributions of stacking weights ($>0.001$ only) for the first (a), {second (b)}, third (c) and {fourth (d)} simulation. Each distribution is depicted through a violin plot. The horizontal line in each violin plot indicates the median. }
\label{fig: w_distr}
\end{figure}

\subsection{Interpolated maps for the simulation studies}\label{app: maps_w}

See Figures~\ref{fig: y_pred_compar_sim1}--\ref{fig: w_compar_sim3}

\begin{figure}[t]
\centering
\subfloat[ raw $y(s)$ \label{subfig: y_held_sim1}]{\includegraphics[width=0.3\textwidth, height=0.2\textwidth]{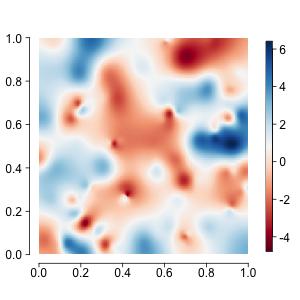}} \hspace{0.6cm}
\subfloat[ raw $x(s)^\T \beta + z(s)$ \label{subfig: y_held_denoise_sim1}]{\includegraphics[width=0.3\textwidth, height=0.2\textwidth]{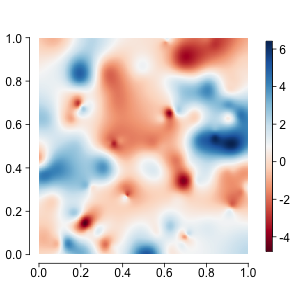}} \hspace{0.6cm}
\subfloat[ stacking of mean \label{subfig: y_pred_LSE_sim1}]{\includegraphics[width=0.3\textwidth, height=0.2\textwidth]{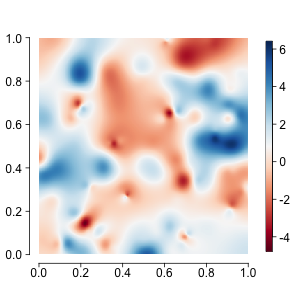}}\\
\centering
\subfloat[ stacking of predictive densities \label{subfig: y_pred_LP_sim1}]{\includegraphics[width=0.3\textwidth, height=0.2\textwidth]{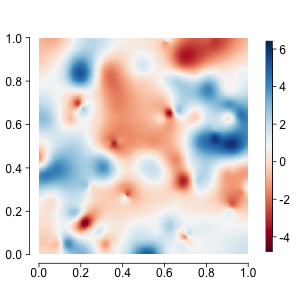}} 
\hspace{0.6cm}
\subfloat[ M0 \label{subfig: y_pred_M0_sim1}]{\includegraphics[width=0.3\textwidth, height=0.2\textwidth]{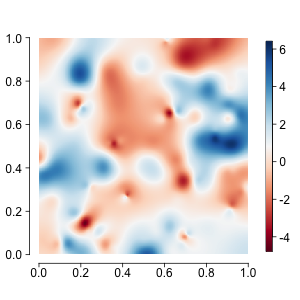}}
\hspace{0.6cm}
\subfloat[ MCMC \label{subfig: y_pred_MCMC_sim1}]{\includegraphics[width=0.3\textwidth, height=0.2\textwidth]{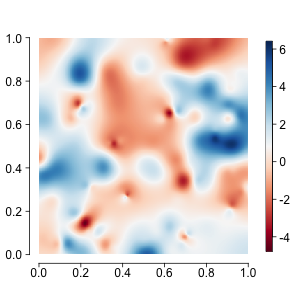}}\\
\caption{Interpolated maps of (a) the response $y(s)$, (b) the denoised response $x(s)\beta + z(s)$ and (c-f) the expected $y(s)$ on the $n_h = 100$ held out locations generated by all competing algorithms for the example with $800$ observations from the first simulation.}\label{fig: y_pred_compar_sim1}
\end{figure}

\begin{figure}[t]
\centering
\subfloat[ raw \label{subfig: w_raw_sim1}]{\includegraphics[width=0.3\textwidth, height=0.2\textwidth]{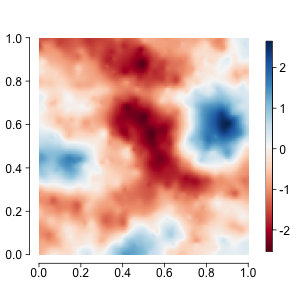}} \hspace{0.6cm}
\subfloat[ stacking of mean \label{subfig: w_LSE_sim1}]{\includegraphics[width=0.3\textwidth, height=0.2\textwidth]{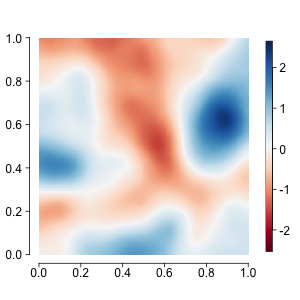}}
\hspace{0.6cm}
\subfloat[ stacking of predictive densities \label{subfig: w_LP_sim1}]{\includegraphics[width=0.3\textwidth, height=0.2\textwidth]{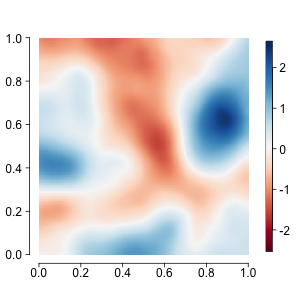}} \\
\centering
\subfloat[ M0 \label{subfig: w_M0_sim1}]{\includegraphics[width=0.3\textwidth, height=0.2\textwidth]{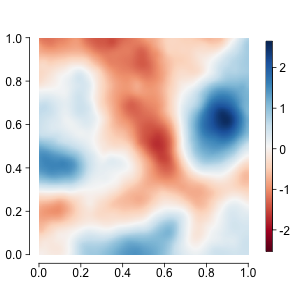}}
\hspace{0.6cm}
\subfloat[ MCMC \label{subfig: w_MCMC_sim1}]{\includegraphics[width=0.3\textwidth, height=0.2\textwidth]{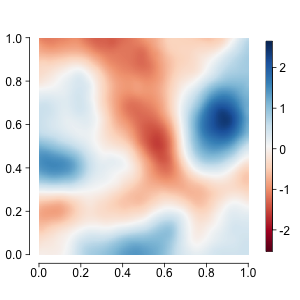}}
\caption{Interpolated maps of (a) the latent process $z(s)$ and (b-g) the expected $z(s)$ on all $n = 900$ sampled locations generated by all competing algorithms for the example from the first simulation. The $n = 900$ locations include both observed and unobserved locations}\label{fig: w_compar_sim1}
\end{figure}

\begin{figure}[t]
\centering
\subfloat[ raw $y(s)$ \label{subfig: y_held_sim4}]{\includegraphics[width=0.3\textwidth, height=0.2\textwidth]{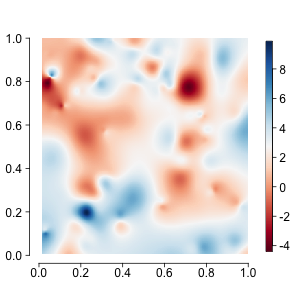}} \hspace{0.6cm}
\subfloat[ raw $x(s)^\T \beta + z(s)$ \label{subfig: y_held_denoise_sim4}]{\includegraphics[width=0.3\textwidth, height=0.2\textwidth]{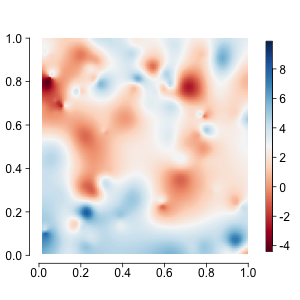}} \hspace{0.6cm}
\subfloat[ stacking of mean \label{subfig: y_pred_LSE_sim4}]{\includegraphics[width=0.3\textwidth, height=0.2\textwidth]{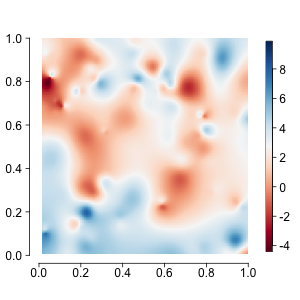}}
\hspace{0.6cm}\\
\centering
\subfloat[ stacking of predictive densities \label{subfig: y_pred_LP_sim4}]{\includegraphics[width=0.3\textwidth, height=0.2\textwidth]{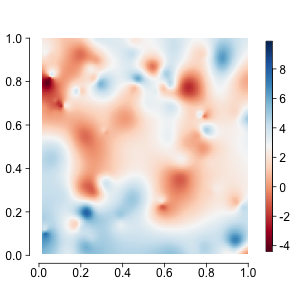}} 
\hspace{0.6cm}
\subfloat[ M0 \label{subfig: y_pred_M0_sim4}]{\includegraphics[width=0.3\textwidth, height=0.2\textwidth]{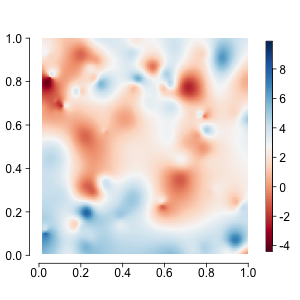}}
\hspace{0.6cm}
\subfloat[ MCMC \label{subfig: y_pred_MCMC_sim4}]{\includegraphics[width=0.3\textwidth, height=0.2\textwidth]{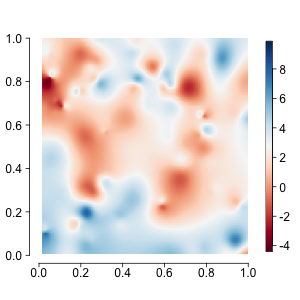}}\\
\caption{Interpolated maps of (a) the response $y(s)$, (b) the denoised response $x(s)\beta + z(s)$ and (c-f) the expected $y(s)$ on the $n_h = 100$ held out locations generated by all competing algorithms for the example with $600$ observations from the second simulation.}\label{fig: y_pred_compar_sim4}
\end{figure}

\begin{figure}[t]
\centering
\subfloat[ raw \label{subfig: w_raw_sim4}]{\includegraphics[width=0.3\textwidth, height=0.2\textwidth]{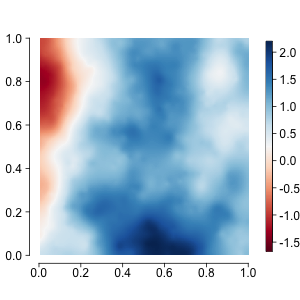}} \hspace{0.6cm}
\subfloat[ stacking of mean \label{subfig: w_LSE_sim4}]{\includegraphics[width=0.3\textwidth, height=0.2\textwidth]{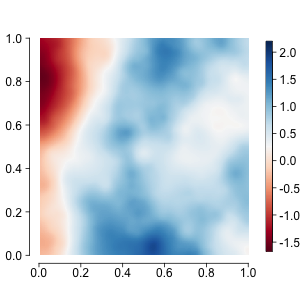}}
\hspace{0.6cm}
\subfloat[ stacking of predictive densities \label{subfig: w_LP_sim4}]{\includegraphics[width=0.3\textwidth, height=0.2\textwidth]{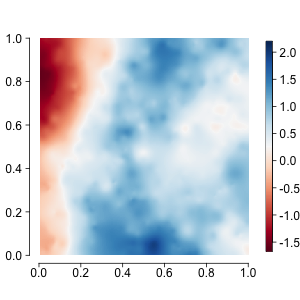}} \\
\centering
\subfloat[ M0 \label{subfig: w_M0_sim4}]{\includegraphics[width=0.3\textwidth, height=0.2\textwidth]{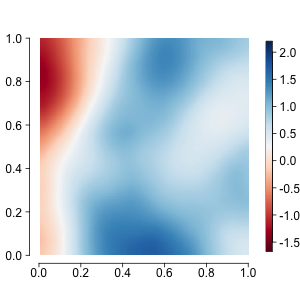}}
\hspace{0.6cm}
\subfloat[ MCMC \label{subfig: w_MCMC_sim4}]{\includegraphics[width=0.3\textwidth, height=0.2\textwidth]{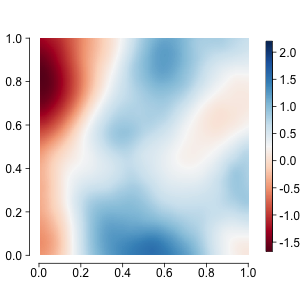}}
\caption{Interpolated maps of (a) the latent process $z(s)$ and (b-g) the expected $z(s)$ on all $n = 700$ sampled locations generated by all competing algorithms for the example from the second simulation. The $n = 700$ locations include both observed and unobserved  locations}\label{fig: w_compar_sim4}
\end{figure}
\captionsetup{labelfont={color=black},textfont={color=black}}

\begin{figure}[t]
\centering
\subfloat[ raw $y(s)$ \label{subfig: y_held_sim2}]{\includegraphics[width=0.3\textwidth, height=0.2\textwidth]{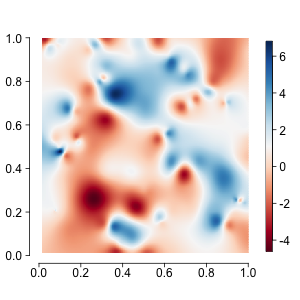}} \hspace{0.6cm}
\subfloat[ raw $x(s)^\T \beta + z(s)$ \label{subfig: y_held_denoise_sim2}]{\includegraphics[width=0.3\textwidth, height=0.2\textwidth]{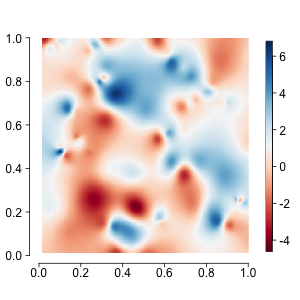}} \hspace{0.6cm}
\subfloat[ stacking of mean \label{subfig: y_pred_LSE_sim2}]{\includegraphics[width=0.3\textwidth, height=0.2\textwidth]{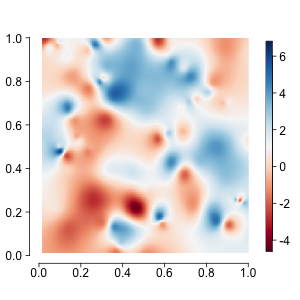}}
\hspace{0.6cm}\\
\centering
\subfloat[ stacking of predictive densities \label{subfig: y_pred_LP_sim2}]{\includegraphics[width=0.3\textwidth, height=0.2\textwidth]{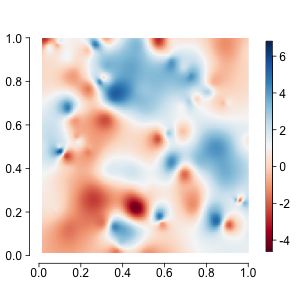}} 
\hspace{0.6cm}
\subfloat[ M0 \label{subfig: y_pred_M0_sim2}]{\includegraphics[width=0.3\textwidth, height=0.2\textwidth]{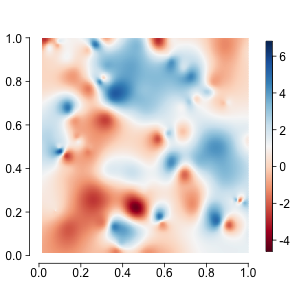}}
\hspace{0.6cm}
\subfloat[ MCMC \label{subfig: y_pred_MCMC_sim2}]{\includegraphics[width=0.3\textwidth, height=0.2\textwidth]{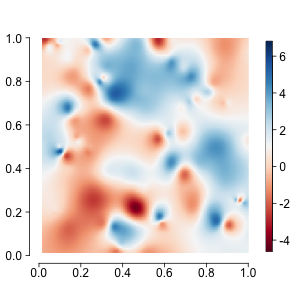}}\\
\caption{Interpolated maps of (a) the response $y(s)$, (b) the denoised response $x(s)\beta + z(s)$ and (c-f) the expected $y(s)$ on the $n_h = 100$ held out locations generated by all competing algorithms for the example with $400$ observations from the third simulation.}\label{fig: y_pred_compar_sim2}
\end{figure}

\begin{figure}[t]
\centering
\subfloat[ raw \label{subfig: w_raw_sim2}]{\includegraphics[width=0.3\textwidth, height=0.2\textwidth]{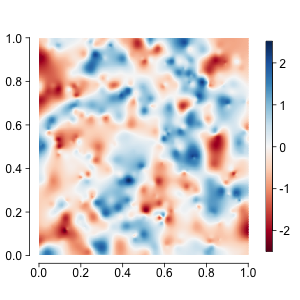}} \hspace{0.6cm}
\subfloat[ stacking of mean \label{subfig: w_LSE_sim2}]{\includegraphics[width=0.3\textwidth, height=0.2\textwidth]{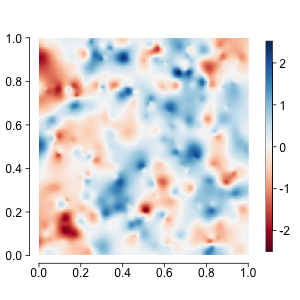}}
\hspace{0.6cm}
\subfloat[ stacking of predictive densities \label{subfig: w_LP_sim2}]{\includegraphics[width=0.3\textwidth, height=0.2\textwidth]{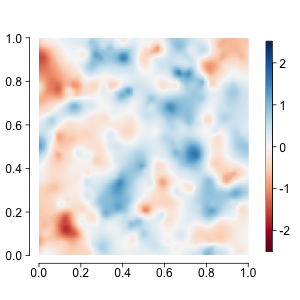}} \\
\centering
\subfloat[ M0 \label{subfig: w_M0_sim2}]{\includegraphics[width=0.3\textwidth, height=0.2\textwidth]{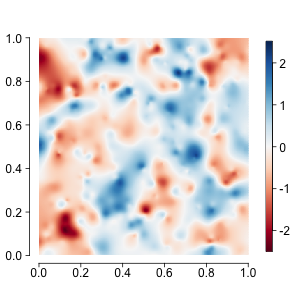}}
\hspace{0.6cm}
\subfloat[ MCMC \label{subfig: w_MCMC_sim2}]{\includegraphics[width=0.3\textwidth, height=0.2\textwidth]{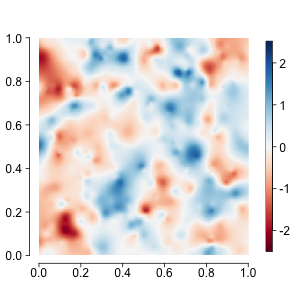}}
\caption{Interpolated maps of (a) the latent process $z(s)$ and (b-g) the expected $z(s)$ on all $n = 500$ sampled locations generated by all competing algorithms for the example from the third simulation. The $n = 500$ locations include both observed and unobserved locations}\label{fig: w_compar_sim2}
\end{figure}

\begin{figure}[t]
\centering
\subfloat[ raw $y(s)$ \label{subfig: y_held_sim3}]{\includegraphics[width=0.3\textwidth, height=0.2\textwidth]{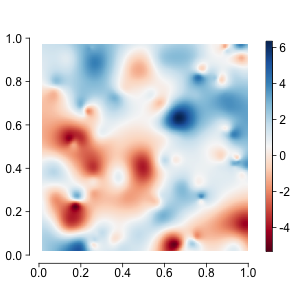}} \hspace{0.6cm}
\subfloat[ raw $x(s)^\T \beta + z(s)$ \label{subfig: y_held_denoise_sim3}]{\includegraphics[width=0.3\textwidth, height=0.2\textwidth]{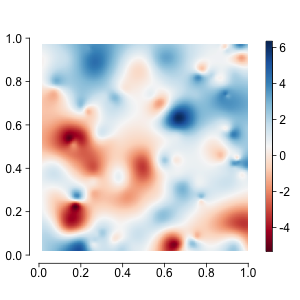}} \hspace{0.6cm}
\subfloat[ stacking of mean \label{subfig: y_pred_LSE_sim3}]{\includegraphics[width=0.3\textwidth, height=0.2\textwidth]{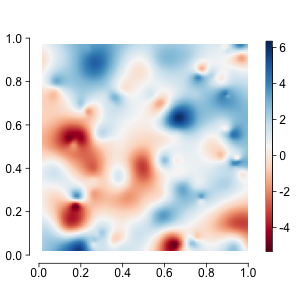}}
\hspace{0.6cm}\\
\centering
\subfloat[ stacking of predictive densities \label{subfig: y_pred_LP_sim3}]{\includegraphics[width=0.3\textwidth, height=0.2\textwidth]{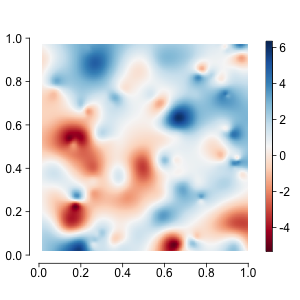}} 
\hspace{0.6cm}
\subfloat[ M0 \label{subfig: y_pred_M0_sim3}]{\includegraphics[width=0.3\textwidth, height=0.2\textwidth]{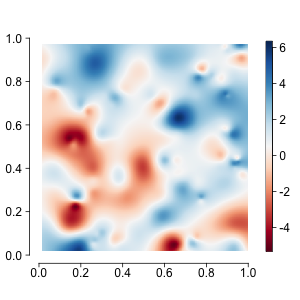}}
\hspace{0.6cm}
\subfloat[ MCMC \label{subfig: y_pred_MCMC_sim3}]{\includegraphics[width=0.3\textwidth, height=0.2\textwidth]{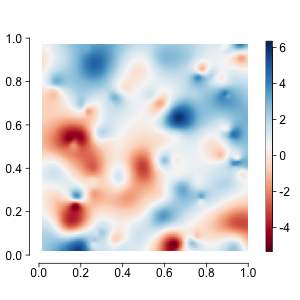}}\\
\caption{Interpolated maps of (a) the response $y(s)$, (b) the denoised response $x(s)\beta + z(s)$ and (c-f) the expected $y(s)$ on the $n_h = 100$ held out locations generated by all competing algorithms for the example with $200$ observations from the fourth simulation.}\label{fig: y_pred_compar_sim3}
\end{figure}

\begin{figure}[t]
\centering
\subfloat[ raw \label{subfig: w_raw_sim3}]{\includegraphics[width=0.3\textwidth, height=0.2\textwidth]{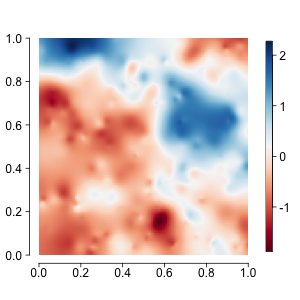}} \hspace{0.6cm}
\subfloat[ stacking of mean \label{subfig: w_LSE_sim3}]{\includegraphics[width=0.3\textwidth, height=0.2\textwidth]{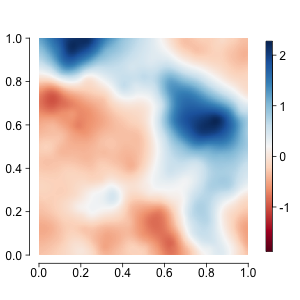}}
\hspace{0.6cm}
\subfloat[ stacking of predictive densities \label{subfig: w_LP_sim3}]{\includegraphics[width=0.3\textwidth, height=0.2\textwidth]{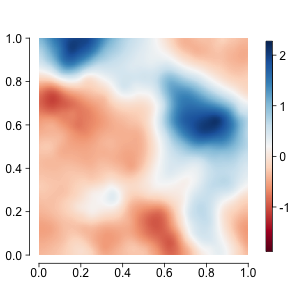}} \\
\centering
\subfloat[ M0 \label{subfig: w_M0_sim3}]{\includegraphics[width=0.3\textwidth, height=0.2\textwidth]{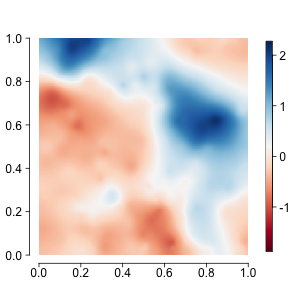}}
\hspace{0.6cm}
\subfloat[ MCMC \label{subfig: w_MCMC_sim3}]{\includegraphics[width=0.3\textwidth, height=0.2\textwidth]{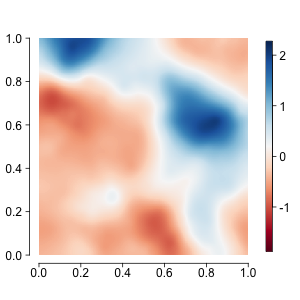}}
\caption{Interpolated maps of (a) the latent process $z(s)$ and (b-g) the expected $z(s)$ on all $n = 300$ sampled locations generated by all competing algorithms for the example from the fourth simulation. The $n = 300$ locations include both observed and unobserved locations}\label{fig: w_compar_sim3}
\end{figure}
\captionsetup{labelfont={color=black},textfont={color=black}}


\subsection{Inference of prefixed hyper-parameters}\label{app: Inf_hyper}
A limitation of stacking, compared to full Bayesian inference (e.g., using MCMC), is that it does not provide inference for the prefixed hyper-parameters. If we treat the grid of the candidate values for the hyper-parameters in our stacking algorithms as a discrete uniform prior, then, intuitively, one might be tempted to treat the stacking weights as probability masses on the support of the hyper-parameters. This intuition, however, is incorrect. Figure~\ref{fig: phi_compar} compares the point estimates of $\phi$ based on stacking for the simulation studies. It is clear that stacking of means yields unstable estimates. Stacking of predictive densities has a smaller variance, but the bias can be large. Also, since $\phi$ is not identifiable, we observe that the posterior interval estimates for $\phi$ inferred from MCMC algorithms are wide, showing that the inference for $\phi$ is relatively unstable for all candidate algorithms in this simulation study. Figures~\ref{fig: nu_compar}~and~\ref{fig: deltasq_compar} presents similar comparisons for the other two hyper-parameters. 

\begin{figure}[t]
\centering
\subfloat[\label{subfig: sim1_phi}]{
\includegraphics[width=0.49\textwidth, keepaspectratio]{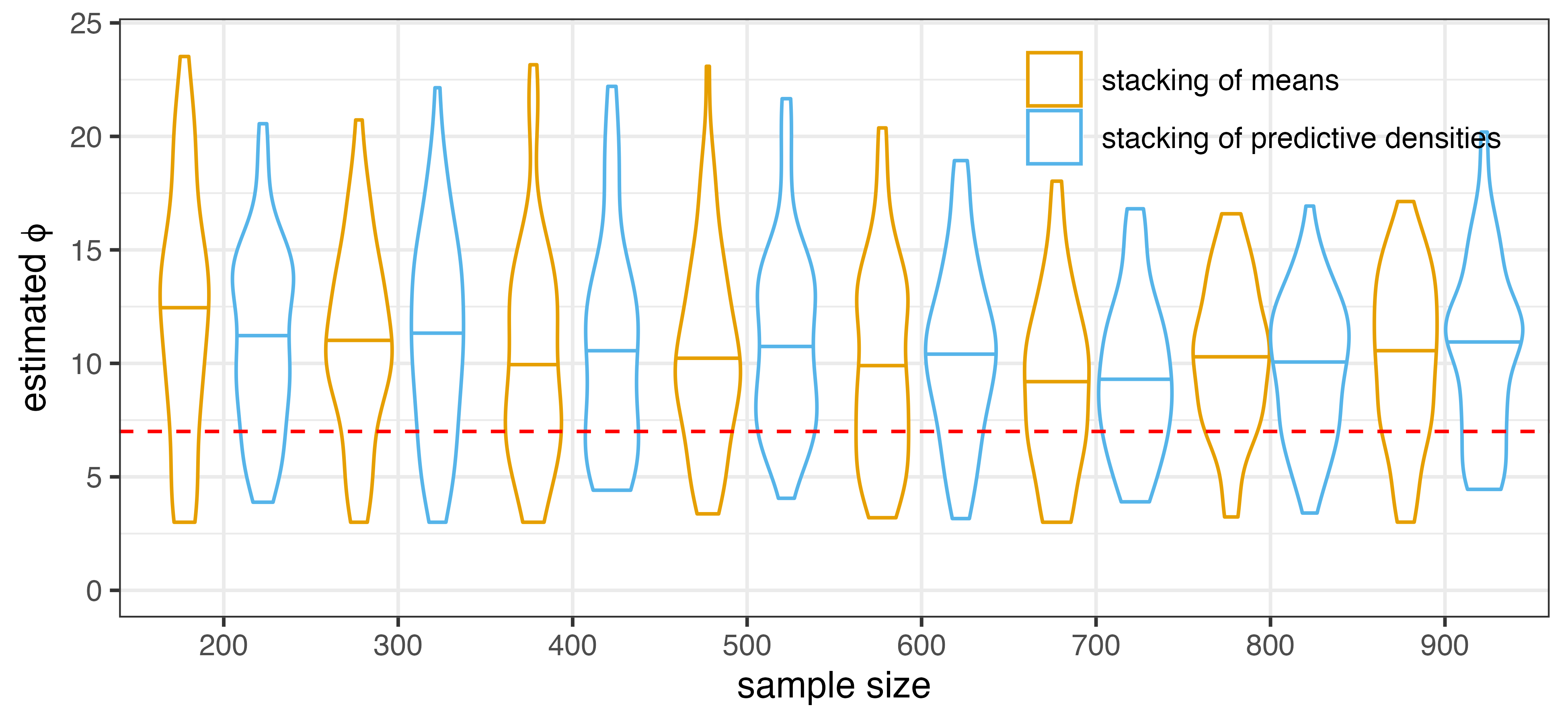}}
\subfloat[\label{subfig: sim4_phi}]{
\includegraphics[width=0.49\textwidth, keepaspectratio]{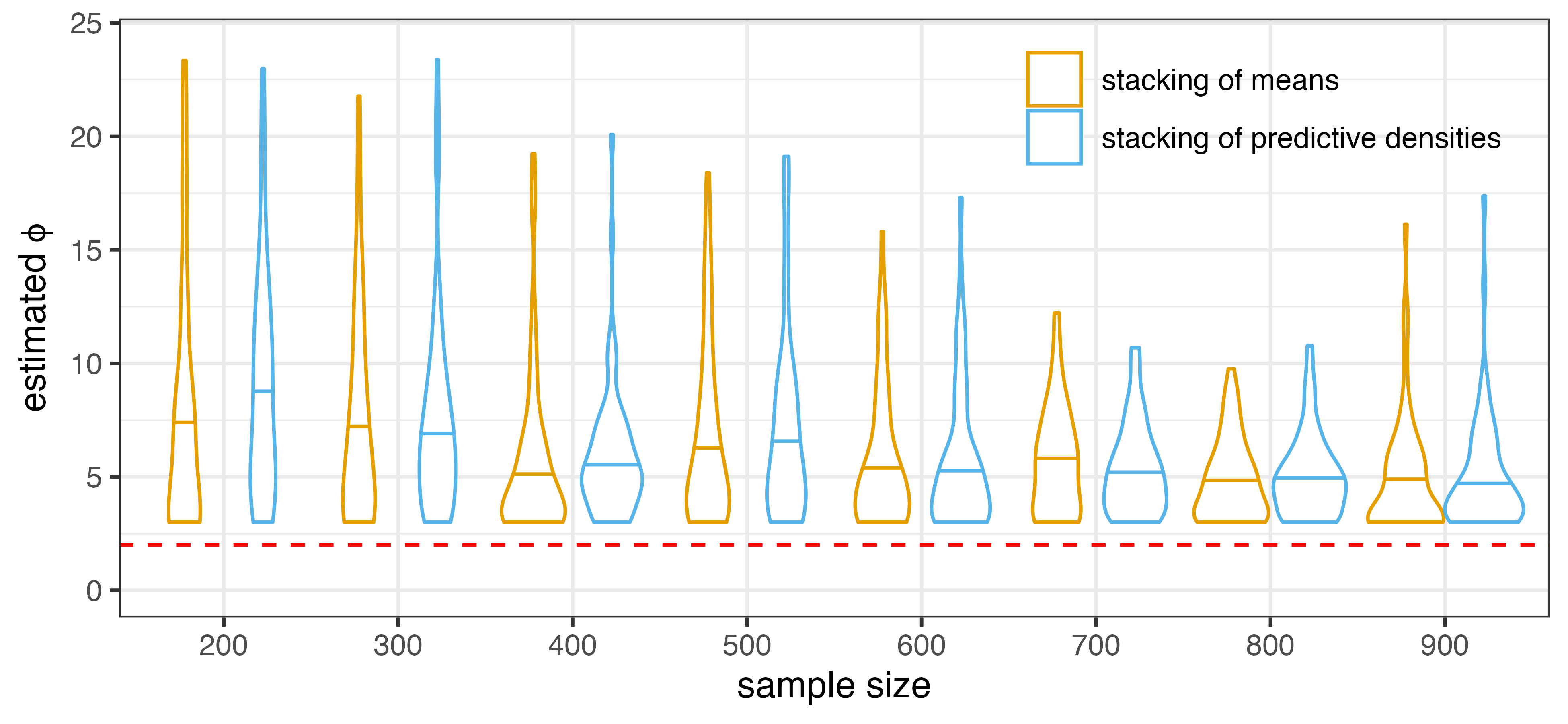}}\\
\subfloat[ \label{subfig: sim2_phi}]{
\includegraphics[width=0.49\textwidth, keepaspectratio]{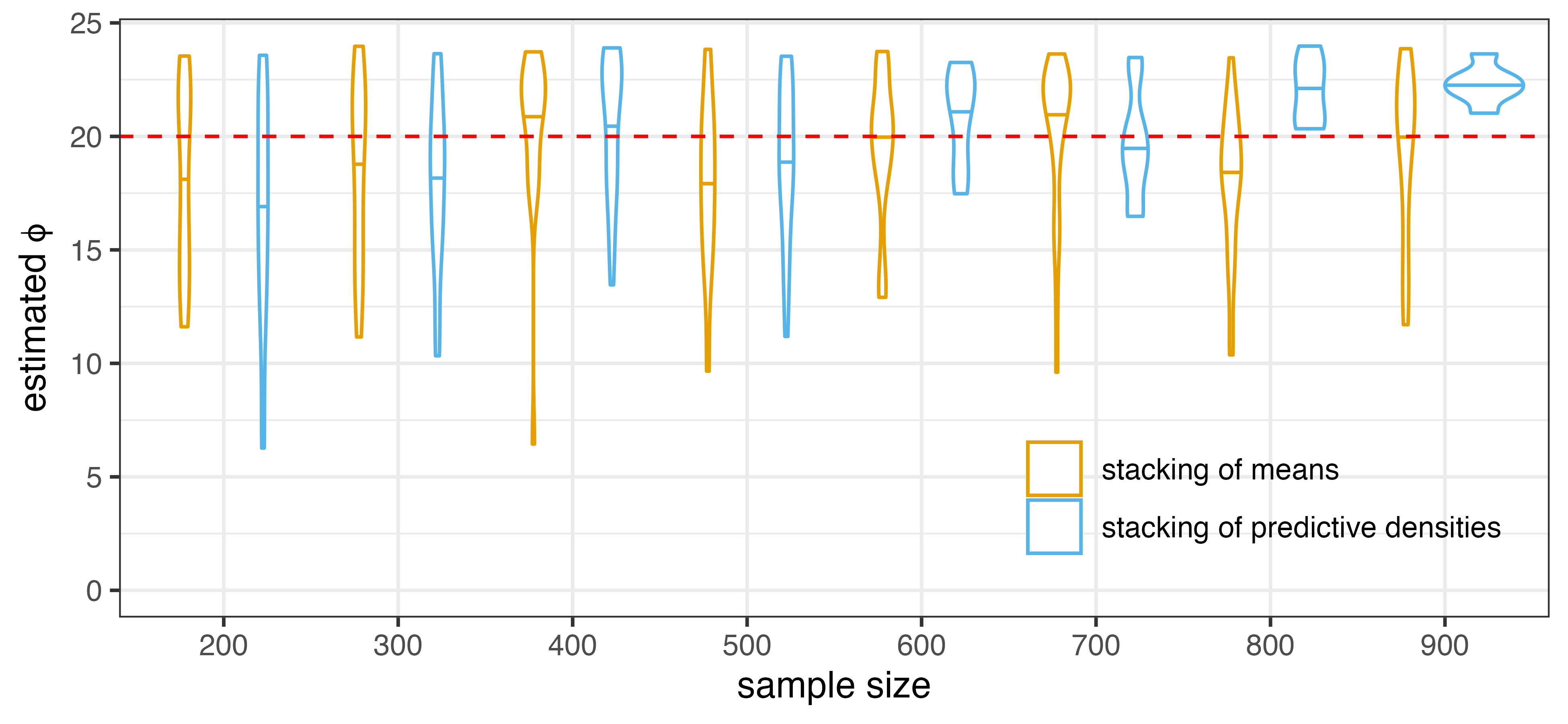}}
\subfloat[ \label{subfig: sim3_phi}]{
\includegraphics[width=0.49\textwidth, keepaspectratio]{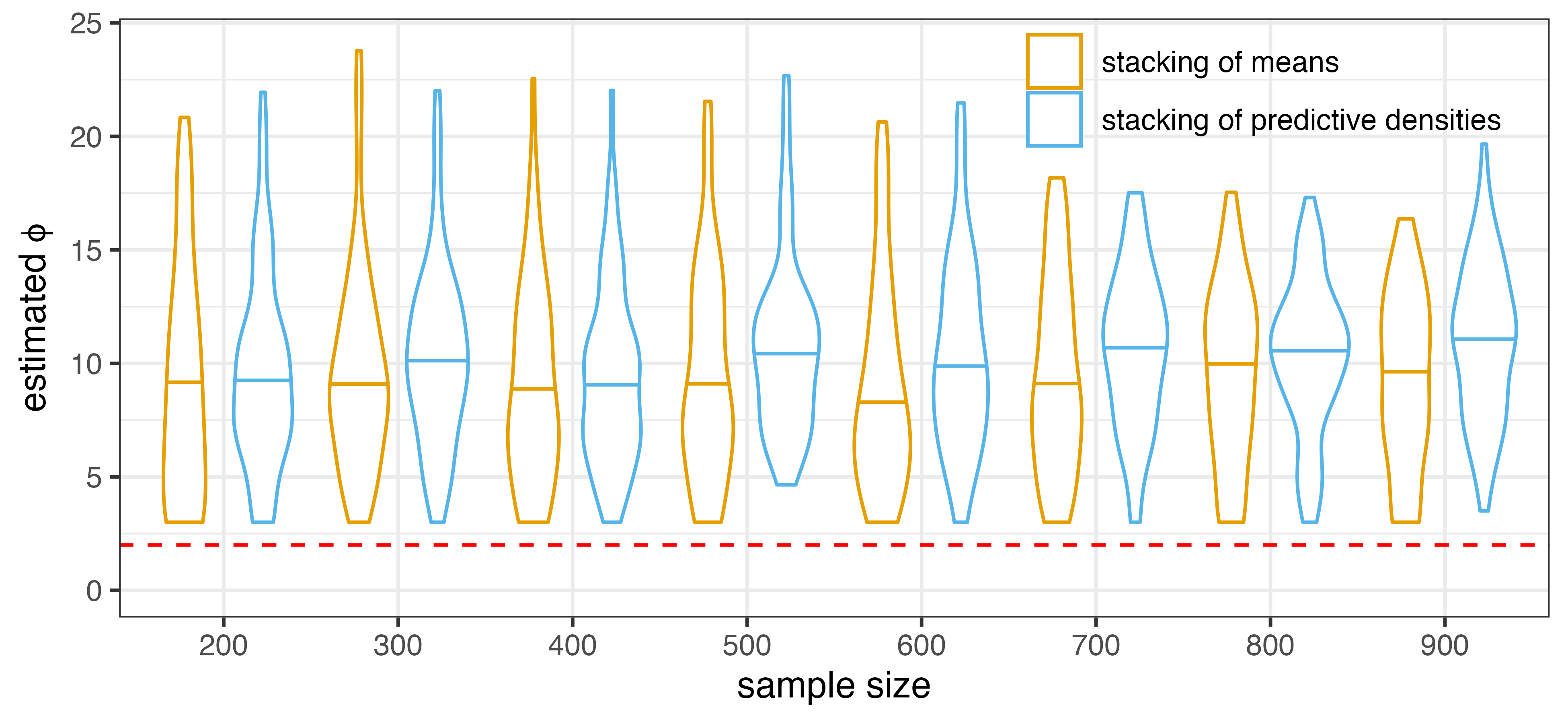}}
\caption{Distributions of the estimated $\phi$ in the first (a), second (b), third (c) and fourth (d) simulations. The distributions are described by violin plots whose horizontal lines indicate the medians. The red dashed horizontal line indicates the actual value of $\phi$}\label{fig: phi_compar}
\end{figure}

\begin{figure}[t]
\vspace{-0.5cm}
\centering
\subfloat[\label{subfig: sim1_nu}]{\includegraphics[width=0.49\textwidth, keepaspectratio]
{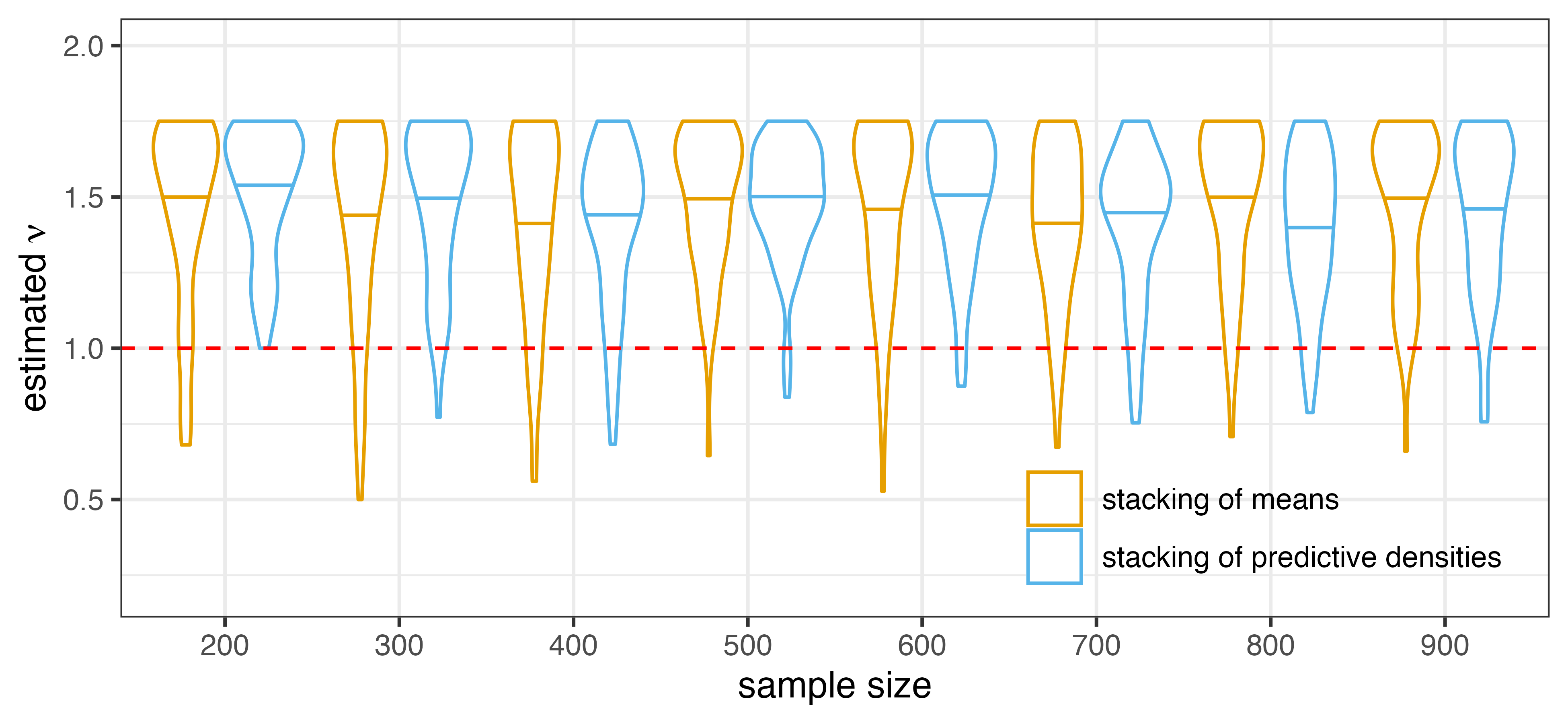}}
\subfloat[ \label{subfig: sim4_nu}]{\includegraphics[width=0.49\textwidth,  keepaspectratio]
{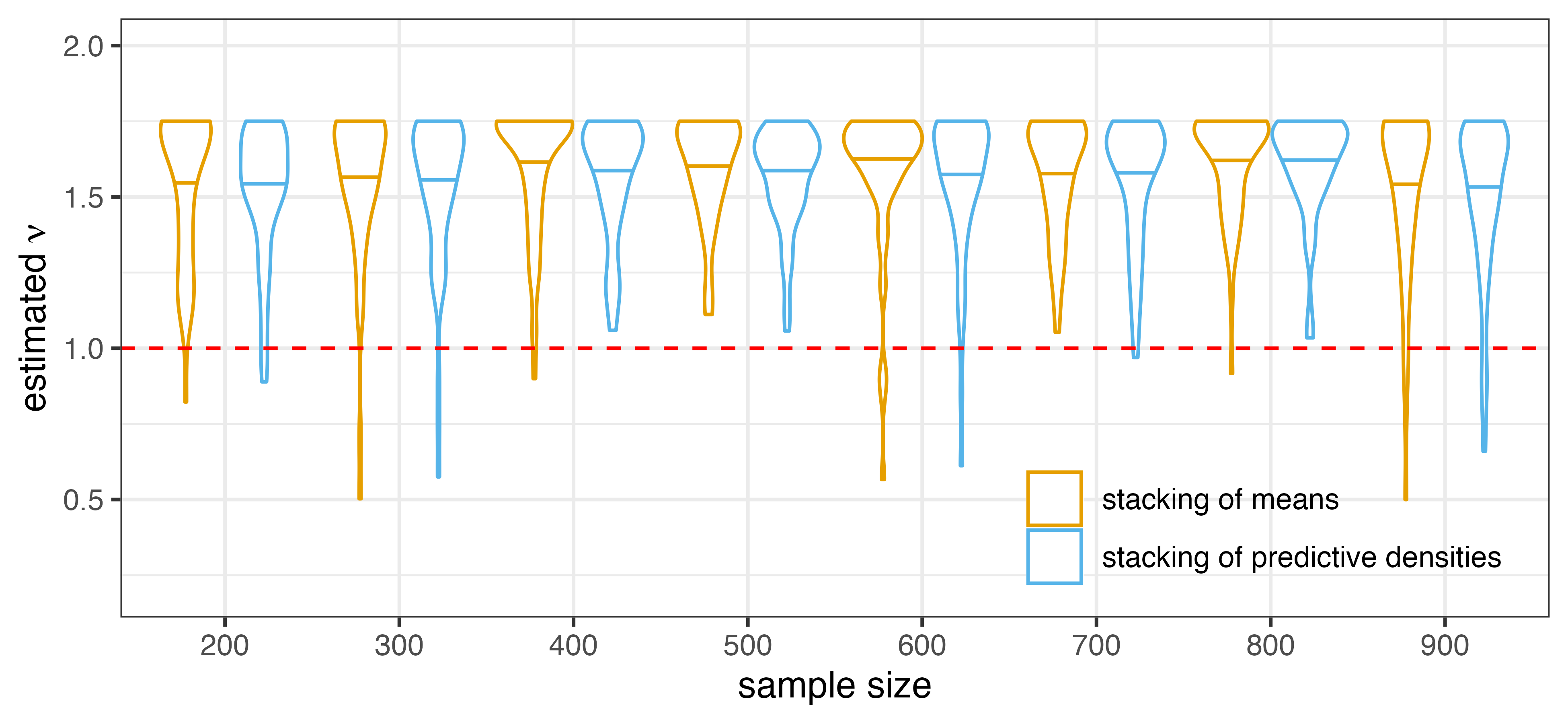}}\\
\subfloat[ \label{subfig: sim2_nu}]{\includegraphics[width=0.49\textwidth,  keepaspectratio]
{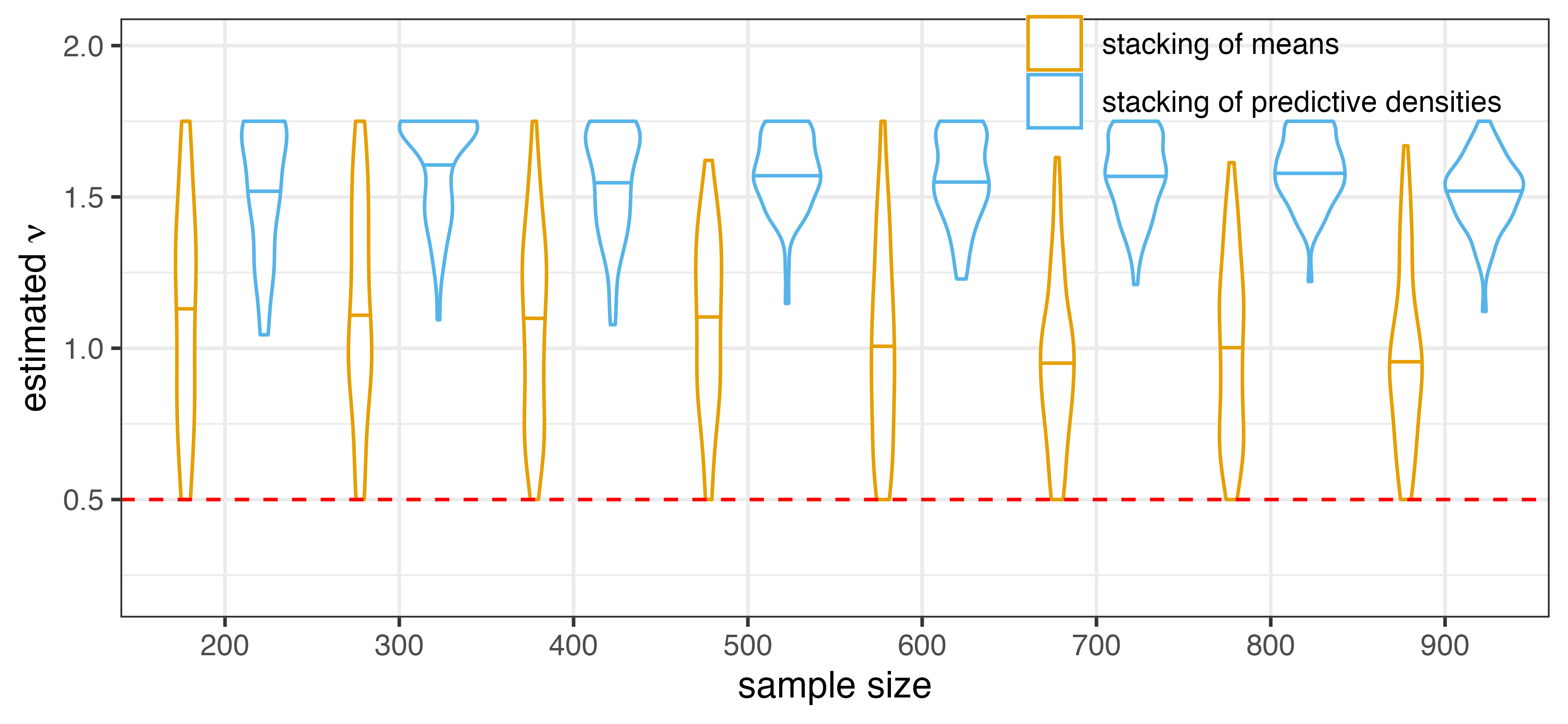}}
\subfloat[ \label{subfig: sim3_nu}]{\includegraphics[width=0.49\textwidth,  keepaspectratio]
{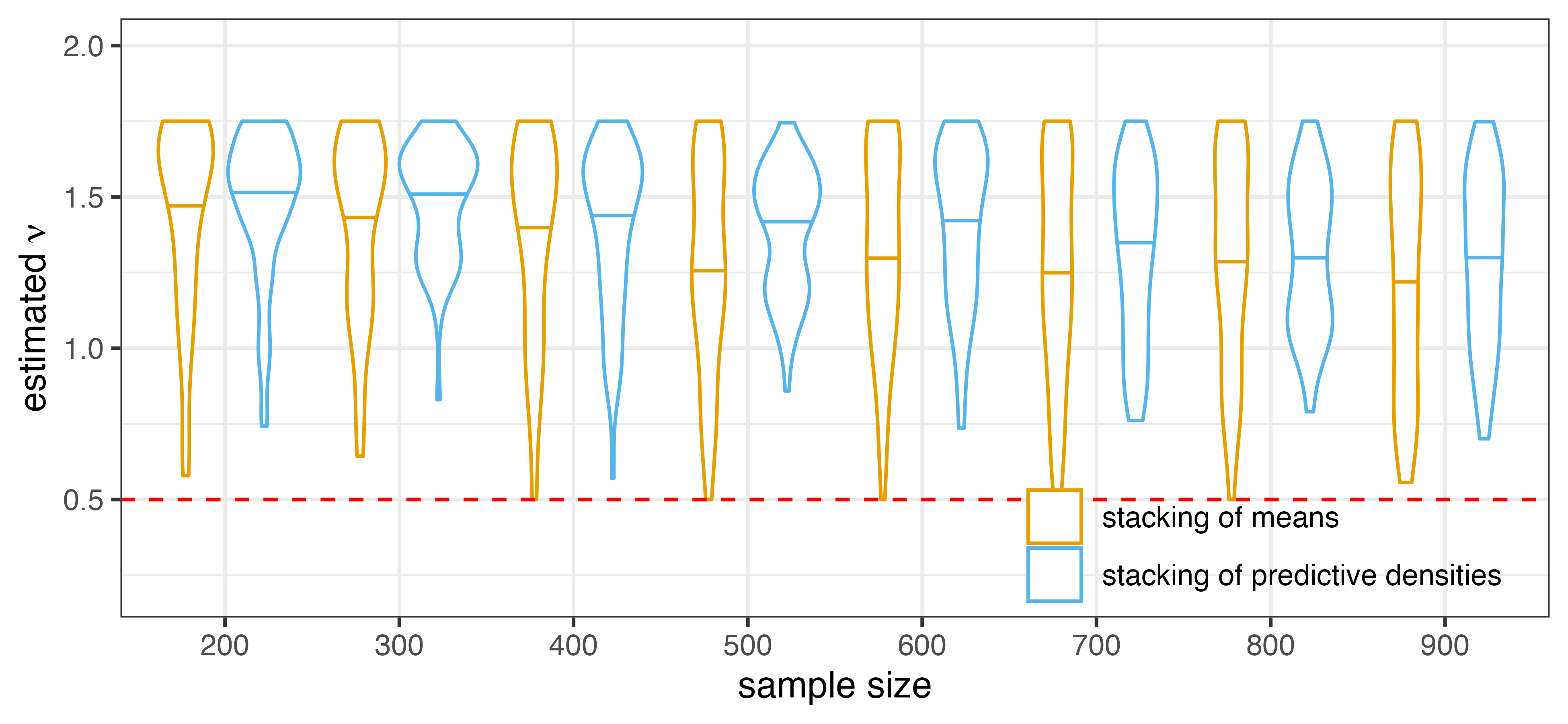}}
\vspace{-0.2cm}
\caption{Distributions of the estimated $\nu$ in the first (a), second (b), third (c) and fourth (d) simulations. The distribution of the counts are described through violin plots whose horizontal lines indicate the medians. The red dashed horizontal line indicates the actual value of $\nu$}\label{fig: nu_compar}
\end{figure}

\begin{figure}[t]
\centering
\subfloat[\label{subfig: sim1_deltasq}]{\includegraphics[width=0.49\textwidth, keepaspectratio]
{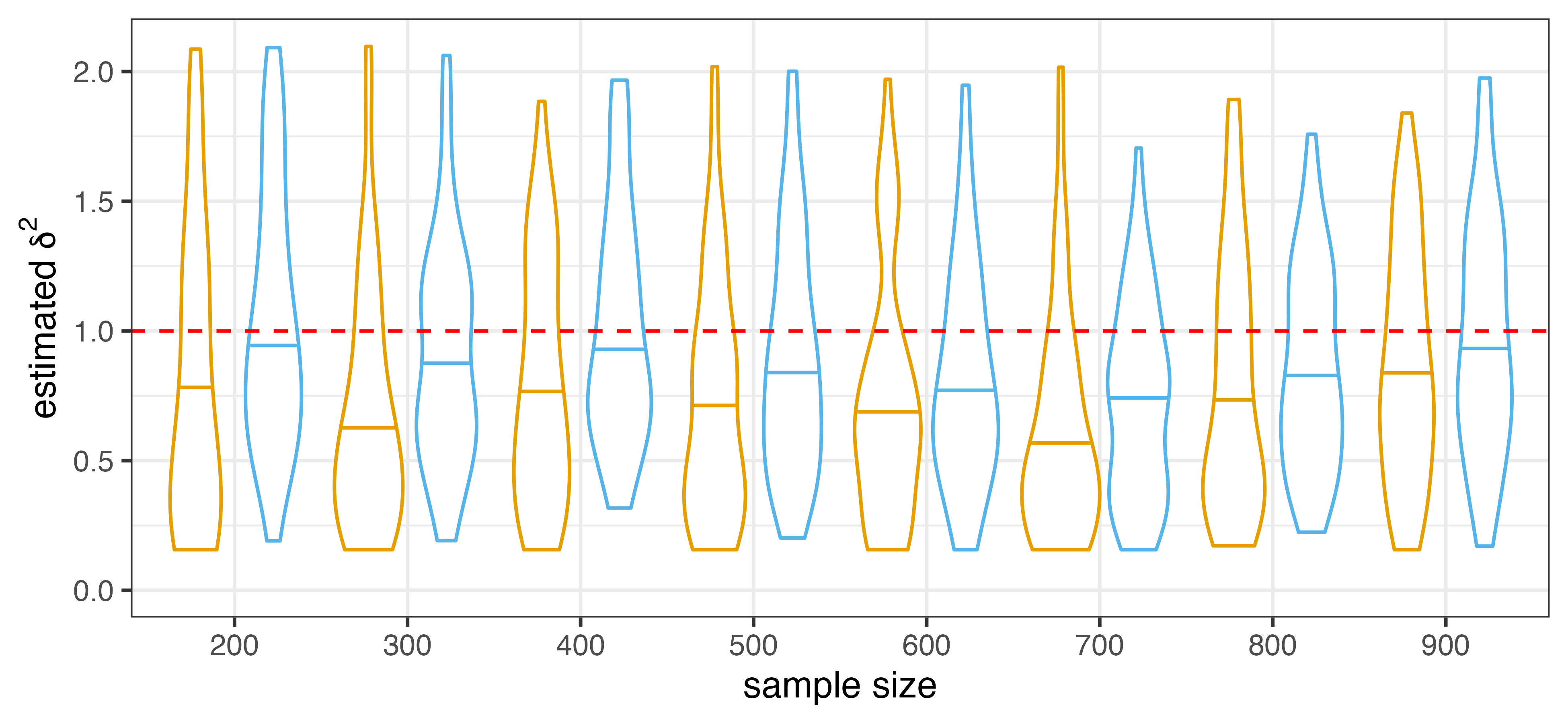}}
\subfloat[ \label{subfig: sim4_deltasq}]{\includegraphics[width=0.49\textwidth, keepaspectratio]
{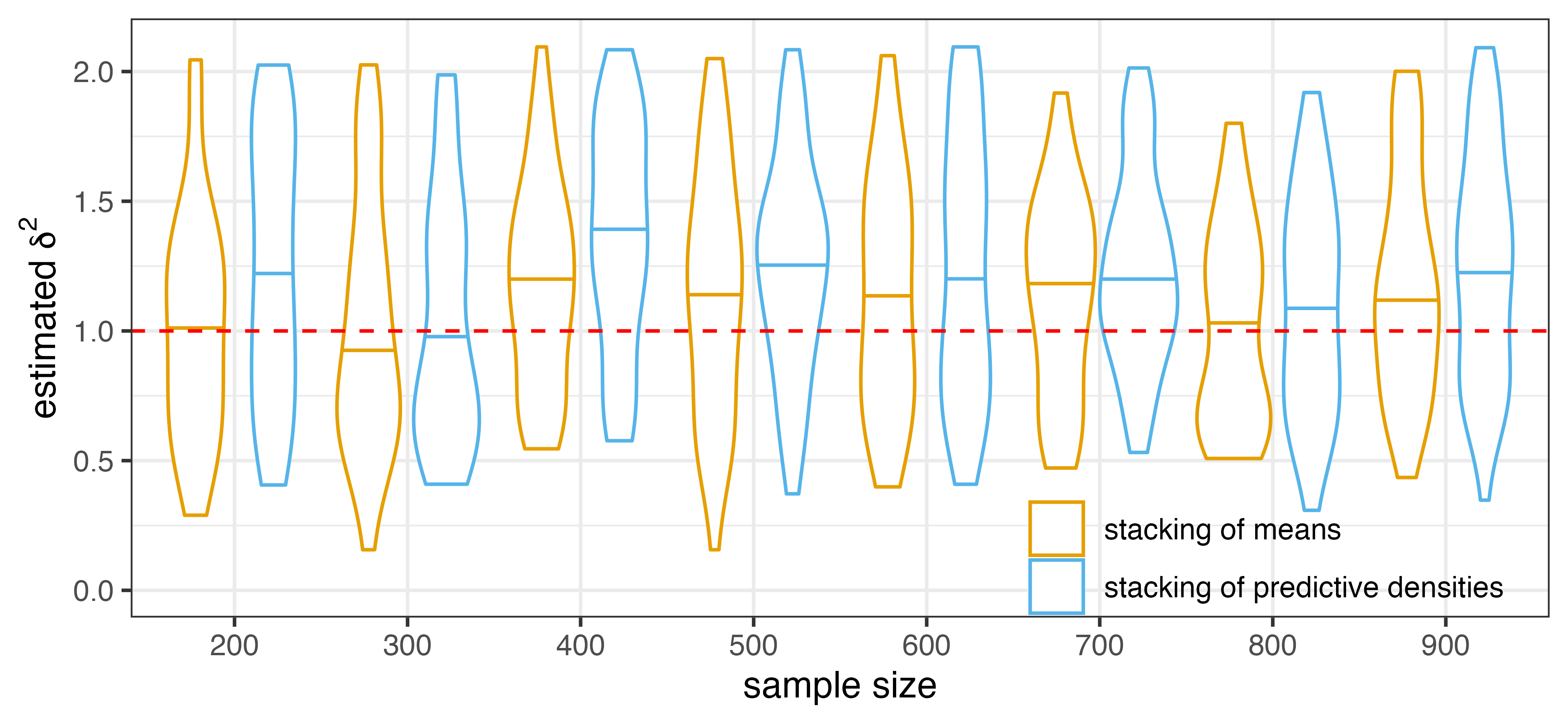}}\\
\subfloat[ \label{subfig: sim2_deltasq}]{\includegraphics[width=0.49\textwidth, keepaspectratio]
{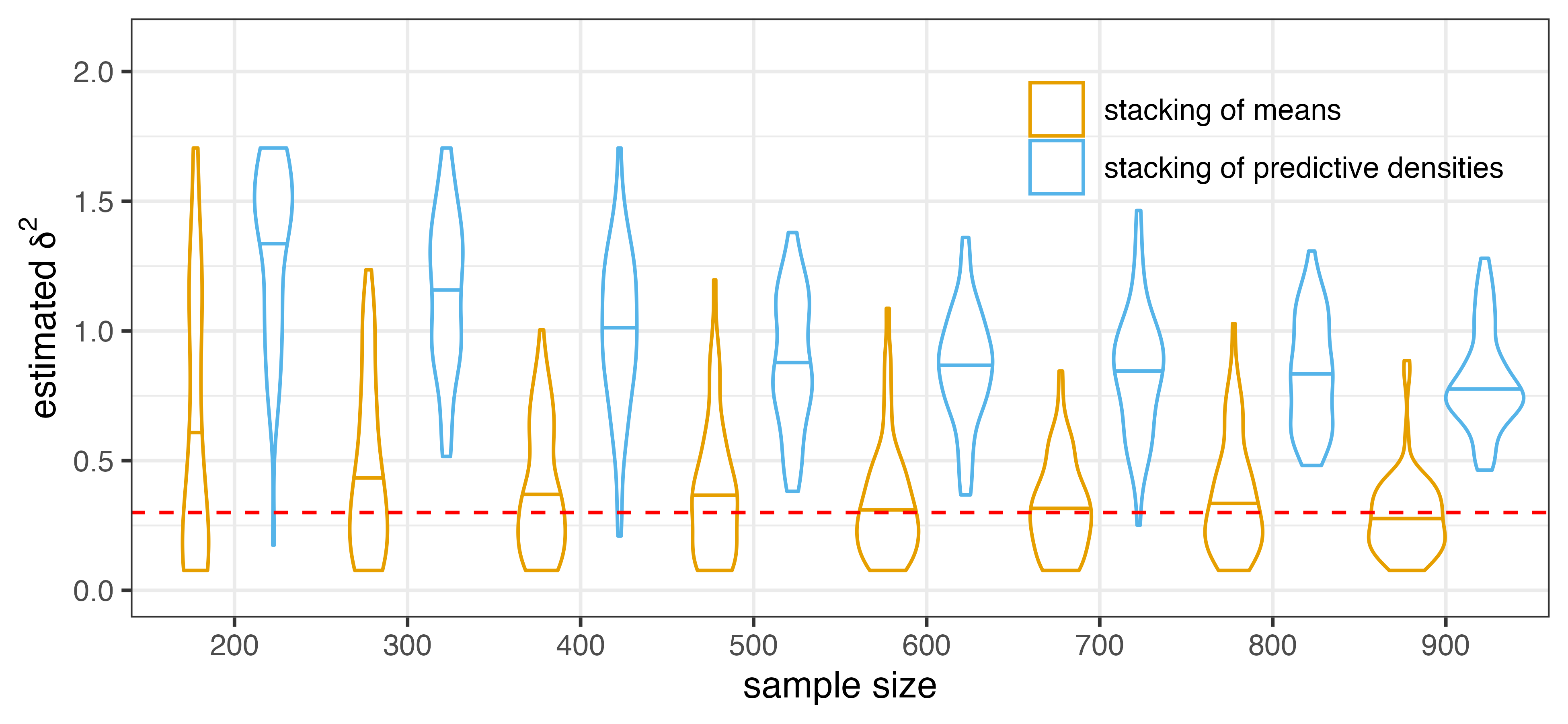}}
\subfloat[ \label{subfig: sim3_deltasq}]{\includegraphics[width=0.49\textwidth, keepaspectratio]
{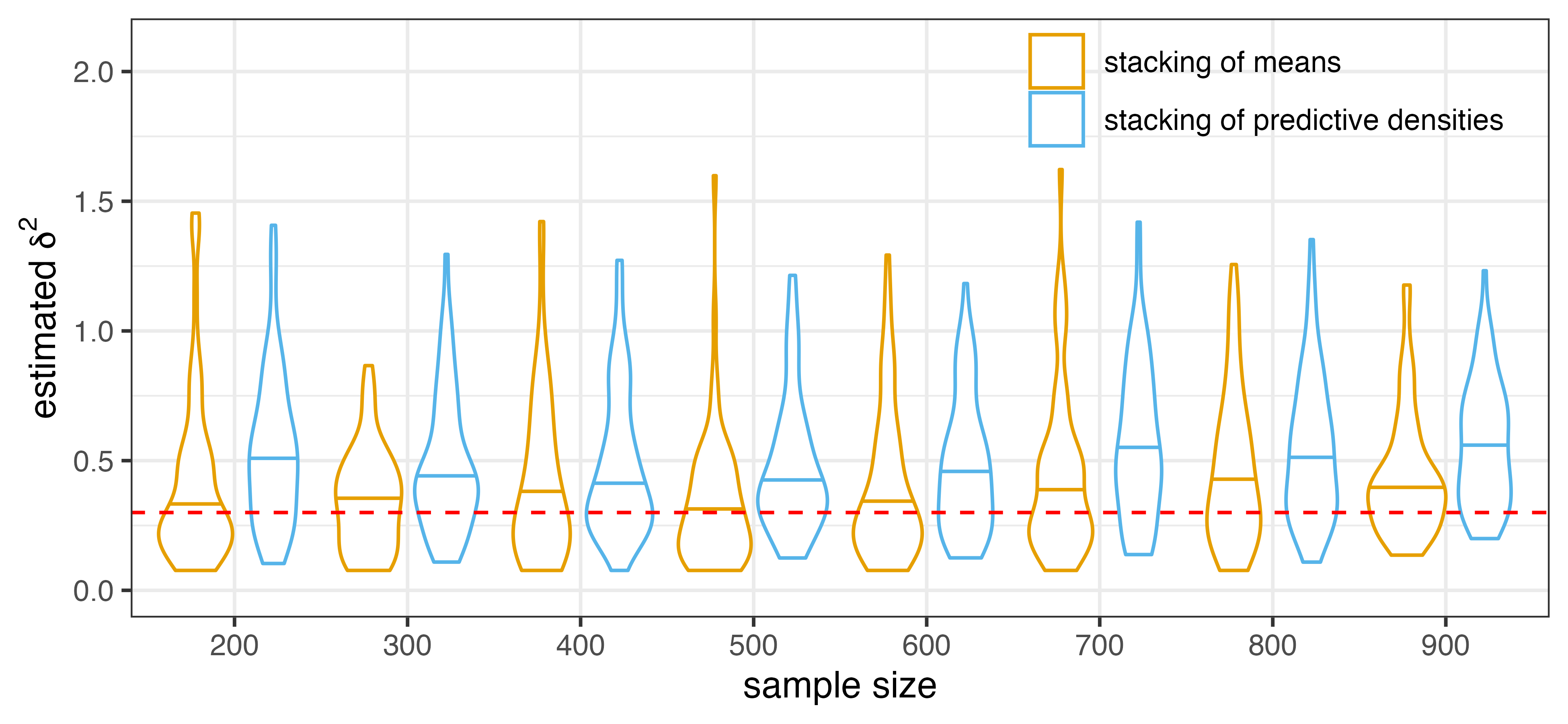}}
\caption{Distributions of the estimated $\delta^2$ in the first (a), second (b), third (c), and fourth (d) simulations. The distribution of the counts are described through violin plots whose horizontal lines indicate the medians. The red dashed horizontal line indicates the actual value of $\delta^2$}\label{fig: deltasq_compar}
\end{figure}

\subsection{Plots for the simulation study in Section~\ref{subsec:Impact_hyper_select}}
See Figures~\ref{fig: sim_prefix_compar}, ~\ref{fig: y_U_CI_compar_all}, ~\ref{fig: sim_prefix_tau_compar}~and~\ref{fig: sim_prefix_beta_compar}.

\begin{figure}[t]
\centering
\subfloat[\label{subfig: sim1_prefix_compar}]{\includegraphics[width=0.5\textwidth, keepaspectratio]{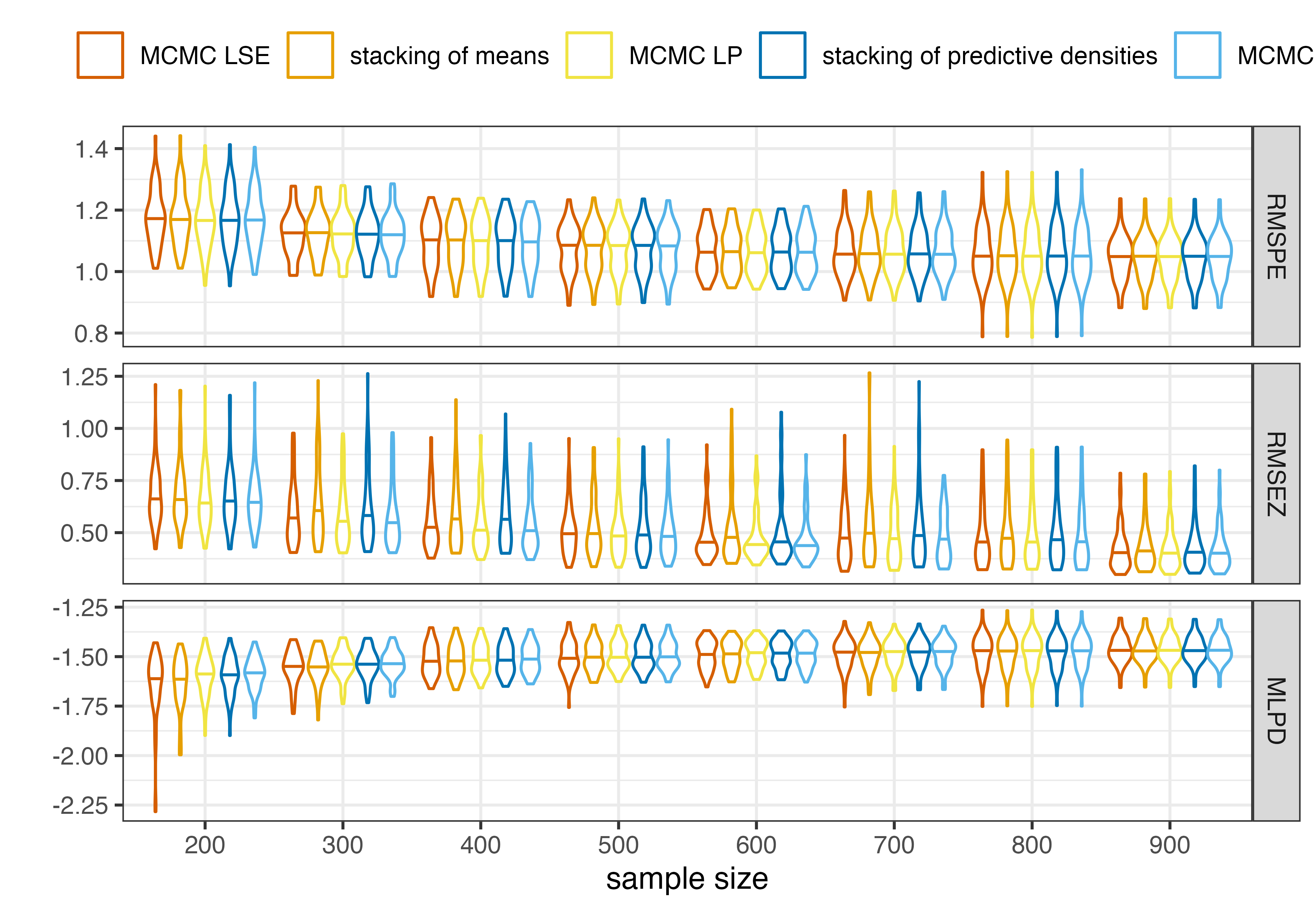}}
\subfloat[ \label{subfig: sim4_prefix_compar}]{\includegraphics[width=0.5\textwidth, keepaspectratio]{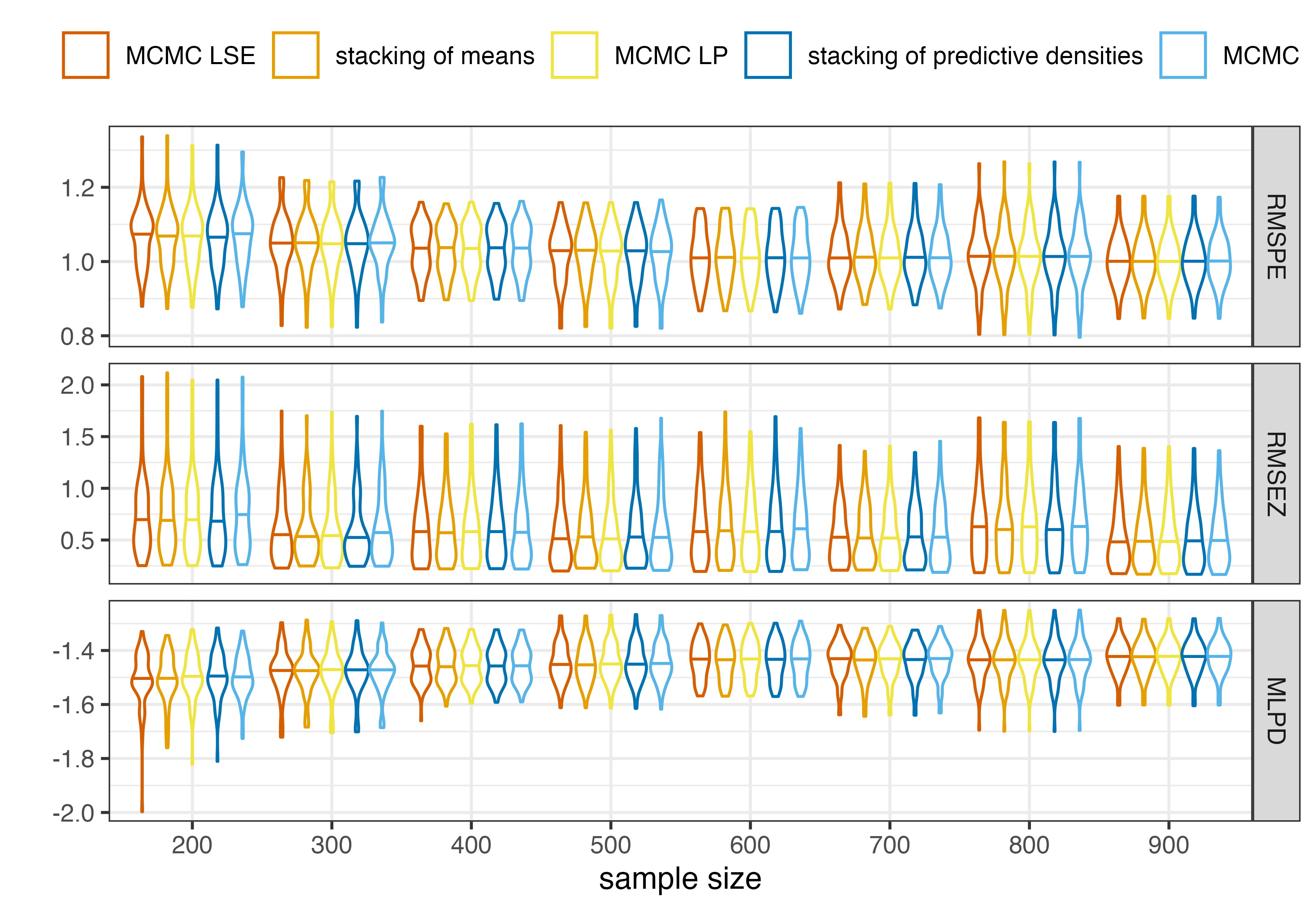}}\\
\subfloat[ \label{subfig: sim2_prefix_compar}]{\includegraphics[width=0.5\textwidth, keepaspectratio]{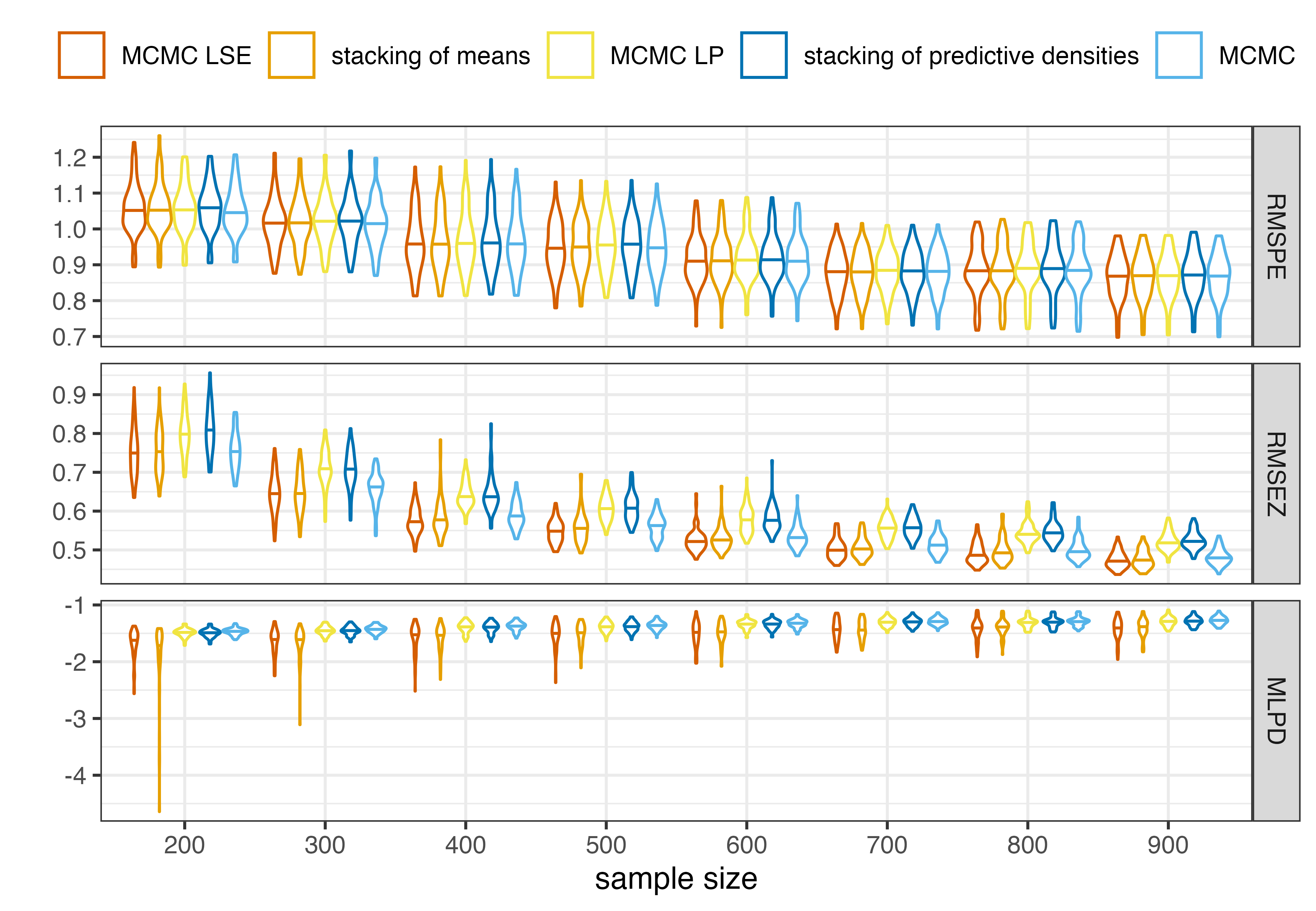}}
\subfloat[ \label{subfig: sim3_prefix_compar}]{\includegraphics[width=0.5\textwidth, keepaspectratio]{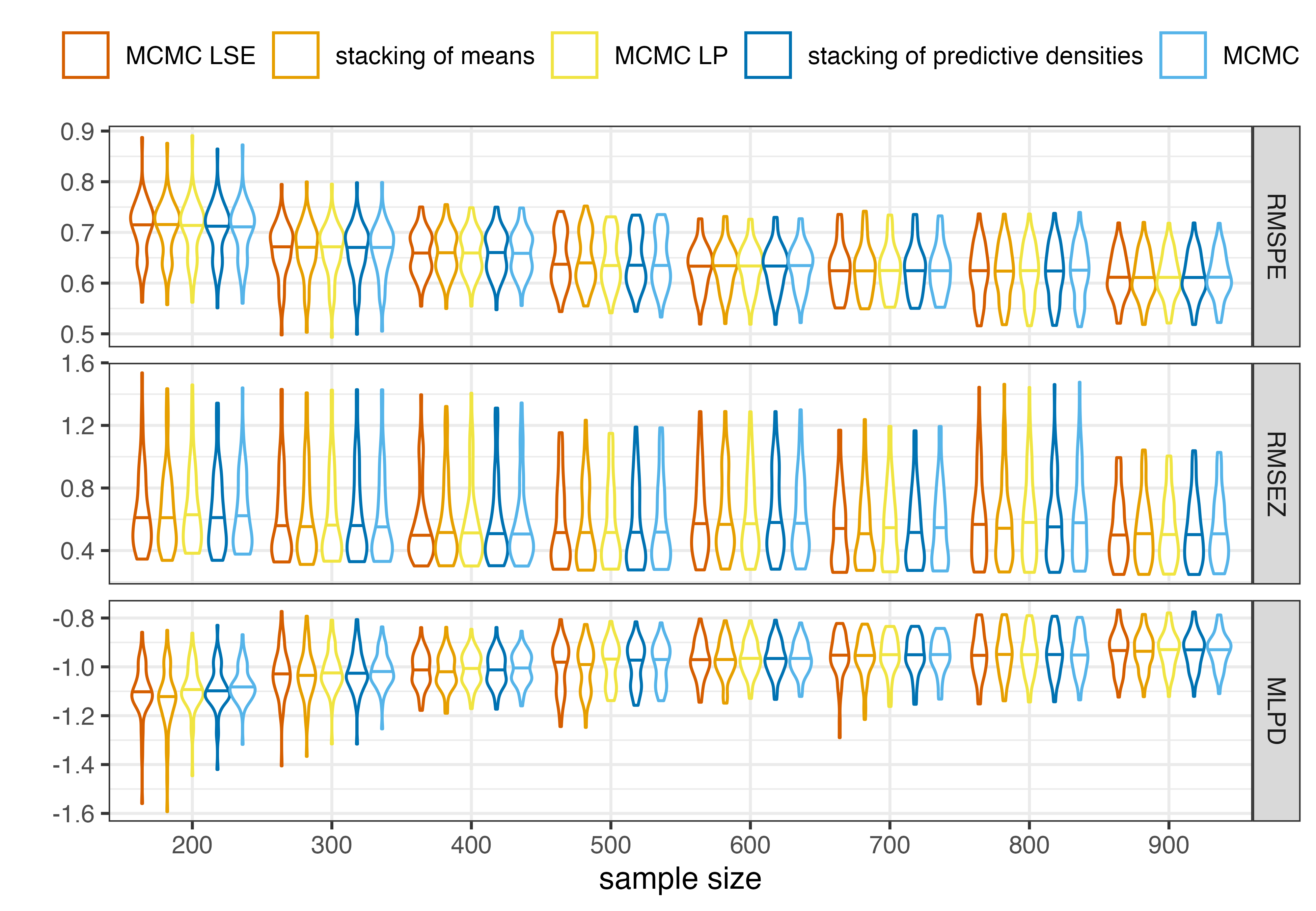}}
\caption{Distributions of the diagnostic metrics for prediction performance for the first simulation (a), second (b), third (c) and fourth (d) simulation. Label `MCMC LSE' and `MCMC LP' denote stacking of mean and stacking of predictive densities using $\phi$, $\nu$, $\delta^2$ sampled through MCMC, respectively. Each distribution is depicted through a violin plot. The horizontal line in each violin plot indicates the median. }
\label{fig: sim_prefix_compar}
\end{figure}

\begin{figure}[t]
\centering
\includegraphics[width=0.8\textwidth, height=0.275\textwidth]{pics/y_U_95CIsim1_r8.png}\\
\includegraphics[width=0.8\textwidth, height=0.275\textwidth]{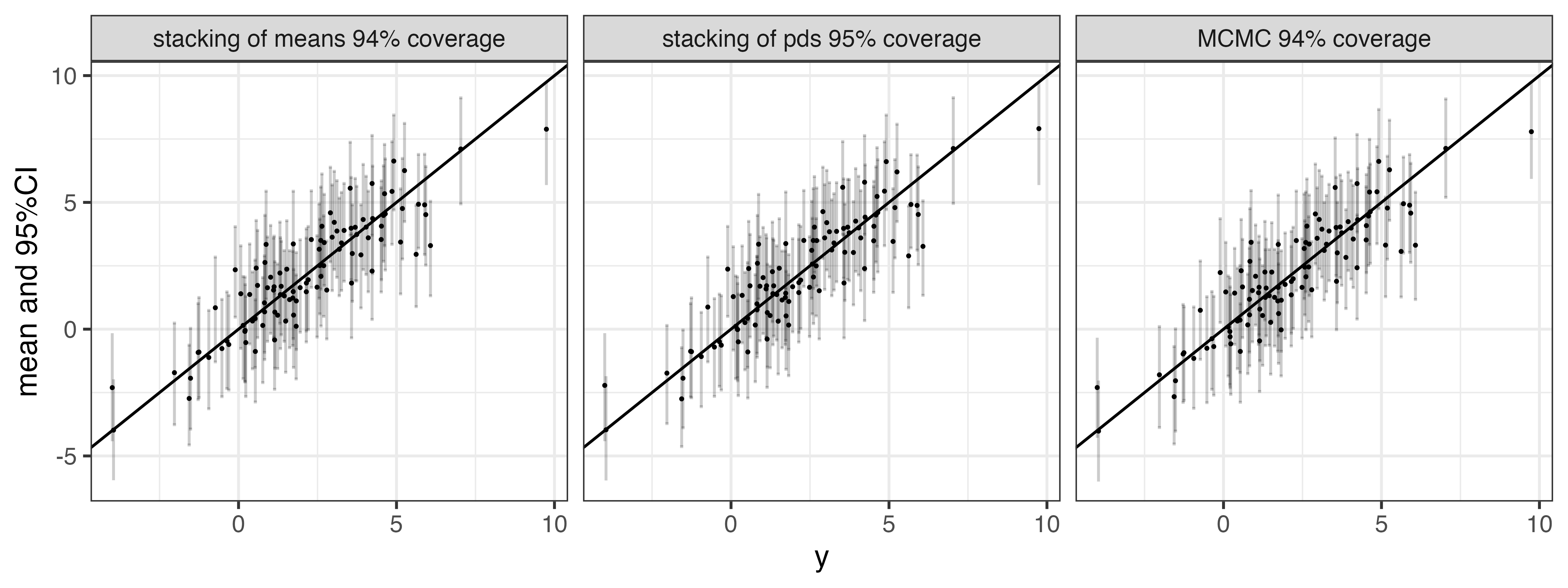}\\
\includegraphics[width=0.8\textwidth, height=0.275\textwidth]{pics/y_U_95CIsim2_r4.png}
\includegraphics[width=0.8\textwidth, height=0.275\textwidth]{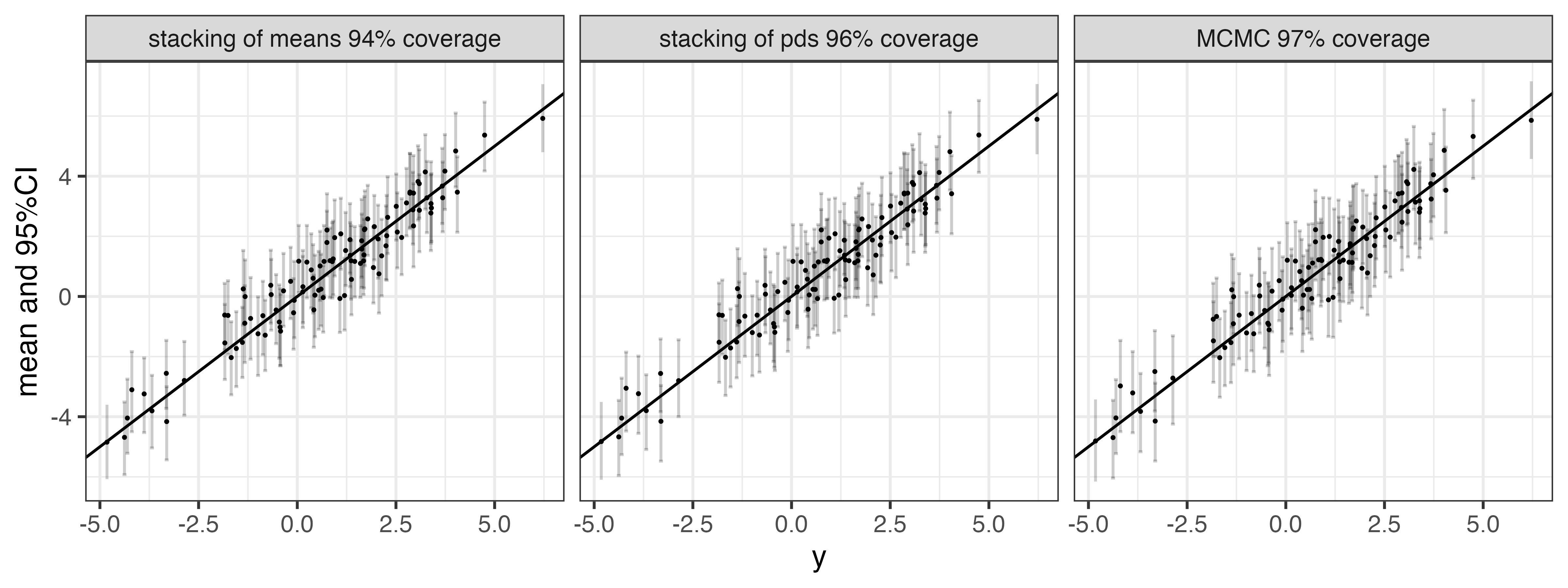}
\caption{95\% credible intervals for predicted and actual outcomes at 100 unobserved locations with $45$-degree (solid black) line indicating equality. Top row corresponds to Simulation 1 (800 observations); second row presents Simulation 2 (600 observations); third row presents Simulation 3 (400 observations) and the bottom row presents Simulation 4 (200 observations). The captions indicate coverage of 95\% credible intervals. `pds' denotes predictive densities.}\label{fig: y_U_CI_compar_all}
\end{figure}

\begin{figure}[t]
\centering
{\includegraphics[width=0.45\textwidth, keepaspectratio]{pics/sigmasq_prefix_compar1_r8.png}}
{\includegraphics[width=0.45\textwidth, keepaspectratio]{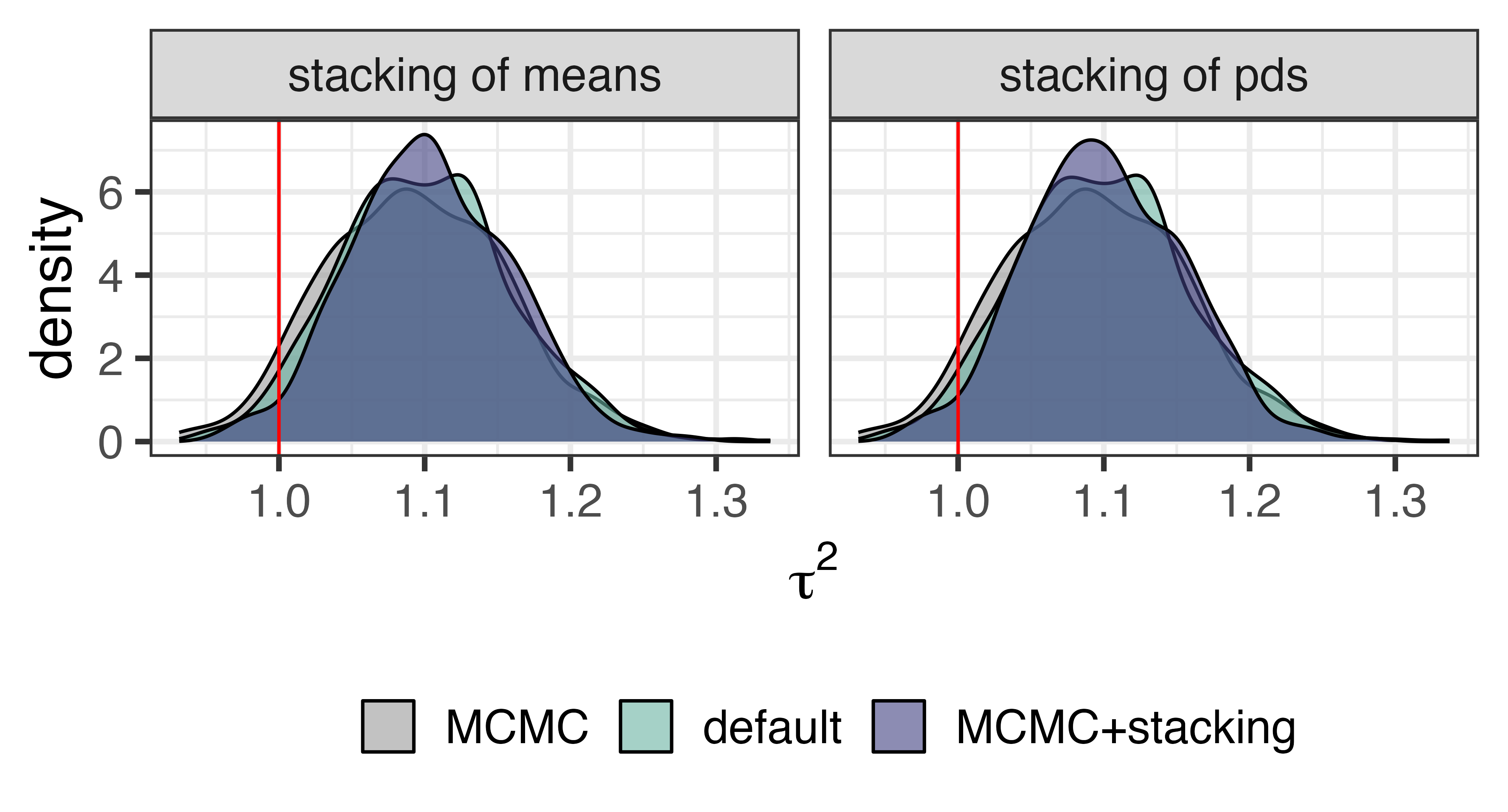}}\\
{\includegraphics[width=0.45\textwidth, keepaspectratio]{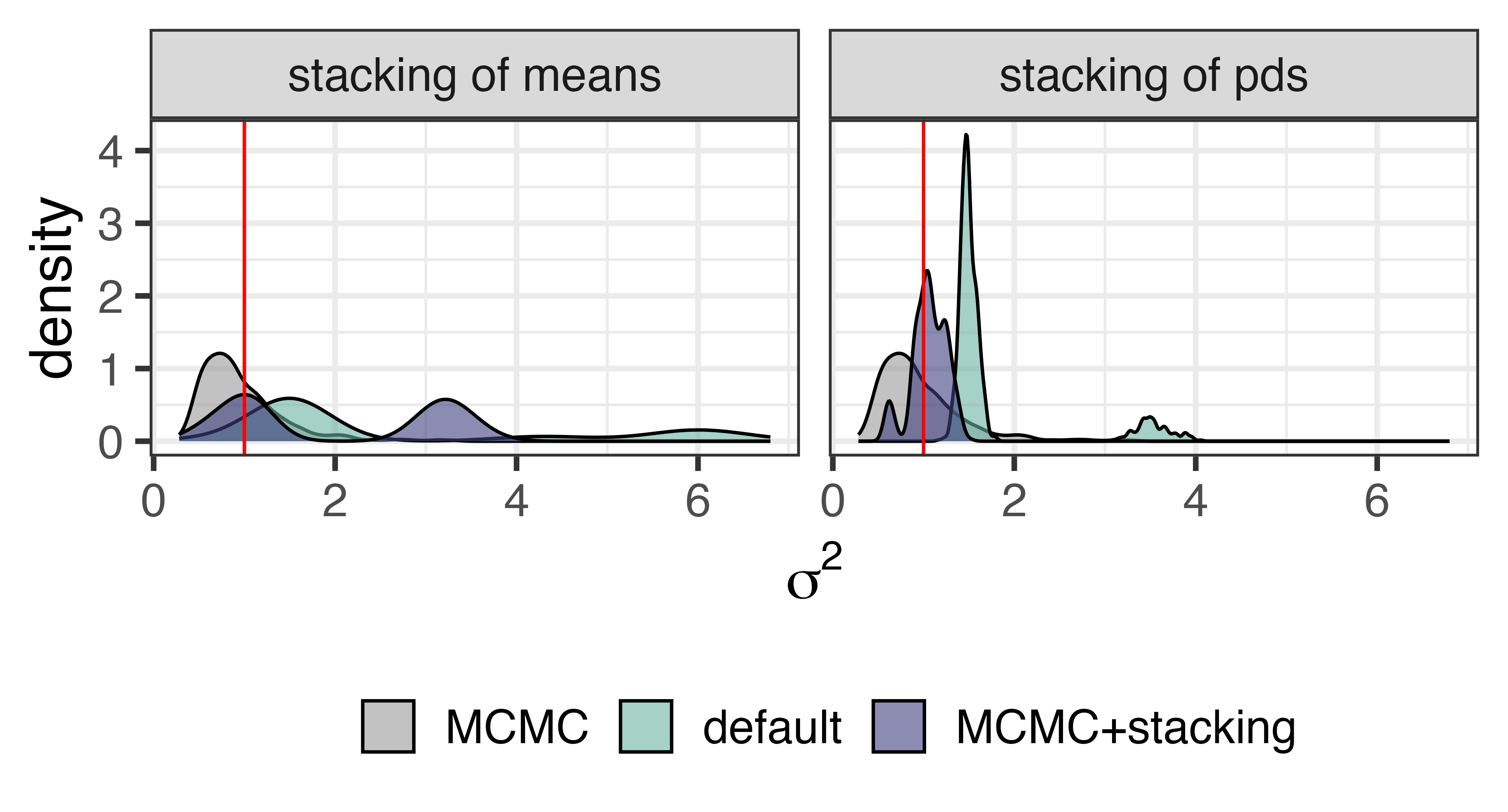}}
{\includegraphics[width=0.45\textwidth, keepaspectratio]{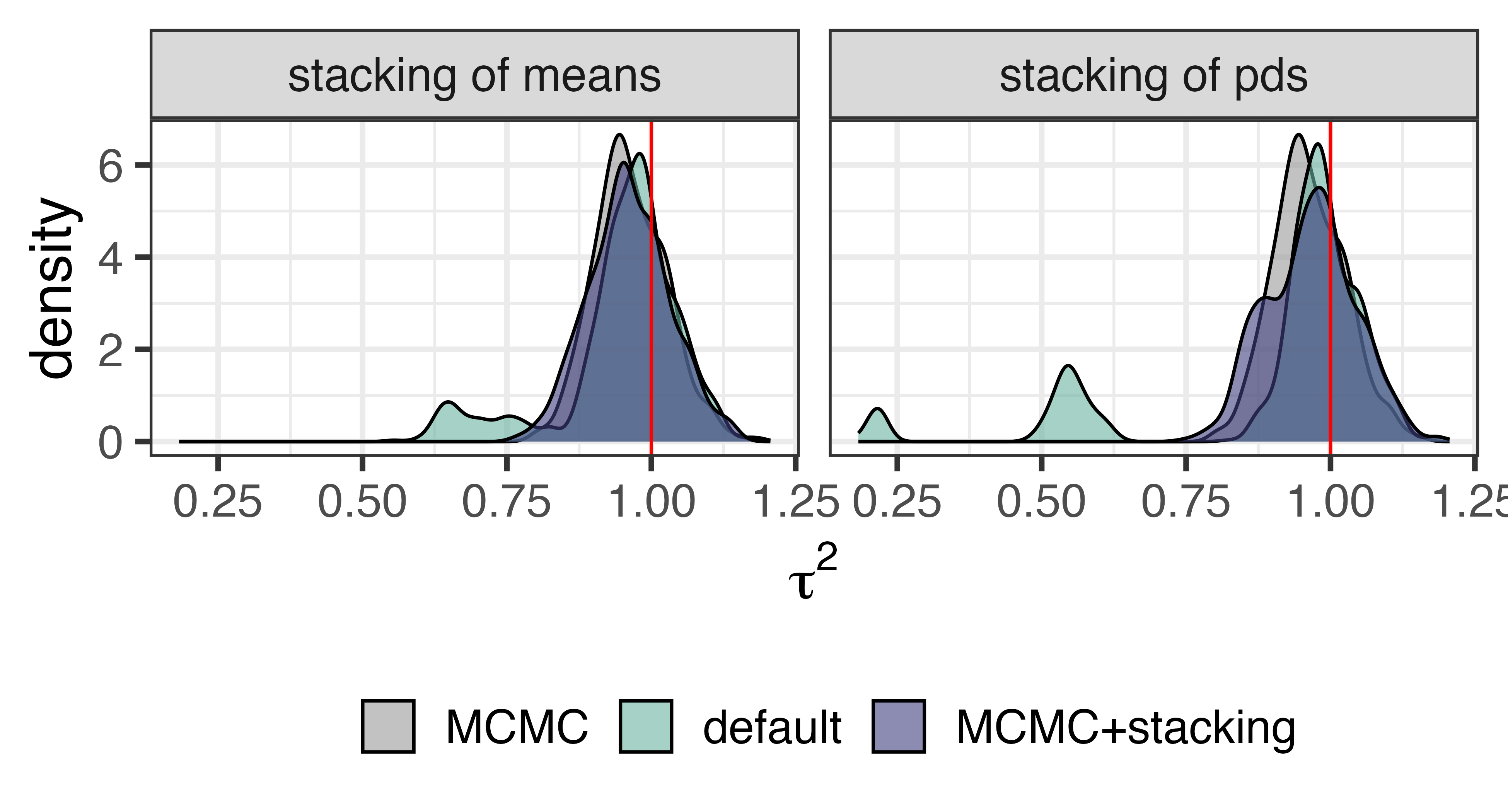}}\\
{\includegraphics[width=0.45\textwidth, keepaspectratio]{pics/sigmasq_prefix_compar2_r4.png}}
{\includegraphics[width=0.45\textwidth, keepaspectratio]{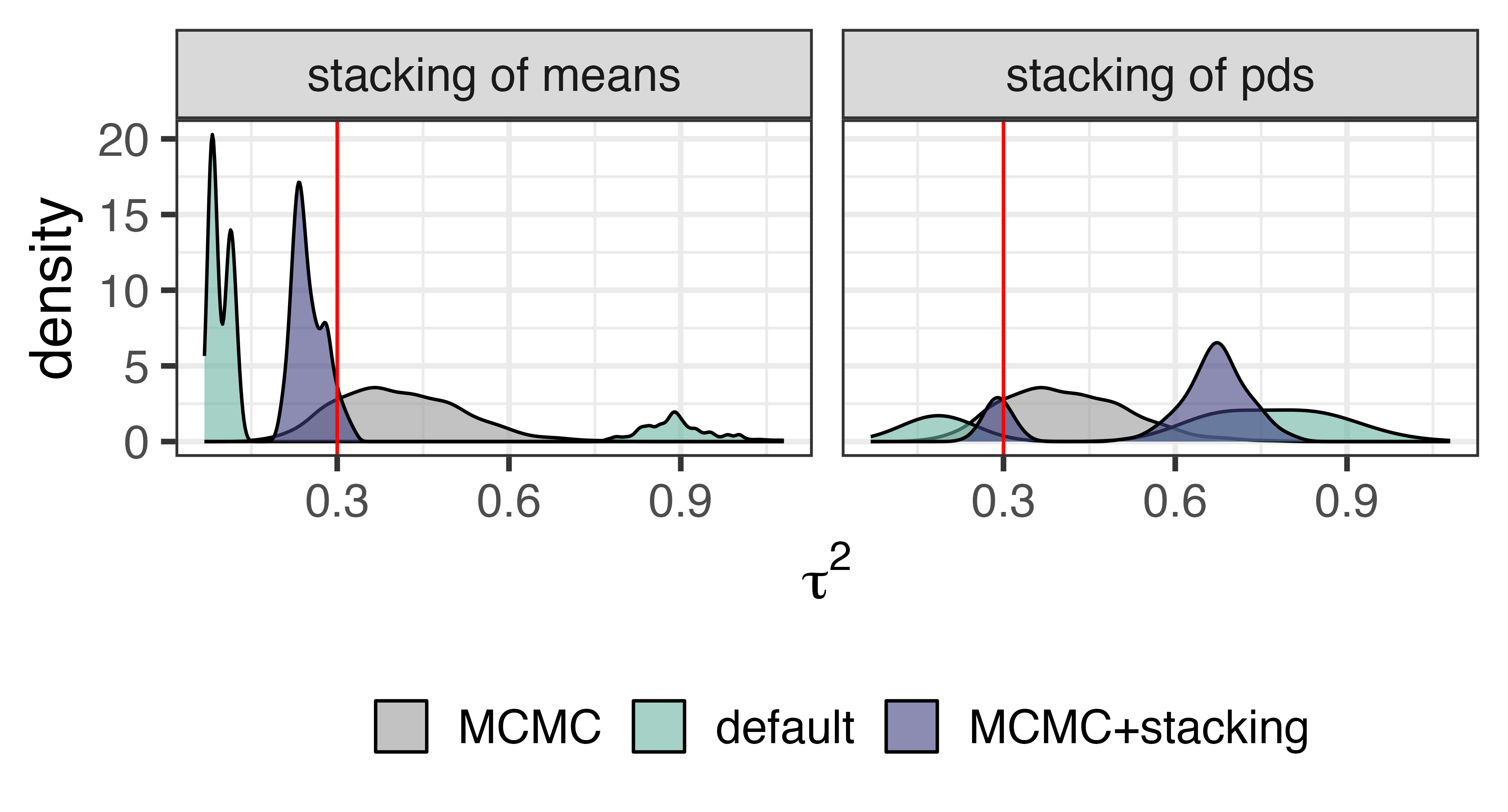}}\\
{\includegraphics[width=0.45\textwidth, keepaspectratio]{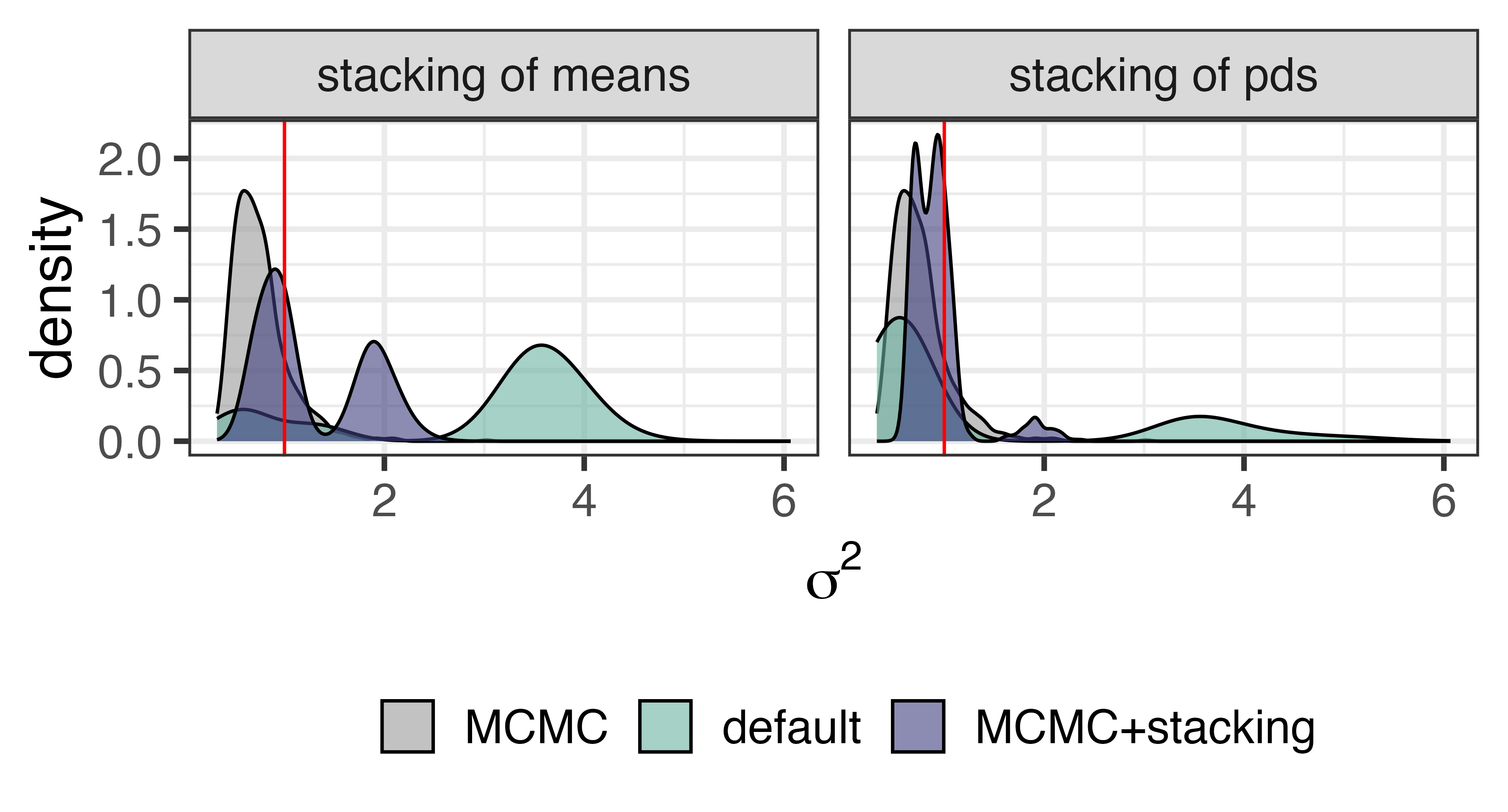}}
{\includegraphics[width=0.45\textwidth, keepaspectratio]{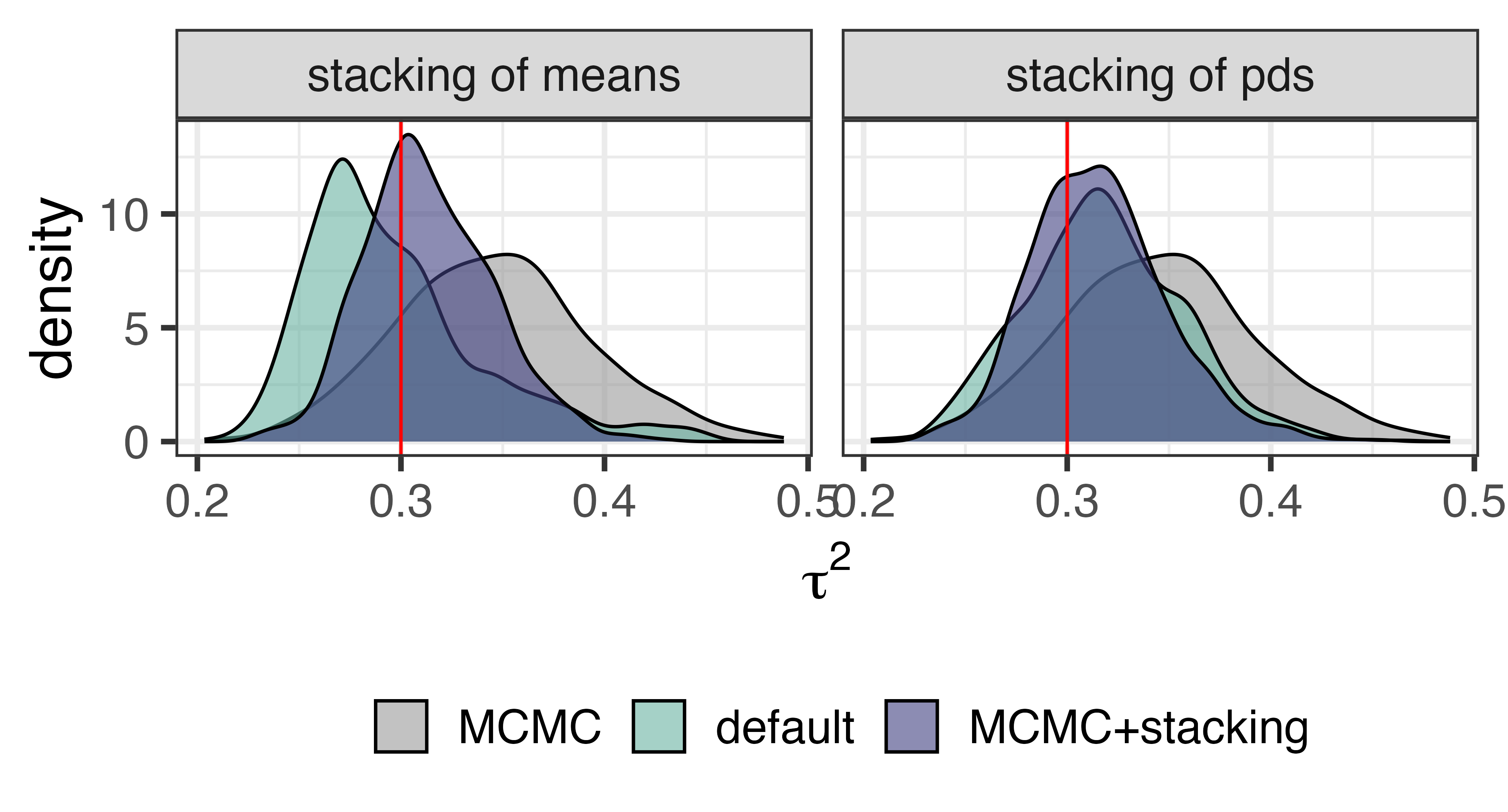}}
\caption{Densities of $\sigma^2$ (left column) and $\tau^2$ (right column) for the example with 800 observations from simulation 1 (top row), the example with 600 observations from simulation 2 (second row)
the example with 400 observations from simulation 3 (third row), and the example with 200 observations from simulation 4. Vertical red lines indicate the actual $\tau^2$ values. Grey densities represent MCMC-recovered posterior distributions of $\tau^2$. 'Default' and 'MCMC+Stacking' show stacking results using two methods for selecting ${\phi, \nu, \delta^2}$ candidates. Left panel: stacking of means. Right panel: stacking of predictive densities}
\label{fig: sim_prefix_tau_compar}
\end{figure}

\begin{figure}[t]
\centering
\includegraphics[width=0.45\textwidth, keepaspectratio]{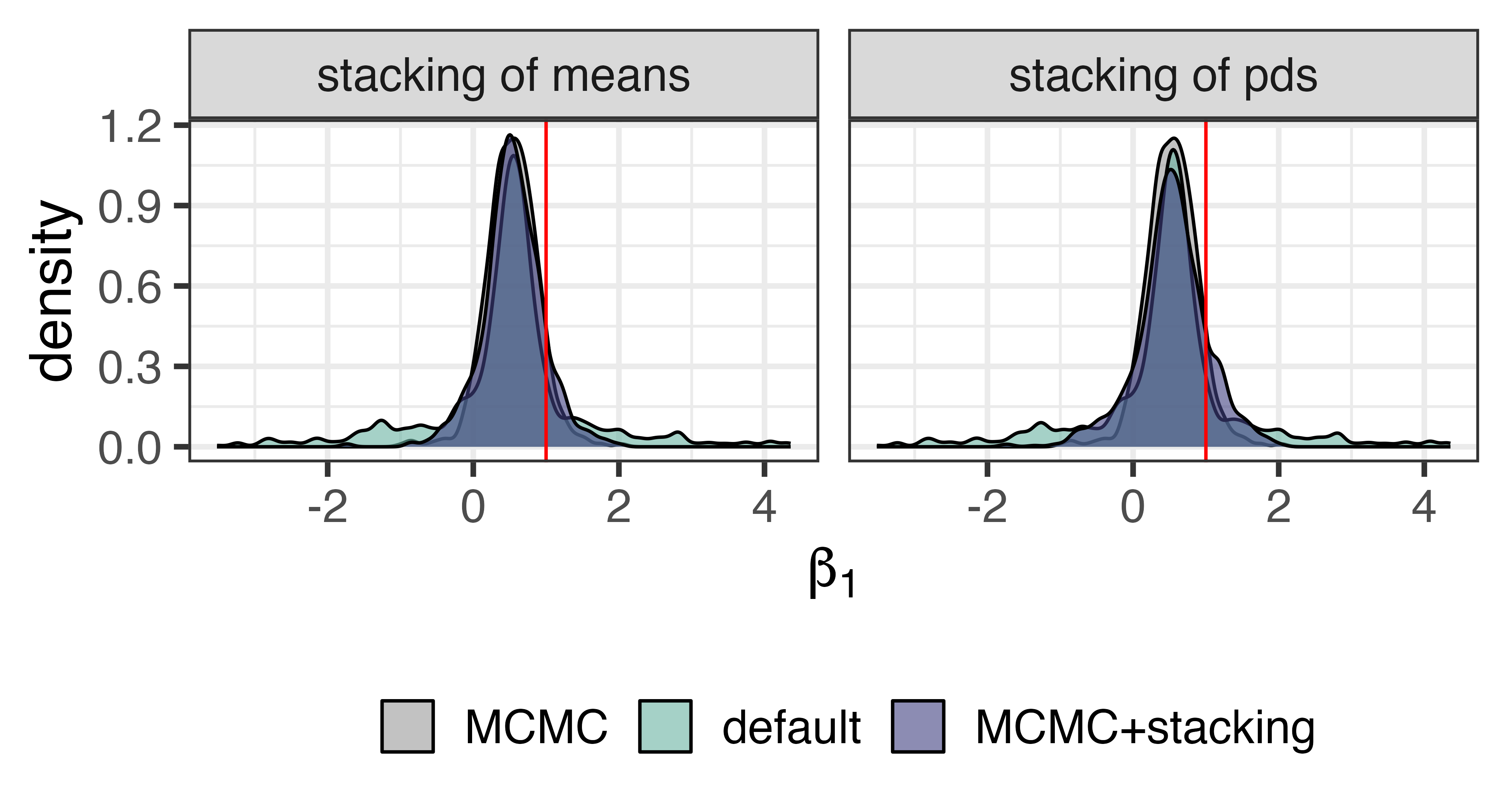} 
\includegraphics[width=0.45\textwidth, keepaspectratio]{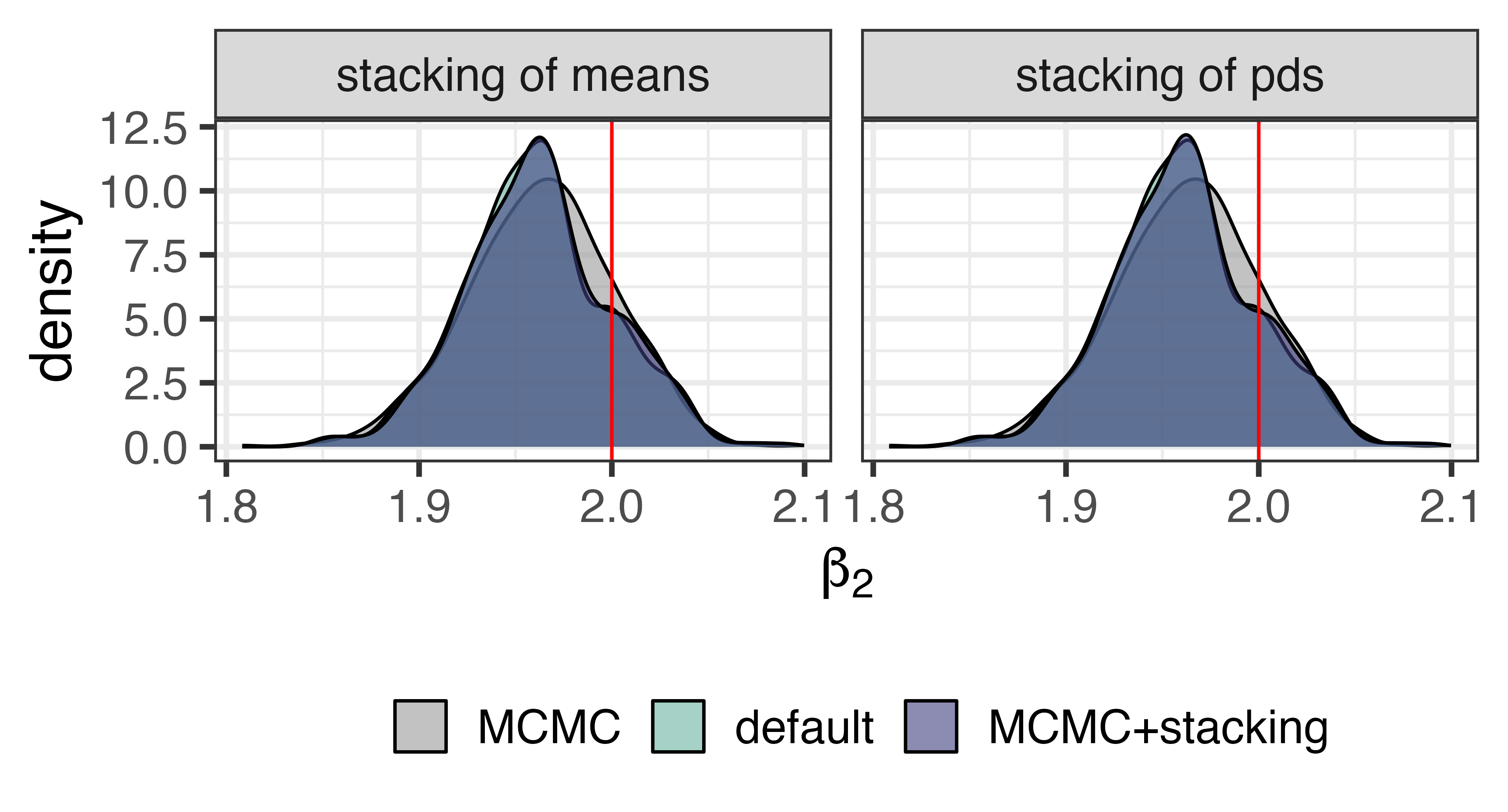}\\
\includegraphics[width=0.45\textwidth, keepaspectratio]{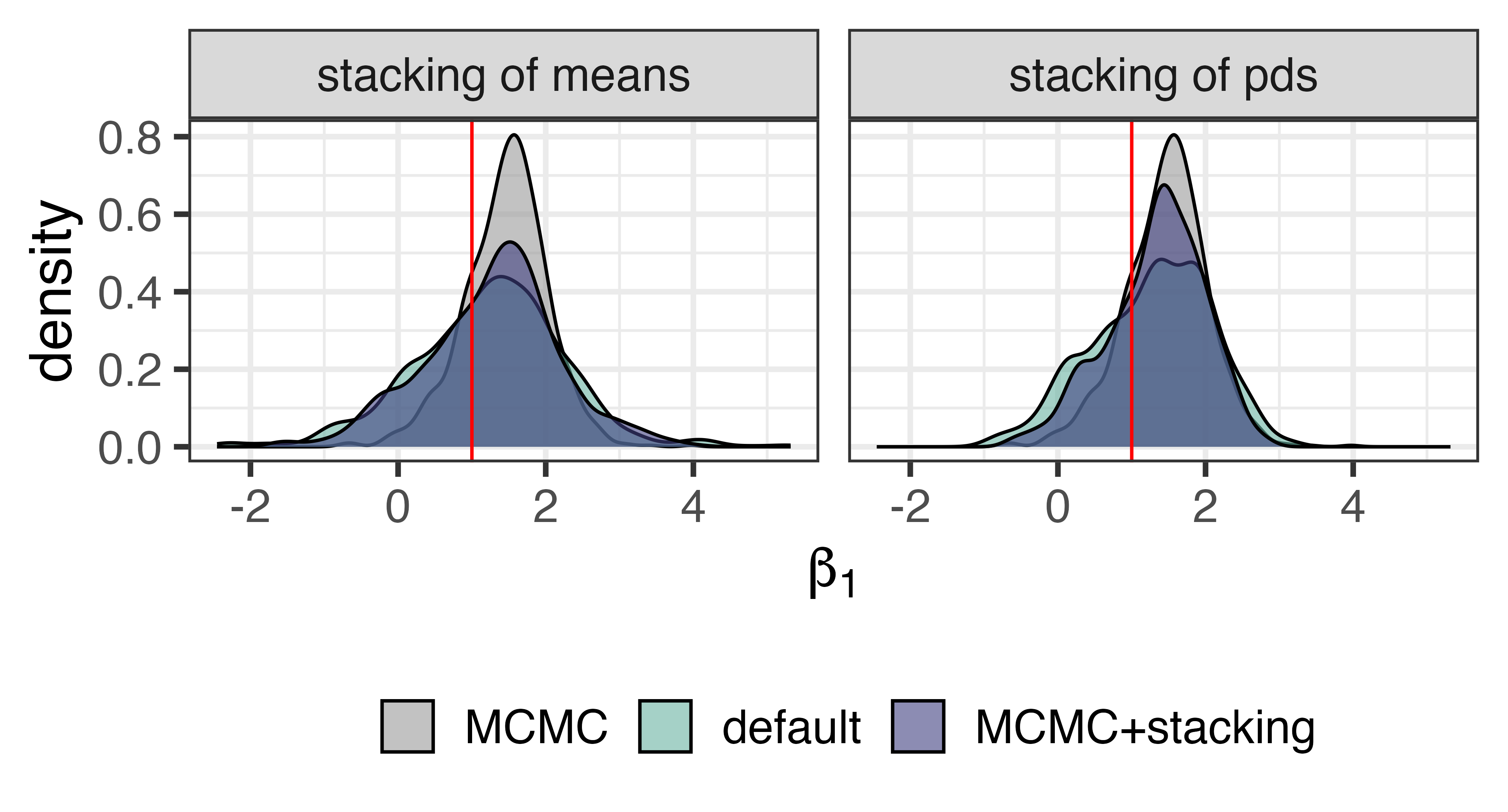}
\includegraphics[width=0.45\textwidth, keepaspectratio]{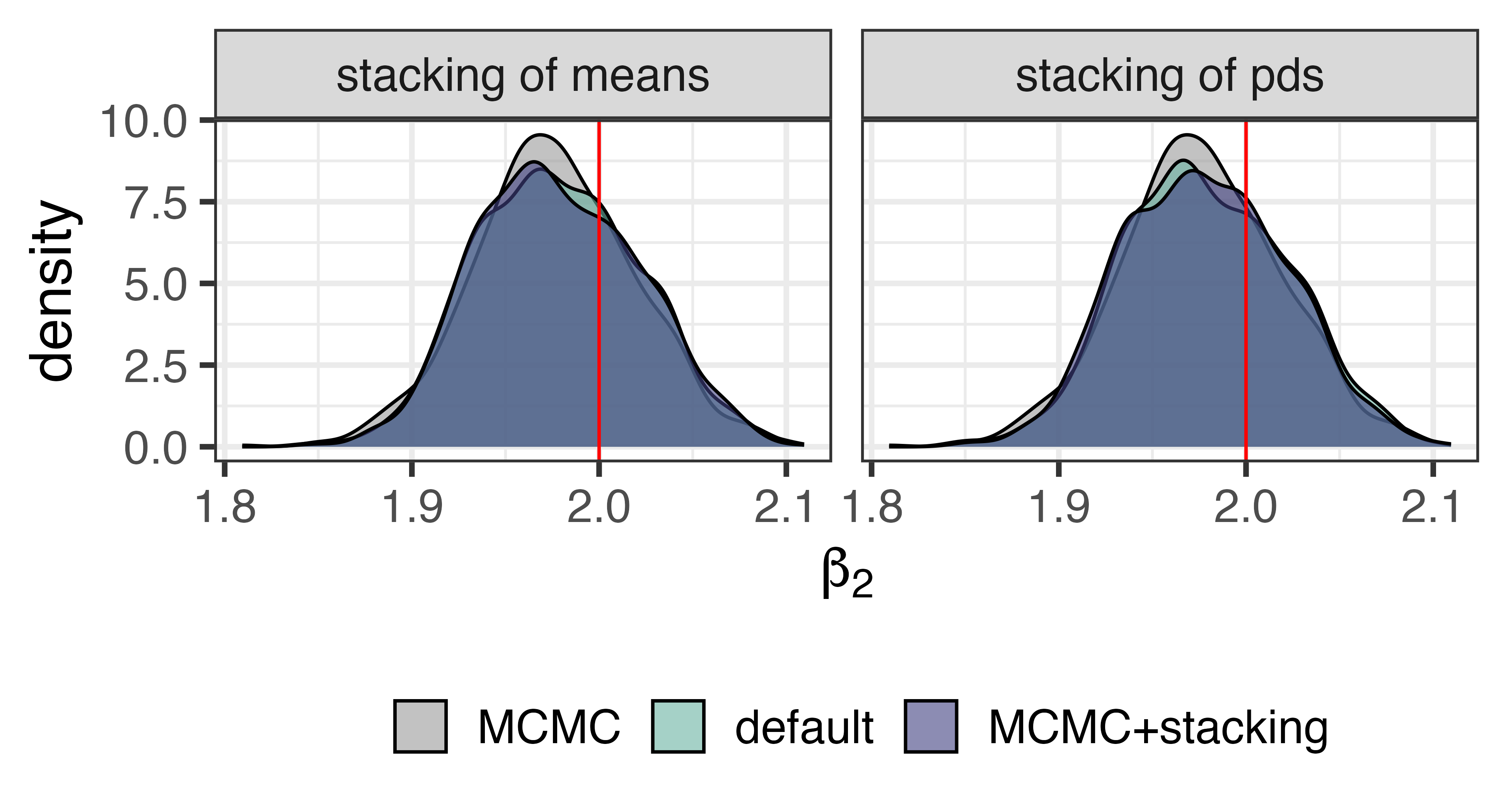}\\
\includegraphics[width=0.45\textwidth, keepaspectratio]{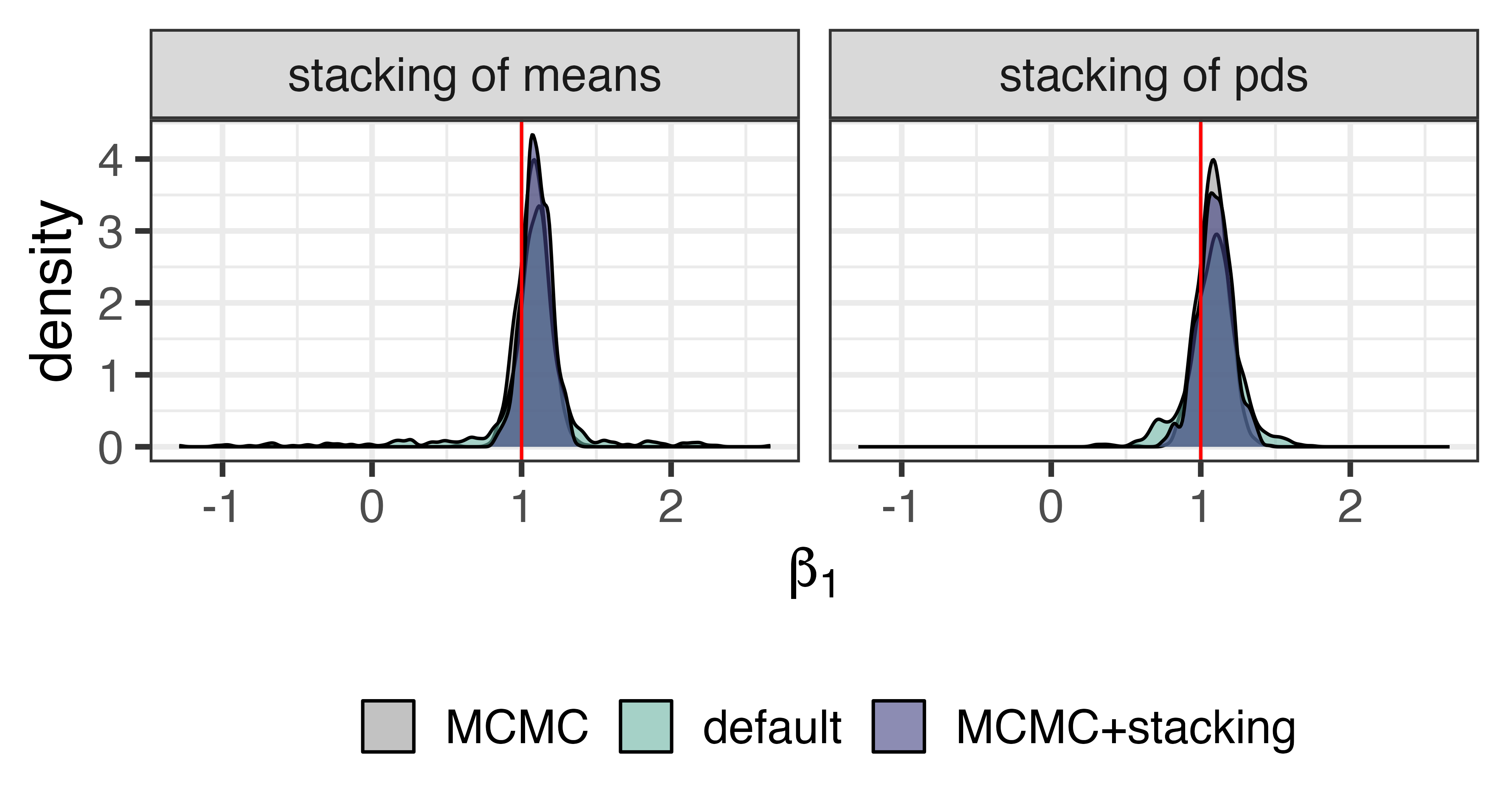}
\includegraphics[width=0.45\textwidth, keepaspectratio]{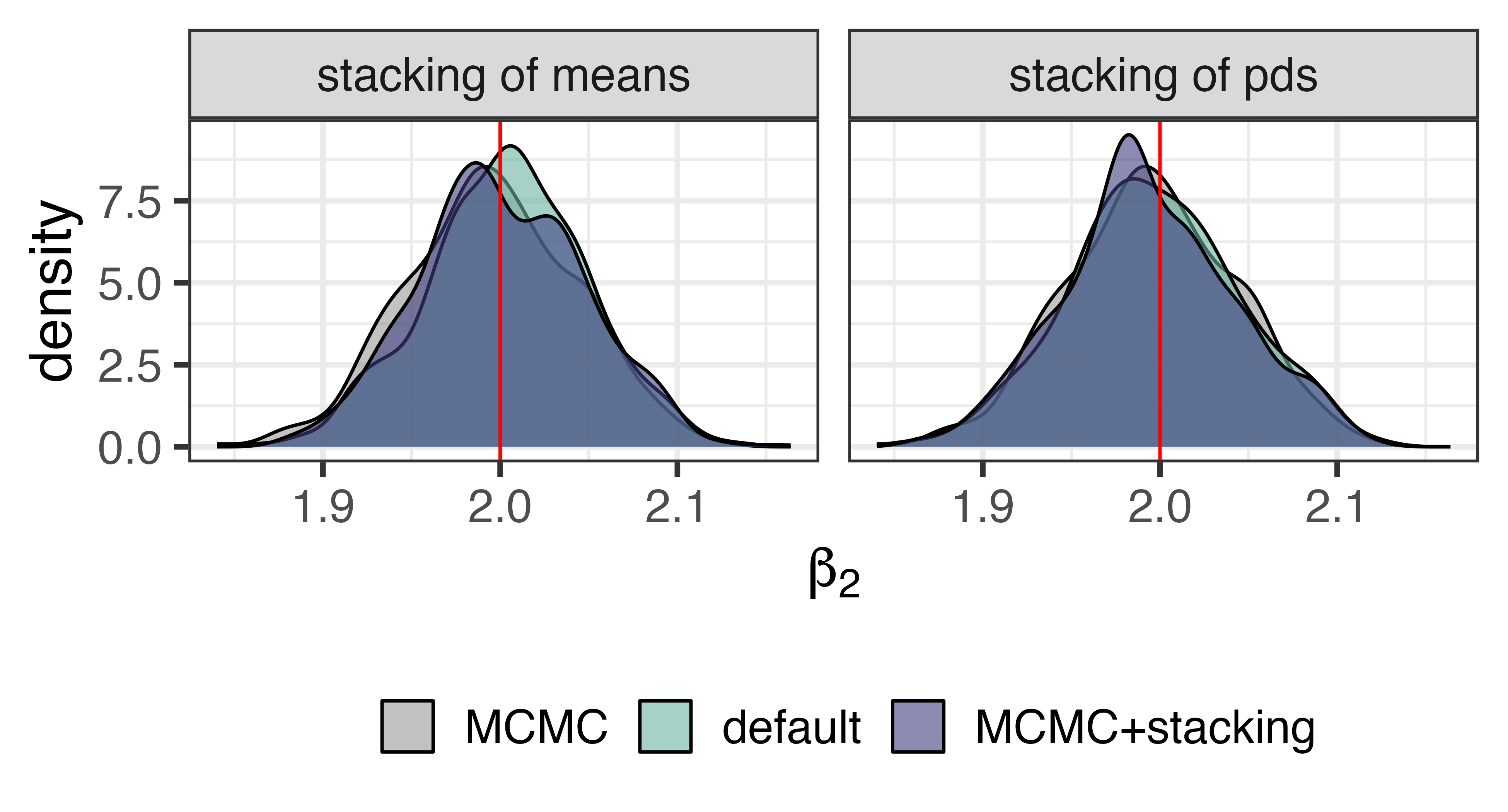}\\
\includegraphics[width=0.45\textwidth, keepaspectratio]{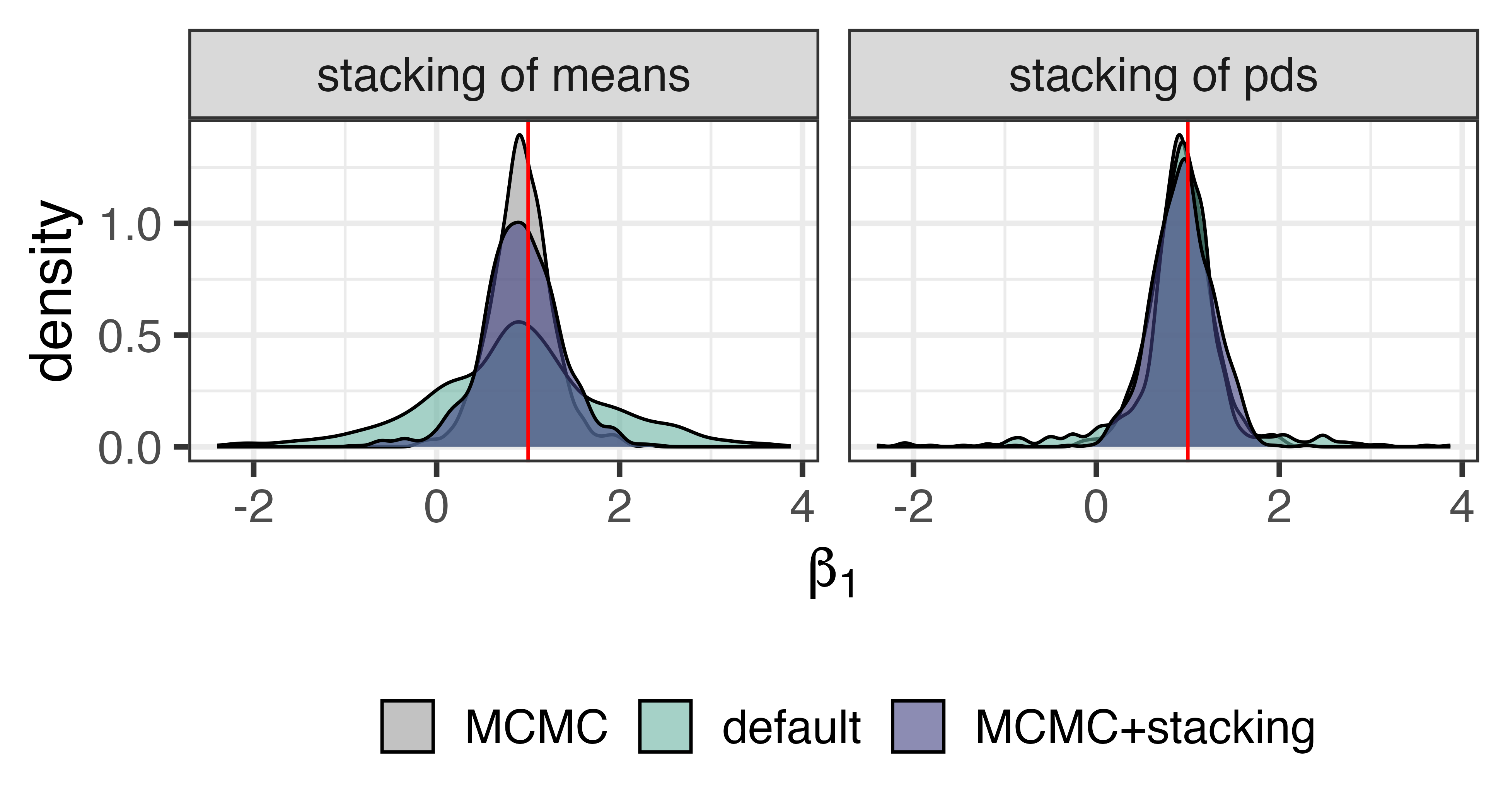}
\includegraphics[width=0.45\textwidth, keepaspectratio]{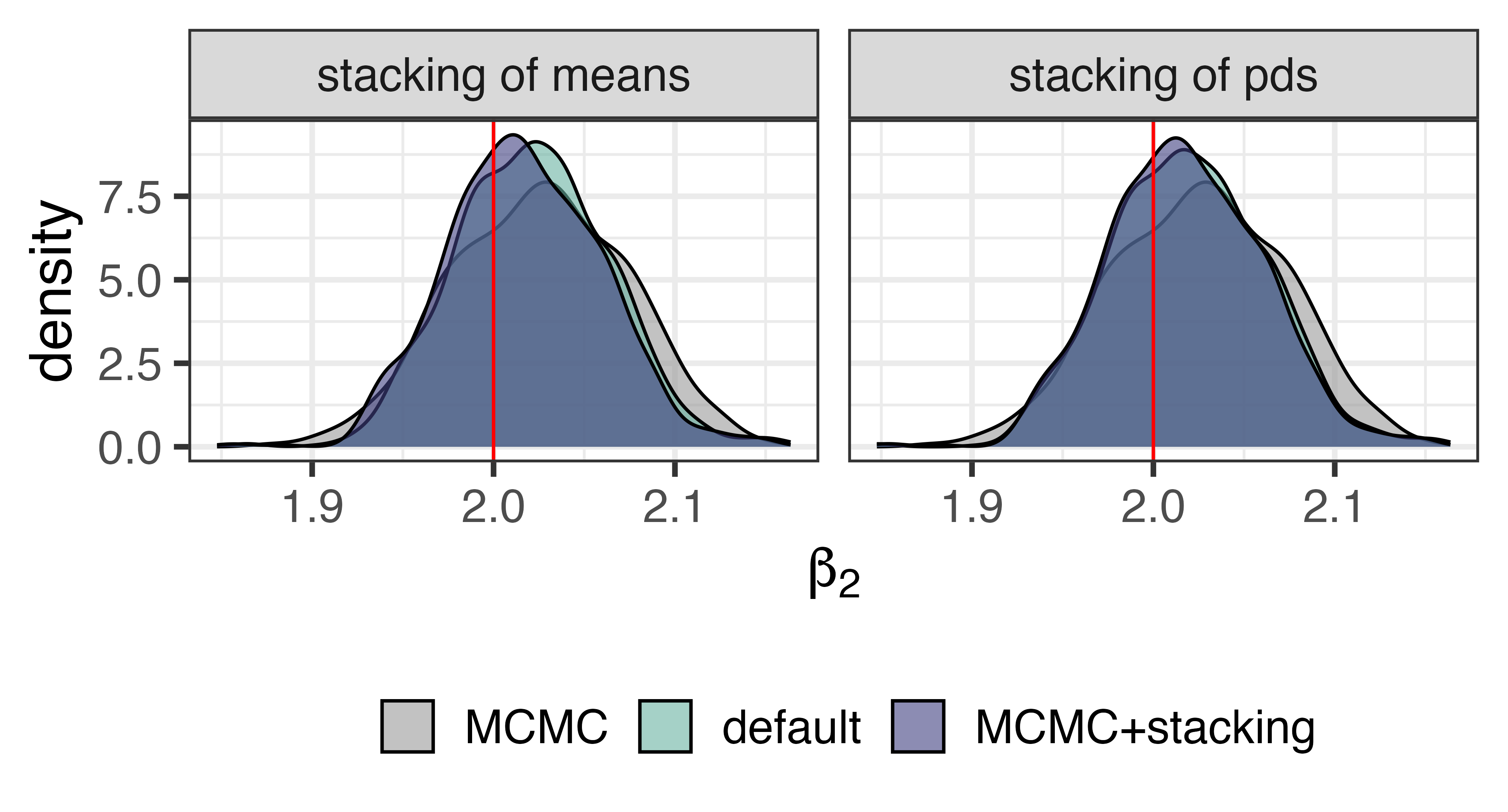}
\caption{Densities of $\beta_1$ (left column) and $\beta_2$ (right column) for the example with 800 observations from simulation 1 (top row), the example with 600 observations from simulation 2 (second row), the example with 400 observations from simulation 3 (third row), and the example with 200 observations from simulation 4 (bottom row). Vertical red lines indicate the actual values. Grey densities represent MCMC-recovered posterior distributions. 'Default' and 'MCMC+Stacking' show stacking results using two methods for selecting ${\phi, \nu, \delta^2}$ candidates. Left panel: stacking of means. Right panel: stacking of predictive densities}
\label{fig: sim_prefix_beta_compar}
\end{figure}
\captionsetup{labelfont={color=black},textfont={color=black}}

\begin{figure}[t]
\centering
{\includegraphics[width=0.45\textwidth, keepaspectratio]{pics/indi50_prefix_compar1ICI_r8.png}}
{\includegraphics[width=0.45\textwidth, keepaspectratio]{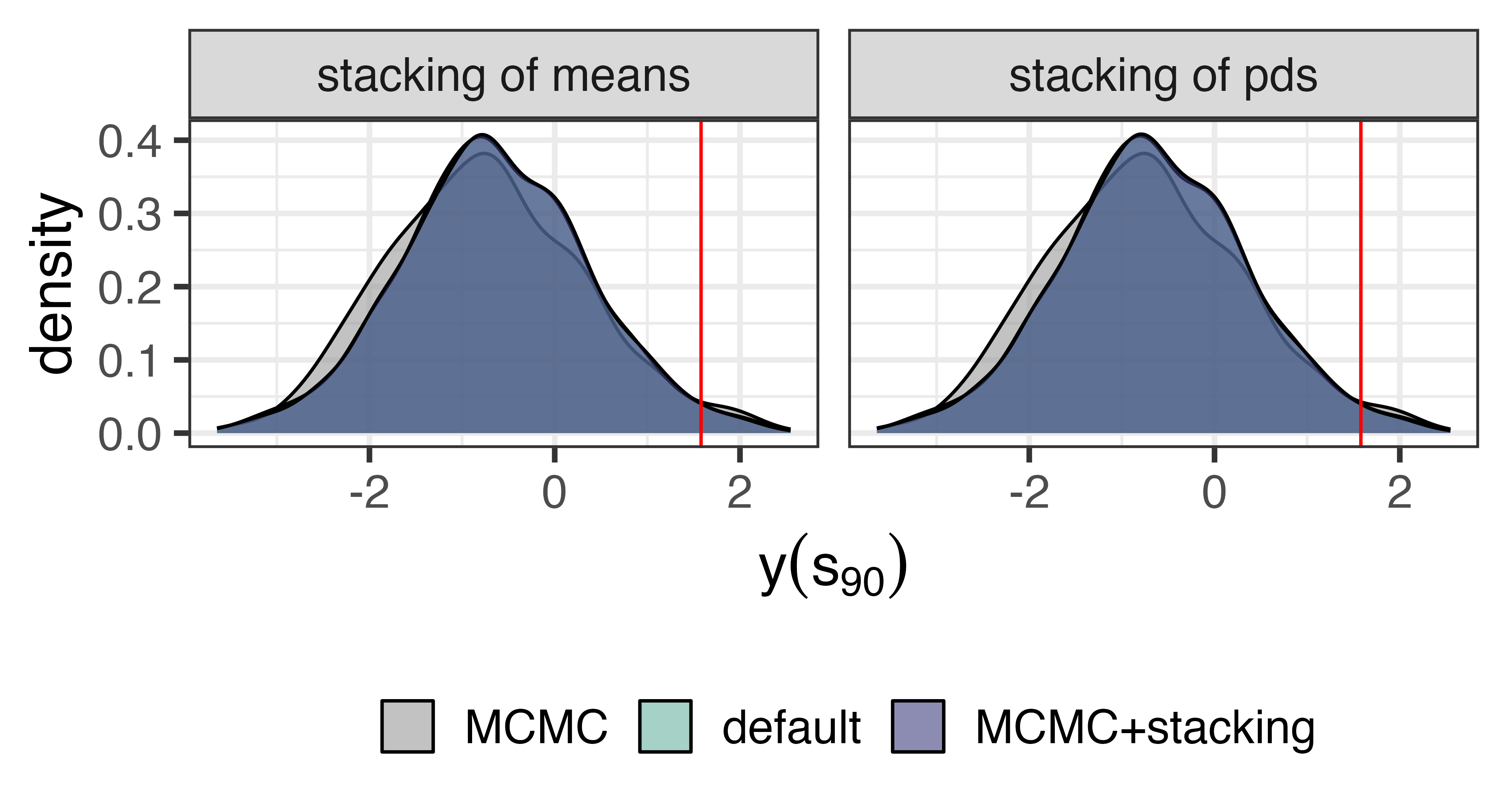}}\\
{\includegraphics[width=0.45\textwidth, keepaspectratio]{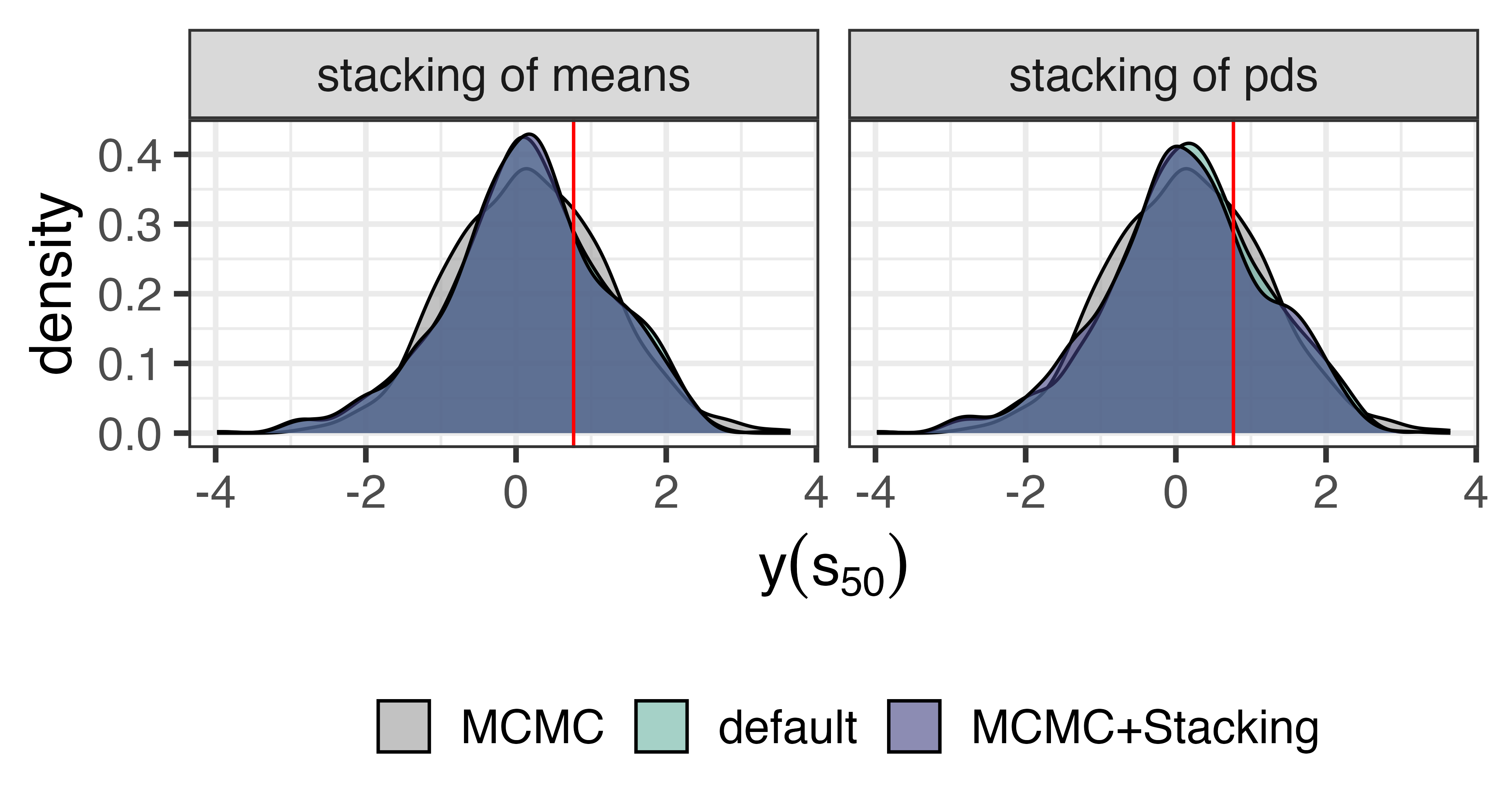}}
{\includegraphics[width=0.45\textwidth, keepaspectratio]{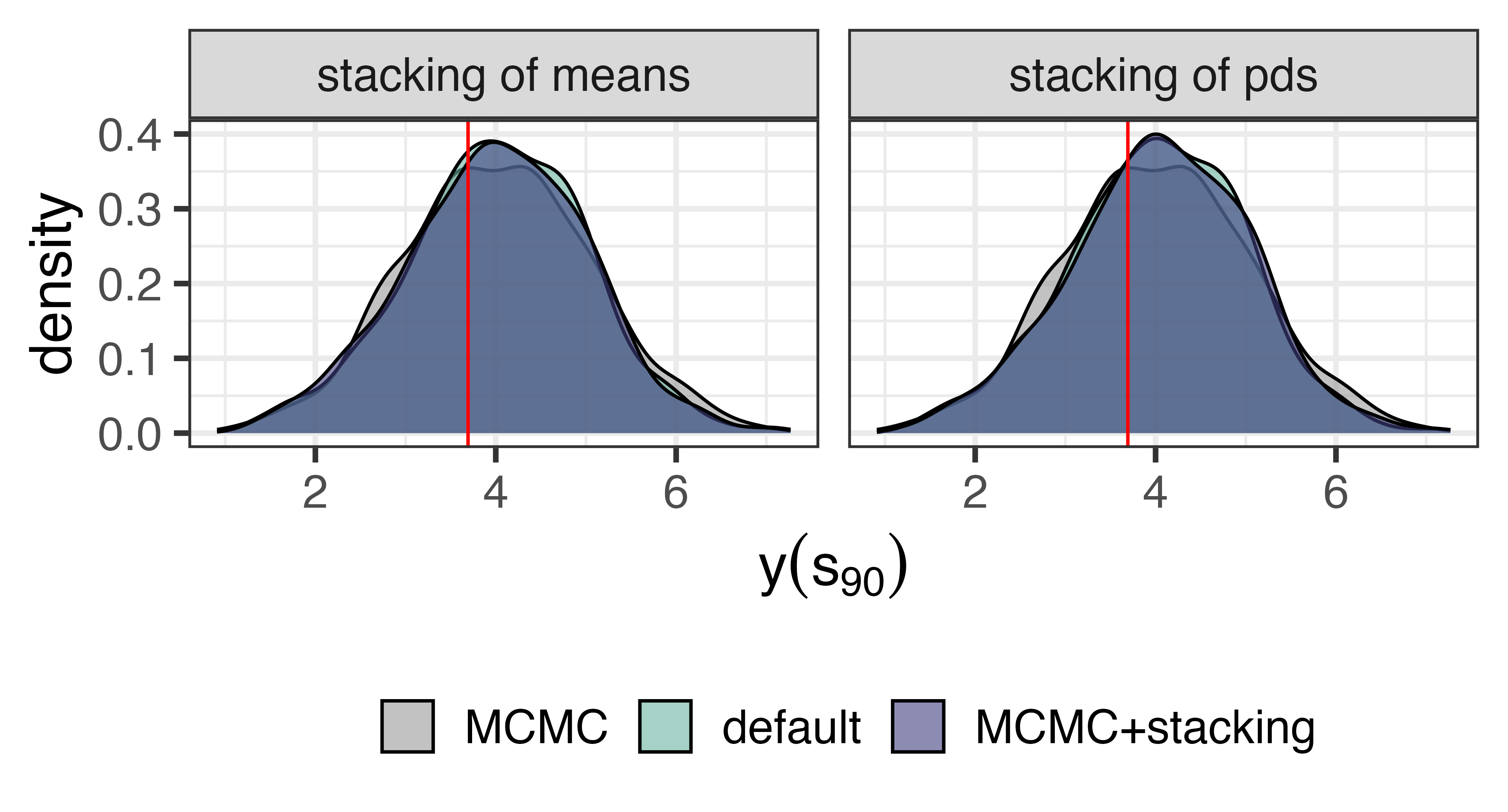}}\\
{\includegraphics[width=0.45\textwidth, keepaspectratio]{pics/indi50_prefix_compar2ICI_r4.png}}
{\includegraphics[width=0.45\textwidth, keepaspectratio]{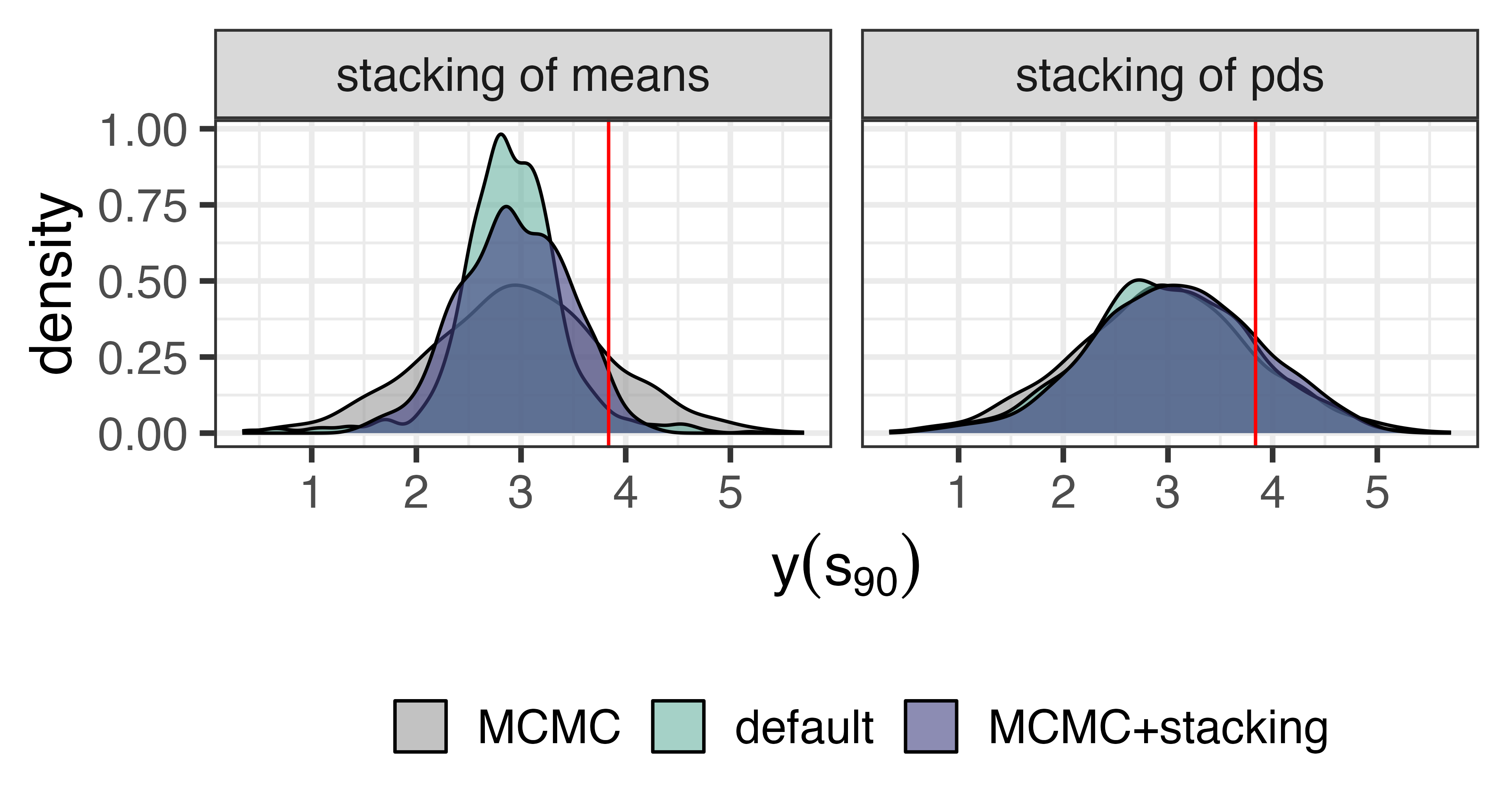}}\\
{\includegraphics[width=0.45\textwidth, keepaspectratio]{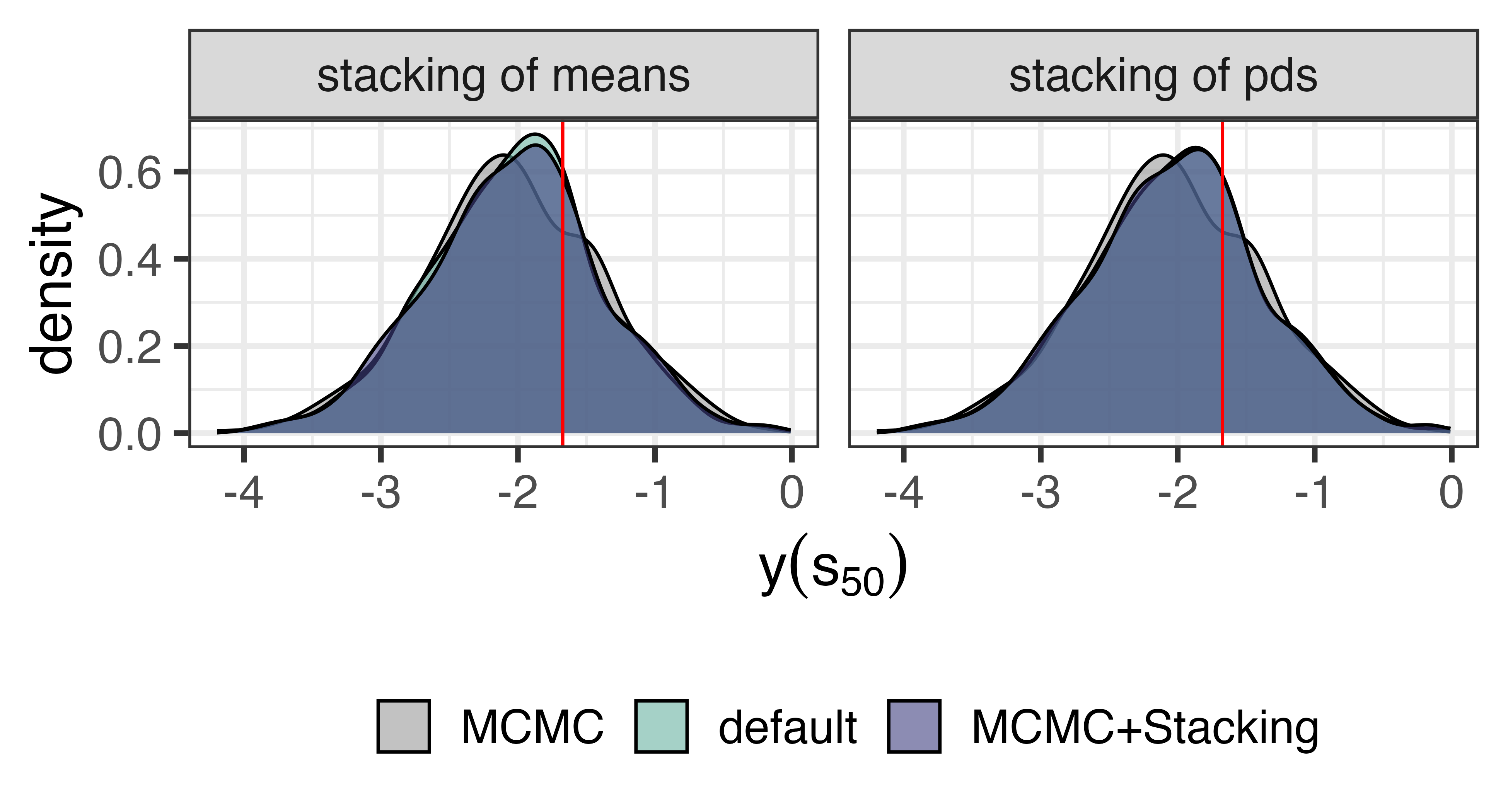}}
{\includegraphics[width=0.45\textwidth, keepaspectratio]{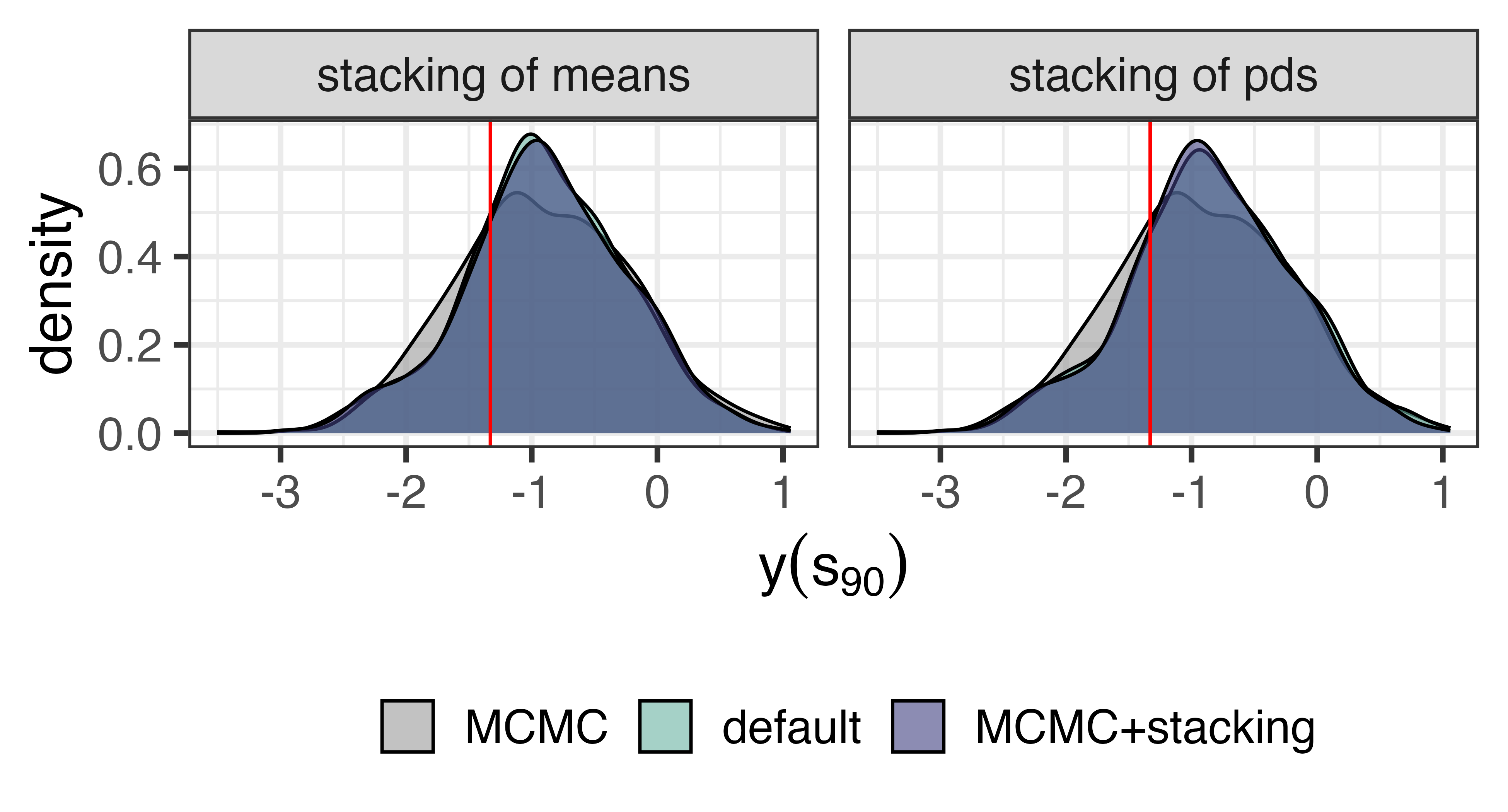}}
\caption{Predictive densities of the outcome at the $50$-th point (left column) and $90$-th point (right column) in the example with 800 observations from simulation 1 (top row), 
the example with 600 observations from simulation 2 (second row)
the example with 400 observations from simulation 3 (third row), and the example with 200 observations from simulation 4 (bottom row). Vertical red lines indicate the actual values. Grey densities represent MCMC-recovered posterior distributions. 'Default' and 'MCMC+Stacking' show stacking results using two methods for selecting ${\phi, \nu, \delta^2}$ candidates. Left panel: stacking of means. Right panel: stacking of predictive densities}
\label{fig: sim_prefix_50_90_compar}
\end{figure}

\section{Plots for AOD prediction analysis}
See Figures~\ref{fig: RDA_variogram}~and~\ref{fig: RDA_CI}.
\begin{figure}[t]
\centering
\includegraphics[width=0.5\textwidth, keepaspectratio]{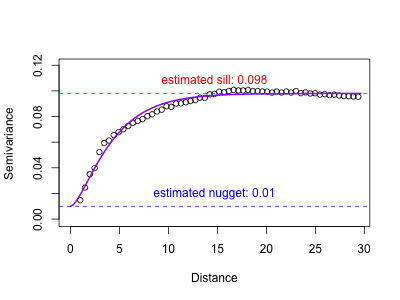} \caption{Semivariogram of the residuals from AOD linear regression model}
\label{fig: RDA_variogram}
\end{figure}

\begin{figure}[t]
\centering
\includegraphics[width=0.8\textwidth, keepaspectratio]{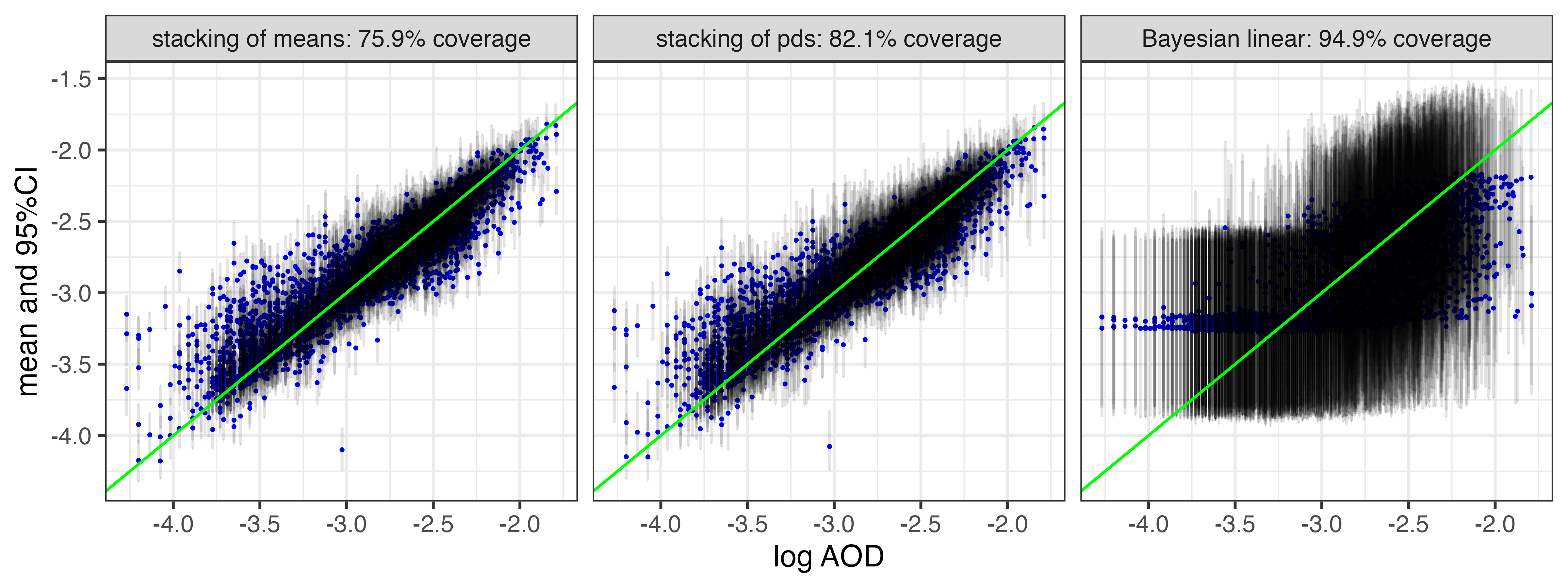} \caption{Scatterplots for log of interpolated and testing data AOD with 95\% credible intervals. Solid green line denotes the 45-degree line. Titles include 95\%CI coverage.}
\label{fig: RDA_CI}
\end{figure}

\clearpage

\end{document}